\documentclass[11pt]{report}

\usepackage{AMonographSections/arxiv}
\renewcommand{\headeright}{Preprint}
\renewcommand{\undertitle}{Preprint}
\renewcommand{\shorttitle}{Quadratic Forms in Gaussian Random Variables}

\usepackage[T1]{fontenc}
\usepackage[utf8]{inputenc}
\usepackage{lmodern}
\usepackage{amsmath,amsfonts,amsthm,mathrsfs,amssymb}
\usepackage{graphicx}
\usepackage{pdflscape}
\usepackage{makecell}
\usepackage{adjustbox}
\usepackage{algorithm,algpseudocode}
\usepackage{tikz}
\makeatletter
\input{WierdTexLibrary/TikZNetworkText}
\makeatother
\usepackage[hidelinks,colorlinks,unicode,linkcolor=blue,citecolor=blue,urlcolor=blue]{hyperref}
\usepackage{doi}
\usepackage{longtable}
\usepackage{multirow}
\usepackage{xparse}
\usepackage{colortbl}
\usepackage{booktabs}
\usepackage{array}
\usepackage{ragged2e}
\usepackage{ltablex}
\keepXColumns
\usepackage[flushleft]{threeparttable}
\usepackage{xcolor}
\usetikzlibrary{calc}
\usetikzlibrary{shapes.symbols}
\usepackage[caption=false,font=footnotesize,subrefformat=parens,labelformat=parens]{subfig}
\usepackage[subfigure]{tocloft}
\usepackage{tabularx}
\usepackage{tabulary}
\usepackage{cleveref}
\usepackage[numbers,sort&compress]{natbib}

\title{Quadratic Forms in Gaussian Random Variables\\\large Theoretical Results and Applications}
\author{%
Mohanad Ahmed\\
\href{mailto:mohanad@caltech.edu}{mohanad@caltech.edu}
\And
Mahmoud Ghazal\\
\href{mailto:mahmoud.ghazal@kaust.edu.sa}{mahmoud.ghazal@kaust.edu.sa}
\And
Maaz Mahadi\\
\href{mailto:maaz.mahadi@kaust.edu.sa}{maaz.mahadi@kaust.edu.sa}
\And
Tareq Y. Al-Naffouri\\
\href{mailto:tareq.alnaffouri@kaust.edu.sa}{tareq.alnaffouri@kaust.edu.sa}\\[0.75em]
Computer, Electrical and Mathematical Science and Engineering Division (CEMSE)\\
King Abdullah University of Science and Technology (KAUST)\\
Thuwal 23955, Saudi Arabia
}
\date{}

\ExplSyntaxOn
\NewDocumentCommand{\mref}{m}{\quinn_mref:n {#1}}
\seq_new:N \l_quinn_mref_seq
\cs_new:Npn \quinn_mref:n #1
 {
  \seq_set_split:Nnn \l_quinn_mref_seq { , } { #1 }
  \seq_pop_right:NN \l_quinn_mref_seq \l_tmpa_tl
  (
  \seq_map_inline:Nn \l_quinn_mref_seq
    { \ref{##1},\nobreakspace }
  \exp_args:NV \ref \l_tmpa_tl
  )
 }
\ExplSyntaxOff

\newcommand{\mhndcite}[2][]{\textcolor{blue}{\cite[#1]{#2}}, \citeauthor*{#2}, \citeyear{#2}}

\newif\ifWeWantMhnd
\newif\ifWeWantMhndO
\WeWantMhndtrue
\ifWeWantMhnd
    \newcommand{\citeF}[2][]{\mhndcite[#1]{#2}}
\else
\ifWeWantMhndO
    \newcommand{\citeF}[2][]{\mhndciteO[#1]{#2}}
\else
    \newcommand{\citeF}[2][]{\cite[#1]{#2}}
\fi
\fi

\DeclareMathOperator{\diag}{diag}
\DeclareMathOperator{\tr}{tr}
\DeclareMathOperator{\etr}{etr}
\DeclareMathOperator{\erf}{erf}
\DeclareMathOperator{\rank}{rank}
\DeclareMathOperator{\Var}{Var}
\DeclareMathOperator{\sign}{sign}
\newcommand\Rank[1]{\mathscr{R}_\mathbf{#1}}

\newif\ifWeWantUglyEqns
\WeWantUglyEqnstrue
\ifWeWantUglyEqns
    \newcommand{\UglyEquation}[1]{\begin{equation}#1\end{equation}}
    \newcommand{\UglyAlign}[1]{\begin{align}#1\end{align}}
    \newcommand{\UglyAlignS}[1]{\begin{align*}#1\end{align*}}
    \newcommand{\UglyTabular}[1]{\begin{tabular}#1\end{tabular}}
    \newcommand{\UglyInline}[1]{\(#1\)}
\else
    \newcommand{\UglyEquation}[1]{\begin{equation}\text{There was an ugly equation here!}\end{equation}}
    \newcommand{\UglyAlign}[1]{\begin{align}\text{There was an ugly align here!}\end{align}}
    \newcommand{\UglyAlignS}[1]{\begin{align*}\text{There was an ugly align* here!}\end{align*}}
    \newcommand{\UglyTabular}[1]{\\ \text{There was an ugly tabular here!}\\}
    \newcommand{\UglyInline}[1]{\textcolor{red}{There was ugly inline math here!}}
\fi

\newcolumntype{Y}{>{\RaggedRight\arraybackslash}X}
\newcolumntype{L}[1]{>{\raggedright\let\newline\\\arraybackslash\hspace{0pt}}m{#1}}
\newcolumntype{C}[1]{>{\centering\let\newline\\\arraybackslash\hspace{0pt}}m{#1}}
\newcolumntype{R}[1]{>{\raggedleft\let\newline\\\arraybackslash\hspace{0pt}}m{#1}}
\newcommand{\rankFinal}{{\tilde{N}}}
\newcommand{\rankSigma}{{\Rank{\dmt{\Sigma}}}}
\newcommand{\dv}[1]{{\boldsymbol{#1}}}
\newcommand{\dmt}[1]{\MakeUppercase{\boldsymbol{#1}}}
\newcommand{\rs}[1]{\mathsf{#1}}
\newcommand{\rv}[1]{\boldsymbol{{\mathsf{#1}}}}
\newcommand{\rmt}[1]{\boldsymbol{\MakeUppercase{\mathsf{#1}}}}

\newcommand{\expect}[1]{\mathbb{E}\left[ #1 \right]}
\newcommand{\trace}[1]{\operatorname{Tr}\left[{#1}\right]}
\newcommand{\real}[1]{\mathcal{R}\left({#1}\right)}

\newcommand{\reals}[1]{\mathbb{R}^{#1}}
\newcommand{\complexs}[1]{\mathbb{C}^{#1}}

\newtheorem{theorem}{Theorem}[chapter]
\newtheorem{proposition}[theorem]{Proposition}

\newtheorem{corollary}[theorem]{Corollary}

\theoremstyle{definition}

\theoremstyle{remark}
\newtheorem{remark}[theorem]{Remark}
\newtheorem{openProblem}[theorem]{Open Problem}

\begin{document}

\maketitle

\begin{abstract}
This manuscript reviews theoretical results and applications related to quadratic forms in Gaussian random variables. It summarizes definitions, canonical representations, exact and approximate distributional results, numerical inversion methods, applications, and selected open problems for real and complex quadratic forms, multiforms, and ratios of quadratic forms.
\end{abstract}

\tableofcontents

\chapter{Introduction}

\section{Introduction}
     Quadratic forms in Gaussian Random Variables (QFGRVs) are frequently encountered in the signal processing literature, either through proposing pure theoretical results, e.g., \cite{ramirez-espinosaNewApproachStatistical2019, wiegandSeriesApproximationsRayleigh2019, cuiExactDistributionProduct2016, nadimiRatioIndependentComplex2018}, or analyzing the performance of different metrics across diverse signal processing applications. For example, transient analysis in adaptive filters \cite{naffouriTransientAnalysisOfData-normalizedAdaptiveFilters2003}, directional of arrival estimation \cite{athleyThresholdRegionPerformance2005}, periodogram, \cite{gismallaPeriodogramEnergyDetection2011,gismallaPerformanceEnergyDetectionBartlett2012}, differential modulation for wireless systems \cite{zhaoDifferentialModulation2007}, and more \cite{jinAsymptoticSEROutageMRC2006, gerkmannStatisticsSprctralAmplitudes2009,havaryOptimalDistributedBeamforming2009,raghavanFalseAlarm2018,bessonAnalysisSNRLossDistribution2020}. QFGRVs also find applications in other fields where statistical signal processing is applied, such as communication  \cite{khammassiNewAnalyticalApproximation2022,beaulieuNovelSimpleRepresentations2011,kimSINRRandomBeamforming2013}, genetics \cite{xuGenomeWideAssociationStudies2022,jiangGeneralizedLinearMixed2021,zhouEfficientlyControllingGenetic2018,wangMethodologyAssociationStudies2021,chenSequenceKernelAssociation2013,crawfordDetectingEpistasisGenetic2017,limBATComboPowerfulTest2023}, reliability, performance engineering and vibration engineering \cite{maoModelQuantifyingUncertainty2012, yanUnifiedSchemeSolving2020}.     
     
    Consider a vector of random quadratic forms, $ \boldsymbol{\mathsf{Q}} = [\mathsf{Q}_1, \mathsf{Q}_2, \ldots, \mathsf{Q}_M]^T$, where the $m$th ($m =1,2,\ldots M$) quadratic form \footnote{ \color {black} Strictly speaking, this is a quadratic polynomial! A quadratic form should have $\boldsymbol{b}_m = \boldsymbol{0}$ and $c_m = 0$. We use the term quadratic form to include both and leave the distinction to  context.} is expressed as follows: 
    \begin{equation}
        \label{eqn:QF}
        \mathsf{Q}_m =\boldsymbol{\mathsf{x}}^T\boldsymbol{A}_m\boldsymbol{\mathsf{x}} + \boldsymbol{b}_m^T\boldsymbol{\mathsf{x}} + c_m. 
    \end{equation}
     In this formula, $\boldsymbol{\mathsf{x}}$ is a Gaussian random vector of dimension $N$ with mean $\dv{\mu}$ and covariance $\dmt{\Sigma}$, $\boldsymbol{A}_m$ is a symmetric matrix,  $\boldsymbol{b}_m$ is a vector of linear coefficients and  $c_m$ is an additive constant. 
    From this generic formula, we can have a rough classification of quadratic forms. 
    If $M=1$, we deal with  \emph{single} quadratic forms; otherwise, we deal with \emph{quadratic multiform}. If $\boldsymbol{\mathsf{x}}$ is complex we call \eqref{eqn:QF} \emph{complex} quadratic forms \footnote{In this case, we should replace the transpose `$T$' by Hermitian `$H$'. }. An extended classification of quadratic forms is given later in this monograph. 

    Researchers have a particular interest in various quantities related to quadratic forms, including probability density (PDF), cumulative distribution function (CDF), moments, and moment-generating function (MGF), cumulants, cumulant generating function, and charactersitc function. This work mainly surveys the results not included in Mathai and Provost \cite{mathaiQuadraticFormsRandom1992}. There are at least 60 references that have proposed new results and more applying quadratic forms results. New results pertaining to these quantities are scattered throughout the literature (sometimes from very different fields/venues) and would not appear applicable to the same problems without a deeper inspection and a considerable effort in collecting them. The contributions of this monograph are
\begin{enumerate}
    \item \textbf{Unified Presentation and Compilation of Diverse Results:} We present the results relating to QFGRVs using unified notation. We also reduce all QFGRVs collected from the literature to a common set of parameters that can be used to judge whether a particular method/result about a QFGRV is applicable to another QFGRV. The results and computational methods are collected from literature in various venues in signal processing, communications, statistics, graphics, and various other fields.  
    \item \textbf{Numerical Implementation of Select Methods:} In single QFGRVs there are several numerical libraries that allow computation. In the quadratic multiform case, this is not so. We select some of the most general methods (that have no publicly available implementations) and provide numerical implementations. In addition, we provide analysis of the computational properties and performance of these algorithms/implementations. We note here, particularly the works of Royen which provide quite general results but lack implementation.
    \item \textbf{Applications:} We provide concise descriptions and derivations of QFGRVs from various fields to demonstrate how they arise and the computational requirements that arise from these diverse applications. We also believe this will allow the reader to appreciate the broad applicability of QFGRVs and the importance of computational tools in this area.    
    \item \textbf{Open Problems:} we highlight several of the open problems related to statistical quantities of QFGRVs, the algorithms for computing these statistical quantities and why we believe these open problems are important in light of the applications and popularity of some of the methods.
    \item \textbf{Appendices / Supplementary Materials:} We also provide appendices/supplementary document containing 
    \begin{enumerate}
        \item \textit{Handbook style formulae collection} collected from the literature in unified and harmonized notation to enable the reader to quickly locate methods/results that will apply to their particular problem.
        \item \textit{Pseudocodes and/or Implementations} we provide implementations for the most general methods. For the ones we did not implement we provide pseudocode and computational complexity analysis that would allow the reader to pursue their own implementation if they so desire.
        \item \textit{Original authors notation} to allow easy comparison of an extracted result against the original authors result we provide the method/formula in their original notation and indicate any algebraic manipulations or changes we have made to make the result more suitable for the computation of QFGRVs.
    \end{enumerate}
\end{enumerate}

    Also, we support our monograph with many visual aids. The inclusion of classification diagrams and summary tables aids in organizing and synthesizing complex information, facilitating comprehension. Some of the summary lists/tables/figures produced are: 
    \begin{enumerate} 
        \item \textit{Classification Diagrams}: We identify the main categories along which quadratic forms appear in the literature., e.g., Fig.~ 
        \item \textit{Formulae Type Table}: We identified six categories: (a) Finite expressions, (b) Infinite Series, (c) Numerical Integration, (d) Sequence of Approx. RV, (e) Moment Matching and (f) Saddle Point Approximations. These types can be seen in Table 
        \item \textit{Citation Count}: We can ask ``which methods are more frequently used?''. An indirect measure of this is the citation count.
        \item \textit{M vs N Graph}: We plot the number of QFRGV (M) vs the number of Gaussian variables (N) Fig. 2 as they appear in the literature.
        \item \textit{Subsumption graphs}: Not all treatments of QFGRV provide the same generality. Some methods are more general than others. Organizing the literature along lines of which formula/method applies to which QFGRV produces Figs. 
        \item \textit{Applications Summary Table}: We collect references to a wide variety of applications of QFGRVs and organize these by field and venue Table and Table. 
    \end{enumerate}

    \section{Definitions}
        This section introduces the different forms of QFGRVs studied in this monograph. The definition of a single quadratic form is the foundation for subsequent discussions on multiforms and ratios of quadratic forms. The formulation of quadratic forms can vary depending on whether the random variables are real or complex. Therefore, there are six different quadratic forms defined as follows: 

            \subsection{(Single) Quadratic Form in Gaussian Random Variables} \label{sssec:QFGRV}

                A (single) quadratic form in Gaussian random variables is simply a quadratic polynomial in Gaussian random variables \begin{equation}
                    \rs{Q} = \rv{x}^T\dmt{A}\rv{x}+\dv{b}^T\rv{x}+c, \qquad \rv{x}\sim \mathcal{N}_N(\dv{\mu},\dmt{\Sigma}), \label{eq:QFGRV} \tag{QF}
                \end{equation}
                where $\dmt{A}\in\reals{N\times N}$, $\dv{b}\in\reals{N}$ and $c\in\mathbb{R}$. Without loss of generality, we can assume that the matrix $\dmt{A}$ is symmetric: $\dmt{A}=\dmt{A}^T$. This follows from the fact that for any matrix $\dmt{A}\in\reals{N\times N}$ and vector $\dv{x}\in\reals{N}$, $\dv{x}^T\dmt{A}\dv{x}=\dv{x}^T\left(\dfrac{\dmt{A}+\dmt{A}^T}{2}\right)\dv{x}$. The form is said to be \emph{complete} if $\dv{b}=\dv{0}$ and $c=0$. Otherwise, the form is said to be \emph{incomplete} \footnote{Note that other authors \cite{mathaiQuadraticFormsRandom1992,mohsenipourDistributionQuadraticExpressions2012} use the terms quadratic forms and quadratic expressions to differentiate between what we call complete and incomplete quadratic forms.}.

            \subsection{Quadratic Multiform in Gaussian Random Variables} \label{sssec:QMFGRV}

                A quadratic multiform in Gaussian random variables is simply a tuple (or a vector) of quadratic polynomials in Gaussian random variables \begin{equation}
                    \rs{Q}_m = \rv{x}^T\dmt{A}_m\rv{x}+\dv{b}_m^T\rv{x}+c_m, \qquad \rv{x}\sim \mathcal{N}_N(\dv{\mu},\dmt{\Sigma}),\qquad m=1,2,\ldots,M \label{eq:QMFGRV} \tag{MultiQF}
                \end{equation}
                where $\dmt{A}_m\in\reals{N\times N}$, $\dv{b}_m\in\reals{N}$ and $c_m\in\mathbb{R}$ for $m=1,2,\ldots,M$. Similarly, we can assume that the matrices $\{\dmt{A}_m\}_{m=1}^M$ are symmetric. From now on, the capital letter $N$ is always used to denote the number of Gaussian variables, and the capital letter $M$ is always used to denote the number of forms.

            \subsection{(Single) Quadratic Form in Complex Gaussian Random Variables} \label{sssec:QFCGRV}

                A (single Hermitian) quadratic form in complex Gaussian random variables is also a quadratic polynomial in complex Gaussian random variables \begin{equation}
                    \rs{Q} = \rv{x}^H\dmt{A}\rv{x}+\real{\dv{b}^H\rv{x}}+c, \qquad \rv{x}\sim \mathcal{CN}_N(\dv{\mu},\dmt{\Sigma}) \label{eq:QFCGRV} \tag{CQF}
                \end{equation}
                where $\dmt{A}\in\complexs{N\times N}$, $\dv{b}\in\complexs{N}$ and $c\in\mathbb{R}$. We assume that the matrix $\dmt{A}$ is Hermitian. Note that, unlike in the real case, not all quadratic forms in complex variables can be rewritten with a Hermitian matrix. In fact, $\dmt{A}$ is Hermitian, if and only if, for all $\dv{x}\in\mathbb{C}$, $\dv{x}^H\dmt{A}\dv{x}\in\mathbb{R}$.

            \subsection{Quadratic Multiform in Complex Gaussian Random Variables} \label{sssec:QMFCGRV}

                A (Hermitian) quadratic multiform in complex Gaussian random variables is a tuple of quadratic polynomials in complex Gaussian random variables \begin{equation}
                    \rs{Q}_m = \rv{x}^H\dmt{A}_m\rv{x}+\real{\dv{b}_m^H\rv{x}}+c_m, \quad \rv{x}\sim \mathcal{CN}_N(\dv{\mu},\dmt{\Sigma}), \quad m=1,2,\ldots,M \label{eq:QMFCGRV} \tag{MultiCQF}
                \end{equation}
                where $\dmt{A}_m\in\complexs{N\times N}$, $\dv{b}_m\in\complexs{N}$ and $c_m\in\mathbb{R}$ for $m=1,2,\ldots,M$. We assume that the matrices $\{\dmt{A}_m\}_{m=1}^M$ are Hermitian. 

            \subsection{Ratio of Quadratic Forms in Gaussian Random Variables} \label{sssec:RQFGRV}

                A ratio of quadratic forms in Gaussian random variables \footnote{We could include incomplete forms (i.e., $\frac{\rv{x}^T\dmt{A}\rv{x} + \dv{b}^T \rv{x} + c}{\rv{x}^T\dmt{B}\rv{x} +\dv{d}^T \rv{x} + e}$); however, we did not encounter such forms in the literature.} is defined by \begin{equation}
                    \rs{R}=\dfrac{\rv{x}^T\dmt{A}\rv{x}}{\rv{x}^T\dmt{B}\rv{x}}, \qquad \rv{x}\sim \mathcal{N}_N(\dv{\mu},\dmt{\Sigma}) \label{eq:RQFGRV} \tag{QFRatio}
                \end{equation}
                where $\dmt{A},\dmt{B}\in\reals{N\times N}$. Without loss of generality, we can assume that the matrices $\dmt{A}$ and $\dmt{B}$ are symmetric.

            \subsection{Ratio of (Hermitian) Quadratic Forms in Complex Gaussian Random Variables} \label{sssec:RQFCGRV}

                A ratio of quadratic forms in complex is defined by \begin{equation}
                    \rs{R}=\dfrac{\rv{x}^H\dmt{A}\rv{x}}{\rv{x}^H\dmt{B}\rv{x}}, \qquad \rv{x}\sim \mathcal{CN}_N(\dv{\mu},\dmt{\Sigma}) \label{eq:RQFCGRV} \tag{CQFRatio}
                \end{equation}
                where $\dmt{A}\in\complexs{N\times N}$ and $\dmt{B}\in\complexs{N\times N}$ are assuemd to be Hermititan.
        
        \section{Representation of Quadratic Forms}
            In the previous section, we introduced the various forms that will be analyzed in detail throughout the upcoming sections. When examining the distributions of these quadratic forms, it is often advantageous to identify equivalent distributions where possible, simplifying the analysis and providing insights into their behavior. In this section, we show that two specific quadratic forms, \eqref{eq:QFGRV} and \eqref{eq:QFCGRV}, can be expressed mainly as linear combinations of chi-square distribution, offering a more flexible framework for further study. 
            \subsection{Quadratic Forms in Gaussian Random Variables}
                Consider an incomplete quadratic form in \(N\) Gaussian variables $\rs{Q}=\rv{x}^T\dmt{A}\rv{x}+\dv{b}^T\rv{x}+c$, where $\rv{x}\sim \mathcal{N}_N(\dv{\mu},\dmt{\Sigma})$. We can rewrite the form in terms of the central random vector $\rv{x}-\dv{\mu}$ as follows:
                \begin{equation}
                    \rs{Q}=(\rv{x}-\dv{\mu})^T\dmt{A}(\rv{x}-\dv{\mu})+(2\dv{\mu}^T\dmt{A}+\dv{b}^T)(\rv{x}-\dv{\mu})+\dv{\mu}^T\dmt{A}\dv{\mu}+\dv{b}^T\dv{\mu}+c.
                \end{equation}
                Suppose that the rank of the covariance matrix is \(\Rank{\dmt{\Sigma}}\) and write \(\dmt{\Sigma}=\dmt{B}\dmt{B}^T\) with \(\dmt{B}\in \mathbb{R}^{N\times 
                \Rank{\Sigma}}\) \footnote{Hence $\rank(\dmt{B})=\Rank{\Sigma}$. $\dmt{B}$ can be computed in a number of different ways, for example, start from the sorted diagonalization: $\dmt{\Sigma}=\dmt{U}^T\dmt{DU}$, write $\dmt{U}^T=[\dmt{U}_1^T,\dmt{U}_2^T]$ and $\dmt{D}=\diag(\dmt{D}_1,\dmt{0})$ with $\dmt{U}_1\in\reals{N\times \rankSigma}$ and $\dmt{D}_1\in\reals{\rankSigma\times\rankSigma}$, hence choose $\dmt{B}=\dmt{U}_1^T\dmt{D}_1^{1/2}$.}. Let \(\rv{u}\sim\mathcal{N}(\dv{0},\dmt{I}_{\Rank{\Sigma}})\). Then it can be easily shown that \(\dv{\mu}+\dmt{B}\rv{u}\) is identically distributed as \(\rv{x}\). Hence we can write 
                \begin{equation}
                    \rs{Q}\stackrel{d}{=}\rv{u}^T\dmt{B}^T\dmt{A}\dmt{B}\rv{u}+(2\dv{\mu}^T\dmt{A}+\dv{b}^T)\dmt{B}\rv{u}+\dv{\mu}^T\dmt{A}\dv{\mu}+\dv{b}^T\dv{\mu}+c \label{eq:B^TAB}
                \end{equation} 
                By applying orthogonal diagonalization, \(\dmt{B}^T\dmt{AB}=\dmt{P}\dmt{\Lambda}\dmt{P}^T\), and taking \(\rv{z}=\dmt{P}^T\rv{u}=[\rs{z}_1,\ldots,\rs{z}_{\Rank{\Sigma}}]^T\) (hence \(\rv{z}\sim \mathcal{N}(\dv{0},\dmt{I}_{\Rank{\Sigma}})\) since the standard multivariate normal distribution is invariant under orthogonal transformation), we can write 
                \begin{equation}
                    \rs{Q}\stackrel{d}{=}\rv{z}^T\dmt{\Lambda}\rv{z}+[\dmt{P}^T\dmt{B}^T(2\dmt{A}\dv{\mu}+\dv{b})]^T\rv{z}+(\dv{b}^T\dv{\mu}+\dv{\mu}^T\dv{A}\dv{\mu}+c).
                \end{equation}  Let \(\dmt{\Lambda}=\diag(\lambda_1,\ldots,\lambda_{\Rank{{\Sigma}}})\), \(\dv{d}=\dmt{P}^T\dmt{B}^T(2\dmt{A}\dv{\mu}+\dv{b})=[d_1,\ldots,d_{\Rank{{\Sigma}}}]^T\), and \(c'=\dv{b}^T\dv{\mu}+\dv{\mu}^T\dv{A}\dv{\mu}+c\). Then \begin{equation}
                    \rs{Q}\stackrel{d}{=}\rv{z}^T\dmt{\Lambda}\rv{z}+\dv{d}^T\rv{z}+c'=\sum_{i=1}^{\Rank{{\Sigma}}}\lambda_iz_i^2+\sum_{i=1}^{\Rank{{\Sigma}}}d_iz_i+c'\label{RawDecomposition}
                \end{equation}
                Note that the eigenvalues of ($\dmt{B}^T\dmt{A}\dmt{B}$) and ($\dmt{\Sigma}\dmt{A}$) are equal. These eigenvalues play an important role in determining the distribution of $\rs{Q}$, as will be clearer soon. We may refer from now on to theses eigenvalues as \emph{the eigenvalues of the quadratic form}. The eigenvalues of ($\dmt{B}^T\dmt{A}\dmt{B}$) are always real (since $\dmt{A}$ is symmetric). Depending on the definiteness of $\dmt{A}$, the eigenvalues can all be positive, negative, or a combination of both. To provide a general treatment, let the number of nonzero eigenvalues is $\tilde{N}$  composed of $N_P$ positive eigenvalues and $N_N$ negative eigenvalues, i.e., \(\rankFinal=N_N+N_P\). We can easily show that \(\rankFinal=\rank(\dmt{B}^T\dmt{AB})=\rank(\dmt{\Sigma}^{1/2}\dmt{A}\dmt{\Sigma}^{1/2})=\rank(\dmt{\Sigma}\dmt{A})\). From now on, arrange the eigenvalues as follows:
                \[\lambda_1\geq\lambda_2\geq\ldots\geq\lambda_{N_P}>0>\lambda_{N_{P+1}}\geq \lambda_{N_{P+2}}\geq\ldots\geq \lambda_{N_P+N_N},\]
                and \[\lambda_{N_P+N_N+1}=\lambda_{N_P+N_N+2}=\ldots=\lambda_{\Rank{\Sigma}}=0.\]
                So we have put the \(N_P\) positive eigenvalues first, then the \(N_N\) negative ones next, then the \([\Rank{\Sigma}-(N_N+N_P)]\) zero eigenvalues last. 
                The following proposition characterizes the distribution of a single quadratic form in Gaussian random variables \eqref{eq:QFGRV}. 
                \begin{proposition} \label{prop:1}
                    Any single quadratic form in Gaussian random variables \eqref{eq:QFGRV} can be written as the sum of
                    \begin{enumerate}
                        \item a linear combination of independent non-central chi-square variables \footnote{Note that \(0\) is a linear combination of chi-square variables. In general, the chi-squares disappear whenever \(\dmt{\Sigma}^{1/2}\dmt{A}\dmt{\Sigma}^{1/2}=\dv{0}\), or equivalently, whenever the eigenvalues of the matrix \(\dmt{\Sigma}\dmt{A}\) are all zero.},
                        \item possibly a Gaussian variable that is independent of the chi-square variables,
                        \item possibly a constant.
                    \end{enumerate}
                \end{proposition}
                \begin{proof}
                    Equation \eqref{RawDecomposition} can be rewritten as \[\rs{Q}\stackrel{d}{=}\sum_{i=1}^{\rankFinal
                    }\lambda_i\rs{z}_i^2+\sum_{i=1}^{\rankFinal}d_i\rs{z}_i+\sum_{i=\rankFinal+1}^{\Rank{{\Sigma}}}d_i\rs{z}_i+c'.\]
                    By completing the square, we can write
                    \begin{align}
                        \rs{Q}&\stackrel{d}{=}\sum_{i=1}^{\rankFinal}\lambda_i\left(\rs{z}_i+\dfrac{d_i}{2\lambda_i}\right)^2-\dfrac{1}{4}\sum_{i=1}^{\rankFinal}\dfrac{d_i^2}{\lambda_i}+\sum_{i=\rankFinal+1}^{\Rank{{\Sigma}}}d_i\rs{z}_i+c'\nonumber\\&=\sum_{i=1}^{\rankFinal}\lambda_i(\rs{z}_i+h_i)^2+\sum_{i=\rankFinal+1}^{\Rank{\Sigma}}d_i\rs{z}_i+c'', \label{eqNC}\\
                        &\sim \left[\sum_{i=1}^{\rankFinal}\lambda_i\chi_1^2(h_i^2)\right]+\sigma\mathcal{N}(0,1)+c'', \label{eqn:QF_rep1_ch1}
                    \end{align}
                    where $h_i=\frac{d_i}{2\lambda_i},\;c''=c'-\sum_{i=1}^{\rankFinal} \lambda_i h_i^2$ and $\sigma=\sqrt{\sum_{i=\rankFinal+1}^{\Rank{{\Sigma}}}d_i^2}$. Note that each \((\rs{z}_i+h_i)^2\) follows a non-central chi-square distribution with a single degree of freedom and \emph{non-centrality parameter} \(h_i^2\). The random variable $\sum_{i={\rankFinal+1}}^{\Rank{\Sigma}}d_i\rs{z}_i$  follows the normal distribution as it is a linear combination of jointly normal random variables and is independent of the aforementioned chi-squares due to the independence of \(\rs{z}_i\)'s.
                \end{proof}
                
                Note that in Equation \eqref{eqNC}, in case $\rankFinal=N=\Rank{\Sigma}$, the set of zero eigenvalues is empty, so we can follow the convention of a zero sum over the empty set, thus there is no Gaussian term ($\sigma = 0$) and equality holds true.
                An alternative representation can be obtained by grouping repeated eigenvalues, we may write  \begin{equation}
                     \sum_{i=1}^{\rankFinal}\lambda_i\chi_1^2(h_i^2) = \sum_{\ell=1}^L\omega_\ell\chi_{\nu_\ell}^2(\delta_\ell^2), 
                \end{equation}
                where $\{\omega_\ell\}_{\ell=1}^L$ are the distinct eigenvalues, $L$ is the number of distinct nonzero eigenvalues, $\nu_\ell$ is the multiplicity of the eigenvalue $\omega_\ell$, i.e., $$\nu_\ell=\operatorname{Card}\left(\{n:\lambda_n=\omega_\ell\}\right),$$
                where $\delta_\ell^2$ is the sum of the squares of $h_n$ corresponding to the same eigenvalue $\omega_\ell$, i.e., 
                $$
                    \delta_\ell^2=\displaystyle\sum_{\lambda_n=\omega_\ell}h_n^2.
                $$
                Hence, \eqref{eqn:QF_rep1_ch1} can be equivalently written as 
                \begin{equation}
                    \label{eqn:QF_rep2_ch1}
                    \rs{Q}\sim \sum_{\ell=1}^L\omega_\ell\chi_{\nu_\ell}^2(\delta_\ell^2)+\sigma\mathcal{N}(0,1)+c''.
                \end{equation}
                                 
                The next proposition states the necessary and sufficient condition for a quadratic form \eqref{eq:QFGRV} to be distributed as a linear combination of chi-square random variables up to an additive constant (i.e., no Gaussian term).
                \begin{proposition}
                    Consider a single quadratic form in Gaussian random variables as previously defined and refer to Equation \eqref{eqNC}.  Then, the form is distributed as a linear combination of independent non-central chi-square variables (up to an additive constant) if and only if
                    \[\forall j\in\{\rankFinal+1,\ldots,\Rank{\Sigma}\},\; d_j = 0.\label{condition}\]
                    In this case, we have 
                    \begin{equation}
                        \label{eqn:QF_rep3}
                         \rs{Q}\stackrel{d}{=}\sum_{i=1}^{\rankFinal}\lambda_i\left(\rs{z}_i+h_i\right)^2+c''\sim\sum_{i=1}^{\rankFinal}\lambda_i\chi_1^2(h_i^2)+c'',    
                    \end{equation}
                    or by grouping equal eigenvalues, we get 
                    \begin{equation}
                        \rs{Q}\sim\sum_{\ell=1}^{L}\omega_\ell\chi_{\nu_\ell}^2(\delta_\ell^2)+c'',    
                    \end{equation}
                \end{proposition}
                \begin{proof}  
                    The sufficient condition can be proved by completing the squares. For the necessary condition, suppose, without loss of generality, that the form can be written as \(\sum_{i=1}^{\Rank{{\Sigma}}-1}(\rs{z}_i+h_i)^2+\sigma \rs{z}_{\Rank{{\Sigma}}}+c\),
                    where $\rv{z}=[\rs{z}_1,\ldots,\rs{z}_{\Rank{\Sigma}}]$ is a standard Gaussian vector. We can do so since any linear combination of jointly Gaussian variables is Gaussian. Since $\rs{z}_{\Rank{{\Sigma}}}$ is independent from \(\{\rs{z}_1,\ldots,\rs{z}_{\Rank{{\Sigma}}-1}\}\), the MGF can be written as
                    \[\begin{aligned}
                        M_\rs{Q}(t)&=\exp(c t+\sigma^2 t^2/2)\exp\left[t\sum_{j=1}^{\Rank{{\Sigma}}-1}h_j^2\lambda_j(1-2t\lambda_j)^{-1}\right]\\ &\times\prod_{j=1}^{\Rank{{\Sigma}}-1}(1-2\lambda_jt)^{-1/2}.
                    \end{aligned}\]
                    This is not the MGF of a linear combination of independent non-central chi-square variables (which is the product of scaled MGFs of chi-square variables). Since the MGF defines the distribution uniquely, the expression cannot be written as a form.
                \end{proof}
                
                \begin{corollary}
                    A quadratic form $\rs{Q}=\rv{x}^T\dmt{A}\rv{x}+\dv{b}^T\rv{x}+c$, where $\rv{x}\sim \mathcal{N}_N(\dv{\mu},\dmt{\Sigma})$, $\dmt{\Sigma}$ is nonsingular and \(\det(\dmt{A})\neq 0\) (equivalently, \(\rankFinal=\Rank{\dmt{\Sigma}}\)) is distributed as a linear combination of independent non-central chi-square random variables (up to an additive constant)
                    \[\rs{Q}\stackrel{d}{=}\sum_{i=1}^{\Rank{{\Sigma}}}\lambda_i\left(\rs{z}_i+h_i\right)^2+c''\sim\sum_{i=1}^{\Rank{{\Sigma}}}\lambda_i\chi_1^2(h_i^2)+c''.\]
                \end{corollary}

                As the constant does not present significant complications in studying the distribution functions, it is rarely considered. Actually, except for one author \cite{daviesNumericalInversionCharacteristic1973,daviesAlgorithm155Distribution1980}, the literature usually focuses on the case of a linear combination of independent non-central chi-squares which  will be denoted by \begin{equation}
                     \rs{Q}\stackrel{d}{=}\sum_{i=1}^\rankFinal \lambda_i(\rs{z}_i+h_i)^2 \sim\sum_{i=1}^{\rankFinal}\lambda_i\chi_1^2(h_i^2)\label{eq:lcOnlyChiQFGRV}
                \end{equation}
                Similarly, taking into account multiple eigenvalues, we write \begin{equation}
                     \rs{Q}\sim \sum_{\ell=1}^L \omega_\ell \chi_{\nu_\ell}^2(\delta_\ell^2) \label{eq:lcOnlyChi}
                \end{equation}
                
                \textbf{In summary}, we have precisely four possible equivalent distributions for a quadratic form in Gaussian random variables:
                \begin{enumerate}
                    \item constant random variable \footnote{Consider, for example, $\rs{Q}=\rv{x}^T\dmt{A}\rv{x}$, where $\rv{x}\sim\mathcal{N}(\dv{\mu},\dmt{\Sigma})$, \(\dmt{A}=\left[\begin{smallmatrix}
                    1 &0\\ 0 &0
                \end{smallmatrix}\right]\), \(\dmt{\Sigma}=\left[\begin{smallmatrix}
                    0 &0\\ 0 &1
                \end{smallmatrix}\right]\) and \(\dv{\mu}=[1,1]^T\).},
                    \item Gaussian random variable \footnote{Consider, for example,
                    $\rs{Q}=\rv{x}^T\dmt{A}\rv{x}$, where $\rv{x}\sim\mathcal{N}(\dv{\mu},\dmt{\Sigma})$, \(\dv{\mu}=[1,1,1]^T\), \(\dmt{\Sigma}=\left[\begin{smallmatrix}
                    1 &0 &0\\ 0 &1 &0\\ 0 &0 &0
                \end{smallmatrix}\right]\) and \(\dmt{A}=\left[\begin{smallmatrix}
                    0 &0 &0\\ 0 &0 &1\\ 0 &1 &1
                \end{smallmatrix}\right]\). In this case, $\rs{Q}\sim \mathcal{N}(3,4)$.},
                    \item linear combination of independent (non-central) chi-square random variables (up to an additive constant),
                    \item sum of \begin{enumerate}
                        \item a linear combination of independent (non-central) chi-square random variables, and
                        \item a Gaussian variable that is independent of the previous ones.
                    \end{enumerate} 
                \end{enumerate}
                
                The first two cases are trivial, so they will be disregarded. The third case is the one that is studied mostly in the literature. Up to our knowledge, only Davies \cite{daviesNumericalInversionCharacteristic1973} studies the fourth case.

                Finally, it is useful from a practical point of view to characterize three different types of input parameters of a quadratic form in Gaussian random variables to study its distribution. The basic set, or what we will refer to as the \emph{raw input}, is just the parameters in the basic definition \eqref{eq:QFGRV}, i.e., $\dmt{\mu}\in\reals{N}$, $\dmt{\Sigma}\in\reals{N\times N}$, $\dmt{A}\in\reals{N\times N}$, $\dv{b}\in\reals{N}$ and $c\in\mathbb{R}$. Another set of parameters that defines the final representation will be called the \emph{effective input}. It incorporates $\lambda_n,h_n\in\mathbb{R}$ for $n=1,\ldots,\rankFinal$ \footnote{We may drop the tilde from $\rankFinal$.}, $\sigma\geq 0$, and $c''\in\mathbb{R}$. Furthermore, if we take repeated eigenvalues into consideration and ignore the constant, we have the \emph{reduced effective input} $\omega_\ell\in \mathbb{R}$, $\delta_\ell^2>0$ for $\ell=1,2,\ldots,L$, and $\sigma>0$. These three sets of input are summarized in Table \ref{tab:Input_QFGRV}.
                \begin{table}[ht!]
		          \centering
		          \begin{adjustbox}{max width=\textwidth}
                        \begin{tabular}{|>{\centering\arraybackslash}m{3cm}|>{\raggedright\arraybackslash}m{6cm}|>{\raggedright\arraybackslash}m{5cm}|}
                        \hline
                        \textbf{Input} & \textbf{Form} & \textbf{Parameters}\\ \hline
				    Raw Input & $\rs{Q} = \rv{x}^T\dmt{A}\rv{x}+\dv{b}^T\rv{x}+c,\;\; \rv{x}\sim                    \mathcal{N}_N(\dv{\mu},\dmt{\Sigma})$ & $\dv{\mu}\in\mathbb{R}^N$,                               $\dmt{\Sigma}\in\mathbb{R}^{N\times N}$, $\dmt{A}\in\mathbb{R}^{N\times N}$,                    $\dv{b}\in\mathbb{R}^N$, $c\in\mathbb{R}$ \\ 
                        \hline
                        Effective Input & $\rs{Q}\stackrel{d}{=}\displaystyle\sum_{n=1}^\rankFinal \lambda_n(\rs{z}_n+h_n)^2 +\sigma \rs{z}_{\rankFinal+1} + c''$, $[\rs{z}_1,\ldots,\rs{z}_\rankFinal,\rs{z}_{\rankFinal+1}]\sim\mathcal{\rankFinal}_{\rankFinal+1}(\dv{0},\dmt{I}_{\rankFinal+1})$ & $\lambda_n,h_n\in\mathbb{R}$ for $n=1,\ldots,\rankFinal$, $\sigma\geq 0$, $c''\in\mathbb{R}$ \\ 
                        \hline
                        Reduced Effective Input & $\rs{Q}\sim\displaystyle\sum_{\ell=1}^L\omega_\ell\chi_{\nu_\ell}^2(\delta_\ell^2) + \mathcal{N}(0,\sigma^2)$ & $\omega_\ell\in \mathbb{R},\;\delta_\ell^2>0$ for $\ell=1,2,\ldots,L$, $\sigma>0$, for $k\neq \ell$, $\omega_k\neq\omega_\ell$ \\ 
                        \hline
			        \end{tabular}
		          \end{adjustbox}
    		      \caption{Input Classification of a Quadratic Form in Gaussian Random Variables.}
    		      \label{tab:Input_QFGRV}
    	    \end{table}
                
                From now on, we will classify forms that happen to be linear combinations of non-central chi-squares.
                \paragraph{Central Quadratic Forms in Gaussian Random Variables}
                Consider a quadratic form that can be written as a linear combination of non-central chi-squares as in Equation \eqref{eq:lcOnlyChiQFGRV},
                \begin{equation*}
                     \rs{Q}\stackrel{d}{=}\sum_{i=1}^\rankFinal \lambda_i(\rs{z}_i+h_i)^2,
                \end{equation*} such a form is said to be \emph{central} if the non-centrality parameters $\{h_i\}_{i=1}^\rankFinal$ are all null: $h_i=0 \text{ for } i=1,\ldots,\rankFinal$. Equivalently, the form can be written as a linear combination of independent (central) chi-squares \begin{equation*}
                     \rs{Q}\sim \sum_{\ell=1}^L \omega_\ell \chi_{\nu_\ell}^2.
                \end{equation*}
                Evidently, a complete quadratic form, $$\rs{Q}=\rv{x}^T\dmt{A}\rv{x},$$ with $\rv{x}\sim\mathcal{N}_N(\dv{\mu},\dmt{\Sigma})$ with a non-singular covariance matrix $\dmt{\Sigma}$ is central if and only if the random vector $\rv{x}$ is central, i.e., $\dmt{A}\dv{\mu}=\dv{0}$. The following proposition is a straightforward result of Proposition \ref{prop:1}.
                \begin{proposition}
                    Any central complete quadratic form is distributed as a linear combination of independent chi-square random variables:
                    \[\rs{Q}\stackrel{d}{=}\sum_{i=1}^{\Rank{{\Sigma}}}\lambda_i\rs{z}_i^2\sim\sum_{i=1}^{\Rank{{\Sigma}}}\lambda_i\chi_1^2\]
                \end{proposition}

                \paragraph{Positive Definite Quadratic Forms in Gaussian Random Variables} \label{par:Positve_Definite}
                Consider a quadratic form that can be written as a linear combination of chi-squares as in Equation \eqref{eq:lcOnlyChiQFGRV}, \begin{equation*}
                     \rs{Q}\stackrel{d}{=}\sum_{i=1}^\rankFinal \lambda_i(\rs{z}_i+h_i)^2, 
                \end{equation*} such a form is said to be \emph{positive definite} if all the linear coefficients $\{\lambda_i\}_{i=1}^\rankFinal$ are positive: $\lambda_i > 0 \text{ for } i=1,\ldots,\rankFinal$. Null eigenvalues are simply ignored. Note that, in the case of a complete form with a non-degenerate Gaussian distribution, i.e., $\dv{b}=\dv{0}$, $c=0$ and $\dmt{\Sigma}$ is positive definite, positive definiteness of the form is simply equivalent to the positive semi-definiteness of the matrix $\dmt{A}$.

                On the other hand, a form is said to be \emph{negative-definite} if all the eigenvalues of the quadratic form $\{\lambda_i\}_{i=1}^\rankFinal$ are negative. Equivalently, a quadratic form is negative definite if and only if its additive inverse is a positive definite form. Due to that fact, negative forms do not need to be studied separately. Finally, a form in which eigenvalues alternate signs, i.e., $\exists i_1, i_2, \; \lambda_{i_1}\lambda_{i_2}<0$, is called an indefinite form. The latter can be written as the difference of two positive definite forms.

                \textbf{Semi-definite Forms}
                Some writers \cite{mathaiQuadraticFormsRandom1992} distinguish positive semi-definite quadratic forms in Gaussian random variables from positive forms following the standard algebraic notions. If we are interested in the form itself, and this form can be written as a linear combination of independent chi-squares \eqref{eq:lcOnlyChi}, it does not matter if some eigenvalues are zero. In the end, we can perceive the form as the identically distributed one with a smaller number of variables. However, in a ratio \ref{eq:RQFGRV}, (1) we cannot necessarily write both forms in terms of the same chi-squares, and (2) if even if we can, i.e., the matrices $\dmt{A}$ and $\dmt{B}$ are simultaneously orthogonally diagonalizable, the ``dead" chi-squares in one form may be living and breathing in the other. Due to these considerations, in the context of a ratio, we will use the standard linear algebraic terminology for a quadratic form, i.e., a form $\rv{x}^T\dmt{A}\rv{x}$ is positive definite (respectively, semi-definite), if and only, the matrix $\dmt{A}$ is positive definite (respectively, positive semi-definite). 
                
                Here is an illustrative example. Consider the forms $\rs{Q}_1=\rv{x}^T\dmt{A}_1\rv{x}$ and $\rs{Q}_2=\rv{x}^T\dmt{A}_2\rv{x}$ where $\dmt{A}_1=\dmt{I}_2$, $\dmt{A}_2=\begin{bmatrix}
                     0 & 0\\
                     0 & 1
                \end{bmatrix}$, and $\rv{x}=[\rs{x}_1,\rs{x}_2]^T\sim\mathcal{N}_2(\dv{0},\dmt{I}_2)$. In itself, the quadratic form $\rs{Q}_2$ can be perceived as a form in one Gaussian variable, namely $\rs{x}_2$. So we can say it is a positive definite single form. However, within the ratio, it does not contain any $\rs{x}_1$, which appears in the numerator $\rs{Q}_1$. In the ratio, we would say $\rs{Q}_2$ is positive semi-definite, not positive definite.
         
                We conclude this section with two non-trivial numerical examples that summarize the section.
    
                Consider a complete quadratic form 
                \begin{equation}
                \rs{Q}=\rv{x}^T\dmt{A}\rv{x}, \label{eq:First-Example}
                \end{equation}
                where $\rv{x}=[\rs{x}_1,\rs{x}_2,\rs{x}_3,\rs{x}_4]^T\sim \mathcal{N}_4(\dv{\mu},\dmt{\Sigma})$, and raw input parameters $$\dmt{\Sigma}=\dfrac{1}{4}\begin{bmatrix}
                    5 & 5 & 3 & 3\\
                    5 & 5 & 3 & 3\\
                    3 & 3 & 9 & 1\\
                    3 & 3 & 1 & 9
                \end{bmatrix},\;\dv{\mu}=\begin{bmatrix}
                    0 \\
                    1 \\
                    0 \\
                    1
                \end{bmatrix},\; \text{and}~ \dmt{A}=\dfrac{1}{2}\begin{bmatrix}
                    -1 & -1 & 1 & -1\\
                    -1 & -1 & -1 & 1\\
                    1 & -1 & 1 & 1\\
                    -1 & 1 & 1 & 1
                \end{bmatrix} .$$
         
                Note first that the covariance matrix $\dmt{\Sigma}$ is rank-deficient with $\Rank{\Sigma}=3$. To factorize this matrix, it is possible to start from the orthogonal diagonalization $$\dmt{\Sigma}=\dmt{QD}\dmt{Q}^T$$ with $$\dmt{Q}=\begin{bmatrix}
                    \frac12 &\frac12 &\frac12 &\frac12\\
                    0 &0 &-\frac{1}{\sqrt{2}} &\frac{1}{\sqrt{2}}\\
                    -\frac12 &-\frac12 &\frac12 &\frac12\\
                    -\frac{1}{\sqrt{2}} &\frac{1}{\sqrt{2}} &0 &0
                \end{bmatrix} \text{ and } \dmt{D} = \begin{bmatrix}
                    4 & 0 & 0 & 0\\
                    0 & 2 & 0 & 0\\
                    0 & 0 & 1 & 0\\
                    0 & 0 & 0 & 0\\
                \end{bmatrix}.$$ Partition the matrices as $\dmt{Q}=[\dmt{Q}_1\;\; \dmt{Q}_2]$, $\dmt{Q}_1\in\mathbb{R}^{4\times 3}$ and $\dmt{D}=\diag(\dmt{D}_1,0)$, to finally write $\dmt{\Sigma}=\dmt{B}\dmt{B}^T$ with $$\dmt{B}=\dmt{Q}_1\dmt{D}_1^{1/2}=\begin{bmatrix}
                    1 &0 &-\frac12\\
                    1 &0 &-\frac12\\
                    1 & -1 &\frac12\\
                    1 &1 &\frac12
                \end{bmatrix}.$$ Referring to Equation \eqref{eq:B^TAB}, we calculate the matrix $$\dmt{B}^T\dmt{AB}=\begin{bmatrix}
                    0 &0 &2\\
                    0 &0 &0\\ 
                    2 &0 &0
                \end{bmatrix}.$$ Orthogonally eigendecomposing, we get $\dmt{B}^T\dmt{AB}=\dmt{P\Lambda}\dmt{P}^T$, with $$\dmt{\Lambda}=\begin{bmatrix}
                    2 & 0 & 0\\
                    0 & -2 & 0\\
                    0 & 0 & 0
                \end{bmatrix} \text{ and } \dmt{P}=\begin{bmatrix}
                    \frac{1}{\sqrt{2}} & -\frac{1}{\sqrt{2}} & 0\\
                    0 & 0 & 1\\
                    \frac{1}{\sqrt{2}} & \frac{1}{\sqrt{2}}  & 0
                \end{bmatrix}.$$
                The eigenvalues of the quadratic form are $\lambda_1= 2, \lambda_2 = -2$ and $\lambda_3 = 0$. Hence, the form has 2 non-zero eigenvalues: $\rankFinal=N_P+N_N=2<\Rank{\Sigma}=3<N=4$. Referring to Equation \eqref{RawDecomposition}, calculate the vector $$\dv{d}=2\dmt{P}^T\dmt{B}^T\dmt{A}\dv{\mu}=\begin{bmatrix}
                    \sqrt{2} & \sqrt{2} & 2
                \end{bmatrix}^T.$$ Referring to the same equation, the constant $c'$ is given by $c'=\dv{\mu}^T\dmt{A}\dv{\mu}=1$. Hence, we have \begin{equation*}
                    \rs{Q}\stackrel{d}{=}\rv{z}^T\dmt{\Lambda}\rv{z}+\dv{d}^T\rv{z}+c'=2\rs{z}_1^2-2\rs{z}_2^2+\sqrt{2}\rs{z}_1+\sqrt{2}\rs{z}_2+2\rs{z}_3+1
                \end{equation*} where $\rv{z}\sim\mathcal{N}_3(\dv{0},\dmt{I}_3)$. Completing the two squares as in the proof of Proposition \ref{prop:1} (Equation \ref{eqNC}), we can  write $$\begin{aligned}
                    \rs{Q}&\stackrel{d}{=}\sum_{i=1}^{\rankFinal}\lambda_i(\rs{z}_i+h_i)^2+\sum_{i=\rankFinal+1}^{\Rank{\Sigma}}d_i\rs{z}_i+c''\\&=2\left(\rs{z}_1+\frac{\sqrt{2}}{4}\right)^2-2\left(\rs{z}_2+\dfrac{\sqrt{2}}{4}\right)^2+2\rs{z}_3 + 1.
                \end{aligned}$$ Therefore, we characterize the distribution of the quadratic form in Equation \eqref{eq:First-Example} by $$\rs{Q}\sim 2\chi_1^2(\frac18)-2\chi_1^2(\frac18)+2\mathcal{N}(0,1)+1$$
                Note that the complete quadratic form is the sum of a linear combination of independent non-central chi-squares, an independent central Gaussian, and a constant (which clearly can be used to shift the Gaussian).The effective input parameters are $\lambda_1 =2$, $\lambda_2=-2$, $h_1^2=h_2^2=\frac{1}{8}$, $\sigma=2$ and $c''=1$. We classify $\rs{Q}$ as a non-central indefinite quadratic form.

                Consider now the incomplete form $$\rs{Q}=\rv{x}^T\dmt{A}\rv{x}+\dv{b}^T\rv{x},$$ where $\rv{x}\sim\mathcal{N}_3(\dv{0},\dmt{I}_3)$, and raw input parameters $$
                \dmt{\Sigma}=\dmt{I}_3,\; \dv{\mu} = \dv{0},\; \dmt{A}=\begin{bmatrix}
                    7  & 24 & 0\\
                    24 & -7 & 0\\
                    0  &  0 & 25
                \end{bmatrix} \text{ and } \dv{b}=\begin{bmatrix}
                    40\\
                    50\\
                    30
                \end{bmatrix}.$$ As $\rv{x}$ is a standard Gaussian vector, we can replace every $\dmt{B}$ by the identity matrix. Eigendecomposing $\dmt{B}^T\dmt{AB}=\dmt{A}$, we get $\dmt{A}=\dmt{P}^T\dmt{\Lambda P}$ where $$\dmt{P}=\begin{bmatrix}
                    -0.8 &  0.6 & 0\\
                    -0.6 & -0.8 & 0\\
                    0   &  0   & 1
                \end{bmatrix} \text{ and } \dmt{\Lambda} = \begin{bmatrix}
                    25 &  0  & 0\\
                    0  &  25 & 0\\
                    0  &  0  & -25
                \end{bmatrix}$$ The eigenvalues of the quadratic form are $\lambda_1 = \lambda_2 =25$, and $\lambda_3 = -25$, or $\omega_1 = 25$ and $\omega_2=-25$. So the number of non-zero eigenvalues is $\rankFinal = 3$, and the distinct eigenvalues is $L=2$. 
                Now, we can write $\dv{d}=\dmt{P}^T\dv{b}=[-62,30,-16]^T$; hence, the noncentrality parameters are $h_1^2 =\dfrac{961}{625} $, $h_2^2 = \dfrac{9}{25}$ and $h_3^2 =\dfrac{64}{625}$. Reexpress the form as 
                \begin{equation*}
                \begin{aligned}
                    \rs{Q}&\stackrel{d}{=}\rv{z}^T\dmt{\Lambda}\rv{z}+\dv{d}^T\rv{z}+c'=25\rs{z}_1^2+25\rs{z}_2^2-25\rs{z}_3^2-62\rs{z}_1+30\rs{z}_2-16 \rs{z}_3 \\
                    &\sim 25\chi_1^2\left( \dfrac{961}{625}\right)+25\chi_1^2 \left(\dfrac{9}{25} \right)-25\chi_1^2 \left( \dfrac{64}{625}\right)-\dfrac{1122}{25} 
                \end{aligned}
                \end{equation*} or $$\rs{Q}\sim 25\chi_2^2\left(\dfrac{1186}{625}\right)-25\chi_1^2\left(\dfrac{64}{625} \right)-\dfrac{1122}{25}.$$
                The reduced effective input parameters are $\omega_1=25$, $\omega_2 = -25$, $\delta_1^2 = \dfrac{1186}{625}$, $\delta_2^2 =\dfrac{64}{625}$, $\sigma = 0$ and $c'' =-\dfrac{1122}{25}$.

            \subsection{Quadratic Forms in Complex Gaussian Random Variables}
                Consider now a quadratic form in \(N\) complex Gaussian random variables \eqref{eq:QFCGRV}, \(\rs{Q}=\rv{x}^H\dmt{A}\rv{x}+\Re(\dv{b}^H\rv{x})+c\), where \(\rv{x}\sim \mathcal{CN}_N(\dv{\mu},\dmt{\Sigma})\) and \(\rank(\dmt{\Sigma})=\Rank{\Sigma}\). Following a similar approach from the previous section, it can be easily shown that \[\rs{Q}=(\rv{x}-\dv{\mu})^H\dmt{A}(\rv{x}-\dv{\mu})+\Re\left[(2\dmt{A}\dv{\mu}+\dv{b})^H(\rv{x}-\dv{\mu})\right]+\dv{\mu}^H\dmt{A}\dv{\mu}+\Re(\dv{b}^H\dv{\mu})+c.\] 
                Write \(\dmt{\Sigma}=\dmt{B}\dmt{B}^H\), where \(\dmt{B}\in \mathbb{C}^{N\times \Rank{\Sigma}}\). Let \(\rv{u}\in \mathcal{CN}(\dv{0},\dmt{I}_{\Rank{\Sigma}})\). We can show that the random vector \(\dv{\mu}+\dmt{B}\rv{u}\) is identically distributed as \(\rv{x}\). Then, \[\rs{Q}\stackrel{d}{=}\rv{u}^H\dmt{B}^H\dmt{AB}\rv{u}+\Re\left[(2\dmt{A}\dv{\mu}+\dv{b})^H\dmt{B}\rv{u}\right]+\dv{\mu}^H\dmt{A}\dv{\mu}+\Re(\dv{b}^H\dv{\mu})+c.\]
                Unitarily diagonalizing, write \(\dmt{B}^H\dmt{AB}=\dmt{P}\dmt{\Lambda}\dmt{P}^H\), where \[\dmt{\Lambda}=\diag (\lambda_1,\ldots,\lambda_{\Rank{\Sigma}})\] Then,
                \[\rs{Q}\stackrel{d}{=}\rv{u}^H\dmt{P}\dmt{\Lambda}\dmt{P}^H\rv{u}+\Re\left[(2\dmt{A}\dv{\mu}+\dv{b})^H\dmt{B}\rv{u}\right]+\dv{\mu}^H\dmt{A}\dv{\mu}+\Re(\dv{b}^H\dv{\mu})+c.\]
                Let \(\rv{z}=\dmt{P}^H\rv{u}\). Since \(\dmt{P}\) is unitary, \(\rv{z}\sim \mathcal{CN}(\dv{0},\dmt{I}_{\Rank{\Sigma}})\). Therefore,
                \[\rs{Q}\stackrel{d}{=}\rv{z}^H\dmt{\Lambda}\rv{z}+\Re\left[\dv{d}^H\rv{z}\right]+c',\]
                where \(\dv{d}=\dmt{P}^H\dmt{B}^H(2\dmt{A}\dv{\mu}+\dv{b})\) and  \(c'=\Re[\dv{b}^H\dv{\mu}]+\dv{\mu}^H\dv{A}\dv{\mu}+c\). Writing \(\dv{d} =[d_1,\ldots,d_{\Rank{\Sigma}}]^T\) and \(\rv{z}=[\rs{z}_1,\ldots,\rs{z}_{\Rank{\Sigma}}]^T\), we have
                \[\rs{Q}\stackrel{d}{=}\sum_{j=1}^{\Rank{\Sigma}}\lambda_j|z_j|^2+\Re\left[\sum_{j=1}^{\Rank{\Sigma}}\bar{d}_jz_j\right]+c'.\]
                Now since \(\rv{z}\sim \mathcal{CN}(\dv{0},\dmt{I}_{\Rank{\Sigma}})\), we can write \(\rs{z}_j=\rs{w}_j+i\rs{v}_j\), where \(\{\rs{w}_j,\rs{v}_j,j=1,\ldots,\Rank{\Sigma}\}\) are independent and follow \(\mathcal{N}(0,\frac{1}{2})\). Let \(d_j=f_j+ig_j\) where \(f_j,g_j\in \mathbb{R}\). Then, we have \begin{equation}
                    \rs{Q}\stackrel{d}{=}\sum_{j=1}^{\Rank{\Sigma}}\left[\lambda_j(\rs{w}_j^2+\rs{v}_j^2)+f_j\rs{w}_j+g_j\rs{v}_j\right]+c'.\label{eqNCC}
                \end{equation}
                Similarly, 
                the eigenvalues of the quadratic form can have $\rankFinal = N_P+N_N$ nonzero eigenvalues and \(\Rank{\Sigma}-\rankFinal\) zero eigenvalues. We can easily show that \(\rankFinal=\rank(\dmt{B}^H\dmt{AB})=\rank(\dmt{\Sigma}^{1/2}\dmt{A}\dmt{\Sigma}^{1/2})\). The eigenvalues are arranged by putting the \(N_P\) positive eigenvalues first, the \(N_N\) negative ones, then the zero eigenvalues last,
                \[\lambda_1\geq\lambda_2\geq\ldots\geq\lambda_{N_P}>0>\lambda_{N_P+1}\geq \lambda_{N_P+2}\geq\ldots\geq \lambda_{N_P+N_N}\]
                and \[\lambda_{N_P+N_N+1}=\lambda_{N_P+N_N+2}=\ldots=\lambda_{\Rank{\Sigma}}=0\]
                 Therefore, we have
                \begin{align}
                    \rs{Q}&\stackrel{d}{=}\sum_{j=1}^{\rankFinal}\lambda_j\left[\left(\rs{w}_j+\dfrac{f_j}{2\lambda_j}\right)^2+\left(\rs{v}_j+\dfrac{g_j}{2\lambda_j}\right)^2\right]+\sum_{j=\rankFinal+1}^{\Rank{\dmt{\Sigma}}}(f_j\rs{w}_j+g_j\rs{v}_j)+c''\\
                    &\sim \sum_{j=1}^{\rankFinal}\dfrac{\lambda_j}{2}\chi_2^2(h_j^2)+\sigma\mathcal{N}(0,1)+c''
                \end{align}
                where \begin{align*}
                    h_j^2&=\dfrac{|d_j|^2}{2\lambda_j^2},\\
                    \sigma&=\dfrac{1}{\sqrt{2}}\sqrt{\displaystyle\sum_{j=\rankFinal+1}^{\Rank{\dmt{\Sigma}}}|d_j|^2},\\
                    c''&=c'-\dfrac{1}{4}\sum_{j=1}^{\rankFinal}\dfrac{|d_j|^2}{\lambda_j}.
                \end{align*}
                or by grouping the repeated eigenvalues 
                \begin{equation}
                     \rs{Q}\sim \sum_{\ell=1}^L \tilde{\omega}_\ell \chi_{2\nu_\ell}^2(\delta_\ell^2), 
                \end{equation}
                 where $\{\tilde{\omega}_\ell\}_{\ell=1}^L$ are the distinct eigenvalues, $L$ is the number of distinct nonzero eigenvalues, $\nu_\ell$ is the multiplicity of the eigenvalue $\omega_\ell$, i.e., $$\nu_\ell=\operatorname{Card}\left(\{n:\frac{\lambda_n}{2}=\tilde{\omega}_\ell\}\right),$$
                where $\delta_\ell^2$ is the sum of the squares of $h_n$ corresponding to the same eigenvalue $\tilde{\omega}_\ell$, i.e., 
                $$
                    \delta_\ell^2=\displaystyle\sum_{\lambda_n=\tilde{\omega}_\ell}h_n^2.
                $$
                Similar propositions hold true:
                \begin{proposition}
                    Any quadratic form in complex Gaussian random variables can be written as the sum of
                    \begin{enumerate}
                        \item a linear combination of independent non-central chi-square variables with even degrees of freedom,
                        \item possibly a Gaussian variable that is independent of the chi-square variables,
                        \item possibly a constant.
                    \end{enumerate}
                \end{proposition}
                As a result, a Hermitian quadratic form in complex Gaussian random variables can be viewed as a special case of a form in real ones. In fact, we could have shown the equivalence starting from the raw input of the complex case.
                
                Sorting the eigenvalues in non-increasing order, we can write the following proposition:
                \begin{proposition}
                Consider a quadratic form in complex Gaussian random variables as previously defined and refer to Equation \eqref{eqNCC}.  Then the form is distributed as a linear combination of independent non-central chi-square variables with even degrees of freedom (up to an additive constant) if and only if
                \[\forall j\in\{\rankFinal,\rankFinal+1,\ldots,\Rank{\Sigma}\},\; d_j = 0 \label{condition2}.\]
                \end{proposition}

                As a numerical example, consider the complete Hermitian form $\rs{Q}=\rv{x}^H\dmt{A}\rv{x}+\Re[\dv{b}^H\rv{x}]$, where $\rv{x}\sim\mathcal{CN}_2(\dv{\mu},\dmt{\Sigma})$, and raw input parameters $$\dmt{A}=\begin{bmatrix}
                    1 & -i\\
                    i &  1
                \end{bmatrix},\;\;\dv{b}=\begin{bmatrix}
                    1 \\ 1
                \end{bmatrix},\;\;\dv{\mu}=\begin{bmatrix}
                    1\\
                    1+i
                \end{bmatrix}, \text{ and } \dmt{\Sigma}=\begin{bmatrix}
                    10 & -6i\\
                    6i &  10
                \end{bmatrix}.$$ Firstly, as the covariance matrix $\dmt{\Sigma}$ is non-singular, we can use the factorization $\dmt{\Sigma}=\dmt{\Sigma}^{1/2}\dmt{\Sigma}^{1/2}$, where $$\dmt{\Sigma}^{1/2}=\begin{bmatrix}
                    3 & -i\\
                    i &  3
                \end{bmatrix}.$$ Proceeding, we compute the matrix $$\dmt{\Sigma}^{1/2}\dmt{A}\dmt{\Sigma}^{1/2}=\begin{bmatrix}
                    16  & -16i\\
                    16i &  16
                \end{bmatrix}.$$ Eigendecomposing, we get $\dmt{\Sigma}^{1/2}\dmt{A\Sigma}=\dmt{P\Lambda}\dmt{P}^H$, with $$\dmt{P}=\dfrac{1}{\sqrt{2}}\begin{bmatrix}
                    -i & i\\
                    1  & 1
                \end{bmatrix} \text{ and } \dmt{\Lambda}=\begin{bmatrix}
                    32 & 0\\
                    0  & 0
                \end{bmatrix}.$$ 
                There are two eigenvalues $\lambda_1 =32$ and $\lambda_2 =0$. Hence, $ \rankFinal=1<\Rank{\Sigma}=N=2$. Compute the vector $\dv{d}=\dmt{P}^H\dmt{\Sigma}^{1/2}(2\dmt{A}\dv{\mu}+\dv{b})=\sqrt{2}[10+18i,1-i]^T$. The first (and only) non-centrality parameter is then given by $h_1^2=|d_1|^2/2\lambda_1^2=53/128$. Then, the form is identically distributed as $$\rs{Q}\stackrel{d}{=}32|z_1|^2+\Re\left[\sqrt{2}(10+18i)z_1+\sqrt{2}(1-i)z_2\right]$$ As the non-zero second entry of $\dv{d}$ corresponds to a zero eigenvalue, an independent Gaussian appears, with a standard deviation of $\sqrt{2}$. Therefore, the quadratic form $\rs{Q}$ follows: $$\rs{Q}\sim 16\chi_2^2\left(\dfrac{53}{128}\right)+\sqrt{2}\mathcal{N}(0,1)+\dfrac{3}{8}.$$ 
                
                In fact, unlike what was derived in the last example, the literature is only concerned with quadratic forms in complex Gaussian variables of the form \begin{equation}
                     \rs{Q}\stackrel{d}{=}\sum_{n=1}^N \lambda_n|\rs{z}_n+h_n|^2 \label{eq:lcOnlyChiQFCGRV}
                \end{equation}
                where $\rv{z}=[\rs{z}_1,\rs{z}_2,\ldots,\rs{z}_N]^T\sim\mathcal{CN}_N(\dv{0},\dmt{I}_N)$.
                Similarly, taking into account multiple eigenvalues, we write \begin{equation}
                     \rs{Q}\sim \sum_{\ell=1}^L \omega_\ell \chi_{2\nu_\ell}^2(\delta_\ell^2) \label{eq:lcOnlyChiC}
                \end{equation}

        \section{Quantities of Interest}
            In this section, we briefly review basic definitions of the quantities of interest that concern quadratic forms in Gaussian random variables, namely, the cumulative distribution function, the probability density function, moments, moment generating function, characteristic function, cumulant generating function, and cumulants. We restrict these definitions to the continuous random variables since all the quadratic forms under study are functions of continuous random variables. 
            \subsection{Cumulative Distribution Function}
                The \emph{cumulative distribution function (CDF)} of a random variable $ \rs{x}$, denoted $  F_\rs{x}(x) $, is a function that calculates the probability that the value of a random variable $\rs{x}$  is less than or equal to a specific real number $x$. For any given $\rs{x} \in \mathbb{R} $, $  F_\rs{x}(x) $  represents the probability that  $ \rs{x} $  takes on a value within the interval  $  (-\infty, x] $, i.e., $$
                F_\rs{x}(x) = \mathbb{P}[\rs{x} \le x].
                $$ This probability is computed over the \emph {sample space} $ \Omega $  by finding all \emph{outcomes} $ \omega $  such that $  \rs{x}(\omega) \le x $, capturing the probability that $\rs{x}$ does not exceed $x$. Mathematically, we write
                \begin{equation}
                  \label{eqn:CDF}
                  F_{\rs{x}}(x) = \mathbb{P}[\{\omega \in \Omega| \rs{x}(\omega) \le x\}], \quad x \in \mathbb{R}.
                \end{equation}
                The subscript denotes the random variable. For notation simplicity, it is wide common to use $\mathbb{P}[\rs{x} \le x]$ instead of $\mathbb{P}[\{\omega | \rs{x}(\omega) \le x\}]$. 
                The CDF is a monotonic non-decreasing function and must satisfy $\lim_{x \rightarrow -\infty} F_\rs{x}(x)= 0$ and $\lim_{x \rightarrow \infty}F_\rs{x}(x) = 1$.

                The joint CDF of  $N$ random variables, $\rv{x}=[\rs{x}_1, \rs{x}_2, \ldots, \rs{x}_N]$ is the mapping, $F_{\rv{x}}: \mathbb{R}^N \rightarrow [0,1]$, and given by 
                 \begin{equation}
                    F_{\rv{x}}(x_1, x_2, \ldots x_N) = \mathbb{P}[\rs{x}_1 \le x_1, \rs{x}_2 \le x_2, \ldots, \rs{x}_N \le x_N] 
                 \end{equation}
                 A related quantity of interest in some applications is the inverse of the CDF, known as the \emph{quantile function}, which gives the value of the random variable when the probability is known, i.e.,
                 \begin{equation}
                     \label{eqn:quntileFn}
                     x = F_\rs{x}^{-1}(p), 
                 \end{equation}
                 where $p = F_\rs{x}(x)$.
                 
                 The \emph{median} is the value at which the CDF is $\frac{1}{2}$. To be more precise, if $w$ is the median of a random variable $\rs{x}$ it means that $F_{\rs{x}}(w)\ge \frac{1}{2}$ and $1-F_{\rs{x}}(w)\ge \frac{1}{2}$.

                 As will be seen later, much of the literature is concerned with finding the CDF of different quadratic forms. 
                 
            \subsection{Probability Density Function}
                For an absolutely continuous random vector $\rv{x}$, there exists a nonnegative function, $f_{\rs{x}} : \mathbb{R}^N \rightarrow [0,\infty)$, called \emph{probability density function (PDF)}, such that 
                \begin{equation}
                    \mathbb{P}[\rv{x}\in A] = \int_A f_{\rv{x}}(\dv{x}) d \dv{x},
                \end{equation}
                for every $A \subseteq \mathbb{R}^N$. Consequently, if $F_{\rv{x}}$ is sufficiently differentiable at $(x_1, x_2, \ldots, x_N)$, then the joint PDF is
                \begin{equation}
                    f_{\rv{x}}(x_1, x_2, \ldots, x_N) = \dfrac{ \partial^N F_{\rv{x}}(x_1, x_2, \ldots, x_N)}{\partial x_1, x_2, \ldots, x_N}. 
                \end{equation}
                A non-constant quadratic multiform is absolutely continuous.
                
                One quantity of interest related to the PDF is the \emph{mode}, which is the value at which the PDF of a random variable  $\rs{x}$, is maximum, i.e., 
                \begin{equation}
                    \text{Mode}(\rs{x}) = \arg \max_x f_\rs{x} (x)
                \end{equation}
            \subsection{Moments}
                The $n$th \emph{moment} of a continuous random variable, $\rs{x}$, is given 
                \begin{equation}
                    m_n = \mathbb{E}[\rs{x}^n] = \int x^n  f_{\rs{x}}(x) dx .
                \end{equation}
                Another quantity of  interest is the $n$-th central moment, which is defined as 
                \begin{equation}
                    \mu_n = \mathbb{E}[(\rs{x}-m_1)^n] = \int (x-m_1)^n  f_{\rs{x}}(x) dx .
                \end{equation}
                
            \subsection{Moment Generating Function and Characteristic Function}
                The \emph {moment generating function (MGF)} of a random variable $\rs{x}$ is useful for computing moments, identifying distribution and studying convergence in distribution, etc \cite{GengMGF}. It is defined as follows: 
                \begin{equation}
                    M_{\rs{x}}(t) = \mathbb{E}[e^{t\rs{x}}] = \int_{-\infty}^\infty e^{tx} dF_\rs{x}(x),
                \end{equation}
                where $t \in \mathbb{R}$. By separating the integration, we can study its existence \cite{2013-Gallager-StochasticProcessesTheoryForApplications}, 
                \begin{equation}
                    M_{\rs{x}}(t) = \int_0^\infty e^{tx} dF_\rs{x}(x) + \int_{-\infty}^0 e^{tx} dF_\rs{x}(x).
                \end{equation}
                From the definition, the MGF is always positive. The first integral is $\mathbb{P}[\rs{x}>0]$ and the second is $\mathbb{P}[\rs{x} \le 0]$; hence, both integrals exist for $t=0$. If $t_{+}$ is the supremum of values of $t$ for which the first integral exists, it should be in the range $0 \le t_{+} \le \infty$ and the first integral exists for all $t < t_{+}$. Similarly,  If $t_{-}$ is the infimum of values of $t$ for which the second integral exists, the range is $-\infty \le t_{-} \le 0$ and the second integral exists for all $t > t_{-}$. By combining the two integrals one can deduce the interval  $t$ over which $M_\rs{x}(t)$ exists which is $t_{-} \le0 $ to $t_{+} \ge 0$. This interval could be degenerate, i,e, $\{0\} $, finite or infinite without further implications at the endpoints, i.e., open or closed \cite{GengMGF}.

                A complex variable can replace the real value $t$ in the MGF. This results in a different transformation. Specifically, if $M_\rs{x}$ is well defined on some interval $I$, the complex function 
                \begin{equation}
                    \Psi_\rs{x}(z) \triangleq \mathbb{E}[e^{zx}],
                \end{equation}
                is well-defined and holomorphic in the strip $\{ z = t + is: t\in I, s \in \mathbb{R}\}$ \cite{GengMGF}. 
                If $z$ is restricted to the imaginary axis, i.e, $t=0$, this corresponds to \emph{the characteristic function (CF)}, 
                \begin{equation}
                     \Psi_{\rs{x}}(is) = \int_0^\infty e^{isx} dF_\rs{x}(x) + \int_{-\infty}^0 e^{isx} dF_\rs{x}(x).
                \end{equation}

                The MGF is a useful tool for finding the moments of random variables. If $M_\rs{x}(t)$ is well defined on a neighborhood of 0,  It can be shown that the random variable, $x$, is uniquely determined by its moments. The $n$th moment is given by
                \begin{equation}
                    \mathbb{E}[\rs{x}^n] =  \left. \frac{d^n M_\rs{x}(t)}{dt^n} \right\rvert_{t=0}.
                \end{equation}
            \subsection{Cumulant Generating Function}
                The \emph{cumulant generating function (CGF)} of a random variable $\rs{x}$ is the logarithm of the MGF,
                \begin{equation}
                    K_{\rs{x}}(t) = \ln M_{\rs{x}}(t).
                \end{equation}
            \subsection{Cumulants}
                If $M_{\rs{x}}(t)$ is finite in a neighborhood about 0, then $K_\rs{x}$ has the following convergent Taylor series:
                \begin{equation}
                    K_{\rs{x}}(t) = \sum_{j=1}^\infty \dfrac{1}{j!} \kappa_j t^j.
                \end{equation}
                The coefficient $\kappa_j$ is called the $j$th \emph{cumulant} of $\rs{x}$.
            \subsection{Support}
                The \emph{support} of a random vector $\rv{x}\in\mathbb{R}^N$ can be defined as the smallest closed subset of $\mathbb{R}^N$ whose probability content is $1$. For example, the support of a uniformly distributed random variable over $[0,1]$ is $[0,1]$, and the support of bivariate normal vector is $\mathbb{R}^2$. Equivalently, the support can be defined as the collection of points whose arbitrarily small open neighborhoods are not null, i.e., $$\operatorname{Supp}(\rv{x})=\{\dv{a}\in\mathbb{R}^N:\;\forall \epsilon>0, \mathbb{P}\left(\rv{x}\in\mathcal{B}(\dv{a},\epsilon)\right)>0\}$$ where $$\mathcal{B}(\dv{a},\epsilon)=\{\dv{y}\in\mathbb{R}^N:\|\dv{y}-\dv{a}\|<\epsilon\}.$$ If the random vector $\rv{x}$ is absolutely continuous, then $$\operatorname{Supp}(\rv{x})=\operatorname{Closure}\left(\{\dv{x}\in\mathbb{R}^N:f_{\rv{x}}(\dv{x})>0\}\right).$$

\chapter{Applications}
\section{Introduction}
    Quadratic forms in Gaussian random variables appear in a diverse range of applications. Central limit theorems provide a theoretical basis for the approximation of random variables by Gaussian ones. A study of magnitudes or distances may incorporate second-degree polynomials in such variables.  

    The diversity of the applications comes with a diversity in the parameters, from covariance structures to the number of forms and variables. In Figure \ref{fig:MvsN}, we show the number of forms $M$ versus the number of variables $N$ as they appear in different papers. Note that the axes do not have a uniform scale; entire intervals are used to represent 2, 3, 4, and 5 variables, as well as 1, 2, 3, and 4 forms, while the remaining parts of the axes are "continuous" with varying log scales as revealed in the green grid. Each point refers to a single implementation of a formula, and implementations of the same paper are joined by a solid line. The single dashed line represents an implementation at every single point. It is noteworthy that the highest number of variables is a single form that appears in genetics (\mhndcite{chenNumericalEvaluationMethods2019}), and the highest number of forms appear in a theoretical work on fluid antenna systems (\mhndcite{khammassiNewAnalyticalApproximation2022}).

    \begin{figure}[ht!]
        \centering
        \input{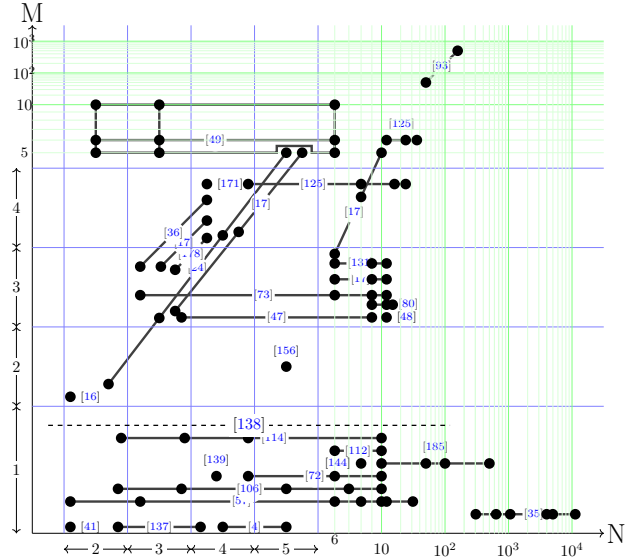}
        \caption{Number of Forms $M$ vs Number of Variables $N$}
        \label{fig:MvsN}
    \end{figure}
    
\section{Telecommunication}    
    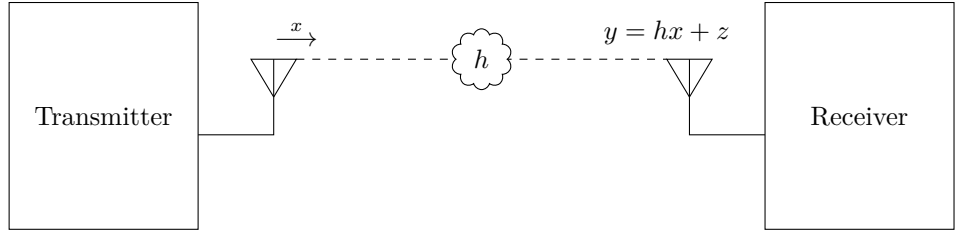
\begin{figure}[ht!]
        \centering
        \begin{tikzpicture}
    \draw (0,2.5) rectangle (2.5,5.5) node[midway]{Transmitter};
    \draw (10,2.5) rectangle (12.5,5.5) node[midway]{Receiver};
    \draw (2.5,3.75) -- (3.5,3.75) -- (3.5,4.75);
    \draw (3.5,4.25) -- (3.2,4.75) -- (3.8,4.75) -- (3.5,4.25);
    \draw (10,3.75) -- (9,3.75) -- (9,4.75);
    \draw (9,4.25) -- (8.7,4.75) -- (9.3,4.75) -- (9,4.25);
    \draw[dashed] (3.8,4.75) -- (8.7,4.75);
    \node[cloud, draw, fill=white] (c) at (6.25,4.75) {$h$};
    \node[fill=white] (d) at (8.7,5.1) {$y=hx+z$};
    \node[fill=white] (e) at (3.8,5.1) {$\stackrel{x}{\longrightarrow}$};
\end{tikzpicture}
        \caption{SISO Channels}
        \label{fig:SISO}
    \end{figure}
    A very straightforward way to assess the quality of a wireless communication system (see single-input-single-output system Fig. \ref{fig:SISO}) is to answer the question ``How much is it likely that the connection will break?" A numerical answer to this question is called the outage probability. The signal in a system undergoes what is known as multi-path loss, which arises due to different sorts of reflections and delays. Under some modeling assumptions, this loss is characterized by a Gaussian attenuation factor, i.e., a Gaussian channel gain. The magnitude of such a gain is used to evaluate the outage probability. Moreover, to enhance connectivity the multiple antennas might be used at either side of the channel, and several combination methods are employed to obtain a better signal. In some instances, the resulting channel gain might be a quadratic form in complex Gaussian random variables. In others, it might be the maximum component of a multivariate quadratic form.

    \subsection{SIMO}
    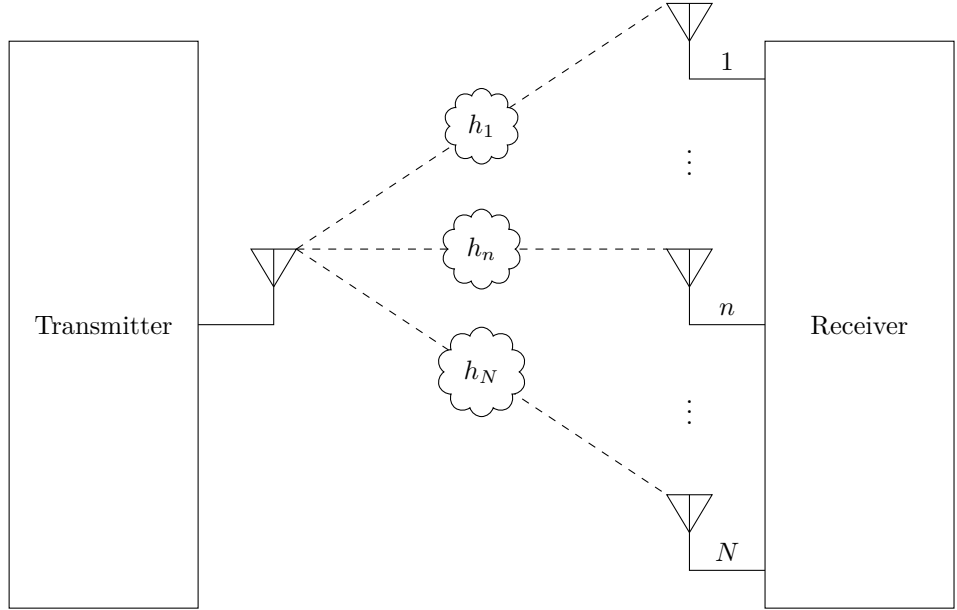
\begin{figure}[ht!]
        \centering
        \begin{tikzpicture}
    \draw (0,0) rectangle (2.5,7.5) node[midway]{Transmitter};
    \draw (10,0) rectangle (12.5,7.5) node[midway]{Receiver};
    \draw (2.5,3.75) -- (3.5,3.75) -- (3.5,4.75);
    \draw (3.5,4.25) -- (3.2,4.75) -- (3.8,4.75) -- (3.5,4.25);

    \draw (10,3.75) -- (9,3.75) -- (9,4.75);
    \node at (9.5,3.75) [anchor=south]{$n$};
    \draw (9,4.25) -- (8.7,4.75) -- (9.3,4.75) -- (9,4.25);

    \draw (10,0.5) -- (9,0.5) -- (9,1.5);
    \node at (9.5,0.5) [anchor=south]{$N$};
    \draw (9,1) -- (8.7,1.5) -- (9.3,1.5) -- (9,1);

    \draw (10,7) -- (9,7) -- (9,8);
    \node at (9.5,7) [anchor=south]{$1$};
    \draw (9,7.5) -- (8.7,8) -- (9.3,8) -- (9,7.5);
    
    \node at (9,2.7) {\vdots};
    \node at (9,6) {\vdots};

    \draw[dashed] (3.8,4.75) -- (8.7,8);
    \node[cloud, draw, fill=white] (c) at (6.25,6.375) {$h_1$};
    \draw[dashed] (3.8,4.75) -- (8.7,4.75);
    \node[cloud, draw, fill=white] (c) at (6.25,4.75) {$h_n$};
    \draw[dashed] (3.8,4.75) -- (8.7,1.5);
    \node[cloud, draw, fill=white] (c) at (6.25,3.125) {$h_N$};
\end{tikzpicture}
        \caption{SIMO Channels}
        \label{fig:SIMO}
    \end{figure}
    Consider a single-input-multiple-output (SIMO) system (see Fig. \ref{fig:SIMO}). Assume that the complex signal received at the $n^{\text{th}}$ antenna is modeled as  
    $$
    \rs{y}_n=\rs{h}_n x + \rs{z}_n,\;n=1,2,\ldots,N, 
    $$ where $x$ is the deterministic transmitted symbol, $\rs{h}_n$ is the corresponding channel gain and $\rs{z}_n$ is additive zero-mean complex Gaussian noise. Stacking the received signal gives 
    $$
    \rv{y}=\rv{h}x+\rv{z}
    $$ 
    where $\rv{y},\rv{h},\rv{z}\in \mathbb{C}^N$. 

    The outage probability is defined in terms of the signal-to-noise ratio (SNR), say $\gamma$, as  $$P_{\text{out}}=\mathbb{P}[\gamma<\gamma_{\text{th}}]$$ 
    where $\gamma_{\text{th}}$ is a threshold. The connection is considered failed if the SNR falls below $\gamma_{\text{th}}$. 
    Hence, to determine the outage probability, it is important to assess the probabilistic distribution of the SNR. Other performance metrics, such as the bit error probability (BEP), and the channel capacity, also depend on the SNR.
    
    \subsection{Fading Models}
        Multi-path fading yields a random complex channel gain. Several regimes are used to model the gain, among which are the Rayleigh and Rician fading regimes. The channel gain follows a complex normal distribution in the mentioned cases. In particular:
        \begin{enumerate}
            \item In Rayleigh fading, $\rs{h}\sim\mathcal{CN}(0,\Omega)$. The underlying physical theory assumes no dominant component in the multi-path fading, and the real and imaginary parts of the signal converge under CLT to Gaussian random variables. The variance $\Omega$ can be viewed as the sum of powers of different paths.
            \item In Rician fading, $\rs{h}\sim \mathcal{CN}(\nu^2,\sigma^2)$. A similar underlying physical theory is used, except for the presence of a dominant component, say, Line Of Sight (LOS) component. The power of the latter component is given by $\nu^2$, whereas the powers of the non-dominant paths sum up to $\sigma^2$. Note that the distribution can be re-parametrized in terms of $\Omega=\sigma^2+\nu^2$, the sum of the powers of the dominant and non-dominant components, and $K=\nu^2/\sigma^2$, the ratio of the dominant to non-dominant signal powers.
        \end{enumerate}
    
        Ideally, the antennas must be placed such that the channels are independent. However, this isn't always the case. Hence, we can generalize and claim $\rv{h}\sim \mathcal{CN}(\dv{\mu},\dmt{\Sigma})$ for some $\dv{\mu}\in\mathbb{C}^N$ and positive definite $\dmt{\Sigma}\in \mathbb{C}^{N\times N}$. It is important to note that a non-real covariance of two channel gains corresponds to a correlation between real and imaginary parts of the signals, hence between the in-phase and quadrature parts. In many papers, it is assumed that such a correlation does not exist, which yields a pure real covariance matrix $\dmt{\Sigma}\in\mathbb{R}^{N\times N}$.
    
    \subsubsection{Diversity Techniques}
    
        To enhance the quality of the received signal, the $N$ different received signals may be combined in some way. Numerous diversity techniques exist in the literature, among which two give rise to quadratic forms, maximum ratio combining and selection combining. 
    
        \paragraph{Maximum Ratio Combining (MRC)}
            Assume that the received signals will be linearly combined
            $$\rs{y}=\sum_{n=1}^N w_n \rs{y}_n$$
            Denoting by $\sigma_z^2=\mathbb{E}[|\rs{z}_n|^2]$ the power of the noise, and by $P=|x|^2$ the power of the signal, the instantaneous signal-to-noise ratio (SNR) $\gamma$ is hence given by $$\gamma = \dfrac{\left|\sum_{n=1}^N w_n \rs{h}_n\right|^2}{\mathbb{E}\left[\left|\sum_{n=1}^N w_n \rs{z}_n\right|^2\right]}=\dfrac{P}{\sigma_z^2}\dfrac{\left|\sum_{n=1}^N w_n \rs{h}_n\right|^2}{\left|\sum_{n=1}^N w_n\right|^2}$$
            Assume further that the combiner has access to the channel state information (CSI), i.e., the channel gains $\mathbf{h}$, then it can be easily demonstrated that the optimum linear combination is given by: $w_n=h_n^*$. Hence, the instantaneous SNR for MRC becomes $\gamma = \frac{P}{\sigma_z^2}\sum_{n=1}^N |h_n|^2= \frac{P}{\sigma_z^2} \|h\|^2$. Now in the case of independent channels, the SNR follows a scaled chi-square distribution under Rayleigh fading, and a scaled non-central chi-square distribution under Rician fading. However, this does not (necessarily) hold in the presence of correlation. In the latter case, we have a non-trivial quadratic form $\gamma=\rv{h}^H\dmt{A}\rv{h}$, where $\dmt{A}=\dfrac{P}{\sigma_z^2}\mathbf{I}_N$, and $\rv{h}\sim \mathcal{CN}(\mathbf{0},\mathbf{\Sigma})$ or $\mathcal{CN}(\mathbf{\mu},\mathbf{\Sigma})$. After unitarily diagonalizing, as discussed in the previous chapter, the form will be equivalent to a linear combination of chi-square variables $$\gamma\sim\sum_{n=1}^N\lambda_n\chi_2^2(d_n)$$  As an example, this case have aspired Ramirez-Espinoza et al. \cite{ramirez-espinosaNewApproachStatistical2019} to develop a new method for evaluating the CDF of quadratic forms in complex Gaussian random variables. As we will explain later, they approximate the quadratic form by a sequence of random variables with tractable distributions. 
        
        \paragraph{Selection Combining}
    
            Another diversity combining technique is selection combining. The idea is fairly simple: given the channel state information, only utilize the best channel, i.e., the channel with the largest gain. Hence, the received signal can be written as $$\rs{y}=\rs{y}_{n_0},\;n_0=\arg \max_{n=1,\ldots,N} |\rs{h}_n|^2$$ Hence the instantaneous SNR is given by $\gamma=\frac{P}{\sigma_z^2}\max_{n=1,\ldots,N}|\rs{h}_n|^2$. Thus we are concerned with the CDF of the maximum power, which can be easily computed using the joint distribution as follows \begin{align*}
                \mathbb{P}\left[\max_{n=1,\ldots,N}\frac{P}{\sigma_z^2}|\rs{h}_n|^2<\gamma_{\text{th}}\right]&=\mathbb{P}\left[\forall n\in{1,\ldots,N},\;\frac{P}{\sigma_z^2}|\rs{h}_n|^2<\gamma_{\text{th}}\right]\\
                &=F_{\rv{h}}\left(\gamma_{\text{th}}\frac{\sigma_z^2}{P},\ldots,\gamma_{\text{th}}\frac{\sigma_z^2}{P}\right)
            \end{align*}
            Therefore, we are interested in the joint distribution of the squared moduli of the complex normal vector $\rv{h}$. These squared moduli follow a multivariate (non-central) chi-square distribution (see Section 4), and the moduli themselves are correlated Rayleigh or Rician variables. This has motivated Beaulieu and Hemachandra to study these distributions \cite{beaulieuNovelRepresentationsBivariate2011,beaulieuNovelSimpleRepresentations2011}.

    \subsection{A MIMO Case}
    
    Consider the case of a MIMO with $K$ receive antennas and $M$ transmit ones. Assume that the receiver combines the signals using MRC, whereas the transmitter uses SC to send the signal. Under a Rayleigh or Rician fading regime, we can similarly verify that the powers follow a multivariate quadratic form, that does not necessarily admit marginal (non-central) chi-square distributions.

    For the sake of simplicity, consider the Rayleigh case. Now we have $N=KM$ channel gains: the signal sent from the $m^{\text{th}}$ transmit antenna to the $k^{\text{th}}$ receive antenna is modeled by $$\rs{y}_k=\rs{h}_{km}x+\rs{z}_k$$ For $M=1$, we recover the MRC scheme. Since the transmitter will select one antenna, we need to compare the possible gains corresponding to each antenna. If the $m^{\text{th}}$ antenna is chosen, the corresponding SNR will be proportional to the quadratic form $\rs{Q}_m=\|\rv{h}_m\|^2$, where $\rv{h}_m=[\rs{h}_{1m},\rs{h}_{2m},\ldots,\rs{h}_{Km}]^T$. In case of independence of channel gains $\{\rs{h}_{km}\}_{k=1}^K$, the forms $\rs{Q}_m$ will be scaled chi-squares. However, in a general setting, this is not the case. Again, the SNR is the maximum among the possible SNRs. So to compute the outage probability in such a system, we need the CDF of 
    $$
    \max_{m=1,\ldots,M}\rs{Q}_m.
    $$
    Therfore, the outage probability in this MIMO system requires the joint CDF of the quadratic multiform $\rv{\rs{Q}} = [\rs{Q}_1, \rs{Q}_2, \cdots, \rs{Q}_M]^T`$.

\section{Direction of Arrival}
 
\begin{figure}[!htb]
    \centering
    \subfloat[A plane wavefront impinging on a ULA.]
    {\includegraphics[width=0.75\textwidth]{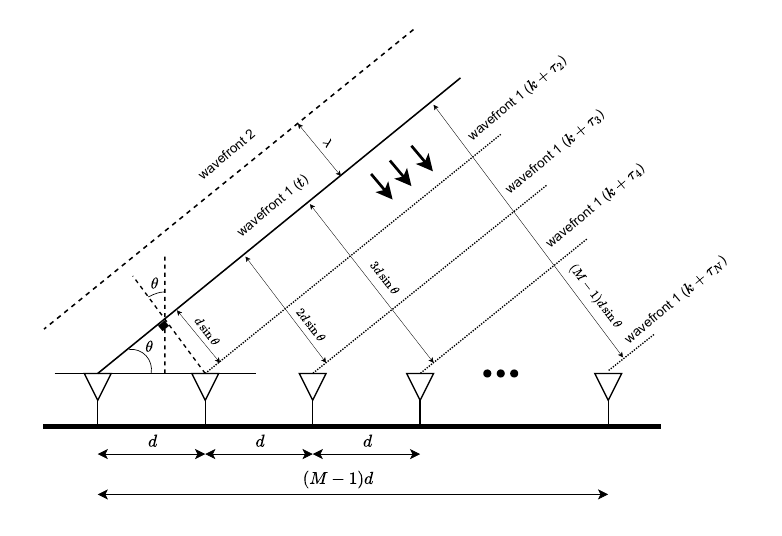}
    \label{subfig:ULA1}}
    \hspace{0.02\textwidth} 
    \subfloat[Noise-free output of all sensors model.]{\includegraphics[width=0.75\textwidth]{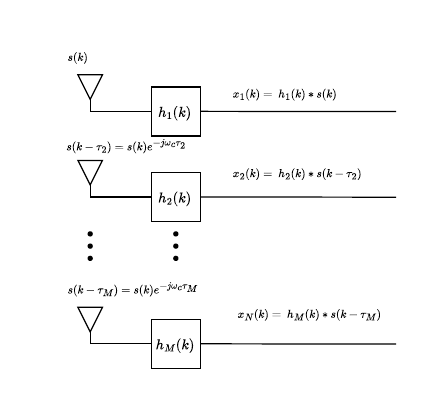}%
    \label{subfig:ULA2}}
    \caption{Uniform linear array.}
    \label{fig:ULA}
\end{figure}

Many modern systems, including radar, wireless communications, and sonar, rely on arrays of sensors, antennas, or hydrophones 
to achieve high-resolution spatial processing. The study of signal processing techniques for such arrays is commonly referred to as \emph{array processing} \cite{vanOptimumArray2002}. One of the central problems in this area is the estimation of the direction of arrival (DOA) of an incoming wave. The performance DOA  estimators has been studied extensively, especially for maximum-likelihood methods and their threshold region behavior. 

Consider a far-field source emitting electromagnetic energy toward a uniform linear array (ULA) (Fig.~\ref{subfig:ULA1}) of $M$ equally spaced sensors \cite{athleySpaceTime2003}. Under the far-field assumption, the impinging wavefront can be treated as planar across the array. If the source and array are assumed coplanar, then the DOA is characterized by a single azimuth angle $\theta$. Let $s(k)$ denote the source signal. Because the wave reaches the sensors at different times, the received signal at each element differs by a propagation delay and possibly sensor-dependent gain and phase effects. Under the narrowband assumption, the noise-free output of the $m$th sensor can be written as \cite{athleyThresholdRegionPerformance2005}
$$
 x_m(k) = h_m(k)\ast s(k-\tau_m),
$$
where $h_m(k)$ is the impulse response of the $m$th sensor and $\tau_m = \tau_m(\theta)$ is the propagation delay from the reference sensor (the left-most sensor in Fig.~\ref{subfig:ULA1}) to the $m$th sensor. The noise-free output of all sensors is shown in Fig.~\ref{subfig:ULA2}. Because of the narrowband assumption, the Fourier transform of the $h_m(k)$, $H_m(\theta)$, is constant over the signal bandwidth and hence, 
$$
x_m(k) = H_m(\theta)e^{-j \omega_c \tau_m (\theta)} s(k),
$$ 
where $\omega_c$ is the carrier frequency. Stacking the sensor outputs yields the familiar array model
$$
    \boldsymbol{x}(k) = \boldsymbol{a}(\theta) s(k),
$$
where $\boldsymbol{a}(\theta) \triangleq [ H_1(\theta) e^{-j \omega_c \tau_1 (\theta)}, H_2(\theta)e^{-j \omega_c \tau_2 (\theta)}, \ldots, H_M(\theta)e^{-j \omega_c \tau_M (\theta)} ]^T $ is the steering vector.
For a ULA with inter-element spacing $d$, it can be seen that (Fig.~\ref{subfig:ULA1}) the propagation delay at the $m$th sensor is 
$$
    \tau_m = \dfrac{(m-1)d \sin \theta}{c}, \quad m= 1,2, \ldots, M,
$$
where $c$ is the propagation speed. If the sensors are isotropic with unit gain, i.e., $H_m(\theta) = 1, (m=1,2, \ldots, M) $ then the steering vector is  
$$
\boldsymbol{a}(u) = [1, e^{-j 2\pi\frac{d}{\lambda}u}, \ldots, e^{-j 2\pi(M-1)\frac{d}{\lambda}u} ]^T, \quad u \triangleq \sin \theta
$$
where $\lambda = 2\pi c/\omega_c$ is the wavelength. In practice, noise and modeling errors are unavoidable. The complex baseband array snapshot is therefore more realistically modeled as
\begin{equation}
    \label{eqn:ArraySignsl}
     \rv{x}(k) = \dv{a}(u_0)s(k) + \rv{z}(k), 
\end{equation}
where $u_0$ is the true direction parameter and $ \rv{z}(k)\sim \mathcal{CN}(\dv{0},\sigma_z^2\dmt{I}) $ is additive Gaussian noise. 

A basic technique for estimating the DOA is \emph{conventional beamforming}. Since the array response is maximized when the steering vector matches the true direction, one may estimate $u_0$ by maximizing the beamformer output power. Specifically, define
$$
    \hat{\rs{u}} = \arg \max_u \rs{v}(u),
$$
where 
\begin{equation*}
             \rs{v}(u) = \dfrac{1}{K} \sum_{k=1}^{K} |\dv{a}^H(u)\rv{x}(k)|^2,
        \end{equation*}
and $K$ is the number of data samples or \emph{snapshots}. The resulting estimator can be interpreted as a maximum-likelihood estimator under the conventional beamforming model. 

A useful way to study the performance of this estimator is through its mean-square error, $ \mathbb{E}\left[(\hat{\rs{u}}-{u_0})^2\right]$. In threshold-region analyses, this error is decomposed into two contributions: one corresponding to small local errors around the true direction, and another corresponding to large \emph{outlier} events in which the estimator locks onto a sidelobe instead of the mainlobe. Thus, 
\begin{equation}
    \label{eqn:MSE_Athley}
    \begin{aligned}
        \mathbb{E}\left[(\hat{\rs{u}}-{u_0})^2\right] &=\mathbb{P}\left[\text{no outlier}\right]\mathbb{E}\left[(\hat{\rs{u}}-{u_0})^2|\text{no outlier}\right]\\&+ \mathbb{P}[\text{outlier}]\mathbb{E}\left[(\hat{\rs{u}}-{u_0})^2|\text{outlier}\right].
    \end{aligned}
\end{equation}
The first term describes estimation errors near the true direction and is often approximated by the Cramer-Rao bound. The second term accounts for the threshold effect, namely, the possibility that a sidelobe peak exceeds the mainlobe peak.

Let $u_p$, $p=1,\cdots,P,$ denote the sidelobe peak locations in the discretized beamformer spectrum. Then the outlier contribution may be approximated by summing over the possible sidelobe events, which leads to
\begin{align*}
    &\mathbb{E}\left[(\hat{\rs{u}}-{u_0})^2\right] \approx \left( 1 - \sum_{p=1}^{P}E_p \right) \text{CRB} + \sum_{p=1}^{P}E_p (u_p -u_0)^2.
\end{align*}
where $ E_p \triangleq \mathbb{P} \left[ \rs{v}(u_p) > \rs{v}(u_0) \right]$ is the pairwise error probability that the $p$th sidelobe exceeds the mainlobe peak. Thus, the threshold region performance is controlled by these pairwise error probabilities. 

For a single source observed over $K$ snapshots, the pairwise error probability can be written explicitly as 
\begin{equation}
    \label{eqn:PairwiseErrProb_Athley}
    E_p = \mathbb{P} \left[\sum_{k=1}^K\left[  \rv{x}^H(k) \left( \dv{a}_0\dv{a}_0^H -\dv{a}_p\dv{a}_p^H \right)  \rv{x}(k)  \right]<0\right],
\end{equation}
where $\dv{a}_0 = \dv{a}(u_0)$ and $\dv{a}_p = \dv{a}(u_p)$. This is already suggestive: the event is defined through a quadratic form in Gaussian vectors.

The Gaussian model of $\rv{x}$ depends on the signal model. If the source signal $s(k)$ is deterministic, then
$$
\rv{x}(k) \sim \mathcal{CN} (\dv{a}_0s(k), \sigma_{\rv{z}}^2\dmt{I}_M).
$$
If, instead, the source is modeled as a zero-mean Gaussian random signal with variance $\sigma_\rs{s}^2$, independent of the noise, then 
$$
\rv{x}(k) \sim \mathcal{CN} (\dv{0}, \sigma_\rs{s}^2\dv{a}_0\dv{a}_0^H+ \sigma_{\rv{z}}^2\dmt{I}_M).
$$

In either case, the stacked observation vector remains complex Gaussian. 

To express the problem more compactly, define the data matrix and its vectorized form as 
$$
     \rmt{X}        \triangleq [\rv{x}(1), \rv{x}(2), \cdots ,\rv{x}(K)], \qquad
     \tilde{\rv{x}} \triangleq \text{vec} (\rmt{X}).
$$
Then, $\tilde{\rv{x}}$ is a complex Gaussian vector  in $ \mathbb{C}^{KM \times 1}$. Also, define
$$
    \tilde{\dmt{A}} \triangleq \dmt{I}_K \otimes (\dv{a}_0\dv{a}_0^H -\dv{a}_p\dv{a}_p^H).
$$
With these definitions, the pairwise probability becomes 
$$
    E_p = \mathbb{P}[\tilde{\rv{x}}^H\tilde{\dmt{A}}\tilde{\rv{x}} < 0].
$$
Therefore, the pairwise error probability is exactly the CDF of an indefinite single Hermitian quadratic form in complex Gaussian random variables, evaluated at zero.

\section{Performance of Adaptive Filters}

\begin{figure}[!ht]
    \centering
    \includegraphics[width=0.7\textwidth]{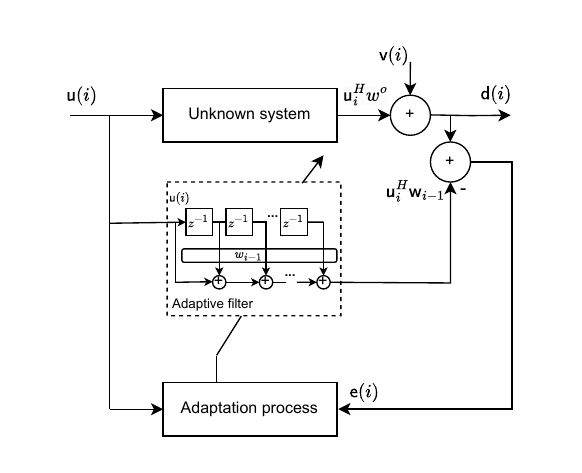}
    \caption{Adaptive system identifier.}
    \label{fig:Adaptive system identifier}
\end{figure}

Adaptive signal processing is concerned with estimating and tracking unknown or time-varying systems from streaming data. Since many practical systems are too complex to model exactly, they are often approximated by linear or finitely parameterized nonlinear models. Adaptive filtering is one of the principal tools for doing so, and many adaptive signal-processing problems can be viewed as instances of adaptive system identification.

Consider an unknown linear system driven by the input sequence \(\{\rs{u}(i)\}\). Define the regression vector
\[
\rv{u}_i \triangleq [\rs{u}(i), \rs{u}(i-1), \ldots, \rs{u}(i-N+1)]^T,
\]
and let \(\dv{w}^o \in \mathbb{C}^N\) denote the optimal coefficient vector. The desired response is modeled as
\[
\rs{d}(i)=\rv{u}_i^H\dv{w}^o+\rs{v}(i),
\]
where \(\rs{v}(i)\stackrel{\mathrm{iid}}{\sim}\mathcal{CN}(0,\sigma_{\rs{v}}^2)\) is additive noise. The optimal coefficient vector is the minimizer of the mean-square error criterion
\[
\dv{w}^o
=
\arg\min_{\dv{w}}
\mathbb{E}\!\left[
\left|\rs{d}(i)-\rv{u}_i^H\dv{w}\right|^2
\right].
\]
Under the usual assumptions, the Wiener solution is given by
\[
\dv{w}^o=\dmt{R}_{\rv{u}}^{-1}\dv{p}_{ud},
\]
where
\[
\dmt{R}_{\rv{u}}\triangleq\mathbb{E}[\rv{u}_i\rv{u}_i^H],
\qquad
\dv{p}_{ud}\triangleq\mathbb{E}[\rv{u}_i \rs{d}^*(i)].
\]
In practice, however, the statistics \(\dmt{R}_{\rv{u}}\) and \(\dv{p}_{ud}\) are unknown, so one resorts to stochastic-gradient methods that estimate \(\dv{w}^o\) recursively from the available data.

A broad class of adaptive algorithms can be written as
\[
\dv{w}_i
=
\dv{w}_{i-1}
+
\mu\frac{\rv{u}_i}{g(\rv{u}_i)}\rs{e}^*(i),
\qquad i\ge 0,
\]
where \(\mu\) is the step size, \(g(\rv{u}_i)\) is a positive normalization function, and
\[
\rs{e}(i)=\rs{d}(i)-\rv{u}_i^H\dv{w}_{i-1}
\]
is the estimation error. The choice \(g(\rv{u}_i)=1\) yields the LMS algorithm, whereas \(g(\rv{u}_i)=\|\rv{u}_i\|^2\) yields the NLMS algorithm \cite{naffouriTransientAnalysisOfData-normalizedAdaptiveFilters2003}.

To study the transient behavior of the adaptive filter, define the weight-error vector
\[
\tilde{\dv{w}}_i\triangleq \dv{w}^o-\dv{w}_i.
\]
Under the standard independence assumptions, two quantities are especially important: the weighted variance relation and the mean weight-error recursion. The weighted variance relation has the form
\begin{equation}
\label{eqn:VarianceRelation}
\mathbb{E}\!\left[\|\tilde{\dv{w}}_i\|_{\dmt{S}}^2\right]
=
\mathbb{E}\!\left[\|\tilde{\dv{w}}_{i-1}\|_{\dmt{S}'}^2\right]
+
\mu^2\sigma_{\rs{v}}^2
\mathbb{E}\!\left[
\frac{\|\rv{u}_i\|_{\dmt{S}}^2}{g^2(\rv{u}_i)}
\right],
\end{equation}
where
\[
\|\dv{a}\|_{\dmt{S}}^2\triangleq \dv{a}^H\dmt{S}\dv{a},
\qquad \forall \dv{a}\in\mathbb{C}^N,
\]
and \(\dmt{S}\) is a Hermitian positive-definite matrix. The matrix \(\dmt{S}'\) is given by
\begin{equation}
\label{eqn:Smatrix}
\dmt{S}'
=
\dmt{S}
-\mu \dmt{S}\,\mathbb{E}\!\left[\frac{\rv{u}_i\rv{u}_i^H}{g(\rv{u}_i)}\right]
-\mu \mathbb{E}\!\left[\frac{\rv{u}_i\rv{u}_i^H}{g(\rv{u}_i)}\right]\dmt{S}
+\mu^2\mathbb{E}\!\left[
\frac{\|\rv{u}_i\|_{\dmt{S}}^2}{g^2(\rv{u}_i)}\rv{u}_i\rv{u}_i^H
\right].
\end{equation}
The mean weight-error recursion is
\begin{equation}
\label{eqn:MeanWeightRecursion}
\mathbb{E}[\tilde{\dv{w}}_i]
=
\left(
\dmt{I}_N
-
\mu \mathbb{E}\!\left[\frac{\rv{u}_i\rv{u}_i^H}{g(\rv{u}_i)}\right]
\right)
\mathbb{E}[\tilde{\dv{w}}_{i-1}].
\end{equation}

Therefore, transient analysis requires the evaluation of expectations of the form
\begin{equation}
\label{eqn:ADF_moments}
\begin{aligned}
\dmt{A}
&\triangleq
\mathbb{E}\!\left[
\frac{\rv{u}_i\rv{u}_i^H}{g(\rv{u}_i)}
\right], \\
\dmt{B}
&\triangleq
\mathbb{E}\!\left[
\frac{\|\rv{u}_i\|_{\dmt{S}}^2}{g^2(\rv{u}_i)}
\rv{u}_i\rv{u}_i^H
\right], \\
c
&\triangleq
\mathbb{E}\!\left[
\frac{\|\rv{u}_i\|_{\dmt{S}}^2}{g^2(\rv{u}_i)}
\right].
\end{aligned}
\end{equation}
These quantities determine the transient behavior of the adaptive filter.

With additional assumptions, these expectations can be reduced to moments of scalar random variables that are ratios, or more generally rational functions, of quadratic forms. This reduction is especially transparent in the NLMS case \cite{al-naffouriMeansquareAnalysisNLMS2021}.

In the NLMS case,
\[
g(\rv{u}_i)=\|\rv{u}_i\|^2.
\]
The denominator is therefore itself a Hermitian quadratic form. Moreover, under the standard independence assumptions used in NLMS analysis, it is convenient to restrict attention to diagonal weighting matrices,
\[
\dmt{S}=\operatorname{diag}(\boldsymbol{\sigma}),
\qquad
\boldsymbol{\sigma}=[\sigma_1,\sigma_2,\ldots,\sigma_N]^T.
\]
Let us define
\begin{equation}
\label{eqn:ADF_moments_NLMS}
\begin{aligned}
\dmt{A}_{\mathrm{NLMS}}
&\triangleq
2\mathbb{E}\!\left[
\frac{\rv{u}_i\rv{u}_i^H}{\|\rv{u}_i\|^2}
\right], \\
\dmt{B}_{\mathrm{NLMS}}
&\triangleq
\mathbb{E}\!\left[
\frac{(\rv{u}_i\rv{u}_i^H)^T \odot (\rv{u}_i\rv{u}_i^H)}{\|\rv{u}_i\|^4}
\right], \\
\dmt{C}_{\mathrm{NLMS}}
&\triangleq
\mathbb{E}\!\left[
\frac{\rv{u}_i\rv{u}_i^H}{\|\rv{u}_i\|^4}
\right].
\end{aligned}
\end{equation}
The entries of \(\dmt{A}_{\mathrm{NLMS}}\), \(\dmt{B}_{\mathrm{NLMS}}\), and \(\dmt{C}_{\mathrm{NLMS}}\) determine the corresponding quantities required in the NLMS transient analysis.

Note that the diagonal entries of the general matrix \(\dmt{B}\) are generated by a coefficient matrix, denoted by \(\dmt{B}_{\mathrm{NLMS}}\), according to
\[
\operatorname{diag}(\dmt{B})=\dmt{B}_{\mathrm{NLMS}}\boldsymbol{\sigma}.
\]
Thus, \(\dmt{B}_{\mathrm{NLMS}}\) should not be identified with the general matrix \(\dmt{B}\); rather, it acts as the matrix that maps the diagonal weighting vector \(\boldsymbol{\sigma}\) to the diagonal of \(\dmt{B}\). Also, the scalar r \(c_{\mathrm{NLMS}}\) is the trace-weighted version of the  matrix \(\dmt{C}_{\mathrm{NLMS}}\), 
$$
c_{\mathrm{NLMS}} = \trace{\dmt{C}_{\mathrm{NLMS}}\diag(\sigma)}.
$$

To compute these matrices, it is convenient to work in the eigenbasis of the input covariance. Let
\[
\dmt{R}_{\rv{u}}=\dmt{\Lambda},
\qquad
\rv{u}_i=\dmt{\Lambda}^{1/2}\bar{\rv{u}}_i,
\qquad
\bar{\rv{u}}_i\sim\mathcal{CN}(\dv{0},\dmt{I}),
\]
where \(\dmt{\Lambda}=\operatorname{diag}(\lambda_1,\ldots,\lambda_N)\). Then
\[
\|\rv{u}_i\|^2
=
\bar{\rv{u}}_i^H\dmt{\Lambda}\bar{\rv{u}}_i
=
\|\bar{\rv{u}}_i\|_{\dmt{\Lambda}}^2.
\]
With this whitening, the moment matrices admit the factorizations
\begin{equation}
    \begin{aligned}
    \dmt{A}_{\mathrm{NLMS}}
    &=
    \dmt{\Lambda}^{1/2}\,
    \overline{\dmt{A}}_{\mathrm{NLMS}}\,
    \dmt{\Lambda}^{1/2},
    \quad
    \dmt{B}_{\mathrm{NLMS}}
    =
    \dmt{\Lambda}\,
    \overline{\dmt{B}}_{\mathrm{NLMS}}\,
    \dmt{\Lambda},
    \\
    \dmt{C}_{\mathrm{NLMS}}
    &=
    \dmt{\Lambda}^{1/2}\,
    \overline{\dmt{C}}_{\mathrm{NLMS}}\,
    \dmt{\Lambda}^{1/2},    
    \end{aligned}
\end{equation}
where
\begin{equation}
\label{eqn:ADF_moments_NLMS2}
\begin{aligned}
\overline{\dmt{A}}_{\mathrm{NLMS}}
&\triangleq
2\mathbb{E}\!\left[
\frac{\bar{\rv{u}}_i\bar{\rv{u}}_i^H}{\|\bar{\rv{u}}_i\|_{\dmt{\Lambda}}^2}
\right], \\
\overline{\dmt{B}}_{\mathrm{NLMS}}
&\triangleq
\mathbb{E}\!\left[
\frac{
(\bar{\rv{u}}_i\bar{\rv{u}}_i^H)^T \odot (\bar{\rv{u}}_i\bar{\rv{u}}_i^H)
}{
\left(\|\bar{\rv{u}}_i\|_{\dmt{\Lambda}}^2\right)^2
}
\right], \\
\overline{\dmt{C}}_{\mathrm{NLMS}}
&\triangleq
\mathbb{E}\!\left[
\frac{\bar{\rv{u}}_i\bar{\rv{u}}_i^H}{
\left(\|\bar{\rv{u}}_i\|_{\dmt{\Lambda}}^2\right)^2}
\right].
\end{aligned}
\end{equation}

The entries of these matrices can be expressed in terms of scalar random variables such as
\[
s_k
\triangleq
\frac{|\bar{u}(k)|^2}{\|\bar{\rv{u}}\|_{\dmt{\Lambda}}^2},
\qquad
s_{k\bar{k}}
\triangleq
\sqrt{\frac{\lambda_k}{\lambda_{\bar{k}}}}
\frac{|\bar{u}(k)|^2}{\|\bar{\rv{u}}\|_{\dmt{\Lambda}}^2}
+
\sqrt{\frac{\lambda_{\bar{k}}}{\lambda_k}}
\frac{|\bar{u}(\bar{k})|^2}{\|\bar{\rv{u}}\|_{\dmt{\Lambda}}^2},
\]
\[
q_k
\triangleq
\frac{|\bar{u}(k)|^2}{\left(\|\bar{\rv{u}}\|_{\dmt{\Lambda}}^2\right)^2},
\qquad
z_k
\triangleq
\frac{|\bar{u}(k)|^2+1}{\|\bar{\rv{u}}\|_{\dmt{\Lambda}}^2},
\qquad
r
\triangleq
\frac{1}{\|\bar{\rv{u}}\|_{\dmt{\Lambda}}^2}.
\]
As summarized in Table~\ref{tab:ADF_moments}, the entries of \(\overline{\dmt{A}}_{\mathrm{NLMS}}\), \(\overline{\dmt{B}}_{\mathrm{NLMS}}\), and \(\overline{\dmt{C}}_{\mathrm{NLMS}}\) can be expressed in terms of moments of these scalar random variables. These variables are ratios, or more generally rational functions, of quadratic forms, and this is precisely the structure that makes quadratic-form-ratio methods relevant in the transient analysis of adaptive filters.

\begin{table}[!ht]
 \caption{Relation of $\overline{\dmt{A}}_{\mathrm{NLMS}}$, $\overline{\dmt{B}}_{\mathrm{NLMS}}$, and $\overline{\dmt{C}}_{\mathrm{NLMS}}$ to moments of scalar random variables arising from quadratic-form denominators.}
    \renewcommand{\arraystretch}{1.5}
    \centering
    \begin{adjustbox}{max width=\textwidth}
    \begin{tabular}{|>{\centering\arraybackslash}m{2.6cm}|>{\raggedright\arraybackslash}m{9.2cm}|}
    \hline
    \textbf{Moment matrix entries} & \textbf{Relation to scalar random variables} \\ \hline

    $\overline{\dmt{A}}_{\mathrm{NLMS}}(k,k)$
    & $2\,\mathbb{E}[s_k]$ \\ \hline

    $\overline{\dmt{A}}_{\mathrm{NLMS}}(k,\bar{k})$, $k\neq \bar{k}$
    & $0$ \\ \hline

    $\overline{\dmt{B}}_{\mathrm{NLMS}}(k,k)$
    & $\mathbb{E}[s_k^2]$ \\ \hline

    $\overline{\dmt{B}}_{\mathrm{NLMS}}(k,\bar{k})$, $k\neq \bar{k}$
    & $\mathbb{E}[s_k s_{\bar{k}}]
    =
    \dfrac{1}{2}\!\left(
    \mathbb{E}[s_{k\bar{k}}^2]
    -\dfrac{\lambda_k}{\lambda_{\bar{k}}}\mathbb{E}[s_k^2]
    -\dfrac{\lambda_{\bar{k}}}{\lambda_k}\mathbb{E}[s_{\bar{k}}^2]
    \right)$ \\ \hline

    $\overline{\dmt{C}}_{\mathrm{NLMS}}(k,k)$
    & $\mathbb{E}[q_k]
    =
    \dfrac{1}{2}\mathbb{E}[z_k^2]
    -\dfrac{1}{2}\mathbb{E}[s_k^2]
    -\dfrac{1}{2}\mathbb{E}[r^2]$ \\ \hline

    $\overline{\dmt{C}}_{\mathrm{NLMS}}(k,\bar{k})$, $k\neq \bar{k}$
    & $0$ \\ \hline
    \end{tabular}
    \end{adjustbox}
    \label{tab:ADF_moments}
\end{table}

The identities in Table~\ref{tab:ADF_moments} were also verified numerically by Monte Carlo simulation with \(\bar{\rv{u}}\sim\mathcal{CN}(\dv{0},\dmt{I})\). The empirical errors were found to be at the level of numerical precision, up to the expected Monte Carlo fluctuations in entries that vanish theoretically.

Therefore, this application differs from the previous ones in an important way: the main task is not directly the evaluation of a PDF or a CDF, but rather the computation of moments of random variables built from ratios, or more generally rational functions, of quadratic forms. 
\section{Spectral Analysis and Periodograms}

\begin{figure}[!ht]
    \centering
    \subfloat[Parametric estimation.]{\includegraphics[width=0.7\textwidth]{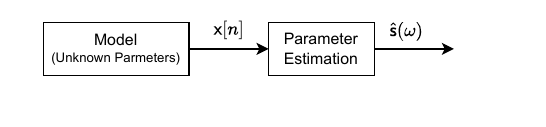}
    \label{subfig:parametric}}
    \hspace{0.02\textwidth}
    \subfloat[Nonparametric estimation.]{\includegraphics[width=0.7\textwidth]{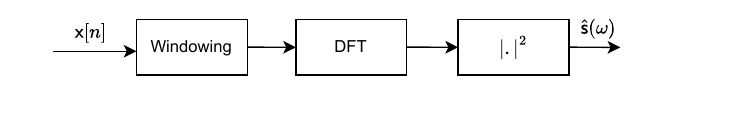}
    \label{subfig:nonparametric}}
    \caption{Power spectrum estimation.}
    \label{fig:PSD_estimation}
\end{figure}

Spectral analysis is concerned with estimating how the power of a stationary random process is distributed across frequency. It plays an important role in many areas, including signal processing, vibration monitoring, economics, meteorology, and astronomy. The central object of interest is the power spectral density (PSD), which characterizes the frequency-domain distribution of power in a time series.

Let \(\rs{x}[n]\) be a discrete-time, second-order stationary random process. The PSD may be defined as \cite{Stoica_SpectralAnalysisOfSignals_2005}
\[
s(\omega)
=
\lim_{N\to\infty}
\mathbb{E}\!\left[
\frac{1}{N}
\left|
\sum_{n=1}^N \rs{x}[n]e^{-j\omega n}
\right|^2
\right].
\]
If we denote the finite-record Fourier transform by
\[
\tilde{\rs{x}}_N(\omega)
\triangleq
\sum_{n=1}^N \rs{x}[n]e^{-j\omega n},
\]
then the same definition can be written as
\[
s(\omega)
=
\lim_{N\to\infty}
\mathbb{E}\!\left[
\frac{1}{N}
\left|
\tilde{\rs{x}}_N(\omega)
\right|^2
\right].
\]
Under standard regularity conditions, this definition is equivalent to the Fourier transform of the autocorrelation sequence,
\[
s(\omega)=\sum_{\ell=-\infty}^{\infty} r(\ell)e^{-j\omega \ell}.
\]
Hence, the spectral estimation problem is to estimate \(s(\omega)\) from a finite observation record \(\{\rs{x}[1],\rs{x}[2],\ldots,\rs{x}[N]\}\).

Methods for PSD estimation are commonly divided into two broad classes: parametric and nonparametric methods (Fig.~\ref{fig:PSD_estimation}). In parametric methods, one postulates a stochastic model for the observed process, estimates the model parameters from the data, and then computes the PSD from the fitted model. In nonparametric methods, no explicit model is assumed; instead, the PSD is estimated directly from the observed data, typically through Fourier-based procedures. The relative performance of the two approaches depends on how well the assumed parametric model matches the true underlying process.

A basic nonparametric estimator of the PSD is the periodogram, defined by
\begin{equation}
\label{eqn:Periodogram}
\hat{\rs{s}}(\omega)
=
\rs{g}(\omega)
=
\frac{1}{N}
\left|
\sum_{n=1}^N \rs{x}[n]e^{-j\omega n}
\right|^2
=
\frac{1}{N}
\left|
\tilde{\rs{x}}_N(\omega)
\right|^2.
\end{equation}
When the angular frequencies are discretized as
\[
\omega_k=\frac{2\pi k}{N},
\qquad
k=0,1,\ldots,N-1,
\]
the periodogram ordinates \(\hat{\rs{s}}(\omega_k)\) are asymptotically independent, and for \(\omega_k\neq 0,\pi\) they are asymptotically distributed as scaled chi-square random variables with two degrees of freedom. More precisely,
\[
\hat{\rs{s}}(\omega_k)
\overset{a}{\sim}
s(\omega_k)\,\frac{\chi_2^2}{2}.
\]
Thus, although the periodogram is simple and widely used, it is not a consistent estimator of the PSD, since its variance does not vanish as the sample size increases \cite{Brillinger_TimeSeriesDataAnalysisAndTheory_2001}.

To reduce this variability, one may average neighboring periodogram ordinates. A simple smoothed periodogram is given by
\begin{equation}
\label{eqn:smoothed_periodogram}
\hat{\rs{s}}(\omega_k)
=
\frac{1}{2m+1}
\sum_{j=-m}^{m}
\rs{g}\!\left(\frac{2\pi(k+j)}{N}\right).
\end{equation}
This averaging improves stability, and the resulting estimator is asymptotically distributed as a scaled chi-square random variable with \(4m+2\) degrees of freedom \cite[Theorem 5.4.3]{Brillinger_TimeSeriesDataAnalysisAndTheory_2001}. Hence, smoothing transforms the raw periodogram from a scaled \(\chi_2^2\)-type statistic into a higher-degree chi-square statistic with reduced variance.

A more general approach uses unequal weights \cite[Theorem 5.5.3]{Brillinger_TimeSeriesDataAnalysisAndTheory_2001},
\begin{equation}
\label{eqn:weighted_periodogram}
\hat{\rs{s}}(\omega_k)
=
\sum_{j=-m}^{m}
w_j\,\rs{g}\!\left(\frac{2\pi(k+j)}{N}\right),
\qquad
\sum_{j=-m}^{m} w_j = 1.
\end{equation}
In this case, the estimator becomes a weighted sum of asymptotically independent scaled chi-square random variables. Therefore, spectral estimation through weighted periodograms provides a direct and practically important example of a random quantity that is distributed as a weighted combination of chi-square variables.

\section{Reliability Analysis}

When designing a mechanical, civil, or electrical structure, the designer may resort to a set of parameters: lengths, cross-sectional areas, Young's modulus, resistance, capacitance, etc. We collect these parameters in a vector $x = \begin{bmatrix}x_1, x_2, x_3 \ldots x_N\end{bmatrix}^T$. A system might fail if stresses or forces exceed yield strength, temperatures exceed safe limits, deflections exceed design limits, currents exceed carrying capacity ... etc. We call the quantities that determine failure \emph{performance functions} or \emph{limit state functions}. Such functions are expressed in terms of the former parameters. If we have $M$ of these functions we will call them $g_i(x)\; \forall \;i \in [1,M]$. Manufacturing tolerances, defects, and other issues cause the desired design values of the parameters not to be achieved. Instead, we know them ``probabilistically". For example, the length of a screw might be randomly distributed, as well as the capacitance of a capacitor, the resistance of a resistor, transistor transconductance parameters, concrete Young's Modulus and cross-sectional area, tensile strength of carbon treated steel \cite{billurChallengesSteels2010} (see Fig. \ref{Fig:Steel}), Young's modulus of recycled aggregate concrete $E_{RAC}$, etc. \cite{chenPractcalEquationRAC2022} (see \ref{Fig:Concrete}).

For example, consider an RC circuit used to filter a dual-tone signal. The gain of the signal at frequency $\omega$ is given by:
$$
\frac{|V_{out}|}{|V_{in}|} = \frac{1}{\sqrt{1 + (\omega R C)^2}} 
$$
Assume the tones are 697 Hz and 1209 Hz. We want the higher tone to be attenuated by 1/4th. Hence, we select the RC components with values $R$ and $C$. What is the probability that, due to tolerances, the gain is above $1/4$? Assume $\sigma_R = 0.1R$ and $\sigma_C = 0.05C$.

\begin{figure}[ht!]

    \centering
    \begin{tikzpicture}
    \draw (0,0) -- (2,0);
    \draw (-1,-0.2) -- (-1,0.2) -- (-0.2,0.2) -- (0,0) -- (-0.2,-0.2) -- (-1,-0.2);
    \draw (2,0) -- (2.15,0.35) -- (2.45,-0.35) -- (2.75,0.35) -- (3.05,-0.35) -- (3.35,0.35) -- (3.65,-0.35) -- (3.95,0.35) -- (4.25,-0.35) -- (4.5,0) -- (7,0) -- (7,-2.8);
    \draw (7,-3.2) -- (7,-5.5) -- (7.75,-5.5) -- (7,-6) -- (6.25,-5.5) --(7,-5.5);
    \draw (6,-2.8) -- (8,-2.8);
    \draw (6,-3.2) -- (8,-3.2);
    \node at (3.25,-0.7) {R};
    \node at (5.5,-3) {C};
    \draw[-{Latex[length=3mm]}] (-1.2,-5.5) -- (-1.2,0);
    \node at (-0.8,-2.75) {$V_{\text{in}}$};
    \draw[-{Latex[length=3mm]}] (8.2,-5.5) -- (8.2,0);
    \node at (8.6,-2.75) {$V_{\text{out}}$};
\end{tikzpicture}
    \caption{\textbf{Concrete Example 1}}
\end{figure}
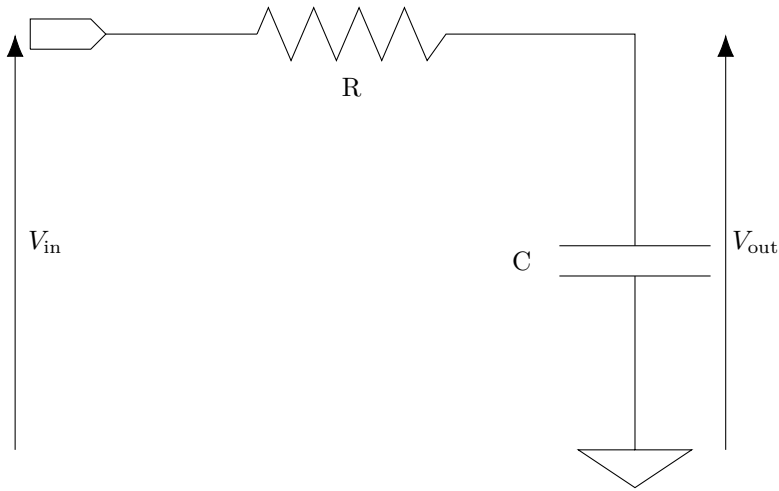

\begin{figure}[ht!]
    \centering
    \includegraphics[scale=0.25]{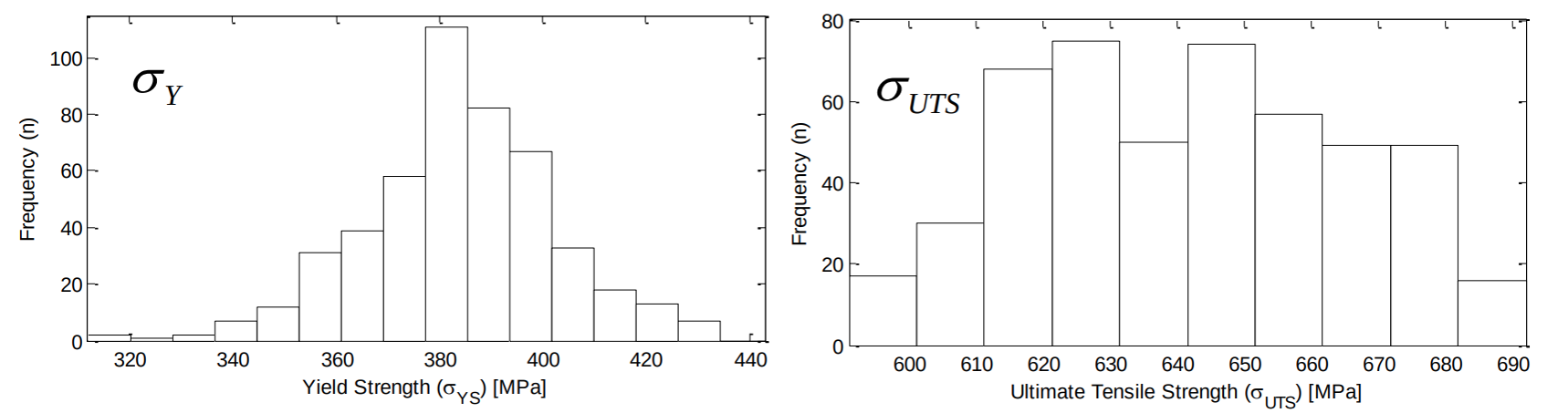}
    \caption{Distribution of yield and tensile strength distributions (in MPa) for DP590 GI material.}
    \label{Fig:Steel}
\end{figure}

\begin{figure}[ht!]
    \centering
    \includegraphics[scale=0.3]{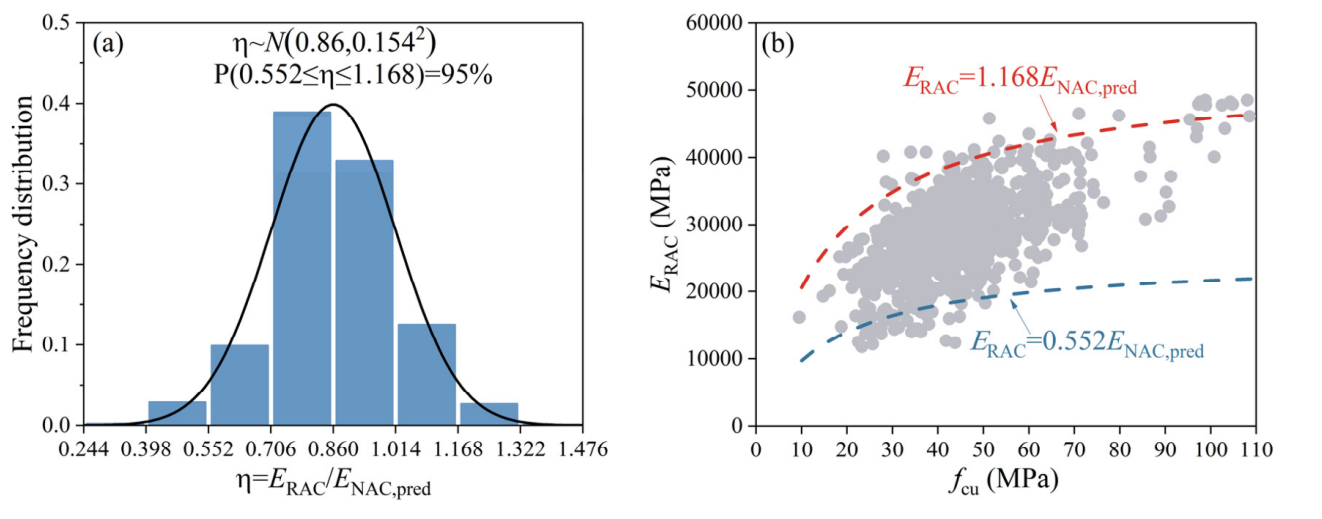}
    \caption{ Distribution of $\eta$ and $E_{\text{RAC}}$ (a) $ \eta \sim \mathcal{N}(0.86,0.154^2)$; (b) $ 1.168E_{\text{NAC,pred}} \ge E_{\text{RAC}} \ge 0.552E_{\text{NAC,pred}}$  }.
    \label{Fig:Concrete}
\end{figure}

    Assuming we have the distribution of each design parameter, we can try to evaluate the probability that one of the performance functions exceeds its failure limit (to calculate the probability of failure). Alternatively, we can evaluate the probability that it stays within bounds to evaluate reliability. $$\begin{aligned}P_f &= \mathbb{P}[g_i(\rs{x}) < 0 ] \\ R &= \mathbb{P}[g_i(\rs{x}) > 0] = 1 - P_f\end{aligned}$$
    Note that since $\rs{x}$ is random, $g_i(\rs{x})$ is a function of a random variable. Let the joint distribution of the parameters $\rs{x}$ be $f_{\rs{x}}(x)$. Define the failure set as: $$ \Omega_i = \{x \in \mathbb{R}^N | g_i(x) < 0\} $$
    Then the probability of a single failure is given by: $$P_{f, i} = \int_{\Omega_i} f_{\rs{x}}(x) dx $$
    The set that defines all failures is simply: $$\Omega = \bigcup_{i} \Omega_i$$ The probability of failure is then: $$P_{f} = \int_{\Omega} f_{\rv{x}}(\dv{x}) d\dv{x} $$
    These integrals are difficult to evaluate, particularly when the distributions are not independent (or not nice!!). However, both the distribution and the domain are so "ugly". Hence, the solutions are to Gaussianize the distribution function and to transform the domain into a manageable one. This can be justified as follows. Distributions usually decay quickly (exponentially). Hence, using approximations close to peaks is good! The domains, although ugly, mostly don't contribute. Thus, taking an easier shape is not very bad. We can make sure to over-approximate such that we "exaggerate" the probability of failure.
     \begin{figure}
        \centering
        \subfloat[Probability integration in $X$-space]{\includegraphics[width=0.49\linewidth]{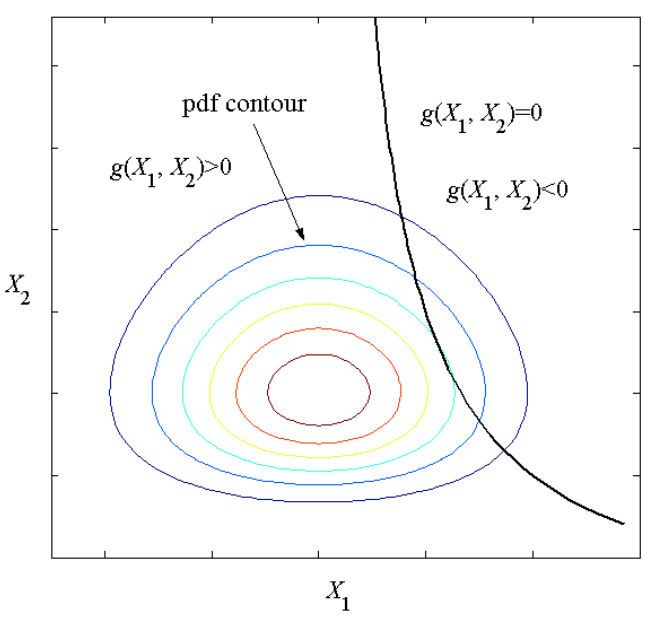}   \label{subfig:1}}
         \hfil
          \subfloat[Probability integration in $U$-space ]{\includegraphics[width=0.49\linewidth]{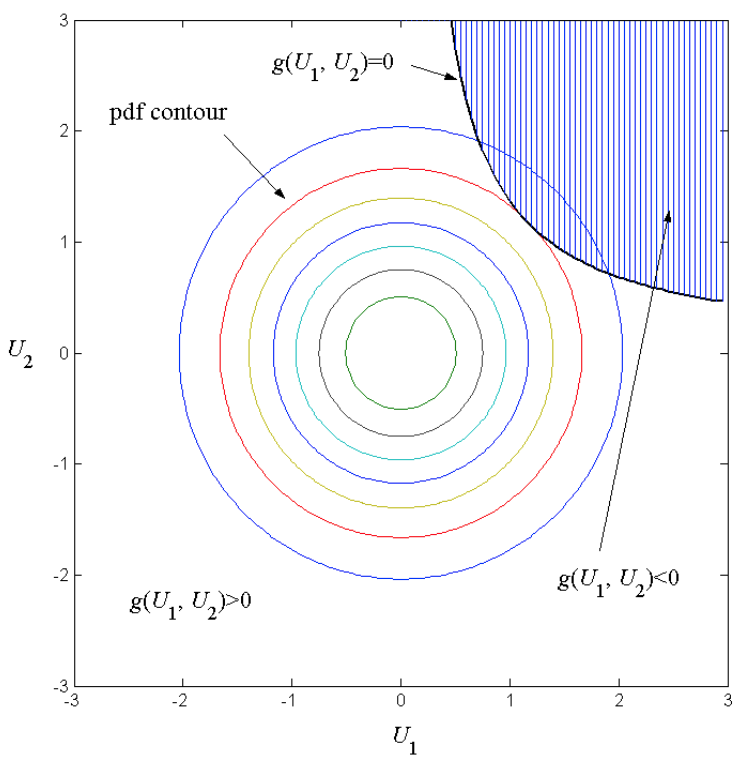}%
        	\label{subfig:2} }
        \caption{Probability Integration in First Order Method \cite{duProbabilisticEngineeringDesign2005}}
    \end{figure}

    To Gaussianize the variables, the \emph{NATAF Transformation}/\emph{Rosenblatt Transformation} is applied. The process starts by applying the CDF on a random variable to obtain a uniform distribution, i.e., $F_{\rs{x}_i}(\rs{x}_i)\sim \mathcal{U}(0,1)$. The inverse CDF of the normal distribution $\Phi^{-1}$ is then applied to obtain a standard normal variable. The mean and variance are ``preserved" for later use. In summary take the parameter $\rs{x}_i$ and transform to $\rs{u}_i$: $$\rs{u}_i = \Phi^{-1}[F_{\rs{X}_i}(\rs{x}_i)]$$
    \begin{figure}[ht!]
            \centering
        \includegraphics[scale=0.35]{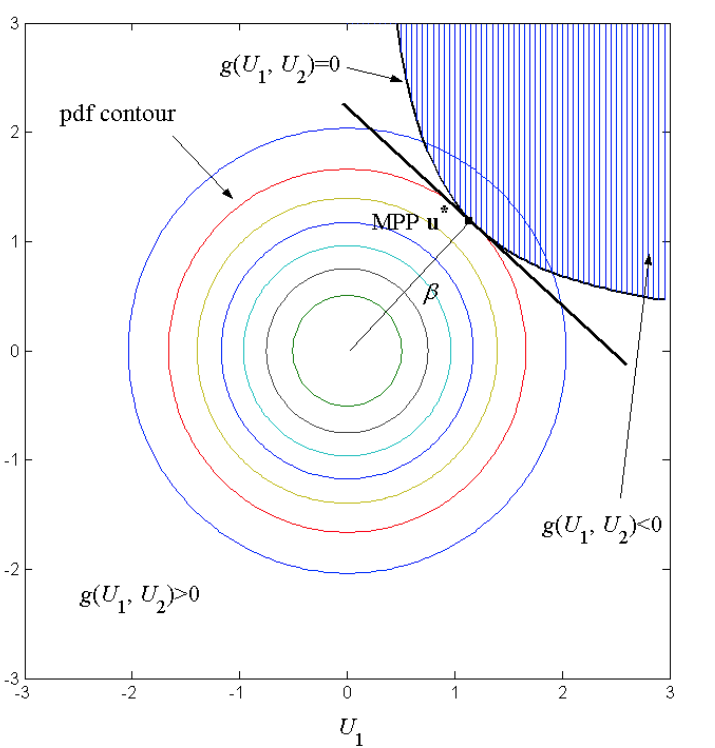}
        \caption{Probability Integration in First Order Method \textcolor{red}{\cite{duProbabilisticEngineeringDesign2005}}}
    \end{figure}
    To make the domain nicer, we can use either Linear domains or Quadratic domains (higher-order polynomials are probably intractable.). We will focus on the quadratic domains in this text. Assume we have one smooth limit state-function and $N$ parameters. To quadratically approximate the domain, the set of failure is approximated using the following Taylor series expansion:
    \[
    \resizebox{\textwidth}{!}{$
    \begin{aligned} 
    \Omega&=\{\dv{x}\in\mathbb{R}^N:g(\dv{x}) < 0\} \\ 
    &=\{\dv{u}\in\mathbb{R}^N:g(\dv{u}_0) + \nabla g(\dv{u}_0)(\dv{u} - \dv{u}_0) + \frac{1}{2}(\dv{u} - \dv{u}_0)^T\nabla^2g(\dv{u}_0)(\dv{u} - \dv{u}_0) + \ldots < 0\} \\
    & \approx \{\dv{u}\in\mathbb{R}^N: g(\dv{u}_0) + \nabla g(\dv{u}_0)(\dv{u} - \dv{u}_0) + \frac{1}{2}(\dv{u} - \dv{u}_0)^T\nabla^2g(\dv{u}_0)(\dv{u} - \dv{u}_0) < 0\}
    \end{aligned}
    $}
    \]
    where $\dv{u}_0$ is an arbitrary point in $\mathbb{R}^N$. The choice of this point is vital for the accuracy of the approximation. We should choose $u_0$ to be the point with the highest probability on the curve $g(u) = 0$. Since probabilities decay quickly, taking the maxima should allow us to account for most of the probability mass. In mathematical notation what we are describing is: Set $\dv{u}_0 = \dv{u}^*$ where: \begin{equation}
        \dv{u}^* = \max_{g(\dv{u}) = 0} f_{\rv{u}}(\dv{u}) = \max_{g(\dv{u}) = 0} c e^{-\|\dv{u}\|^2/2} = \min_{g(\dv{u}) = 0} \|\dv{u}\|^2 \label{eq:reliability-optim}
    \end{equation} 
    The last two inequality manipulations are obvious steps saying that for standard multivariate normal, the point with the highest probability is the closest to the origin with quadratic exponential decay away from the origin.
    
    Hence we have approximated the probability of failure by the following function: \begin{equation} P_f = \mathbb{P}\left[g(\dv{u}^*) + \nabla g(\dv{u}^*)^T(\rv{u} - \dv{u}^*) + \frac{1}{2}(\rs{u} - \dv{u}^*)^T\nabla^2g(\dv{u}^*)(\rv{u} - \dv{u}^*) < 0 \right] \label{eq:reliability-faliure-prob}\end{equation} 
    Therefore, the probability of failure is approximated by the CDF at zero of a quadratic form in Gaussian random variables. 

    A gradient descent-type algorithm with some special termination criteria is used to efficiently calculate $\dv{u}^*, \nabla g(\dv{u}^*), \nabla^2 g(\dv{u}^*)$ used in Equations \eqref{eq:reliability-optim} and \eqref{eq:reliability-faliure-prob}. Moreover, saddlepoint approximation is employed to ``efficiently" calculate the probability of Equation \eqref{eq:reliability-faliure-prob}.

\section{Inference in Genome-Wide Association Studies}

\begin{figure}[!ht]
    \centering
    \includegraphics[width=0.7\linewidth]{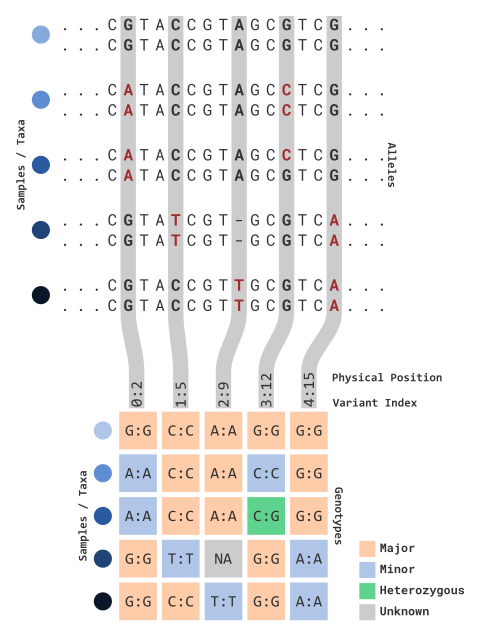}
    \caption{Illustration of genotype coding across multiple loci and samples.}
    \label{fig:genome}
\end{figure}

Genome-wide association studies (GWAS) aim to identify genetic variants associated with a disease or phenotype. In particular, they examine variations in the genome known as single nucleotide polymorphisms (SNPs) and assess whether these variations are statistically associated with the presence of a disease, a treatment response, or another trait of interest. GWAS has become an important tool in studying complex diseases, since individual variants may have small effects, but their collective contribution can still be substantial.

To describe the data structure, consider \(n\) individuals observed across \(p\) genetic variants. At each genomic position, called a locus, there are typically two possible alleles. These combine to form three possible genotypes, corresponding to the two homozygous cases and the heterozygous case. In GWAS, the genotype data are commonly organized into a matrix \(\dmt{G}\in\mathbb{R}^{n\times p}\), where each row corresponds to one individual and each column corresponds to one SNP. The entry \(g_{ij}\) records the number of copies of the minor allele carried by individual \(i\) at variant \(j\), so that the genotype is coded as \(0\), \(1\), or \(2\). The phenotype information is collected in a vector \(\rv{y}\), where, for a binary trait, \(y_i=1\) indicates that the \(i\)-th individual has the disease and \(y_i=0\) otherwise.

To illustrate some of the common terms used in genetic analysis, it is useful to refer to Fig.~\ref{fig:genome}. The figure shows the genetic variation across multiple samples, where each row corresponds to an individual and each highlighted genomic position corresponds to a SNP locus. At each locus, there are two possible alleles, which combine to form three possible genotypes. In the lower part of the figure, these genotypes are displayed in a color-coded grid, illustrating the construction of the genotype matrix from the observed genetic variants.

A standard way to relate phenotype to genotype is through logistic regression. Let \(\dv{x}_i\) denote a vector of covariates for individual \(i\), such as age, sex, or other non-genetic factors, and let \(\dv{g}_i\) denote the \(i\)-th row of the genotype matrix. Then the logistic model is written as
\[
\log\frac{\mathbb{P}(y_i=1)}{1-\mathbb{P}(y_i=1)}
=
\alpha_0+\dv{x}_i^T\dv{\alpha}+\dv{g}_i^T\dv{\beta},
\]
where \(\alpha_0\) is an intercept, \(\dv{\alpha}\) contains the regression coefficients associated with the non-genetic covariates, and \(\dv{\beta}\) contains the regression coefficients associated with the genetic variants. The main inferential question is whether the genetic effects are null, namely,
\[
H_0:\dv{\beta}=\dv{0},
\qquad
H_1:\dv{\beta}\neq \dv{0}.
\]
Thus, GWAS leads naturally to a hypothesis-testing problem.

Quadratic forms arise here through the test statistics used for inference. In many classical settings, statistics such as the likelihood-ratio statistic, Wald statistic, and Rao score statistic are asymptotically chi-square distributed. For example, under a logistic fixed-effects model, Rao’s score statistic may be written as
\[
\rs{T}_n
=
\rs{U}_\beta^T I_{\beta\beta}^{-1}\rs{U}_\beta,
\]
where \(\rs{U}_\beta\) is the score function with respect to \(\dv{\beta}\), and \(I_{\beta\beta}\) is the corresponding Fisher information matrix. Under the null hypothesis, this statistic is asymptotically distributed as \(\chi_p^2\).

In GWAS, however, mixed-effects formulations are often more useful, especially when many variants are tested jointly. In the logistic mixed-effects model, the vector \(\dv{\beta}\) is treated as random, and inference is based on a variance-component score statistic. In that case, under the null hypothesis, the test statistic can be written as
\[
\rs{T}_n
=
(\rv{y}-\hat{\dv{\mu}})^T \dmt{K} (\rv{y}-\hat{\dv{\mu}}),
\]
where \(\hat{\dv{\mu}}\) is the fitted mean vector under the null model and \(\dmt{K}=\dmt{G}\dmt{W}\dmt{G}^T\) is a kernel matrix constructed from the genotype matrix and a weight matrix \(\dmt{W}\). Moreover, under the null hypothesis,
\[
\rs{T}_n
\sim
\sum_{i=1}^n \lambda_i \chi_{1,i}^2,
\]
where \(\lambda_i\) are the eigenvalues of
\[
\dmt{P}_0^{1/2}\dmt{K}\dmt{P}_0^{1/2},
\]
with
\[
\dmt{P}_0
=
\dmt{V}
-
\dmt{V}\tilde{\dmt{X}}
\left(\tilde{\dmt{X}}^T\dmt{V}\tilde{\dmt{X}}\right)^{-1}
\tilde{\dmt{X}}^T\dmt{V},
\qquad
\tilde{\dmt{X}}=[\dv{1},\,\dmt{X}],
\]
and
\[
\dmt{V}
=
\operatorname{diag}\!\bigl(
\hat{\mu}_1(1-\hat{\mu}_1),\ldots,\hat{\mu}_n(1-\hat{\mu}_n)
\bigr).
\]
Here,
\[
\hat{\mu}_i
=
\frac{\exp(\hat{\alpha}_0+\dv{x}_i^T\hat{\dv{\alpha}})}
{1+\exp(\hat{\alpha}_0+\dv{x}_i^T\hat{\dv{\alpha}})},
\]
and the \(\chi_{1,i}^2\) are independent chi-square random variables with one degree of freedom.

Therefore, GWAS provides a practically important example in which statistical inference reduces to evaluating the distribution of a quadratic-form statistic, or equivalently, a weighted sum of independent chi-square random variables. 
\section{Transmissibility of Civil and Mechanical Structures}

Engineers may be interested to assess the robustness of some system (machine, building, ...), then, for instance, they may decide to detect some damage or evaluate the effects of deterioration. Vibration-based techniques are system identification technologies that may give insight into the current state of a system, say the structural properties of a machine part or a structural element in a building. Predetermined artificial excitation is available in some cases, while the only available excitation are unknown natural disturbances in other cases. The latter may be due to safety measures, for example, a significant disturbance may lead to a partial or total failures.  Whenever an element is subject to a predetermined excitation, output-input relations can be measured: for example, the frequency response function (FRF) which is the ratio of output to input in their frequency domains. However, in the case of natural excitation, access to output at different positions allows the estimation of what is known as transmissibility: the ratio of output at a position $k$ to that at position $\ell$. In addition to measurement errors, the methods used to estimate the spectral properties of the output are not perfect. Ratios of correlated chi-squares arise in the study of errors arising from such a spectral estimation. 

Consider an ideal linear time-invariant SIMO model (Fig. \ref{fig:SIMO_Trasmissibility}) in which an excitation $\rs{x}(t)$ is applied, and two outputs $\rs{y}_k(t)$ and $\rs{y}_\ell(t)$ are available for measurement.  The measurements are modeled by the following equalities: 
\begin{equation}
\begin{aligned}
    \rs{y}_k(t)&=(h_k*\rs{x})(t)+\rs{n}_k(t)=\rs{v}_k(t)+\rs{n}_k(t)\\
    \rs{y}_\ell(t)&=(h_\ell*\rs{x})(t)+\rs{n}_\ell(t)=\rs{v}_\ell(t)+\rs{n}_\ell(t)
\end{aligned} \tag{*} \label{Eq:Measurement}
\end{equation}
Here, $\rs{v}_k$ and $\rs{v}_\ell$ are the true responses (velocity, acceleration...) at positions $k$ and $\ell$ respectively, $h_k(t)$ and $h_\ell(t)$ are the corresponding impulse responses, and $\rs{n}_k(t)$ and $\rs{n}_\ell(t)$ are the corresponding additive noises. In the case we are reporting \cite{maoModelQuantifyingUncertainty2012}, $v_k(t)$ is considered to be an acceleration. 
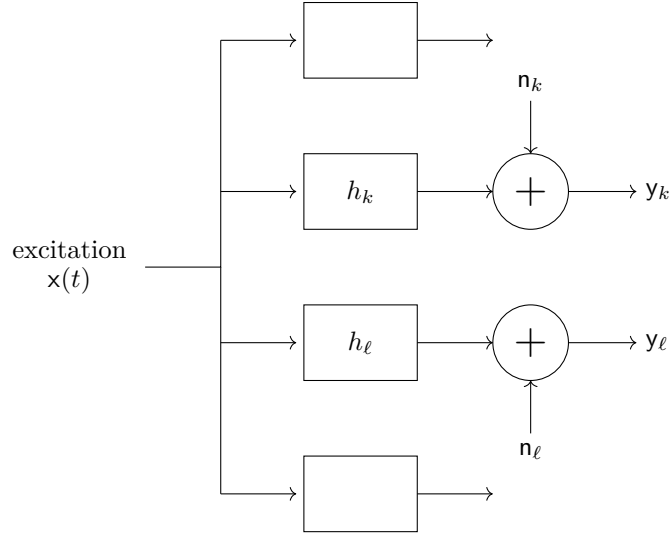
\begin{figure}[ht!]
    \centering
    \begin{tikzpicture}
    \node[text width=2cm,align=center] at (0,0) {excitation  $\rs{x}(t)$};
    \draw (1,0) -- (2,0) -- (2,3);
    \draw[->] (2,3) -- (3,3);
    \draw[->] (2,1) -- (3,1);

    \draw (2,0) -- (2,-3);
    \draw[->] (2,-3) -- (3,-3);
    \draw[->] (2,-1) -- (3,-1);

    \draw (3.1,3.5) rectangle (4.6,2.5);
    \draw (3.1,1.5) rectangle (4.6,0.5);
    \node at (3.85,1) {$h_k$};
    \draw (3.1,-1.5) rectangle (4.6,-0.5);
    \node at (3.85,-1) {$h_\ell$};
    \draw (3.1,-3.5) rectangle (4.6,-2.5);

    \draw[->] (4.6,3) -- (5.6,3);
    \draw[->] (4.6,1) -- (5.6,1);
    \draw (6.1,1) circle (0.5) node [anchor=center]{\LARGE +};
    \draw[->] (6.1,2.2)  node[anchor=south] {$\rs{n}_k$} -- (6.1,1.5);
    \draw[->] (6.6,1) -- (7.5,1) node[anchor=west]{$\rs{y}_k$};
    \draw[->] (4.6,-1) -- (5.6,-1);
    \draw (6.1,-1) circle (0.5) node [anchor=center]{\LARGE +};
    \draw[->] (6.1,-2.2)  node[anchor=north] {$\rs{n}_\ell$} -- (6.1,-1.5);
    \draw[->] (6.6,-1) -- (7.5,-1) node[anchor=west]{$\rs{y}_\ell$};
    \draw[->] (4.6,-3) -- (5.6,-3);
    
\end{tikzpicture}
    \caption{Ideal SIMO Model}
    \label{fig:SIMO_Trasmissibility}
\end{figure}

Moreover, it is preferable to work in the frequency domain \cite{brinckerModalIdentification2001}. Denote by capital letters the Fourier transforms of the lower-case functions, for example, $\mathcal{F}(\rs{v}_k(t))(\omega)=\rs{V}_k(\omega)$. The local acceleration transmissibility is defined as $$T_{k\ell}(\omega):=\dfrac{V_k(\omega)}{V_\ell(\omega)}$$
Its modulus is then $|\rs{T}_{k\ell}|=\sqrt { \dfrac{|\rs{V}_k|^2}{|\rs{V}_\ell|^2} }=\sqrt{\dfrac{\rs{V}_k^*\rs{V}_k}{\rs{V}_\ell^*\rs{V}_\ell}}$. Mao and Todd claim this that we have $$|\rs{T}_{k\ell}|=\sqrt{\dfrac{G_{\rs{v}_k\rs{v}_k}}{G_{\rs{v}_\ell \rs{v}_\ell}}}$$ where $G_{\rs{v}_k\rs{v}_k}$ is the (one-sided) auto-power spectral density function of $\rs{v}_k$. 

Consider the noiseless case, i.e., $n_k=n_\ell=0$. The Fourier transform cannot be computed over the infinite-time domain since we can only access the signal at a finite interval of time. So this estimate will come with its own error even in the absence of the additive noise. 

Denote by $\rs{V}_k(\omega,T)$ the finite Fourier transform $$\rs{V}_k(\omega,T)=\int_0^T v_k(t) e^{-j\omega t}dt$$ If $\rs{v}_k(t)$ is Gaussian, $\rs{V}_k(\omega,T)$ will be a complex Gaussian variable. Now clearly, this is the case if $\rs{x}(t)$ is a Gaussian process. Mao and Todd \cite{maoModelQuantifyingUncertainty2012} note that even if the input $\rs{x}(t)$ is not Gaussian, the output $\rs{v}(t)$ can be approximated by a Gaussian using central limit theorem. The one-sided auto-power spectral density can be written as \cite{bendatRandomData2011}: $$G_{\rs{v}_k\rs{v}_k}(\omega)=2\lim_{T\to \infty}\dfrac{\mathbb{E}\left[|V_k(\omega,T)|^2\right]}{T}$$
Now a rough estimate of this would be obtained by dropping the expectation and the limit: $\tilde{G}_{\rs{v}_k\rs{v}_k}=\dfrac{2|\rs{V}_k|^2}{T}$. According to Bendat \cite{bendatRandomData2011}, $\tilde{G}_{\rs{v}_k\rs{v}_k}$ is related to $G_{\rs{v}_k\rs{v}_k}$ through a multiplicative noise, i.e., $\tilde{G}_{\rs{v}_k\rs{v}_k}\sim\frac12\chi_2^2G_{\rs{v}_k\rs{v}_k}$ whenever $\rs{v}_k(t)$ is a Gaussian process.

Clearly, $\tilde{G}_{\rs{v}_k\rs{v}_k}$ is not a good estimate as $\operatorname{Var}(\frac12\chi_2^2)=1$. Hence, a better estimate $\hat{G}_{\rs{v}_k\rs{v}_k}$ is constructed by averaging over an ensemble of estimates over $N_d$ disjoint records: $$\hat{G}_{\rs{v}_k\rs{v}_k}(\omega)=\dfrac{2}{N_dT}\sum_{s=1}^{N_d}|V_{k,s}(\omega,T)|^2$$ Evidently, $\hat{G}_{v_kv_k}\sim\frac{1}{2N_d}\chi_{2N_d}^2G_{v_kv_k}$. Clearly, this estimate is better as the variance of the coefficient is now $\frac{1}{N_d}$. 

Going back to the transmissibility function $T_{kj}(\omega)$, it can now be approximated by: $$\hat{T}_{k\ell}=\dfrac{\hat{G}_{\rs{v}_k\rs{v}_k}}{\hat{G}_{\rs{v}_\ell \rs{v}_\ell}}$$This estimate is a ratio of two dependent chi-squares (hence a ratio of quadratic forms). Evaluating the distribution of such a ratio may give insight into the accuracy of an estimate, as well as the number of sub-records $N_d$ needed to achieve a predetermined precision.

We conclude this section by presenting \Cref{tab:applicationsOfQFGRVsInLiterature}, which summarizes some representative applications. The table lists the corresponding fields, specific application, relevant references, and quantities of interest (QoI). It shows that quadratic forms arise in a wide range of problems, from communication and array processing to reliability analysis, genetics, and structural dynamics. Thus, the applications surveyed in this section motivate the need for a unified theoretical and computational treatment developed in the following sections.
\begin{landscape}
\footnotesize
\setlength{\tabcolsep}{3pt}
\renewcommand{\arraystretch}{1.1}
\setlength\LTleft{0pt}
\setlength\LTright{0pt}

\begin{tabularx}{\linewidth}{>{\RaggedRight\arraybackslash}p{2.2cm} Y Y >{\RaggedRight\arraybackslash}p{1.5cm}}
\caption{Applications of Quadratic Forms in the Literature}
\label{tab:applicationsOfQFGRVsInLiterature} \\

\toprule
Field/Venue & Applications/Methods & Ref. & QoI \\
\midrule
\endfirsthead

\multicolumn{4}{c}{{\tablename\ \thetable{} -- continued from previous page}} \\
\toprule
Field/Venue & Applications/Methods & Ref. & QoI \\
\midrule
\endhead

\midrule
\multicolumn{4}{r}{{Continued on next page}} \\
\endfoot

\bottomrule
\endlastfoot

Signal Processing 
& Direction of Arrival 
& \cite{athleyThresholdRegionPerformance2005} Eq.~15 (single source), Eq.~34 (multi source) 
& CDF \\

Signal Processing 
& Transient Analysis in adaptive filters 
& \cite[Table~1]{al-naffouriMeansquareAnalysisNLMS2021} 
& Moments \\

Signal Processing 
& Spectral Analysis \& Periodograms 
& \cite{gerkmannStatisticsSprctralAmplitudes2009} Eq.~15; \cite{gismallaPeriodogramEnergyDetection2011,gismallaPerformanceEnergyDetectionBartlett2012} Eq.~8 
& PDF, CDF \\

Signal Processing 
& Spectrum Sensing 
& \cite{raghavanFalseAlarm2018} Secs.~3, 4 
& CDF \\

Signal Processing 
& Adaptive Matched Filters 
& \cite{bessonAnalysisSNRLossDistribution2020} Eq.~22 
& PDF \\

Signal Processing 
& Differential Modulation 
& \cite{zhaoDifferentialModulation2007} App.~II 
& CDF \\

Signal Processing 
& Matched Field Processing 
& \cite{gallMatchedField2014} 
& CDF \\

Communications 
& MIMO Combining Techniques (SNR of MRC, SC, EGC) 
& \cite{beaulieuNovelSimpleRepresentations2011},
  \cite{ropokisQuadraticFormsNormal2008},
  \cite{peppasTrivariateNakagamimDistribution2009},
  \cite{morales-jimenezDiagonalDistributionComplex2011},
  \cite{bithasNovelResultsMultivariate2019}
& CDF, PDF \\

Communications 
& Beamforming 
& \cite{hassanPerformancePerformanceAnalysisOfBeamforming2017},
  \cite[Eq.~4]{kimSINRRandomBeamforming2013},
  \cite[Sec.~2, Eq.~10]{schlegelErrorProbabilityMultibeam1996}
& CDF \\

Communications 
& BER, SER, Capacity and Outage (SISO and MIMO) 
& \cite{raphaeliDistributionNoncentralIndefinite1996},
  \cite{al-naffouriDistributionIndefiniteQuadratic2016},
  \cite{wiegandSeriesApproximationsRayleigh2019},
  \cite[Sec.~III.C]{khammassiNewAnalyticalApproximation2022},
  \cite{simonDifferenceTwoChisquare2001},
  \cite{beaulieuNovelRepresentationsBivariate2011},
  \cite{beaulieuNovelSimpleRepresentations2011},
  \cite[App.~B]{proakisDigitalCommunications2008},
  \cite{tekinayMomentsQuadrivariateRayleigh2020},
  \cite{dharmawansaNewSeriesRepresentation2009}
& CDF, PDF \\

Statistics 
& Goodness of Fit Tests (for different models) 
& \cite[Secs.~7.2--7.5]{mathaiQuadraticFormsRandom1992} 
& CDF \\

Statistics 
& Multivariate Goodness of Fit Tests 
& \cite[Sec.~2]{dickhausSurveyMultivariateChisquare2015} 
& CDF \\

Statistics 
& Distribution of Empirical Quantities: Mean \& Variance, linear regressors, sampled autocovariances, Mahalanobis distance 
& \cite[Eqs.~3,4]{schoneJointDistributionQuadratic2000},
  \cite{blacherMultivariateQuadraticForms2003},
  \cite{andersonMomentsSampleAutocovariances1990},
  \cite{royenMultivariateGammaDistributions1991}
& PDF, CDF, Moments \\

Genetics 
& Association Tests (Genome Wide, Phenome, Sequence Kernel, etc.) 
& \cite[Eq.~3]{tongEfficientCalculation2010},
  \cite[Ch.~3]{xiongBigDataInOmicsAndImagingAssociationAnalysis2017},
  \cite[Eq.~3]{wuRarevariantAssociationTestingForSequencing2011},
  \cite{zhouEfficientlyControllingGenetic2018},
  \cite{crawfordDetectingEpistasisGenetic2017},
  \cite[Sec.~4]{dickhausSurveyMultivariateChisquare2015}
& CDF \\

Other 
& Vibration and Frequency Response 
& \cite[Sec.~3.2]{maoModelQuantifyingUncertainty2012},
  \cite[Sec.~3]{yanCircularlysymmetricComplexNormal2016},
  \cite[Sec.~2]{yanUnifiedSchemeSolving2020}
& PDF \\

Other 
& Reliability 
& \cite[Sec.~7.4]{duProbabilisticEngineeringDesign2005},
  \cite[Sec.~3]{huangAnewDirectSecondOrder2018}
& CDF \\

Other 
& Econometrics 
& \cite{BaoAmanUllahExpectation2010},
  \cite[Sec.~1.4.2]{studerStochasticTaylorExpansions2001},
  \cite{hull2022OptionsFuturesOther}
& Moments, CDF \\

\end{tabularx}
\end{landscape}

\chapter{Classification and Literature Overview: Single Forms and Ratios} \label{Chapter3}

\section{Introduction}
In this section, we provide an overview of the vast literature on single quadratic forms (SQF) in Gaussian random variables. Most works focus on the PDF and CDF of  real and complex SQFs (\Cref{sssec:QFGRV} and \Cref{sssec:QFCGRV}), while others study the PDF/CDF of ratio of SQFs ( \Cref{sssec:RQFGRV} and \Cref{sssec:RQFCGRV}) and their moments.

As explained earlier, both the real and complex SQFs (\eqref{eq:QFGRV} and \eqref{eq:QFCGRV}) are representable as linear combinations of chi-square random variables, possibly together with an independent Gaussian term, and an additive constant. For convenience, we restate these representations here. For the real SQF \eqref{eq:QFGRV}, using the effective parameters ($\{\lambda_n\}_{n=1}^\rankFinal $, $\{h_n^2\}_{n=1}^\rankFinal $, $\sigma$, $c^{\prime\prime} $  ) and the reduced effective parameters ($\{\omega_\ell\}_{\ell=1}^L $, $\{\delta_\ell^2\}_{n=1}^\rankFinal$, $\{\nu_\ell\}_{\ell=1}^L $,$\sigma$, $c^{\prime\prime} $  ), we have
\begin{equation}
    \rs{Q} \sim \sum_{n=1}^{\rankFinal}\lambda_n\chi_1^2(h_n^2) +\sigma\mathcal{N}(0,1)+c'', \label{eqn:QF_rep1} \tag{QFa} \\
\end{equation}
or, equivalently, 
\begin{equation}
     \rs{Q} \sim \sum_{\ell=1}^L\omega_\ell\chi_{\nu_\ell}^2(\delta_\ell^2)+\sigma\mathcal{N}(0,1)+c''.   \label{eqn:QF_rep2}\tag{QFb}
\end{equation}
Similarly, for the complex SQF,
\begin{equation}
    \rs{Q}\sim \sum_{n=1}^{\rankFinal} \frac{\lambda_2}{2}\chi_2^2(h_n^2)+\sigma\mathcal{N}(0,1)+c'' \label{eqn:CQF_rep1} \tag{CQFa}  
\end{equation}
or, equivalently,
\begin{equation}
    \rs{Q} \sim \sum_{\ell=1}^L \tilde{\omega}_\ell \chi_{2\nu_\ell}^2(\delta_\ell^2) +\sigma\mathcal{N}(0,1)+c''. \label{eqn:CQF_rep2} \tag{CQFb}
\end{equation}
The first representation in each case uses the effective parameters, whereas the second uses the reduced effective parameters.The complex case naturally leads to even degrees of freedom in the grouped chi-square representation.

For SQFs, the moment generating function (MGF), the characteristic function (CF), the moments, the cumulant generating function (CGF), and the cumulants are all available in terms of elementary functions of the raw or effective parameters; these quantities are reviewed in Section \ref{sec:MGF_CGF_moments_cumulans}. They form the basis of the exact and approximate methods considered later in this chapter.

In general, obtaining analytical expressions for the PDF/CDF for arbitrary $ \rankFinal$, (or $L$), $\{ \lambda_n \}$, ($\{\omega_\ell\}$), $\{h_n^2\}$, ($\{\delta_\ell^2\}$) and $ \{\nu_\ell\}$ is difficult. We therefore organize the discussion progressively. \Cref{sec:AnalyticalInversionMGF} reviews exact analytical inversion results for tractable cases. \Cref{sec:NumericalInversionMGF} presents exact numerical inversion methods, which are especially useful for CDF evaluation in more general settings. \Cref{sec:Approximate-Formulae} reviews approximation methods, including moment matching, sequence of random variables, and saddlepoint approximations. Finally, \Cref{sec:ratioQF} discusses the moments and the PDF/CDF of ratios of quadratic forms.

\Cref{tab:sqf-pdf-cdf-methods} summarizes the applicability of the main PDF/CDF methods in terms of the reduced effective parameters. It should be read as a guide to the settings in which each method is available. In particular, ``Any'' means no restriction, ``Even'' means that the corresponding degrees of freedom are even, and the condition \(\omega_\ell > 0\) corresponds to the positive-definite case. Throughout the PDF/CDF discussion, we take \(c^{\prime\prime} = 0\), since this additive constant only shifts the distribution and therefore does not introduce significant additional complexity. Indeed, if $\rs{Q} = \rs{Q}_0 + c''$, then $F_{\rs{Q}}(q) = F_{\rs{Q}_0}(q - c^{\prime\prime})$  and ($f_\rs{Q}(q) = f_{\rs{Q}_0}(q - c'')$). Thus, the main difficulty lies in handling the chi-square and Gaussian components, especially whether \(\sigma = 0\) or \(\sigma \neq 0\). In this sense, most exact results reported in the literature concern the chi-square-only setting, namely \(\sigma = 0\).

\begin{table}[t]
\centering
\caption{Applicability of the main PDF/CDF methods for a single quadratic form under the reduced effective representation
\(
Q \sim \sum_{\ell=1}^{L}\omega_\ell \chi^2_{\nu_\ell}(\delta_\ell^2)
   + \sigma \mathcal{N}(0,1).
\)}
\label{tab:sqf-pdf-cdf-methods}
\setlength{\tabcolsep}{4.5pt}
\renewcommand{\arraystretch}{1.12}
\small
\begin{tabular}{@{}llcccccc@{}}
\toprule
& & \multicolumn{5}{c}{Applicability conditions} & \\
\cmidrule(lr){3-7}
Method & Section & \(L\) & \(\omega_\ell\) & \(\nu_\ell\) & \(\delta_\ell^2\) & \(\sigma\) & Rep. result \\
\midrule
Analytic inv.(Cent,Even)   & \S\ref{sec:EvenDegreesHermitianForms} & Any & Any   & Even & \(0\) & \(0\)   & \eqref{eq:Naffouri-Central-CDF}  \\
Analytic inv. (PD)  & \S\ref{sec:SeriesExpansionsforPositiveDefiniteForms}& Any & \(>0\) & Any  & Any   & \(0\)   & \eqref{eqn:Chi-square Expansion_pdf}  \\
Analytic inv. (IndefC)   & \S\ref{ssec:IndefiniteComplex(Even multiplicity)} & Any & Any   & Even & Any   & \(0\)   & \eqref{eq:Raphaelli_pdf}  \\
Analytic inv. (IndefR)   & \S\ref{ssec:IndefiniteRealForms} & Any & Any   & Any  & Any   & \(0\)   & \eqref{eq:ProvostRudiuk1_pdf}  \\
Numerical inv.  & \S\ref{ssec:Gil–Pelaez inversion formula} & Any & Any   & Any  & Any   & Any     & \eqref{eq:Fc2_davies}  \\
Moment matching      &  \S\ref{ssec:Moment-Matching} & Any & \(>0\) & Any  & Any   & \(0\)   & \eqref{eq:wood}  \\
Sequence of RVs      & \S\ref{sssec:IndefiniteForms_SeqRvs}& Any & Any   & Any  & Any   & \(0\)   & \eqref{eq:HaProvost1}   \\
Saddlepoint approx.  & \S\ref{ssec:Saddlepoint-Approx}  & Any & Any   & Any  & Any   & Any     & \eqref{eq:QF_SPA_CDF_LR}  \\
\bottomrule
\end{tabular}
\end{table}

\section{ MGF, CGF, Moments and Cumulants} \label{sec:MGF_CGF_moments_cumulans}
 We collect the main transform and moment quantities associated with SQFs. These quantities will be used repeatedly in the analytical inversion, numerical inversion, and approximation methods presented later. In particular, the MGF and CF are the starting point for exact inversion formulas, while the CGF and cumulants play a crucial role in moment-matching and saddlepoint approximation. 

 Consider $\rs{Q}$ written in effective form \eqref{eqn:QF_rep1}, or equivalently, in reduced effective form \eqref{eqn:QF_rep2}. The MGF of $\rs{Q}$ follows directly from the product of the MGF of independent noncentral chi-square random variables and the MGF of a Gaussian random variable. In terms of effective parameters, it is given by
\begin{equation}
    \label{eq:MGF_incmplete_1}
    M_{\rs{Q}}(t)=\exp\left(t\sum_{n=1}^\rankFinal \dfrac{h_n^2\lambda_n} {1-2\lambda_n t}+\dfrac{\sigma^2t^2}{2}+c'' t \right)\prod_{n=1}^\rankFinal (1-2\lambda_n t)^{-\frac{1}{2}}.
\end{equation}
Equivalently, in terms of the reduced effective parameters, 
\begin{equation}
    \label{eq:MGF_incmplete_2}
    M_{\rs{Q}}(t)=\exp\left(t\sum_{\ell=1}^L \dfrac{\delta_\ell^2\omega_\ell}{1-2\omega_\ell t}+\dfrac{\sigma^2t^2}{2} +c'' t \right)\prod_{\ell=1}^L (1-2\omega_\ell t)^{-\frac{\nu_\ell}{2}}.
\end{equation}            
If $\sigma = c^{\prime\prime} = 0$, the MGF is reduced to the chi-square only form, 
\begin{equation}
    M_{\rs{Q}}(t)=\expect{e^{t\rs{Q}}}=\exp\left(\sum_{n=1}^\rankFinal\dfrac{h_n^2 t}{1-2\lambda_n t}\right)\prod_{n=1}^\rankFinal\dfrac{1}{(1-2\lambda_n t)^{1/2}}, \label{eq:MGF-QFGRV-onlyChi}
\end{equation}
or, equivalently, 
            \begin{equation}
            \label{eq:MGF_cmplete_2}
            M_{\rs{Q}}(t)=\exp\left(t\sum_{\ell=1}^L \dfrac{\delta_\ell^2\omega_\ell}{1-2\omega_\ell t }\right)\prod_{\ell=1}^L (1-2\omega_\ell t)^{-\frac{\nu_\ell}{2}}.
            \end{equation}  
If further the quadratic form is central, i.e., $\{h_n^2\}_{n=1}^\rankFinal = 0$,  then the exponential factor disappears and the MGF simplifies to
\begin{equation}
            M_{\rs{Q}}(t)=\prod_{n=1}^\rankFinal\dfrac{1}{(1-2\lambda_n t)^{1/2}}, \label{eq:MGF_central_1}
\end{equation}
or equivalently, when  $\{\delta_\ell^2\}_{\ell=1}^L = 0,$
\begin{equation}
            M_{\rs{Q}}(t)=\prod_{\ell=1}^L\dfrac{1}{(1-2\omega_\ell t)^{\nu_\ell/2}}. \label{eq:MGF_cenrtal_2}
\end{equation}

The MGF exists for all $t$ in the interval $(t_L, r_r)$ containing the origin, where 
\begin{align}
    t_L &\triangleq -\dfrac{1}{2\displaystyle\max_{\lambda_n<0}|\lambda_n|} \label{eq:tL}\\ 
    t_R &\triangleq \dfrac{1}{2\displaystyle\max_{\lambda_n>0}\lambda_n} \label{eq:tR}
\end{align}
If there are only positive eigenvalues, we take $t_L = -\infty$; if there are only negative eigenvalues, we take $t_R = +\infty$. Thus, for a positive semidefinite form, the MGF exists for all $t<t_R$, whereas for a negative semidefinite form it exists for all $t>t_L$. 

Because the MGF exists in a neighborhood of the origin, there exists some $\epsilon >0$ such that it is finite for all real $t \in (-\epsilon, \epsilon)$. Hence, it admits analytic continuation, denoted by $\Psi(z)$, to the vertical strip \cite{GengMGF}, 
$$
\{ z = \gamma + i\beta : \gamma \in (-\epsilon, \epsilon), \beta \in \mathbb{R} \}.
$$
Thus, $\Psi(z)$ agrees with $M_\rs{Q}(t)$ whenever $z=t$ is real and $t \in (-\epsilon,\epsilon)$. For simplicity, and since no ambiguity will arise, we will continue to use the notation $M_\rs{Q}(t)$ even when the argument $t$ is complex. In particular, the CF is obtained by evaluation along the imaginary axis: $\phi_{\rs{Q}}(\beta) = M_\rs{Q}(i\beta)$. Explicitly, 
\begin{equation}
    \label{eq:CF}
    \phi_{\rs{Q}}(\beta) = \exp \left( i\beta\sum_{n=1}^\rankFinal \dfrac{h_n^2\lambda_n}{1-i2\lambda_n\beta} - \dfrac{\sigma^2\beta^2}{2} + i\beta c^{\prime \prime}\right)\prod_{n=1}^\rankFinal (1-i2\lambda_n \beta)^{-\frac{1}{2}} 
\end{equation}
or, equivalently, 
\begin{equation}
    \label{eq:CF_reduced}
    \phi_{\rs{Q}}(\beta) = \exp \left( i\beta\sum_{\ell=1}^L \dfrac{\delta_\ell^2\omega_\ell}{1-i2\omega_\ell\beta} - \dfrac{\sigma^2\beta^2}{2} + i\beta c^{\prime \prime}\right)\prod_{\ell=1}^L (1-i2\omega_\ell \beta)^{-\frac{\nu_\ell}{2}} 
\end{equation}

Let $f_\rs{Q}(q)$ denote the density of $\rs{Q}$. If $\rs{Q}$ is supported on the nonengative real line, its unilateral Laplace transform is defined by 
$$
\hat{f}_\rs{Q}(s) = \mathcal{L}\{f_\rs{Q}\}(s) =  \int_0^\infty e^{-sq}f_\rs{Q}(q) dq. 
$$

In this case, since $\rs{Q \ge 0}$, the MGF and the unilateral Laplace transform are directly related through 
\begin{equation}
    \hat{f}_\rs{Q}(s)  = M_\rs{Q}(-s). 
\end{equation}
Thus, for positive definite forms, one may work equivalently with the density's unilateral Laplace transform or with the MGF evaluated at $-s$. 

For indefinite quadratic forms, the density is suopported on the whole real line, so the unilateral Laplace transform ove $[0,\inf)$ captures only the positive half-line and is therefore insufficient to characterize the full distribution. In that case, one must instead work with a bilateral transform, or equivalently, with the Fourier inversion of the CF. 



The cumulant generating function (CGF) is defined as the natural logarithm of the MGF, $K_\rs{Q}(t) = \ln M_\rs{Q}(t)$. Therefore,
    \begin{equation}
        \label{eq:CGF}
        K_\rs{Q}(t) = -\dfrac{1}{2}\sum_{n=1}^\rankFinal \ln{(1-2\lambda_n t)} + t\sum_{n=1}^\rankFinal \frac{h_n^2\lambda_n}{1-2\lambda_n t}+\frac{\sigma^2t^2}{2} +c^{\prime\prime}t,
    \end{equation}
or, in reduced effective form, 

\begin{equation}
    \label{eq:CGF_2}
    K_\rs{Q}(t) = -\dfrac{1}{2}\sum_{\ell=1}^L \nu_\ell\ln{(1-2\omega_\ell t)} + t\sum_{\ell=1}^L \frac{\delta_\ell^2\omega_\ell}{1-2\omega_\ell t}+\frac{\sigma^2t^2}{2} +c^{\prime\prime}t.
\end{equation}
    
The $k$th \emph{raw moment} of a SQF is $\tilde{\mu}_k := \mathbb{E}[\rs{Q}^k]$, and can be obtained from the $k$th derivative of the MGF at the origin, 
\begin{equation}
    \tilde{\mu}_k =\dfrac{d^k}{dt^k} M_\rs{Q}(t) |_{t=0}.
\end{equation}
Note that we denote the $k$th raw moment by $\tilde{\mu}_k$ instead of $\mu_k$ to avoid confusion with the elements of the mean vector of the underlying Gaussian random variable. 

The moments can be recursively calculated from the cumulants. Indeed, since the CGF is defined by $$K_{\rs{Q}}(t)=\ln\left(M_{\rs{Q}}(t)\right)=\sum_{\ell=0}^{\infty}\kappa_\ell\dfrac{t^\ell}{\ell!}$$ it holds that $K_{\rs{Q}}'(t)M_{\rs{Q}}(t)=M_{\rs{Q}}'(t)$, where $\kappa_\ell$ denotes the $\ell$th cumulant. Differentiating term-wise and applying Cauchy product, we can recursively calculate the $k^\text{th}$ moment by: 
    \begin{equation}
        \tilde{\mu}_k =\sum_{\ell=0}^{k-1}\binom{k-1}{\ell}\tilde{\mu}_\ell\kappa_{k-\ell}, \quad \tilde{\mu}_\ell=1. \label{eq:Moments-from-Cumulants}
    \end{equation} 
Note that this formula applies to any random variable whose MGF is defined in a neighborhood of zero. Actually, the moments are Bell polynomials \footnote{Bell polynomials arise when we wish to write the exponential of a power series as a power series: $\exp \left(\sum_{j=1}^{\infty} x_j \frac{t^j}{j!}\right)=\sum_{n=0}^{\infty} B_n\left(x_1, \ldots, x_n\right) \frac{t^n}{n!}$. See \cite[P: 13]{comtetAdvancedCombinatorics2012}.} of the cumulants. Now the cumulants can be easily calculated using straightforward McLaurin expansions. The first cumulant is 
\begin{equation}
    \kappa_1=\mathbb{E}[\rs{Q}]=\sum_{n=1}^\rankFinal\lambda_n(1+h_n^2)+c''
    = \sum_{\ell=1}^{L}\omega_\ell(\nu_\ell+\delta_\ell^2)+c''.
\end{equation}
The second cumulant is 
\begin{equation}
    \kappa_2=\Var[\rs{Q}]=2\sum_{n=1}^\rankFinal\lambda_n^2(1+2h_n^2)+\sigma^2 
    =  2\sum_{\ell=1}^{L}\omega_\ell^2(\nu_\ell+2\delta_\ell^2)+\sigma^2.
\end{equation}
For $j \ge 3$, the cumulants are 
\begin{equation}
    \kappa_j=2^{j-1}(j-1) ! \sum_{n=1}^\rankFinal \lambda_n^j\left(1+j h_n^2\right) 
    = 2^{j-1}(j-1)!\sum_{\ell=1}^{L}\omega_\ell^{j}\big(\nu_\ell + j\delta_\ell^{2}\big). \label{eqn:cumulants}
\end{equation}
Note  that the independent Gaussian term contributes only to the first two cumulants. All higher-order cumulants are determined entirely by the chi-square part of the representation.

The additive constant $c^{\prime\prime}$ shifts the random variable. Particularly, it multiplies the MGF by $e^{c^{\prime\prime}t}$ adds $c^{\prime\prime}$ to the first cumulant, and shifts the PDF and CDF without changing their analytical complexity. Thus, for the PDF/CDF problem, it is often convenient to set $c^{\prime\prime} = 0$ and restore the shift afterwards if needed. 

For complex SQF, the same ideas extend directly. Indeed, the MGF, CFm CGF, moments and cumulants are obtained in the same way by applying the corresponding transform formulas to the complex chi-square representation \eqref{eqn:CQF_rep1} and \eqref{eqn:CQF_rep2}. In particular, the grouped complex representation naturally involves even degrees of freedom, which simplifies several inversion results, as will be seen later. 
\section{PDF/CDF: Analytical Inversion of the MGF/CF}
\label{sec:AnalyticalInversionMGF}
In this section,  we review exact analytical methods for deriving the PDF and CDF of SQF by inverting its MGF or CF. Throughout this section, we restrict attention to the chi-square only representation, i.e., $\sigma = c^{\prime\prime} = 0$. As discussed earlier, the additive constant $c^{\prime\prime}$ only shifts the random variable and therefore does not increase the complexity of the PDF/CDF problem. For this reason, it is convenient to suppress this constant here and restore it afterwards if needed.

Broadly speaking, there are two routes for obtaining the PDF/CDF analytically, 
 
\begin{enumerate}
    \item One may integrate, or marginalize, the underlying Gaussian density directly over the appropriate region.
    \begin{align}
        F_\rs{Q}(q)   &= \mathbb{P}[\rs{Q}<q] = \int_{\rs{Q}(\dv{x}) \le q} f_{\rv{x}} (\dv{x}) d \dv{x} \\
        f_{\rs{Q}}(q) &=  \int_{\rs{Q}(\dv{x}) =  q} f_{\rv{x}} (\dv{x}) d \dv{x} = \frac{d}{dq}\int_{\rs{Q}(\dv{x}) \le q} f_{\rv{x}} (\dv{x}) d \dv{x} 
    \end{align}
    \item One may invert the MGF or the CF through Fourier/Laplace inversion theorem:
    \begin{align}
        f_\rs{Q}(q) &= \frac{1}{2\pi i}\int_{\gamma -i\infty}^{\gamma + i\infty} M_\rs{Q}(t) e^{-qt} dt. \label{eq:invertMGF} \\
        & = \frac{1}{2\pi}\int_{-\infty}^{\infty} \phi_\rs{Q}(t) e^{-iqt} dt.  \label{eq:invertCF}
    \end{align}
\end{enumerate}
Depending on the structure of the problem, one route may be more tractable than the other. In general, the transform-inversion route is used more often, since it avoids direct integration over the Gaussian region and shifts the problem to the analysis of a scalar complex integral. It is also useful to note that inversion of the MGF is equivalent to inversion of the two-sided Laplace transform, rather than the unilateral one.

Using \Cref{eq:MGF_cmplete_2}, the PDF can be recovered from the MGF through
\begin{equation}
    f_\rs{Q}(q) = \frac{1}{2\pi i}\int_{\gamma -i\infty}^{\gamma + i\infty}\exp\left(-qt + t\sum_{\ell=1}^L \dfrac{\delta_\ell^2\omega_\ell}{1-2\omega_\ell t }\right)\prod_{\ell=1}^L (1-2\omega_\ell t)^{-\frac{\nu_\ell}{2}} dt,
\end{equation}
Alternatively, one may work with the CF, which yields the PDF representation
\begin{equation}
    f_\rs{Q}(q) = \frac{1}{2\pi}\int_{-\infty}^{\infty}\exp\left(-iq\beta + i\beta\sum_{\ell=1}^L \dfrac{\delta_\ell^2\omega_\ell}{1-i2\omega_\ell\beta  }\right)\prod_{\ell=1}^L (1-i2\omega_\ell \beta)^{-\frac{\nu_\ell}{2}} d\beta,
\end{equation}

Since the MGF is defined in a neighborhood of the origin, the CF-based inversion formula can be viewed as a boundary-line version of the MGF inversion formula. In particular, CF-based inversion corresponds to the boundary line $\gamma = 0$, whereas MGF-based inversion allows interior contours $\gamma \neq 0$, provided that $\gamma$ lies in the region of convergence. This additional flexibility is useful because it can improve numerical stability.  

To obtain the CDF analytically, one may either integrate the PDF, 
$$ 
F_\rs{Q}(q) = \int_{-\infty}^q f_\rs{Q}(\tau) d\tau
$$, 
or whenever ($M_\rs{Q}(t)$) is analytic on the vertical line ($\Re(t)=\gamma$), use the inversion formula:
\begin{equation}
\label{eq: CDF_from_MGF}
F_Q(q)=\frac{e^{-\gamma q}}{2\pi}\int_{-\infty}^{\infty}
e^{-i\beta q}\frac{M_Q(\gamma+i\beta)}{\gamma+i\beta}d\beta,
\end{equation}


Typically, in any mathematical problem, if the issue at hand is too difficult to tackle at once, researchers tend to study special cases, especially when these cases arise in a number of practical applications. In this section, we present two special cases of a single quadratic form: central with even multiplicity forms, and positive definite forms. At the end of this section, the case of indefinite forms with no restrictions on the degrees of freedom solves the QF  for arbitrary parameters.

\subsection{Central Forms with Even Degrees of Freedom}  \label{sec:EvenDegreesHermitianForms}
Mathai and Provost \cite[Chapter~3]{mathaiQuadraticFormsRandom1992} define a central form to be a form in central Gaussian random variables. However, since the practical representation is the final linear combination of chi-squares (as in \eqref{eq:lcOnlyChiQFGRV}), we define a central form here to account for any linear combination of (central) chi-squares, i.e., the non-centrality parameters $h_n$ in \eqref{eq:lcOnlyChiQFGRV} are null \footnote{Equivalently, the non-centrality parameters $\delta_\ell$ in \eqref{eq:lcOnlyChi} are zero.}. As an immediate result, the exponential term in the MGF \eqref{eq:MGF-QFGRV-onlyChi} disappears, and the transform reduces to $$M_{\rs{Q}}(t)=\prod_{n=1}^N\dfrac{1}{(1-2\lambda_n t)^{1/2}}=\prod_{\ell=1}^L\dfrac{1}{(1-2\omega_\ell t)^{\nu_\ell/2}},$$ where $t\in(t_L,t_R)$, $t_L$ and $t_R$ are defined in \eqref{eq:tL} and \eqref{eq:tR}, respectively. 
The MGF, as it stands in the expression above, is still not tractable. Had the square roots not been there, it would have been decomposed into partial fractions. To do so, all eigenvalues must occur in pairs, i.e., the multiplicities $\nu_\ell$ must be even for $\ell=1,2,\ldots, L$.  Define $m_\ell:= \frac{v_\ell}{2} \in \mathbb{N} $. Because each $m_\ell\in\mathbb{N}$, the poles of $M_\rs{Q}(t)$  at $ t_\ell=\frac{1}{2\omega_\ell} $ have integer order $m_\ell$, which enables a finite partial-fraction expansion of the MGF as follows: 
    $$
            M_{\rs{Q}}(t)= \prod_{\ell=1}^L\dfrac{1}{(1-2\omega_\ell t)^{m_\ell}}= \sum_{\ell=1}^L\sum_{k=1}^{m_\ell}\dfrac{A_{\ell k}}{(1-2\omega_\ell t)^k}, \quad t\in(t_L,t_R)\setminus\{t_\ell\}.
    $$
Since the inverse Laplace transform is linear, we can now invert the MGF term by term. 
Consider one term $(1-2\omega t)^{-k}$ . Using \Cref{eq:invertMGF} and residue calculus for a pole of order $k$, we obtain:
\begin{itemize}
    \item For $\omega>0$ (pole on the right, contributing when ($q>0$)):
      \begin{equation}
      \label{eq:right_pdf_CentEven}
      \mathcal{L}_{\text {bilat }}^{-1}\left[(1-2 \omega t)^{-k}\right](q)=\frac{1}{(2 \omega)^k \Gamma(k)} q^{k-1} e^{-q /(2 \omega)} u(q)
      \end{equation}
    \item      For $\omega<0$   (pole on the left, contributing when ($q<0$)):
      \begin{equation}
      \label{eq:left_pdf_CentEven}
      \mathcal{L}_{\text {bilat }}^{-1}\left[(1-2 \omega t)^{-k}\right](q)=\frac{1}{(2 |\omega|)^k \Gamma(k)}(-q)^{k-1} e^{q /(2 |\omega|)} u(-q)
      \end{equation}
    \end{itemize}
    where $u(q) = 1$ for $q>0$ and $u(-q)=1$ for $q<0$. Note that \Cref{eq:right_pdf_CentEven} and \Cref{eq:left_pdf_CentEven} are the PDFs of a scaled chi-square, $\omega \chi^2_{2k}$, when $\omega>0$ and $\omega<0$, respectively. A  compact formula is 
\begin{equation}
    f_{\omega\chi^2_{2k}}(q) = \frac{|q|^{k-1}}{(2|\omega|)^k \Gamma(k)}  \exp\left(-\frac{q}{2 \omega}\right) u(\operatorname{sgn}(\omega) q) .
\end{equation}
      Therefore, the inverse transform is immediate:
\begin{equation}
    \label{eq:PDF_CentralEven}
    f_\rs{Q}(q)=\sum_{\ell=1}^L\sum_{k=1}^{m_\ell}A_{\ell k}f_{\omega_\ell\chi^2_{2k}}(q) = \sum_{\ell=1}^L\sum_{k=1}^{m_\ell}\frac{A_{\ell k}}{|\omega_\ell|}f_{\chi^2_{2k}}\left(\frac{q}{\omega_\ell}\right),
\end{equation}    
Note that $\omega_\ell\chi^2_{2k}$ is equivalent to gamma distribution with shape parameter, $k$ and scale parameter, $2\omega_\ell$. The CDF is obtained by integrating \Cref{eq:PDF_CentralEven} term by term to yield
\begin{equation}
    \label{eq:CDF_CentEven}
    F_{\rs{Q}}(q)=\sum_{\ell=1}^L\sum_{k=1}^{m_\ell}A_{\ell k}F_{\chi_{2k}^2}\left(\dfrac{q}{\omega_\ell}\right),
\end{equation}
        
Note that the CDF of a chi-square is the normalised lower incomplete gamma function, which, in this case of even degrees of freedom, can be written as a linear combination of elementary functions, i.e., $$F_{\chi_{2k}^2}\left(\dfrac{q}{\omega_\ell}\right) = 1 - e^{-q/\omega_\ell} \sum_{j=0}^k \frac{q^j}{\omega_\ell^jj!} $$ 

the CDF can be rewritten as 
\begin{equation}
    \resizebox{\textwidth}{!}{$
    F_{\rs{Q}}(q)=\sum_{\ell=1}^L\sum_{k=1}^{m_\ell/2}A_{\ell k}\left\{1-u(\omega_\ell)+\sign(\omega_\ell)u(\omega_\ell u)\left[1-e^{y/\omega_\ell}\left(\sum_{s=0}^{k/2-1}\dfrac{y^s}{\omega_\ell^s s!}\right)\right]\right\}
    $}
\end{equation}
   
    
The authors in  \cite[Sec 4.3]{mathaiQuadraticFormsRandom1992} report the results for the positive definite case, and the indefinite case

                
Going back to \eqref{eq:lcOnlyChiC}, we can see that a central quadratic form in complex Gaussian random variables satisfies the conditions needed to obtain such a closed-form expression. In fact, Al-Naffouri et al. \cite{al-naffouriDistributionIndefiniteQuadratic2016} derive an equivalent closed-form expression using a different partial fraction decomposition. Specifically, they consider the central complex form \eqref{eqn:CQF_rep2} and write

\[
\dfrac{1}{z\prod_{i=1}^L(1+\tilde{\omega}_iz)^{\nu_i}}=\dfrac{1}{z}+\sum_{\ell=1}^L\sum_{k=1}^{\nu_\ell}\dfrac{\alpha_{\ell k}}{(1+\tilde{\omega}_\ell z)^k}
\]

The corresponding CDF is then  given by \cite{al-naffouriDistributionIndefiniteQuadratic2016} as
   \begin{equation}
        F_{Y_\rs{Q}}(q)=u(q)+\sum_{\ell=1}^L\sum_{k=1}^{\nu_\ell}\dfrac{\alpha_{\ell k}}{(k-1)!|\tilde{\omega}_\ell|^k}q^{k-1}e^{-\frac{q}{\tilde{\omega}_\ell}}u\left(\frac{q}{\tilde{\omega}_\ell}\right)\label{eq:Naffouri-Central-CDF}
   \end{equation}
            
Differentiating the one-dimensional integral expression of the CDF yields the corresponding PDF, which agrees with the central-case expression in \Cref{eq:PDF_CentralEven}. The work in \cite{al-naffouriDistributionIndefiniteQuadratic2016} uses a different partial-fraction decomposition from the one adopted earlier; consequently, that resulting residues differ from the coefficients $A_{\ell k}$. Nevertheless, the two expressions are equivalent and can be related recursively,
\begin{align}
    \alpha_{\ell k} &= - \omega_\ell A_{\ell k} + \alpha_{\ell k}, \quad k < m_\ell \\ 
    \alpha_{\ell m_\ell} &= -\omega_\ell A_{\ell m_\ell}
\end{align}

 Since both involve partial fraction decomposition, which is of complexity $O(L^3)$ together with diagonalization, if needed, ($O(N^3)$), both methods scale cubically with the size of the problem. Thus, even if they are closed-form, they may be computationally prohibitive at large sizes. In this case, approximate methods may be more appropriate, 
  
\subsection{What happens when we lose evenness or centrality?}
The tractability of the previous subsection relies on a very specific structural property, the transform is rational, with poles of finite integer order. Once even multiplicity is lost, or once noncentrality is introduced, this structure is no longer available in general, and the inversion problem becomes harder. This is why the literature contains a number of exact formulas only fore carefully chosen subclasses. In the following, we present some cases for $L=2$.

\subsubsection{Linear Combinations of 2 Chi-Squares} 

The handbook of Simon \cite{simonProbabilityDistributionsInvolving2006} collects a broad range of analytic expressions of the PDF and CDF for linear combinations of two independent chi-square random variables ($L=2$). \Cref{tab:SimonResults} lists a selected subset of these formulae together with their corresponding equation numbers in \cite{simonProbabilityDistributionsInvolving2006}. Although we do not report all of Simon's results here, a few examples are worth reporting because they clearly illustrate how the structure of the formulas changes as the assumptions on evenness and/or centrality are relaxed.  

Consider first the sum of two independent central chi-square random variables with positive coefficients, $\rs{Q} \sim \omega_1\chi_{\nu_1}^2 + \omega_2\chi_{\nu_2}^2$, where $\delta_1^2 = \delta_2^2 = 0$ and $\omega_1, \omega_2 >0$. An attractive feature of this setting is that there is no restriction on the evenness of the degrees of freedom. The PDF as reported by Simon \cite [Eq. (5.26)]{simonProbabilityDistributionsInvolving2006} is \footnote{Simon's notation in \cite{simonProbabilityDistributionsInvolving2006} is slightly different.}: 
\begin{equation}
    \begin{aligned}
           f_{\rs{Q}}(q)&=\dfrac{1}{2 \sqrt{\omega_1\omega_2}\Gamma\left(\dfrac{\nu_1+\nu_2}{2}\right)}\left(\dfrac{q}{2\omega_1}\right)^{\frac{\nu_1-1}{2}}\left(\dfrac{q}{2\omega_2}\right)^{\frac{\nu_2-1}{2}}\exp\left(-\dfrac{q}{2\omega_1}\right)  \\
                    &\times{}_1F_1\left(\dfrac{\nu_2}{2};\dfrac{\nu_1+\nu_2}{2};\dfrac{(\omega_2-\omega_1)^2}{2\omega_1\omega_2}q\right),\quad q\geq 0,
    \end{aligned} \label{eq:Simon-Sum-Central}
    \end{equation}
where $_1F_1\left(.;.;.\right)$ is the confluent hypergeometric function. No CDF formula is reported by Simon for this case.

Now, consider the sixth case in \Cref{tab:SimonResults}, $\rs{Q} \sim \omega_1\chi_{\nu_1}^2(\delta_1^2) + \omega_2\chi_{\nu_2}^2$, where $\delta_1^2 >0, \delta_2^2 = 0$ and  $\omega_1, \omega_2 >0$. In this setting, the evenness of the two chi-squares is preserved (that is, $\nu_1 = 2m_1$, $\nu_2=2m_2$, $m_1, m_2 \in \mathbb{N}$), but the first chi-square is noncentral. Thus, this example loses complete centrality while keeping the favorable evenness structure.  The PDF is given by  \cite [Eq. (5.48)]{simonProbabilityDistributionsInvolving2006}:
            \begin{equation}
                \label{eq:Simon_PDF_noncent}
                \begin{aligned}
                f_{\rs{Q}}(q) & =\frac{1}{2 \omega_1}\left(\frac{\omega_1}{\omega_2}\right)^{m_2}\left(\frac{q}{\delta_1^2}\right)^{\left(m_1+m_2-1\right) / 2} \exp \left(-\frac{q+\delta_1^2}{2 \omega_1}\right) \\
                & \times \sum_{i=0}^{\infty} \frac{\Gamma\left(m_2+i\right)}{i ! \Gamma\left(m_2\right)}\left(\frac{\sqrt{q}\left(\omega_2-\omega_1\right)}{\delta_1 \omega_2}\right)^i I_{m_1+m_2+i-1}\left(\frac{\delta_1 \sqrt{q}}{\omega_1}\right), q \geq 0,
                \end{aligned}
            \end{equation}
            where $I_{\nu}(.)$ denotes the modified Bessel function of the first kind. The corresponding CDF is given by  \cite [Eq. (5.49)]{simonProbabilityDistributionsInvolving2006}:
            \begin{equation}
            \label{eq:Simon_CDF_noncent}
            \begin{aligned}
                 F_{\rs{Q}}(q)&=\left(\frac{\omega_1}{\omega_2}\right)^{m_2} \sum_{i=0}^{\infty} \frac{\Gamma\left(m_2+i\right)}{i ! \Gamma\left(m_2\right)}\left(\frac{\omega_2-\omega_1}{\omega_2}\right)^i\\&\times\left[1-Q_{m_1+m_2+i}\left(\frac{\delta_1}{\sqrt{\omega_1}}, \sqrt{\frac{q}{\omega_1}}\right)\right],\; q \geq 0,
            \end{aligned}
            \end{equation}
            where $Q_{m_1+m_2+i}$ is the generalized Marcum $Q$-function. Compared with the central case above, the loss of centrality replaces a compact single-function expression by a single infinite-series representation involving modified Bessel function and generalized Marcum-Q functions.  

\begin{longtable}{|c|c|c|c|c|c|c|c|c|}
    \caption{Selected results from Simon \cite{simonProbabilityDistributionsInvolving2006} for the PDF/CDF of linear combinations of two independent chi-square random variables.}  \label{tab:SimonResults}\\
    \hline
    \multirow{2}{*}{SN} & \multicolumn{6}{c|}{Parameters} & \multirow{2}{*}{PDF} & \multirow{2}{*}{CDF} \\ \cline{2-7}
    &\(\omega_1\) &\(\omega_2\) &\(\nu_1\) &\(\nu_2\) &\(\delta^2_1\)&\(\delta^2_2\)  &    &            \\ \hline
1   & >0          & > 0          &  Any      &  Any                  & 0     & 0   & (5.26) & --        \\ \hline
2   & >0          & < 0          &  even     &  $\nu_2=\nu_1$        & 0     & 0   & (4.7)  & (4.8)     \\ \hline
3   & >0          & < 0          &  odd      &  $\nu_2 = \nu_1$      & 0     & 0   & (4.10) & --        \\ \hline
4   & >0          & < 0          &  even     &  even                 & 0     & 0   & (4.16) & (4.17)    \\ \hline
5   & >0          & > 0          &  Any      &  $\nu_2 = \nu_1$      & Any   & 0   & (5.45) & (5.46)    \\ \hline
6   & >0          & > 0          &  even     &  even                 & Any   & 0   & (5.48) & (5.49)    \\ \hline
7   & >0          & > 0          &  Any      &  $\nu_2 = \nu_1$      & Any   & Any & (5.60) & (5.61)    \\ \hline
8   & >0          & > 0          &  even     &  even                 & Any   & Any & (5.63) & (5.64)    \\ \hline
\end{longtable}

    Note that we only report the case $L=2$, that is, a linear combination of two chi-square random variables. Extending such analytic results to general $L$ is much more challenging, as already suggested by the previous examples, where even the $L=2$ case may require confluent hypergeometric functions, modified Bessel functions, generalized Marcum-(Q) functions, and infinite summations. Indeed, solving the PDF/CDF problem for arbitrary ($L$), ($\omega_\ell$), ($\nu_\ell$), and ($\delta_\ell^2$) is not easy, and the formulas become more complicated as one moves to greater generality. Next, we report two more interesting results that can be expressed as a linear combination of two independent chi-squares. 
            
    \subsubsection{Sum of correlated chi-squares}
        The sum of correlated chi-square random variables provides a useful example showing that dependence can be absorbed into the covariance matrix of an underlying Gaussian vector. As a result, the random variable can still be represented as a quadratic form and, in turn, as a linear combination of independent chi-square random variables. 
       
       Joarder and Omar \cite{joarderExactDistributionSum2013} study the distribution of the sum of two correlated chi-square variables having the same number of degrees of freedom. Consider two samples each of length $n+1$, $\rv{x}^{(1)} = [\rs{x}^{(1)}_1, \rs{x}^{(1)}_2, \cdots \rs{x}^{(1)}_{n+1} ]^T$ and  $\rv{x}^{(2)} = [\rs{x}^{(2)}_1, \rs{x}^{(2)}_2, \cdots \rs{x}^{(2)}_{n+1} ]^T$. The pairs $(\rs{x}^{(1)}_i, \rs{x}^{(2)}_i),$ for $ i=1,2,\cdots n+1$  are i.i.d. bivariate normal with mean $[\mu^{(1)}, \mu^{(2)} ]^T $ and covariance matrix $ \begin{bmatrix}
                \sigma_1^2 & \rho\sigma_1\sigma_2\\
                \rho\sigma_1\sigma_2 & \sigma_2^2
            \end{bmatrix}$. The unbiased sample variances of sample 1 and 2 can be written as $\rs{S}_1^2 = \dfrac{1}{n}\rv{x}^{(1)^T} \dmt{H} \rv{x}^{(1)} $ and  $\rs{S}_2^2 = \dfrac{1}{n}\rv{x}^{(2)^T} \dmt{H} \rv{x}^{(2)} $, respectively, where $ \dmt{H} = \dmt{I}_{n+1} - \frac{1}{n+1}\dv{1}_{n+1}\dv{1}_{n+1}^T$. Define the scaled sample variances $\rs{U} = n \rs{S}_1^2/\sigma_1^2 $ and $\rs{V} = n \rs{S}_2^2/\sigma_2^2 $ and let $\rs{Z} = \rs{U} + \rs{V}$. Joarder and Omar find the density and distribution of $\rs{Z}$. Interestingly, it can be shown that $\rs{Z} \sim \rs{Q}$ where $\rs{Q} = \rv{x}^T \dmt{A} \rv{x}$, where $\rv{x} =[\rv{x}^{(1)}, \rv{x}^{(2)}]^T\in \mathbb{R}^N $, $N = 2n+2$ with mean $\dv{\mu} = [\mu^{(1)}\dv{1}_{n+1}, \mu^{(2)}\dv{1}_{n+1}]^T$, covariaance matrix $ \begin{bmatrix}
                \sigma_1^2 & \rho\sigma_1\sigma_2\\
                \rho\sigma_1\sigma_2 & \sigma_2^2
            \end{bmatrix} \otimes \dmt{I}_{n+1}$, and  $ \dmt{A} = \begin{bmatrix}
                \dmt{H}/\sigma_1^2 & \dmt{O}\\
                \dmt{O} & \dmt{H}/\sigma_2^2
            \end{bmatrix}$. Since $\dmt{H}\dv{1}_{n+1} = \dv{0}$, we have $\dmt{A}\dv{\mu} = \dv{0}$. Hence, the quadratic form is central. Further, the matrix $\dmt{A}\dmt{\Sigma}$ has $2n$ nonzero eigenvalues, $n$ of them are equal to $1+\rho$ and the remaining $n$ are equal $(1-\rho)$. Therefore, $\rs{Z} \sim \rs{Q} \sim (1+\rho)\chi^2_n + (1-\rho) \chi^2_n$. 
            Thus, the sum of two correlated chi-square random variables is identically distributed as a particular linear combination of two independent central chi-square random variables.
            Joarder and Omar derive the density function as follows  \cite[Eq. 6]{joarderExactDistributionSum2013}, :
        \begin{equation}
            \begin{aligned}
                  f_{\rs{Q}}(q)&=\dfrac{(1-\rho^2)^{-n/2}}{2^n\Gamma(n)}q^{n-1}\exp\left(-\dfrac{q}{2-2\rho^2}\right)\\&\times{}_0F_1\left(\dfrac{n+1}{2};\dfrac{\rho^2q^2}{(4-4\rho^2)^2}\right),\quad q>0. \label{eq:Joarder-PDF}
            \end{aligned}
        \end{equation}
        Note that Joarder’s density is obtained from Simon’s positive-central $L=2$ formula by setting ($\omega_1=1+\rho$), ($\omega_2=1-\rho$), and ($\nu_1=\nu_2=n$), and then applying the identity $({}_1F_1(a;2a;z)=e^{z/2}{}_0F_1(a+\tfrac12;z^2/16))$; hence, the two expressions are equivalent representations of the same density.
        
        The CDF is expanded in an infinite series \cite[Eq. 8]{joarderExactDistributionSum2013}, 
        \begin{equation}
        \label{eq:JoarderOmarCDF}
        F_{\rs{Q}}(q)=\dfrac{(1-\rho^2)^{n/2}}{2^n\Gamma(n)}\sum_{k=0}^\infty \dfrac{\Gamma\left(\frac{n+1}{2}\right)}{\Gamma\left(k+\frac{n+1}{2}\right)k!}\left(\dfrac{\rho^2}{4}\right)^k\gamma\left(2k+2n,\dfrac{q}{2-2\rho^2}\right)    
        \end{equation}        
        
        \subsubsection{Product of two normal variables} 
        Another useful example is the product of two normal random variables. 
        Cui et al. \cite{cuiExactDistributionProduct2016} derive the exact density of the product of two correlated normal random variables. In the quadratic-form framework, this product can be written as \(\rs{Q}=\rv{x}^T\dmt{A}\rv{x}\), \(\rv{x}\sim \mathcal{N}(\dv{\mu},\dmt{\Sigma})\), \(\dv{\mu}=[\mu_1,\mu_2]^T\), \(\dmt{\Sigma}=\begin{bmatrix}
                \sigma_1^2 & \rho\sigma_1\sigma_2\\
                \rho\sigma_1\sigma_2 & \sigma_2^2
            \end{bmatrix}\), \(\dmt{A}=\begin{bmatrix}
                0 & \frac12\\
                \frac12 & 0
            \end{bmatrix}\).
            Indeed, $\rs{Q} = \rs{x}_1 \rs{x}_2$. Since the means \(\mu_1\) and \(\mu_2\) are arbitrary, this is in general noncentral quadratic form. Diagonlaizing $\dmt{A}\dmt{\Sigma}$ shows that $\rs{Q}$ is equivalent to a linear combination of two independent non-central chi-square random variables with coefficients \(\dfrac{1}{2}(1+\rho)\sigma_1\sigma_2\) and \(-\dfrac{1}{2}(1-\rho)\sigma_1\sigma_2\); hence it is indefinite. 
            The PDF is given by \cite[Eq.2]{cuiExactDistributionProduct2016} \begin{equation}
                \begin{aligned}
            & f_{\rs{Q}}(q)=\exp \left\{-\frac{1}{2\left(1-\rho^2\right)}\left(\frac{\mu_1^2}{\sigma_1^2}+\frac{\mu_2^2}{\sigma_2^2}-\frac{2 \rho\left(q+\mu_1 \mu_2\right)}{\sigma_1 \sigma_2}\right)\right\} \\
            & \times \sum_{n=0}^{\infty} \sum_{m=0}^{2 n} \frac{q^{2 n-m}|q|^{m-n} \sigma_1^{m-n-1}}{\pi(2 n) !\left(1-\rho^2\right)^{2 n+1 / 2} \sigma_2^{m-n+1}}\left(\begin{array}{c}
            2 n \\
            m
            \end{array}\right)\left(\frac{\mu_1}{\sigma_1^2}-\frac{\rho \mu_2}{\sigma_1 \sigma_2}\right)^m \\
            & \times\left(\frac{\mu_2}{\sigma_2^2}-\frac{\rho \mu_1}{\sigma_1 \sigma_2}\right)^{2 n-m} K_{m-n}\left(\frac{|q|}{\left(1-\rho^2\right) \sigma_1 \sigma_2}\right)
            \end{aligned} \label{eq:Cui}
            \end{equation} where $K_\nu(.)$ is the modified Bessel function of the second kind of order $\nu$.
            This example shows that even the product of two Gaussian random variables leads to an indefinite noncentral quadratic form whose exact density is considerably more complicated than the positive-definite central cases discussed earlier.
\subsection{Positive Definite Forms} \label{sec:SeriesExpansionsforPositiveDefiniteForms}
Positive-definite forms are especially important subclass, because their support is one-sided and their transforms admits systematic series expansions.

To see where definiteness helps, recall that the MGF is the Laplace transform of the PDF after the substitution $t =-s$. Consider now an expansion of the one-sided Laplace transform/MGF as an infinite sum of a sequence of functions $\hat{h}_k(s)$, 
        \begin{equation}
            \hat{f}_\rs{Q}(s) = \sum_k c_k \hat{h}_k(s).
        \end{equation}
The functions  $\hat{h}_k(s)$ should be chosen so that they are easily invertible, e.g., $(x^k \leftrightarrow \frac{1}{s^{k+1}}), f_{\chi^2_{m}}(x) \leftrightarrow \frac{1}{(1 + 2s)^{k/2}})$. Applying the inverse transform we find 
\begin{equation}
    \label{eqn:inversionOfPositiveDefiniteMGF}
    f_\rs{Q}(q) = \int f_\rs{Q}(s) e^{sq}ds = \int  \sum_k c_k \hat{h}_k(s) e^{sq} ds 
\end{equation}
What we would like to do is interchange the sum and integral, therby inverting the series term by term. This requires the  conditions of the dominated convergence theorem (DCT) must be satisfied. As $ \sum_k c_k \hat{h}_k(s) e^{sq}$ is not yet known, we apply DCT to first equality of \cref{eqn:inversionOfPositiveDefiniteMGF}, 
        \begin{equation}
            \hat{f}(s) = \exp\left( \sum_{n=1}^\rankFinal -\frac{h_n^2\lambda_ns}{(1+2\lambda_n s)^{1/2}}\right)\times\prod_{n=1}^\rankFinal (1+2\lambda_n s)^{-\frac{1}{2}}.
        \end{equation}
An easy choice for the dominating function is $e^{\alpha s}$. This will however not work if any of the $\lambda_i$'s are negative since the exponential would grow indefinitely as $s$  increases. A positive definite quadratic form allows us to dominate this function with $e^{\alpha s}$. Mathai \cite[Sec 4.2, page 91]{mathaiQuadraticFormsRandom1992} provides a similar argument for the PDF. These arguments are equivalent. Note that negative definite forms can also be dominated by pulling out a minus sign. It is indefinite forms that fail, and cannot be dominated on both sides by an exponential.  To summarize the method: 
    \begin{enumerate}
        \item Factor the Laplace transform: $ \hat{f}_\rs{Q}(s) = A(s)B(\theta(s)) $, where $A(s)$ chosen so that $\mathcal{L}^{-1}{A(s)}$ is known, and $\theta(s)$  is a change of variables that makes the remainder $B(\theta)$ analytic around $\theta=0$.
        \item  Expand $B(\theta)$  as a power series using the log-series trick: Compute a series for
        \[
        \log B(\theta)=d_0+\sum_{k\ge1}\frac{d_k}{k}\theta^k,
        \]
        then recover
        \[
        B(\theta)=\sum_{k\ge0}c_k\theta^k
        \]
        using the standard recursion
        \[
        c_k=\frac{1}{k}\sum_{r=0}^{k-1}d_{k-r}c_r,\qquad k\ge1,
        \]
        with $c_0=e^{d_0}$ (or $c_0=1$ depending on normalization).
        We obtain the coefficients $\{c_k\}_{k=0}^K$ up to truncation $K$. This coefficient pipeline is essentially common to all expansions.
        \item Form the series for the Laplace transform as follows:
            \[
            L(s)=A(s)\sum_{k\ge0}c_k\theta(s)^k
            =\sum_{k\ge0} c_k \widehat h_k(s),
            \]
            where $\widehat h_k(s)=A(s)\theta(s)^k$ .
        \item Use the dominated convergence argument to justify
            \[
            f_Q(q)=\mathcal{L}^{-1}{L(s)}
            =\sum_{k\ge0} c_k\mathcal{L}^{-1}{\widehat h_k(s)}
            =\sum_{k\ge0} c_kh_k(q).
            \] 
            Similarly, integrate term-by-term to obtain the CDF:
            \[
            F_Q(q)=\int_0^q f_Q(u)du
            =\sum_{k\ge0} c_kH_k(q).
            \]
    \end{enumerate}
    
\subsubsection{Definite Real Forms} 

Consider first the case of a positive definite form, i.e., \eqref{eq:lcOnlyChiQFGRV} with positive eigenvalues, as follows $$\rs{Q}\stackrel{d}{=}\sum_{n=1}^N \lambda_n(\rs{z}_n+h_n)^2,\quad \lambda_j>0,\;\;j=1,2,\ldots,N$$ where $\rs{z}_1,\rs{z}_2,\ldots,\rs{z}_N$ are independent standard normal variables.  

\underline{Power Series Expansion}:
After several earlier attempts for the positive-definite central forms, e.g.,\cite{shahDistributionDefiniteQuadratic1961,pacharesNoteDistributionDefinite1955}, Kotz \cite{kotzSeriesRepresentationsDistributions1967Central,kotzSeriesRepresentationsDistributions1967NonCentral} gives a widely used expansion that applies in the noncentral setting. Using $A(s) = s^{-N/2}, \theta(s) = \frac{1}{s}$, the PDF and the CDF of the quadratic form can be written as  \cite[Sec 4.2b] {mathaiQuadraticFormsRandom1992}:
\begin{equation}f_{\rs{Q}}(q)=u(q)\sum_{k=0}^\infty (-1)^kc_k\dfrac{q^{\frac{N}{2}+k-1}}{\Gamma\left(\frac{N}{2}+k\right)}\qquad \label{eq:Power Series_PDF}
\end{equation}
where
\begin{align*}
    c_0&=\exp\left(-\dfrac{1}{2}\sum_{j=1}^{N}h_j^2\right)\prod_{j=1}^{N}(2\lambda_j)^{-1/2}\\
    d_k&=\dfrac{1}{2}\sum_{j=1}^N(1-kh_j^2)(2\lambda_j)^{-k} &k\geq 1
\end{align*}
and the CDF is given by:
\begin{equation}
    F_{\rs{Q}}(q)=u(q)\sum_{k=0}^\infty (-1)^kc_k\dfrac{q^{\frac{N}{2}+k}}{\Gamma\left(\frac{N}{2}+k+1\right)}
    \label{eq:Power Series_CDF}
\end{equation}
            
\underline{Laguerre Series}: The PDF of a positive definite form can be written as the product of a chi-square density and an infinite series of Laguerre polynomials. Gurland \cite{gurlandDefiniteIndefinite1955,gurlandQuadraticFormsNormally1956} does so for central positive definite forms. His work is generalized by Kotz \cite{kotzSeriesRepresentationsDistributions1967NonCentral} to account for non-central positive definite forms. The generalized Laguerre polynomials are defined as follows:
\begin{equation}
    \label{eq:Laguerre_polynomial}
    L_k^{(\alpha)}(x)=\dfrac{1}{k!}e^xx^{-\alpha}\left[\dfrac{d^k}{dx^k}(e^{-x}x^{k+\alpha})\right]\;\;\;\;\;\;\; \alpha>-1,k=0,1,\ldots
\end{equation}
            
The PDF and CDF are obtained by choosing $A(s) = (1+2\beta s)^{-N/2}, \quad \theta(s) = \frac{2\beta s}{1+2\beta s}, \quad \gamma_j = 1 - \frac{\lambda_j}{\beta}$,  where $\beta$ is arbitrary positive parameter. These choices are made so that each factor $1 + 2\lambda_js$ can be written as $(1+2\beta s)(1-\gamma_j \theta)$. The PDF is given by  \cite[\S 4.2c]{mathaiQuadraticFormsRandom1992}:
\begin{align}
    f_{\rs{Q}}(q)&=u(q)\sum_{k=0}^\infty c_k\dfrac{k!}{2\beta\Gamma(\frac{N}{2}+k)}\left(\dfrac{q}{2\beta}\right)^{\frac{N}{2}-1}e^{-q/(2\beta)}L_k^{(\frac{N}{2}-1)}\left(\dfrac{q}{2\beta}\right)  \label{eq:Laguerre Series_PDF}
\end{align}
where:
\begin{align*}
    \beta&>\lambda_1/2\\
    c_0&=1\\
    d_k&=\dfrac{1}{2}\left\{-\dfrac{k}{\beta}\sum_{j=1}^{ N}\lambda_jh_j^2\left(1-\dfrac{\lambda_j}{\beta}\right)^{k-1}+\sum_{j=1}^{N}\left(1-\dfrac{\lambda_j}{\beta}\right)^k\right\} \quad k \geq 1
\end{align*}
and the CDF can be written as:
\begin{align}
    F_{\rs{Q}}(q)&=u(q)\left[\dfrac{\gamma(N/2,q/(2\beta))}{\Gamma(N/2)}\right.\nonumber\\&\left.+\sum_{k=1}^\infty c_k \dfrac{(k-1)!}{\Gamma(N/2+k)}\left(\dfrac{q}{2\beta}\right)^{\frac{N}{2}}e^{-\frac{q}{2\beta}}L_k^{(N/2)}\left(\frac{q}{2\beta}\right)\right]
            \label{eq:Laguerre Series_CDF}
\end{align}
        Kotz \cite{kotzSeriesRepresentationsDistributions1967Central,kotzSeriesRepresentationsDistributions1967NonCentral} recommends using \(\beta=\frac{1}{2}(\max \lambda_j+\min \lambda_j)\). 
        
        \underline{Chi-square Densities Expansion}: Ruben \cite{rubenProbabilityContentRegions1962} proposes a related expansion  in terms of chi-square densities. It also entails an arbitrary positive parameter $\beta$ and chooses $A(s) = (1+2\beta s)^{-N/2}, \quad \theta = \frac{1}{1+2\beta s}, \quad \eta_j = 1 - \frac{\beta}{\lambda_j}$. The choices are made so that each factor 
        $1 + 2\lambda_js=\frac{\lambda_j}{\beta}(1+2\beta s)(1-\eta_j \theta)$. The PDF and CDF are given by   \cite{mathaiQuadraticFormsRandom1992}:
        \begin{align}
            f_{\rs{Q}}(q)&=u(q)\sum_{k=0}^\infty c_k\dfrac{q^{N/2+k-1}e^{-q/2\beta}}{(2\beta)^{N/2+k}\Gamma(N/2+k)}=\sum_{k=0}^\infty c_k f_{\beta \chi_{N+2k}^2}(q) \label{eqn:Chi-square Expansion_pdf}\\ 
            F_{\rs{Q}}(q)&=u(q)\sum_{k=0}^\infty c_k \dfrac{\gamma(N+2k,q/\beta)}{N+2k}=\sum_{k=0}^\infty c_k F_{\beta \chi_{N+2k}^2}(q) \label{eqn:Chi-square Expansion_cdf}
        \end{align}
        where 
        \begin{align*}
            c_0&=\exp\left(-\dfrac{1}{2}\sum_{j=1}^Nh_j^2\right)\prod_{j=1}^N(\beta/\lambda_j)^{1/2}\\
            c_k&=(2k)^{-1}\sum_{r=0}^{k-1}d_{k-r}c_r\\
            d_k&=\sum_{j=1}^N(1-\beta/\lambda_j)^k+k\beta \sum_{j=1}^N(h_j^2/\lambda_j)(1-\beta/\lambda_j)^{k-1}
        \end{align*}
        Ruben \cite{rubenProbabilityContentRegions1962} recommends using the harmonic mean of the maximum and minimum eigenvalues for \(\beta\).
    
Farebrother \cite{farebrotherAlgorithm204Distribution1984} develops an algorithm to implement this series expansion. In a similar manner, Moschoupolos and Canada \cite{moschopoulosDistributionFunctionLinear1984} invert the MGF of a linear combination of a linear combination of chi-squares into a CDF written as an infinite series of incomplete gamma functions. The major differences are taking multiple eigenvalues into consideration and fixing $\beta$ to be one of the eigenvalues. (?)
            
For an even number of central variables, Ropokis et al. \cite{ropokisQuadraticFormsNormal2008} provide a bound for the truncation error. Indeed, defining this error by\[e(q)=\sum_{k=K+1}^\infty c_k\beta^{-1}g(N+2k;q/\beta)\] it is bounded by \[e_b(q)=c_0\sum_{i=1}^{N/2}\Delta_ih(q,\xi_{2i-1}2\beta,K+N/2-1)\]
where 
\begin{align*}
    h(q,a,b,m)&=\dfrac{\exp\left[-\dfrac{(1-a)q}{b}\right]}{b}-\sum_{i=0}^ma^k\dfrac{q^{k}\exp\left(-\dfrac{q}{b}\right)}{b^{k+1}\Gamma(k+1)}\\
    &\text{ are decreasing order of } \\
    \eta_i&=1-\dfrac{\beta}{\lambda_i}\\
    \Delta_i&=\prod_{\substack{\ell=1 \\ \ell\neq i}}^{N/2}\dfrac{1}{\xi_{2i-1}-\xi_{2\ell-1}}
\end{align*}
        and $\{\xi_i\}_{i=1}^N$ are the decreasing sorting of $|\eta_i|$, i.e., $$\xi_1\geq \xi_2 \geq \ldots \geq \xi_N$$ and for some permutation $(n_1,n_2,\ldots,n_N)$ of $(1,2,\ldots,N)$, we have $$\xi_1=|\eta_{n_1}|,\;\xi_2=|\eta_{n_2}|,\ldots,\;\xi_N=|\eta_{n_N}|.$$ 
Computational complexity of the bound is $O(K^2N)$. Also, the bound shows the series expansion performs better when the spacing between the eigenvalues is larger. 

The main benefit of the series-expansion method is that it avoids the direct inversion methods, where the inversion cannot be carried out easily. Indeed, series expansion methods make use of coefficients and evaluate well-known functions (Power, Laguerre, chi-squares).  Once the coefficients are computed, the same expansion can be reused efficiently over a grid of values compared to doing the inversion at each value.  However, infinite series must be truncated in practice, with accuracy then depending on how rapidly coefficients decay in the area of interest (near the origin vs moderate values vs tails).


\subsection{Indefinite Forms} \label{sec:Indefinite Forms}
We now turn to indefinite forms, In contrast to the positive-definite case, the support is no longer one-sided, and the transform typically loses the structure that makes the previous expansions straightforward. As a result, exact results still exist, but they are analytically more involved. We present results that solve the PDF/CDF of indefinite quadratic forms.  First, we consider the case of an indefinite complex quadratic form, then the indefinite real quadratic forms, which provide formulae that evaluate the PDF and CDF for arbitrary parameters. The indefinite complex forms are restricted to even multiplicity of the degrees of freedom.

\subsubsection{Indefinite Complex (Even multiplicity)} \label{ssec:IndefiniteComplex(Even multiplicity)}  
Raphaeli \cite{raphaeliDistributionNoncentralIndefinite1996} derives a series expansion for the PDF of an indefinite quadratic form. Raphaeli  derives the PDF and the CDF of a complex quadratic form \eqref{eq:QFCGRV}. Examining the form of the MGF provided in \cite{raphaeliDistributionNoncentralIndefinite1996}, we notice that in the noncentral case, there are essential singularities (unremovable poles) at every point where the exponential becomes infinite. In particular, these points are where $s=1/\lambda_j$. Raphaeli's derivation is based on recognizing that an essential singularity is equivalent to an infinite number of poles which can be obtained through Laurent expansion. Raphalei's derivation consists of successively relocating every pole to the origin, expanding the essential singularity part of the function as Laurent series and the analytic part as a Taylor series.
\begin{align}
    f(q)=&-\exp \left\{-\dfrac{1}{2}\sum_j\tilde{h}_k^2\right\}\sum_{\tilde{\lambda}_k >0} \left[\dfrac{1}{(-\tilde{\lambda}_k)^{\nu_k}}\exp(-\tilde{\lambda}_k^{-1}q)\right. \nonumber\\
    		&\left.\times\sum_{n=\nu_k-1}^\infty g_k^{(n)}(0,q)\dfrac{(-\beta_k)^{n-\nu_k+1}}{n!(n-\nu_k+1)!}\right] \label{eq:Raphaelli_pdf}
    		\end{align}
    where 
\begin{align*}
    g_k^{(n)}&=\sum_{l=0}^{n-1}\binom{n-1}{l}g_k^{(l)}(\ln g_k)^{(n-l)}\\
    [\ln g_k(s,q)]^{(n)}&=\dfrac{1}{2}\sum_{j \neq k}\dfrac{n!\tilde{\lambda}_j^n\tilde{h}_j^2}{(\alpha_{kj}-\tilde{\lambda}_js)^{n+1}}+\sum_{j \neq k}\dfrac{
    (n-1)!\tilde{\lambda}_j^n\nu_j}{(\alpha_{kj}-\tilde{\lambda}_js)^n}\\
    g_k(s,q)&=\exp\left(\frac{1}{2}\sum_{j\neq k}\dfrac{\tilde{h}_j^2}{\alpha_{kj}-\tilde{\lambda}_js}-sq\right)\prod_{j\neq k}\dfrac{1}{(\alpha_{kj}-\tilde{\lambda}_js)^{m_j}}\\
    		  \alpha_{kj} &= 1-\tilde{\lambda}_j/\tilde{\lambda}_k
\end{align*}
        Similarly, The CCDF is given by the following expression: 
         \begin{equation}
         \label{eq:Raphaeli_CCDF}
         \begin{aligned}
         \mathbb{P}[\rs{Q} \geq q] &= -\exp \left(-\frac{1}{2} \sum_{j=1}^K \mu_j^2\right) \sum_{\lambda_k>0} \frac{1}{\left(-\lambda_k\right)^{m_k}} e^{-q / \lambda_k} \\ &\sum_{n=m_k-1}^\infty \hat{g}_k^{(n)}(0, q) \frac{\left(-\beta_k\right)^{n-m_k+1}}{n!\left(n-m_k+1\right)!},
         \end{aligned}
         \end{equation}
It is worth noting that the inner infinite sum converges easily, especially when the eigenvalues are separated 'enough' from each other, allowing for good results using few terms (less than 20, as tested in Matlab).  
Note that \cite{al-naffouriDistributionIndefiniteQuadratic2016} reports that the tail behavior does not perform well. However, their observation is based on integrating the PDF numerically, and this will consequently affect the accuracy of the tail. A direct implementation of \Cref{eq:Raphaeli_CCDF} produces reasonable behavior.

\subsubsection{Indefinite Real Forms} \label{ssec:IndefiniteRealForms}
Consider an indefinite form as in \eqref{eq:lcOnlyChiQFGRV} \[\rs{Q}\stackrel{d}{=}\sum_{n=1}^N \lambda_n(\rs{z}_n+h_n)^2\]
Sort the eigenvalues $\lambda_j$ in non-increasing order. Without loss of generality, we can assume that all the eigenvalues are nonzero. Suppose that the form has precisely $\ell$ positive eigenvalues. Then it can be written as the difference of two positive definite forms:
\[
\rs{Q}\stackrel{d}{=}\overbrace{\sum_{j=1}^\ell \lambda_j (\rs{z}_j+h_j)^2}^{\rs{Q}_1} - \overbrace{\sum_{j=\ell+1}^N|\lambda_j|(\rs{z}_j+h_j)^2}^{\rs{Q}_2}
\]
Provost and Rudiuk  \cite{provostExactDistributionIndefinite1996} derive the exact distribution by representing the densities of ($\rs{Q}_1$) and ($\rs{Q_2}$) via Ruben infinite series (gamma/chi-square mixtures) with coefficients computed recursively, substituting these expansions into the convolution identity for the difference ($\rs{Q}_1- \rs{Q}_2$), and evaluating the resulting term-by-term integrals using a tabulated formula that yields Whittaker functions, producing a doubly-infinite sum series for the PDF.   Finally, they integrate the PDF to obtain the CDF (involving incomplete gamma functions) and treat a special half-integer parameter regime separately by switching to an alternative Whittaker expansion when the standard hypergeometric representation fails.  
        
If $\ell$ and $N-\ell$ are not both odd, the PDF can be written as:
\begin{equation}
    f_{\rs{Q}}(q)=
    \begin{cases}
    \displaystyle\sum_{k=0}^\infty \sum_{\nu=0}^\infty \dfrac{\theta_k \theta_\nu'}{\Gamma(\ell/2+k)}b^{-s}q^{s} e^{\tilde{b}q}W_{\lambda,1-s}(bq) \text{ if } q>0,\\
    \displaystyle\sum_{k=0}^\infty \sum_{\nu=0}^\infty \dfrac{\theta_k \theta_\nu'}{\Gamma\left(\frac{N-\ell}{2} +\nu\right)}b^{-s}(-q)^{s-1}e^{\tilde{b}q}W_{-\lambda,1-s}(-bq) \text{ if } q\leq 0. \label{eq:ProvostRudiuk1_pdf}
    \end{cases}
\end{equation}
         where:
           $\beta$ and $\beta'$ are arbitrarily chosen so that 
                $$
                \min_{j\in \{1,\ldots,\ell\}}\left|1-\dfrac{\beta}{\lambda_j}\right|<1,\qquad \min_{j\in \{\ell+1,\ldots,N\}}\left|1+\dfrac{\beta'}{\lambda_j}\right|<1,
                $$
                
                \UglyAlignS{
                    s&=\frac{N}{4}+\frac{k+\nu}{2} & \lambda&=\frac{2\ell-N}{4}+\frac{k-\nu}{2}\\
                    b&=\frac{\beta^{-1}+\beta'^{-1}}{2} & \tilde{b}&=\frac{\beta'^{-1}-\beta^{-1}}{4}\\
                    \theta_k&=a_k/(2\beta)^{\ell/2+k} & \theta'_k&=a'_k/(2\beta')^{(N-\ell)/2+k}\\
                    a_0&=e^{-d/2}\prod_{j=1}^\ell\left(\dfrac{\beta}{\lambda_j}\right)^{1/2} & a'_0&=e^{-\gamma/2}\prod_{j=\ell+1}^{N}\left(\dfrac{\beta'}{\lambda_j}\right)^{1/2}\\
                    d&=\sum_{j=1}^\ell h_j^2 & \gamma&=\sum_{j=\ell+1}^t h_j^2\\
                    a_k&=(2k)^{-1}\sum_{r=0}^{k-1}e_{k-r}a_r & a'_k&=(2k)^{-1}\sum_{r=0}^{k-1}e'_{k-r}a_r\\
                    e_k&=k\beta \sum_{j=1}^\ell \dfrac{h_j^2}{\lambda_j}\left(1-\dfrac{\beta}{\lambda_j}\right)^{k-1}+\sum_{j=1}^\ell \left(1-\dfrac{\beta}{\lambda_j}\right)^k\\
                    e'_k&=k\beta' \sum_{j=\ell+1}^{N} \dfrac{h_j^2}{\lambda_j}\left(1-\dfrac{\beta'}{\lambda_j}\right)^{k-1}+\sum_{j=\ell+1}^N \left(1-\dfrac{\beta'}{\lambda_j}\right)^k
                }
            
            and $W_{a,b}(z)$ is the Whitaker W-function.
            
            If $\ell$ and $N-\ell$ are both odd, the PDF can be written as:
            \UglyEquation{
                f_{\rs{Q}}(q)=\begin{cases}
                \displaystyle \sum_{k=0}^\infty \sum_{\nu=0}^\infty \dfrac{\theta_k \theta_\nu'}{\Gamma((N-\ell)/2+\nu)} \dfrac{(-1)^{2\mu}e^{q/2\beta'}}{\Gamma(\frac{1}{2}-\mu-l)\Gamma(\frac{1}{2}+\mu-l)}\\
                \times \left\{\displaystyle\sum_{i=0}^\infty \dfrac{\Gamma(-\mu+i-l+\frac{1}{2})}{i!(-2\mu+i)!} b^i\right.\\
                \times \left[\left(\psi(i+1)+\psi(-2\mu+i+1)-\psi(-\mu+i-l+\frac{1}{2})-\ln(b)\right)\right.\\
                \left.\times (-q)^{-2\mu+i}-\displaystyle \sum_{j=1}^\infty \sum_{n=0}^j \dfrac{(-1)^{n-1}}{n!(j-n)!}(-q)^{-2\mu+i+n}\right]\\
                \left.+\displaystyle\sum_{i=0}^{-2\mu-1}\dfrac{\Gamma(-2\mu-i)\Gamma(i+\mu+l+\frac{1}{2})}{i!}b^{2\mu+i}(-1)^{2\mu+i}(-z)^k\vphantom{\displaystyle\sum_{i=0}^\infty}\right\}\text{ if } q\leq 0,\\
                
                \displaystyle \sum_{k=0}^\infty \sum_{\nu=0}^\infty \dfrac{\theta_k \theta_\nu'}{\Gamma(\ell/2+k)} \dfrac{(-1)^{2\mu}e^{q/2\beta}}{\Gamma(\frac{1}{2}-\mu-l_2)\Gamma(\frac{1}{2}+\mu-l_2)}\\
                \times \left\{\displaystyle\sum_{i=0}^\infty \dfrac{\Gamma(-\mu+i-l_2+\frac{1}{2})}{i!(-2\mu+i)!}b^i\right.\\  \times\left[\left(\psi(i+1)+\psi(-2\mu+i+1)-\psi(-\mu+i-l_2+\frac{1}{2})-\ln(b)\right)\right.\\
                \left.\times q^{-2\mu+i}-\displaystyle \sum_{j=1}^\infty \sum_{n=0}^j \dfrac{(-1)^{n-1}}{n!(j-n)!}q^{-2\mu+i+n}\right]\\
                \left.\vphantom{\displaystyle\sum_{i=0}^\infty}+\displaystyle\sum_{i=0}^{-2\mu-1}\dfrac{\Gamma(-2\mu-i)\Gamma(i+\mu+l_2+\frac{1}{2})}{i!}b^{2\mu+i}(-1)^{2\mu+i}z^k\vphantom{\displaystyle\sum_{i=0}^\infty}\right\} \text{ if } q > 0
                \end{cases} \label{eq:ProvostRudiuk2_pdf}
            }
            where $\psi$ is the digamma function, $l=(N/2-\ell-k+\nu)/2$, $l_2=-l$ and $\mu=(1-N/2-k-\nu)/2$.
            Expressions for the CDF incorporating triple infinite summations are available \cite{provostExactDistributionIndefinite1996}.
\section{PDF/CDF: Numerical Inversion of the MGF/CF  }
\label{sec:NumericalInversionMGF}
In the previous section, we reviewed cases in which the PDF and CDF of a SQFs can be obtained through analytical inversion of the MGF or CF. We now turn to numerical inversion methods. These methods are especially useful when a closed-form analytical inversion is unavailable or when the exact representation is not convenient for computation. In practice, the main emphasis is usually on the CDF, since its inversion formulas are generally more numerically stable than those of the PDF.

\subsection{Gil–Pelaez inversion formula} 
\label{ssec:Gil–Pelaez inversion formula}
While the direct inversion formulas provide an exact analytical representation of the PDF and CDF, they are not the most convenient forms for numerical evaluation. Consider the direct CF-CDF inversion (see \Cref{eq: CDF_from_MGF}), 
\[
F_Q(q)=\frac{1}{2\pi}\int_{-\infty}^{\infty}
e^{-i\beta q}\frac{\phi_{\rs{Q}}(\beta)}{i\beta}d\beta,
\]

This inversion is not convenient to be carried out numerically since it involves a two-sided complex oscillatory integral with an apparent singularity at the origin through the factor $(iu)^{-1}$. Also, the positive and negative frequency parts undergo strong cancellation.  
Note that for the MGF-based inversion when $\gamma \neq 0$, the factor $e^{-\gamma q}$ may provide useful damping. Hence, MGF inversion with interior contour shift can sometimes be more stable than the CF inversion. However, this requires a careful contour selection.

Gil--Pelaez formula \footnote{Gurland \cite{gurlandDistributionQuadraticForms1953} derived it independently.} rewrites the CF inversion for the CDF as follows: 
\begin{equation}
\label{eq:gil-pelaez_CDF}
F_Q(q)=\frac12-\frac{1}{\pi}\int_0^\infty 
\Im\!\left\{\frac{e^{-i\beta q}\phi_\rs{Q}(\beta)}{\beta}\right\}d\beta,    
\end{equation}

which is a one-sided real integral exploiting the Hermitian symmetry of the CF. The corresponding PDF inversion is 
\begin{equation}
    \label{eq:gil-pelaezAlike_PDF}
f_Q(q)=\frac{1}{\pi}\int_0^\infty 
\Re\!\left\{e^{-i\beta q}\phi_\rs{Q}(\beta)\right\}d\beta.
\end{equation}

Thus, both the PDF and the CDF can be written as one-sided real integrals by exploiting the Hermitian symmetry of the CF. In this form, the cancellation between positive and negative frequency contributions is removed, which already makes both representations more suitable for numerical evaluation of the original two-sided inversion formulas. However, the CDF representation is usually more attractive numerically because it contains the additional $1/\beta$, which attenuates high-frequency oscillations. The PDF inversion does not contain this factor and is therefore typically more oscillatory and less well conditioned, especially in the tails.

When the Gil--Pelaez formula is evaluated numerically, the resulting CDF approximation is affected by several distinct error sources. To make these errors explicit, it is useful to write the numerical approximation associated with \Cref{eq:gil-pelaez_CDF}. In practice, the improper integral is first truncated to a finite interval $[0,U]$, 
\begin{equation}
\label{eq:gil-pelaez_CDF_truncated}
\tilde{F}_\rs{Q}(q)=\frac12-\frac{1}{\pi}\int_0^U
\Im\!\left\{\frac{e^{-i\beta q}\phi_\rs{Q}(\beta)}{\beta}\right\}d\beta,    
\end{equation}
and the truncated integral is then approximated by a numerical quadrature rule,
\begin{equation}
    \label{eq:gil-Pelaez_CDF_numerical}
    \widehat{F}_\rs{Q}(q) = \dfrac{1}{2} - \dfrac{1}{\pi} \sum_{k=1}^K w_k h_q(\beta_k), \quad \beta_k \in [0,U],
\end{equation}
where $h_q(\beta_k)\triangleq \Im\!\left\{\frac{e^{-i\beta_k q}\phi_\rs{Q}(\beta_k)}{\beta_k}\right\}$, $w_k$ is the quadrature weight associated with the quadrature node, $\beta_k$. Accordingly, the total numerical error can be decomposed as 
\begin{equation}
    F_\rs{Q}(q) - \widehat{F}_{\rs{Q}}(q) = e_\text{trunc} + e_\text{quad} +  e_0 +  e_\text{round},
\end{equation}
where each term has a clear origin as follows:
\begin{itemize}
    \item Truncation error ($e_\text{trunc}$): truncating the semi-infinite interval $[0,\infty)$ to $[0,U]$ introduces a tail error, $e_\text{trunc} = \frac{1}{\pi}\int_U^\infty 
    \Im\!\left\{\frac{e^{-i\beta q}\phi_\rs{Q}(\beta)}{\beta}\right\}d\beta, $
    \item Discretization/Quadrature error ($e_\text{quad}$): the truncated integral must be approximated numerically, and this discretization/quadrature error can be significant because the integrand is oscillatory.
    \item Near origin evaluation error ($e_0$): although the factor $1/\beta$ suggests a singularity, the integrand has a removable singular form at $\beta=0$; nevertheless, special numerical treatment is required near the origin and any approximation used there contributes to $E_0$.   
    \item Floating-point roundoff error ($ e_\text{round}$): finite-precision arithmetic introduces roundoff, which may be amplified by oscillatory cancellation and by the computation of small tail probabilities.
\end{itemize}    

For quadratic forms, where the CF is available in closed form, this structure leads directly to efficient numerical inversion schemes such as those of Imhof and Davies. 

Imhof \cite{imhofComputingDistributionQuadratic1961} applies Gil-Pelaez's formula to a quadratic form in Gaussian random variables written as \eqref{eq:lcOnlyChi}
\begin{equation*}
    \rs{Q}\sim \sum_{\ell=1}^L \omega_\ell \chi_{\nu_\ell}^2(\delta_\ell^2) 
\end{equation*}  
By expressing the complex integrand in modulus-phase form and introducing the change of variable  $u=2\beta$, he obtains the following real oscillatory representation:
\begin{equation}
    F_{\rs{Q}}(q)=\dfrac{1}{2}-\dfrac{1}{\pi}\int_0^{\infty}\dfrac{\sin \theta(u)}{u \rho(u)}du
    \label{eq:Imhof_CDF}
\end{equation}
where 
\begin{align*}
\theta(u)&=\dfrac{1}{2}\sum_{\ell=1}^L\left( \nu_\ell\tan^{-1}(\omega_\ell u)+\delta_\ell^2\omega_\ell u(1+\omega_\ell^2u^2)^{-1} \right)-\dfrac{uq}{2},\\
\rho(u)&=\prod_{\ell=1}^L(1+\omega_\ell^2u^2)^{\nu_\ell/4}\exp[\dfrac{1}{2}\sum_{\ell=1}^L\frac{(\delta_\ell\omega_\ell u)^2}{1+\omega_\ell^2u^2}].
\end{align*} 
The corresponding PDF is obtained from the analogous cosine integral, 
\begin{equation}
    f_{\rs{Q}}(q)=\dfrac{1}{2\pi}\int_0^\infty [\rho(u)]^{-1}\cos[\theta(u)]du
    \label{eq:Imhof_PDF}
\end{equation} 
Imhof then evaluates \eqref{eq:Imhof_CDF} numerically over a finite interval $[0,U]$, which can be written similar to \Cref{eq:gil-Pelaez_CDF_numerical}  
\begin{equation}
    \widehat{F}_\rs{Q}(q) = \dfrac{1}{2} - \dfrac{1}{\pi} \sum_{k=1}^K w_k \dfrac{\sin \theta (u_k)}{u_k \rho(u_k)}, \qquad u_k \in [0,U].
\end{equation}
Using a uniform grid, $u_k = k\Delta $, $k=0,\cdots, K $ and $\Delta = \frac{U}{K}$, the CDF using trapezoidal approximation becomes
\begin{equation}
    \label{eq:Imhof_Trapezoidal}
        \widehat{F}_\rs{Q}(q) = \dfrac{1}{2} - \dfrac{\Delta}{\pi} \left[ \frac{1}{2} \dfrac{\sin \theta (u_0)}{u_0\rho(u_0)} + \sum_{k=1}^{K-1}  \dfrac{\sin \theta (u_k)}{u_k \rho(u_k)} + \frac{1}{2} \dfrac{\sin \theta (u_K)}{u_K\rho(u_K)} \right].
\end{equation}
The term at $u_0 = 0$ is evaluated by its limit,
\begin{equation}
   \lim_{u \to 0} \dfrac{\sin \theta (u)}{u\rho(u)} = \frac{1}{2} \sum_{\ell=1}^L \omega_\ell(\nu_\ell + \delta^2_\ell) - \frac{1}{2}q
\end{equation}

For the omitted tail (truncation error),  $
e_{\text{trunc}}=\frac1\pi\int_U^\infty \frac{\sin \theta(u)}{u\rho(u)}du,$ 
Imhof derives an explicit bound $|e_{\text{trunc}}| < T_U$,
\begin{equation}
    T_U= \frac{1}{\pi kU^k \prod_{\ell=1}^L |\omega_\ell|^{\nu_\ell/2}} \exp\left\{-\frac12\sum_{\ell=1}^L \frac{\delta_\ell^2\omega_\ell^2U^2}{1+\omega_\ell^2U^2}\right\},
\end{equation}
where $ k =\frac{1}{2} \sum_{\ell=1}^L \nu_\ell$. This gives a practical way to choose $U$ so that the truncation error is below a prescribed tolerance.

Davies \cite{daviesNumericalInversionCharacteristic1973} shows that sampling any CF on a lattice does not recover $F_\rs{x}$ alone, but rather $F(x)$ together with additional contributions from translated tails, so the CDF is expressed as
\begin{equation}
\begin{aligned}
F(x)&=\frac{1}{2}-\Delta \sum_{k=-\infty}^{\infty} \Im\left\{\frac{\phi(u+k \Delta) e^{-i(u+k \Delta) x}}{2 \pi(u+k \Delta)}\right\} \\ 
& \underbrace{-\sum_{n=1}^{\infty}\left[F\left(x-\frac{2 \pi n}{\Delta}\right)-\operatorname{Pr}\left(X>x+\frac{2 \pi n}{\Delta}\right)\right] \cos \left(\frac{2 \pi n u}{\Delta}\right)}_{e_\Delta(x)}.
\end{aligned}
\end{equation}
where $e_\Delta (x) $ collects the contributions from translated tails. Davies shows that when sampling at midpoint $u = \Delta/2$, yeilds
$$
e_\Delta(x) = \sum_{n=1}^{\infty}(-1)^n\left[F_X\left(x-\frac{2 \pi n}{\Delta}\right)-\operatorname{Pr}\left(X>x+\frac{2 \pi n}{\Delta}\right)\right],
$$
which can be bounded as follows:
\begin{equation}
\label{eq:davies_eDelta}
\left|e_\Delta(x)\right| \leq \max \left\{F_X\left(x-\frac{2 \pi}{\Delta}\right), \operatorname{Pr}\left(X>x+\frac{2 \pi}{\Delta}\right)\right\}.
\end{equation}
Thus, the lattice spacing $\Delta$ controls the error that depends on translated left and right tails of the distribution. In practice, these tails may be bounded using exponential tail bounds, for example, Chernoff bounds when the MGF is available.

With midpoint sampling, the CDF is computed as 
\[
F_X(x)
=
\frac12 - 
\sum_{k=0}^{K}
\frac{ \Im \left[ \phi_X\bigl((k+\tfrac12)\Delta\bigr)e^{-i(k+\frac12)\Delta x }\right] }{\pi(k+\tfrac12)} + e_\text{trunc} + e_\Delta,
\]
where $$ 
e_\text{trunc} = 
\sum_{k=K+1}^{\infty}\frac{ \Im \left[ \phi_X\bigl((k+\tfrac12)\Delta\bigr)e^{-i(k+\frac12)\Delta x }\right] }{\pi(k+\tfrac12)}
$$
The use of the half-grid $(k+\frac12)\Delta$ is advantageous because it avoids direct evaluation at $u=0$.

Davies later \cite{daviesAlgorithm155Distribution1980} specializes this approach to quadratic forms with Gaussian terms of the form 
\[\rs{Q}\sim \sum_{\ell=1}^L \omega_\ell\chi_{\nu_\ell}^2(\delta_\ell^2)+\sigma\mathcal{N}(0,1)\],

and develops an error-controlled algorithm that exploits the explicit characteristic function of the quadratic form together with explicit truncation bounds. In this case, the infinite midpoint-lattice representation can be written as
\begin{equation}
\label{eq:Fc2_davies}
F_{\rs{Q}}(q)
=
\lim_{\substack{\Delta\to 0\\ K\to\infty\\ (K+\tfrac12)\Delta\to\infty}}
\left[
\frac12-
\sum_{k=0}^{K}
\frac{A(u_k)\sin\!\bigl(\Theta(u_k)-u_k q\bigr)}{\pi(k+\tfrac12)}
\right],
\quad
u_k=(k+\tfrac12)\Delta,
\end{equation}
where
\begin{equation*}
A(u)
=
\exp\!\left\{
-2u^2\sum_{\ell=1}^{L}\frac{\omega_\ell^2\delta_\ell^2}{1+4u^2\omega_\ell^2}
-\frac{u^2\sigma^2}{2}
\right\}
\prod_{\ell=1}^{L}\left(1+4u^2\omega_\ell^2\right)^{-\nu_\ell/4},
\end{equation*}
and
\begin{equation*}
\Theta(u)
=
\sum_{\ell=1}^{L}
\left[
\frac{\nu_\ell}{2}\arctan(2u\omega_\ell)
+
\frac{\delta_\ell^2u\omega_\ell}{1+4u^2\omega_\ell^2}
\right].
\end{equation*}
In practice, the series is truncated at \(k=K\), so the performance of the method depends critically on the choice of the truncation point
\[
U=(K+\tfrac12)\Delta
\]
and the lattice spacing \(\Delta\). Davies derives three computable upper bounds for the truncation tail and uses the smallest applicable one:
\begin{align}
B_1(U)
&=
\frac{2}{\pi s}\,
N(U)\exp\!\left(-\frac{U^2\sigma^2}{2}\right)
\prod_{(\mathrm{i})}\left(1+4U^2\omega_j^2\right)^{-\nu_j/4}
\prod_{(\mathrm{ii})}\left(4U^2\omega_j^2\right)^{-\nu_j/4},
\label{eq:Davies-B1}
\\
B_2(U)
&=
\frac{1}{\pi U^2\sigma^2}\,
N(U)\exp\!\left(-\frac{U^2\sigma^2}{2}\right)
\prod_{j=1}^{L}\left(1+4U^2\omega_j^2\right)^{-\nu_j/4},
\label{eq:Davies-B2}
\\
B_3(U)
&=
\frac{2.5}{\pi}\,
N(U)\exp\!\left(-\frac{U^2\sigma^2}{2}\right)
\prod_{j=1}^{L}\left(1+4U^2\omega_j^2\right)^{-\nu_j/4},
\label{eq:Davies-B3}
\end{align}
where
\begin{equation}
N(U)
=
\exp\!\left[
-2U^2\sum_{j=1}^{L}
\frac{\omega_j^2\delta_j^2}{1+4U^2\omega_j^2}
\right],
\qquad
s=\sum_{(\mathrm{ii})}\nu_j.
\end{equation}
Here, product \((\mathrm{i})\) is over all indices \(j\) such that \(|\omega_j|\le 1\), whereas product \((\mathrm{ii})\) is over all indices \(j\) such that \(|\omega_j|>1\). The bound \(B_3(U)\) is valid under Davies's condition (5), and the algorithm uses the smallest applicable upper bound. To reduce the required truncation point further, Davies introduces a Gaussian convergence factor \(e^{-\tau^2u^2/2}\); the resulting bias is then corrected by an auxiliary integration involving the complementary factor \(1-e^{-\tau^2u^2/2}\).

\subsubsection{Definite Forms}

                Consider a positive definite form $\rs{Q}$, written as in \eqref{eq:lcOnlyChi}
                \begin{equation*}
                     \rs{Q}\sim \sum_{\ell=1}^L \omega_\ell \chi_{\nu_\ell}^2(\delta_\ell^2) 
                \end{equation*} with $\omega_\ell>0$ for all $\ell$.
                
                \noindent Gardini \cite{gardiniMellinTransformManage2022} starts from Ruben's exoression of a positive quadratic form as an infinite mixture of  chi-square densities, and then computes its Mellin transform term-wise to invert it again into a PDF or CDF. Note that the Mellin transform of a function $f(q)$ is given by 
                $$
                \hat{f}(z)=\int_0^\infty q^{z-1}f(q)dq 
                $$ 
                The inverse Mellin transform is given by 
                $$
                f(q)=\frac{1}{2 \pi i} \int_{h-i \infty}^{h+i \infty} q^{-z} \hat{f}(z) \mathrm{d} z
                $$ where $h$ is an arbitrary parameter that belongs to the strip of analyticity of the Mellin transform.
                
                Gardini demonstrates that Mellin Transforms of the PDF and the CDF of a positive definite quadratic form are given by 
                \begin{align} 
                    \hat{f}_{\rs{Q}}(z)&=(2\beta)^{z-1}\sum_{k=0}^\infty c_k\dfrac{\Gamma(z+\alpha+k-1)}{\Gamma(\alpha+k)}, \;\;\Re(z)>1-\alpha \label{eq:Gardini-MellinPDF}\\
                    \hat{F}_{\rs{Q}}(z)&=-\dfrac{(2\beta)^{z}}{z}\sum_{k=0}^\infty c_k\dfrac{\Gamma(z+\alpha+k)}{\Gamma(\alpha+k)}, \;\;-\alpha<\Re(z)<0 \label{eq:Gardini-MellinCDF},
                \end{align}
                where $\alpha = \rankFinal/2$. Gardini goes on to invert the Mellin transform term-wise, truncating the improper integral and the series, and discretizing the resulting finite integral. as shown below
                \begin{align}
                    f_{\rs{Q}}(q) &\approx \dfrac{1}{2\pi i}\sum_{t=-T}^T \widehat f_{\rs{Q}}^{(K)}(z_t) q^{-z_t} \delta \label{eq:Gardini-PDF} \\    
                    F_{\rs{Q}}(q) &\approx -\dfrac{1}{2\pi i}\sum_{t=-T}^T \dfrac{\widehat f_{\rs{Q}}^{(K)}(z_t)}{z_t-1} q^{-(z_t-1)} \delta \label{eq:Gardini-CDF}
                \end{align}
                where $\beta$ is the arbitrary constant in the chi-square density expansion, $\Delta$ is the integration step, $K$ is the truncation index of the infinite series. So as $T$ goes to $\infty$ and $\Delta$ goes to zero, the approximate sign is replaced be an equal one. In addition, $z_t= h + t\delta$, $\delta = i\Delta$, $T=[h-\delta,h+\delta]$ is the truncation of the improper integral, $h$ is an arbitrary real value in the strip of analyticity of the Mellin transform.

                A convenient recursive form is obtained by writing 
\[
\widehat f_{\rs{Q}}^{(K)}(z)
=(2\beta)^{z-1}\sum_{k=0}^{K} c_k\,P_k(\alpha,z-1),
\]
where
\begin{align}
P_0(\alpha,z-1)
&=\frac{\Gamma(z+\alpha-1)}{\Gamma(\alpha)},\\
P_k(\alpha,z-1)
&=P_{k-1}(\alpha,z-1)\left(1+\frac{z-1}{\alpha+k-1}\right),
\qquad k\ge 1,
\end{align}
and the coefficients satisfy
\begin{align}
c_0
&=\exp\!\left(-\frac12\sum_{j=1}^{\rankFinal} h_j^2\right)
\prod_{j=1}^{\rankFinal}\left(\frac{\beta}{\lambda_j}\right)^{1/2},\\
c_k
&=\frac{1}{2k}\sum_{\ell=0}^{k-1} d_{k-\ell}c_\ell,
\qquad k\ge1,\\
d_k
&=\sum_{j=1}^{\rankFinal}\left(1-\frac{\beta}{\lambda_j}\right)^k
+k\beta\sum_{j=1}^{\rankFinal}\frac{h_j^2}{\lambda_j}
\left(1-\frac{\beta}{\lambda_j}\right)^{k-1},
\qquad k\ge1.
\end{align}
Thus, once the truncated Mellin transoform has been computed, it can be reused to evaluate both the PDF and the CDF. 

                For the PDF, \(K\) and \(T\) are determined so that the following error bounds satisfy the required precision criterion after dividing the eigenvalues by the leading one and taking \(\beta=1\):
                \begin{align}
                    \left|e_M\right| &\leq\left(1-\sum_{k=0}^K c_k\right) f((\alpha+K) ; \alpha+K+1)/2\\
                    \left|e_T\right| &\leq\left|\widehat{f}_{\rs{Q}}(h+i \Delta T)\right| \frac{h^2+(\Delta T)^2}{\pi h q^h} \times\left(\frac{\pi}{2}-\arctan \left(\frac{\Delta T}{h}\right)\right)
                \end{align}
                Similarly, for the CDF:
                \begin{align}
                    |e_M|&\leq 1-\sum_{k=0}^Kc_k \label{Gardini:e_M}\\
                    |e_T|&\leq\left|\widehat{f}_{\rs{Q}}(h+i \Delta T)\right| \frac{h^2+(\Delta T)^2}{\pi(h-1)^2} q^{1-h}\left(1-\frac{\Delta T}{\sqrt{(\Delta T)^2+(h-1)^2}}\right)  \label{Gardini:e_T}
                \end{align}

            Gardini provides an R package for his implementation. \footnote{\url{https://cran.r-project.org/web/packages/QF/index.html}}
            
\section{PDF/CDF: Approximate Methods} \label{sec:Approximate-Formulae}
    When analytical expressions for the PDF or CDF of SQFs are unavailable, or when exact inversion methods are too expensive for repeated numerical evaluation, approximate methods become essential. These methods replace the target distribution by a more tractable representation while preserving selected distribution features. In this section, we review three approximation methods. Moment-matching methods replace the quadratic form by a simpler random variable whose moments or cumulants agree with those of the target variable. Sequence of random variable methods construct auxiliary distributions that converge to the target random variable, typically in distribution or MGF. Saddlepoint approximation method works directly with the CGF and often provide highly accurate approximations for both densities and tail probabilities.
    \subsection{Moment-Matching Methods}

    \label{ssec:Moment-Matching}
            Moment-matching methods approximate the distribution of a quadratic form by a simpler parametric distribution whose first few moments or cumulants agree with those of the target random variable (quadratic form). In practice, the first two to four moments are used most often, since they describe the main features of the distribution, namely, location, dispersion, asymmetry and tail behavior. This idea is motivated by the Taylor expansion of the MGF around the origin, 
            \begin{equation}
                M_{\rs{Q}}(t) = 1 + \sum_{i=1}^n \dfrac{\mu_it^i}{i!},
            \end{equation}
            where $\mu_i = \mathbb{E}[\rs{Q}^i]$ is the $i-$th raw moment. Hence, for small values of $t$, the local behavoir of the MGF is mainly determined by the low-order moments. This explains why moment-mathching methods often provided good accuracy in the central region of the distribution, although their accuracy in the tails may worsen when only a few moments are matched. 

            Moments and cumulants carry equivalent information, since each can be computed from the other. Although moments are more intuitive because of their direct interpretation in terms of location, spread and shape, cumulants are usually more attractive analytically. Particularly, higher order vanish for normal distribution, cumulants add under independence and they transform simply under affine mappings.

            The different moment-matching methods share the same basic principle. We have a target quadratic form, $\rs{Q}$, with cumulants, $\kappa^{\rs{Q}}_j$, A moment- or cumulant- based approximation methods select a simpler approximating random variable $\rs{Z}(\dv{\theta})$, with parameter vector $\dv{\theta}$, and determines $\dv{\theta}$ so that the selected moments, cumulants, or standardized shape measures of $\rs{Z}(\dv{\theta})$ agree with those of $\rs{Q}$. The final approximation is then obtained by replacing, 
            \begin{equation}
            \label{eq:moment-matching_general}
            F_\rs{Q}(q) \approx F_{\rs{Z}(\dv{\theta})}(q).
            \end{equation}
            Moment-matching methods differ mainly in the choice of the approximating random variable and in the number and type of mathced quantities. 

            The two-moment approximation of Satterthwaite \cite{satterthwaiteApproximateDistributionEstimates1946} provides a simple moment-mathcing method. Consider  $\rs{Q} \sim \omega_1 \chi^2_{\nu_1} + \omega_2 \chi^2_{\nu_2} + \ldots + \omega_L \chi^2_{\nu_L} $, and approximate it by a scaled chi-square random variable, 
            \begin{equation}
            \label{eq:moment-matching_satterthweight_Z}
                \rs{Z} \sim a \chi_b^2.
            \end{equation} 
            The first two cumulants of $\rs{Q}$ are 
                $$
                    \kappa_1^{\rs{Q}} = \sum_{i=1}^L \omega_i\nu_i, \quad \kappa_2^{\rs{Q}} = 2\sum_{i=1}^L  \omega_i^2\nu_i,
                $$
            whereas those of $\rs{Z}$ are 
                $$
                    \kappa_1^{\rs{Z}} = ab, \quad \kappa_2^{\rs{Z}} = 2a^2b.
                $$
            Equating the two pairs yields,
                \begin{equation}
                \label{eq:moment-matching_satterthweight_parameters}
                    a = \dfrac{1}{2} \dfrac{\kappa_2^{\rs{Q}}}{\kappa_1^{\rs{Q}}}, \quad b = \dfrac{2({\kappa_1^{\rs{Q}}})^2}{\kappa_2^{\rs{Q}}},
                \end{equation}
            Hence, we can write 
            \begin{equation}
                \label{eq:moment-matching_satterthweight}
                F_\rs{Q}(q) \approx F_{\chi^2_b}(\dfrac{q}{a}).
            \end{equation}
            Thus, Satterthwaite's method matches the mean and variance of $\rs{Q}$ by a scaled chi-square random variable. It is a simple and a low computational cost method, but matching only two moments limits its flexibility.

            Pearson \cite{pearsonNoteApproximationDistribution1959} extends this idea by introducing a location shift and matching one additinoal cumulant, Specifically, $\rs{Q}$ is approximated by 
            \begin{equation}
                \label{eq:moment-matching_pearson_Z}
                \rs{Z} \sim a \chi^2_b + c, 
            \end{equation}
            where $a$, $b$ and $c$ are chosen so that the first three cumulants of $\rs{Z}$ and $\rs{Q}$ match. Since 
            $$
            \kappa^\rs{Z}_1 = ab + c, \quad \kappa^{\rs{Z}}_2 = 2a^2b, \quad \kappa^{\rs{Z}}_3 = 8a^3b,
            $$
            matching gives, 
            \begin{equation}
                \label{eq:moment-matching_pearson_parameters}
                 a = \dfrac{\kappa_3^{\rs{Q}}}{4\kappa_2^{\rs{Q}}}, \qquad b = \dfrac{8(\kappa_2^{\rs{Q}})^3}{(\kappa^{\rs{Q}}_3)^2},  \qquad c = \kappa_1^\rs{Q} - \dfrac{2(\kappa^{\rs{Q}}_2)^2}{\kappa_3^{\rs{Q}}}.
            \end{equation}
            Therefore,
            \begin{equation}
                \label{eq:moment-matching_pearson}
                F_\rs{Q}(q) \approx F_{a\chi^2_b+c}(q)
            \end{equation}
            Compared with Satterthwaite’s method, Pearson’s approximation accounts for skewness through the third cumulant and therefore provides greater flexibility.

            Hall–Buckley–Eagleson (HBE)\cite{buckleyApproximationDistributionQuadratic1988} applies moment-matching to a standardized version of the quadratic form rather than $\rs{Q}$, $\rs{Q}^\prime = \frac{\rs{Q} - \kappa^\rs{Q}_1}{\sqrt{\kappa_2^\rs{Q}}}$ so that $\rs{Q}^\prime$ has zero mean and uinit variance. The HBE method then approximates $\rs{Q}^\prime$ by a standardized chi-square variable, $\rs{Z}\sim \dfrac{\chi^2_b - b}{\sqrt{2b}}$, where the parameter $b$ is chosen so that the first three central moments, equivalently the skewness, agree. Since both $\rs{Q}^\prime$ and $\rs{Z}$ are already standardized, the approximation has only one free parameter, $b$. In terms of the original quadratic form, the approximation may be written as, 
            \begin{equation}
                F_{\rs{Q}}(q) \approx F_{\rs{Z}} \left(  \dfrac{q - \kappa_1^\rs{Q}}{\sqrt{\kappa_2^\rs{Q}}}  \right) = F_{\chi^2_b} \left( b + \sqrt{2b} \dfrac{q - \kappa_1^\rs{Q}}{\sqrt{\kappa_2^\rs{Q}}}  \right), \quad  b=\frac{8(\kappa_2^{\mathsf{Q}})^3}{(\kappa_3^{\mathsf{Q}})^2}.
            \end{equation}
            Note that Pearson’s method approximates $\rs{Q}$ directly by a shifted and scaled central chi-square variable, whereas HBE applies the same approximation after standardization. Consequently, HBE may be viewed as the standardized one-parameter form of Pearson’s three-moment central chi-square approximation.
            
            Wood's method \cite{woodApproximationDistributionLinear1989} uses the same three-moment principle, but replaces the shifted chi-square family by a more flexible corrected-$F$ family. The approximation is written as 
            \begin{equation}
                F_{\rs{Q}}(q) \approx F_{F(2\alpha_1,2\alpha_2)}\left( \dfrac{\alpha_2}{\alpha_1\beta}q \right), \label{eq:wood}
            \end{equation}
            where the parameters are determined from the first three cumulants of $\rs{Q}$. The parameters are 
            \begin{align}
                \alpha_1 &= \dfrac{2\kappa_1^{\rs{Q}} \left( \kappa_1^{\rs{Q}}\kappa_3^{\rs{Q}} + (\kappa_1^{\rs{Q}})^2 \kappa_2^{\rs{Q}} - (\kappa_2^{\rs{Q}})^2 \right)}{ 4\kappa_1^{\rs{Q}} (\kappa_2^{\rs{Q}})^2 + \kappa_3^{\rs{Q}}(\kappa_2^{\rs{Q}}-(\kappa_1^{\rs{Q}})^2)}, \\ 
                \alpha_2 &= 3 + \dfrac{2 \kappa_2^{\rs{Q}} \left( \kappa_2^{\rs{Q}} + (\kappa_1^{\rs{Q}})^2 \right)}{\kappa_1^{\rs{Q}}\kappa_3^{\rs{Q}} - 2(\kappa_2^{\rs{Q}})^2}, \\ 
                \beta &= \dfrac{4\kappa_1^{\rs{Q}} (\kappa_2^{\rs{Q}})^2 + \kappa_3^{\rs{Q}}(\kappa_2^{\rs{Q}}-(\kappa_1^{\rs{Q}})^2)}{\kappa_1^{\rs{Q}}\kappa_3^{\rs{Q}} - 2(\kappa_2^{\rs{Q}})^2}.
            \end{align}

            Liu et al. \cite{liuNewChisquareApproximation2009} approximate a nonnegative quadratic form by matching its standardized version to a standardized noncentral chi-square random variable. Let $\mu_{\rs{Q}} = \kappa_1^\rs{Q}$ and $\sigma_{\rs{Q}} = \sqrt{\kappa_2^\rs{Q}}$. Their approximation assumes 
            $$
            \dfrac{\rs{Q - \mu_\rs{Q}}}{\sigma_\rs{Q}} \sim \dfrac{\chi^2_\ell(\delta) - (\ell + \delta)}{\sqrt{2}a}, \quad a =\sqrt{\ell + 2\delta},
            $$
            where $\ell$ and $\delta$ are chosen so  that the skewness is matched exactly and the kurtosis discrepancy is minimized. Define 
            $$
            s_1 = \dfrac{\kappa_3^\rs{Q}}{2\sqrt{2}(\kappa_2^\rs{Q})^{3/2}}, \qquad s_2 = \dfrac{\kappa_4^\rs{Q}}{12(\kappa_2^\rs{Q})^2},
            $$
            Matching skewness gives $ \delta = s_1 a^3-a^2 $. Then, two cases arise. If $s_1^2 > s_2$, the kurtosis discrepancy can be made zero, and the parameters are 
             are given by 
            $$
            a = \dfrac{1}{s_1\sqrt{s_1^2 - s_2}} , \quad \delta = s_1 a^3-a^2, \quad \ell = a^2-2\delta
            $$
            If $s_1^2 \le s_2$, the minimum kurtosis discrepancy is attained at 
            $$
            a =\dfrac{1}{s_1}, \quad \delta=0, \quad\ell = \dfrac{1}{s_1^2}.
            $$
            Consequently, the distribution of the quadratic form is approximated by
            \begin{equation}
                \label{eq:moment-matching_Liu}
                F_{\rs{Q}} (q) \approx F_{\chi^2_\ell(\delta)} \left( \ell + \delta+ \sqrt{2}a \dfrac{q-\kappa^{\rs{Q}}}{\sqrt{\kappa_2^\rs{Q}}} \right)
            \end{equation}
            Thus, Liu's method handles noncentral quadratic forms by incorporating fourth-order shape information through kurtosis matching. When $s_1^2 \le s_2$, one obtains $\delta = 0$, so the approximation reduces to HBE (standardized Pearson's method) central chi-square approximation.

            A more general high-order approach is provided by Lindsay, Pilla and Basak \cite{lindsayMomentBasedApproximationsDistributions2000}, who approximate the distribution of $\rs{Q}$ by a finite mixture of gamma distributions, 
            \begin{equation}
                    F_{\rs{Q}} \approx \sum_{i=1}^n \pi_i F_{\rs{Z}_i}, \quad \rs{Z}_i \sim \Gamma(k,\theta_i), 
                    \label{eq:Lindsay}
            \end{equation}
            with $\pi_i \ge 0$ and $\sum_{i=1}^n \pi_i =1$. In this case, the number of matched moments is increased to $2n$, and the parameters are obtained through a constructive procedure involving pseduo-moment matrices and polynomial root finding. This method uses a much richer approximating family and a larger number of matched quantities. As a result, it can provide better global accuracy, particularly in the tails, at the expense of higher computational complexity. 

            A practical way to compare such methods is to evaluate them as quintiles corresponding to prescribed probability levels as described in \cite{2016-Bodenham-ComparisonofEfficientApproximationsForWeightedSum}. Let $F_\rs{Q}(x)$ denote a highly accurate reference method and let $G(x)$ be an approximation. For probability levels $\{p_1, p_2, \cdots, p_L\}$, suppose the corresponding quintiles $\{ x_1, x_2, \cdots, x_L\}$ satisfy
            $$
            |F_{\rs{Q}(x_j)} - p_j| < \zeta, \qquad \zeta \ll 1.
            $$
            Then the approximation error $x_j$ is 
            $$
            \epsilon_j = |G(x_j) - F_{\rs{Q}(x_j)}|,
            $$
            and the triangle inequality gives 
            $$
            |G(x_j) - p_j  | < \epsilon_j + \zeta
            $$
            Hence, whenever $\zeta$ is negligible relative to $\epsilon_j$, the quantity $|G(x_j)-p_j|$ provides a reliable proxy for the approximation error. Repeating this comparison over many problem instances gives a practical framework for assessing how different moment-based approximations perform across the support, including both the center and tail.

    \subsection{Sequence of Random Variables}
     \label{ssec:Seq-RVs}
                The quadratic form can be approximated by a sequence of random variables. Convergence in distribution, or weak convergence, simply means the convergence of the cumulative distribution functions to the CDF of the limit random variable. Weak convergence does not imply, in general, the convergence of the PDFs to the PDF of the limit. However, under the conditions of boundedness and equicontinuity of the sequence of the PDFs, and continuity of the PDF of the limit RV, it holds that the sequence of the PDFs converges to the PDF of the limit. A useful implication of weak convergence is the convergence of expectations of bounded continuous functions of the random variables \[\rs{Q}_n\stackrel{d}{\rightarrow}\rs{Q}\implies \expect{g(\rs{Q}_n)}\to \expect{g(\rs{Q})}\] given $g$ is a bounded continuous function. This may be particularly useful for wireless communication applications, in which one may be interested in functions of the SNR which may, under some assumptions, be modeled as a quadratic form in Gaussian random variables. Ramirez-Espinoza et al. \cite{ramirez-espinosaNewApproachStatistical2019} coin the term "confluence": the convergence in the MGF. Note that confluence implies weak convergence; however, in general, weak convergence does not imply confluence \cite{GengMGF}.
                
                \subsubsection{Indefinite Complex  Forms}
                Ramirez-Espinoza et al. \cite{ramirez-espinosaNewApproachStatistical2019} study the distribution of a Hermitian form in complex Gaussian random variables. Consider \eqref{eq:lcOnlyChiQFCGRV}, which can be rewritten as \(\rs{Q}=(\rv{z}+\dv{h})^H\dmt{\Lambda}(\rv{z}+\dv {h})\) with \(\dmt{\Lambda}=\diag(\lambda_1,\lambda_2,\ldots,\lambda_N)\). Let
	
                \[\rs{Q}_m=(\rv{z}+\dmt{D}_\xi\dv{h})^H\dmt{\Lambda} (\rv{z}+\dmt{D}_\xi\dv{h})\] 
                with $\dmt{D}_\xi=\diag(\xi_{m,1},\ldots,\xi_{m,N})$, $\xi^2_{m,i} \sim \Gamma(m,1/m)$.
                Then $\rs{Q}_m$ is confluent to $\rs{Q}$, i.e.,
                \[M_{\rs{Q}_m}(t)=\prod_{j=1}^N\dfrac{\left(1-\frac{\lambda_jh_j^2t}{m(1-\lambda_jt)}\right)^{-m}}{1-\lambda_jt}\xrightarrow[m \to \infty]{} M_{\rs{Q}}(t)\]
                The MGF of \(\rs{Q}_m\) can be rewritten as 
                \[M_{\rs{Q}_m}(t)=\prod_{k=1}^{N}\left[(-\lambda_k)\left(1+\dfrac{h_k^2}{m}\right)^m\right]^{-1}\prod_{i=1}^{N}\dfrac{(t-1/\lambda_i)^{m-1}}{(t-\beta_i)^m}\]
                where $\beta_i=[\lambda_i(1+h_i^2/m)]^{-1}$. Taking into consideration multiple poles and zeroes, we can write
                \[M_{\rs{Q}_m}(t)=\prod_{k=1}^{N}\left[(-\lambda_k)\left(1+\dfrac{h_k^2}{m}\right)^m\right]^{-1}\dfrac{\prod_{j=1}^{n_\lambda}(s-1/\tilde{\lambda_j})^{q_j}}{\prod_{i=1}^{n_\beta}(t-\tilde{\beta_i})^{p_i}}\]
                where \(n_\lambda\) and \(n_\beta\) are the numbers of distinct roots and poles. The PDF of $\rs{Q}_m$ is \begin{equation}
                    f_{\rs{Q}_m}(q)=\sum_{i=1}^{n_\beta}\sum_{j=1}^{p_i}\alpha_{i,j}e^{-\tilde{\beta}_iq}x^{j-1}u(\tilde{\beta}_iq)\sign(q)\label{eq:ramirez-complex}
                \end{equation}
                where $\alpha_{i,j}=B_jA_{i,j}$, with
                \[B_j=\dfrac{1}{(j-1)!}\prod_{k=1}^N\left[(-\lambda_k)(1+\dfrac{h_k^2}{m})^m\right]^{-1}\]
                and $A_{i,j}$ are the residues that arise after performing partial fraction decomposition on $R(s)=\dfrac{\prod_{j=1}^{n_\lambda}(s+1/\tilde{\lambda}_j)^{q_j}}{\prod_{i=1}^{n_\beta}(s+\tilde{\beta}_i))^{p_i}}$.\\
                The mean square error between the quadratic form and the confluent one is given by
                \[\mathbb{E}[(\rs{Q}_m-\rs{Q})^2]=\sum_{i=1}^N\lambda_i^2h_i^2\left[4\left(1-\dfrac{\Gamma(m+1/2)}{m^{1/2}\Gamma(m)}\right)+\dfrac{h_i^2}{m}\right]\] Since $$\dfrac{\Gamma(m+1/2)}{\Gamma(m)}\sim m^{1/2} \text{ as } m\to \infty$$ we can see that $\rs{Q}_m$ converges to $\rs{Q}$ in $L^2$.
                
                \subsubsection{Positive Definite Forms}
                Ramirez-Espinoza et al. \cite{ramirez-espinosaNewApproximationDistribution2019} approximate a positive definite quadratic form by a sequence of random variables as follows \(\rs{Q}_m=\rs{Q}/\xi_m\), where \(\xi_m\) is a gamma variate with shape \(m\) and rate \(m-1\). Note that \(\xi_m\to 1\) as \(m\to\infty\), thus \(\rs{Q}_m\) converges in MGF, hence in distribution, to \(\rs{Q}\). The PDF of \(\rs{Q}_m\) is given by \begin{equation}
                    f_{\rs{Q}_M}(q)=\dfrac{(m-1)^m}{q^{m+1}\Gamma(m)}M_\rs{Q}\left(\dfrac{1-m}{q}\right)D_m\left(\dfrac{1-m}{q}\right) \label{eq:Ramirez-PD}
                \end{equation}
                where \(M_{\rs{Q}}\) is the moment generating function of \(\rs{Q}\) and the function \(D_m\) is recursively defined by \begin{align*}
                    D_0(s)&=1\\
                    D_k(s)&=\sum_{j=0}^{k-1}\binom{k-1}{j}g_{k-1-j}(s)D(s) &k\geq 1\\
                    g_i(s)&=2^ii!\sum_{t=1}^n\dfrac{\lambda_t^{i+1}[(i+1)b_t^2+1-2\lambda_ts]}{(1-2\lambda_ts)^{i+2}}
                \end{align*}
                The CDF is given by:
                \[F_{\rs{Q}_m}(q)=M_\rs{Q}\left(\dfrac{1-m}{q}\right)\sum_{k=0}^{m-1}\dfrac{(m-1)^k}{q^kk!}D_k\left(\dfrac{1-m}{q}\right)\]
                Typically, \(m\) seems to be in the range of \(10^2\). In addition to weak convergence, the authors prove the convergence of the probability density functions of the terms to the PDF of the limit random variable.
                
                \subsubsection{Indefinite Forms} \label{sssec:IndefiniteForms_SeqRvs}

                Ha and Provost \cite{haAccurateApproximationDistribution2013} give an approximate PDF, which is, in fact, the PDF of an approximate random variable \(\rs{Q}_d\) that converges in MGF, thus in distribution, to $\rs{Q}$ as $d$ goes to infinity. Start by writing the indefinite form as a difference of two definite forms: $\rs{Q}=\rs{Q}_{(1)}-\rs{Q}_{(2)}$. The PDF of the approximate variable is given by: 
                \begin{equation}
                    f_{\rs{Q}_d}(q)=\begin{cases}
                        f_P(q) &\text{ for }q\geq 0\\
                        f_N(q) &\text{ for }q<0
                        \end{cases}\label{eq:HaProvost1}
                \end{equation}
                with \begin{align}
                    f_P(q)&=\sum_{i=0}^d\sum_{j=0}^d \xi_{\nu_1,i}\xi_{\nu_2,j}\int_0^\infty \left(\dfrac{q+t}{\beta_1}\right)^i\left(\dfrac{q}{\beta_2}\right)^j\gamma_{\nu_1,\beta_1}(q+t)\gamma_{\nu_2,\beta_2}(t)dt\\
                    f_N(q)&=\sum_{i=0}^d\sum_{j=0}^d \dfrac{\xi_{\nu_1,i}\xi_{\nu_2,j}}{\beta_1^i\beta_2^j}\int_0^\infty t^i(t-q)^j\gamma_{\nu_1,\beta_1}(t)\gamma_{\nu_2,\beta_2}(t-q)dt\\
                \end{align}
                where:
                \UglyEquation{
                    \begin{aligned}
                        \beta_i&=\dfrac{V(\rs{Q}_{(i)})}{\mathbb{E}(\rs{Q}_{(i)})}\;\;\nu_i=\dfrac{\mathbb{E}(\rs{Q}_{(i)})^2}{V(\rs{Q}_{(i)})}-1\;\;\;\;i=1,2\\
                        \xi_{\nu_j,k}&=\begin{cases}
                        1+\sum_{i=2}^d \eta_i^{\nu_j}d_{i,k}^{\nu_j}, &\text{ for }k=0,\\
                        \sum_{i=2}^d \eta_i^{\nu_j}d_{i,k}^{\nu_j}, &\text{ for }k=1,\\
                        \sum_{i=k}^d \eta_i^{\nu_j}d_{i,k}^{\nu_j}, &\text{ for }k=2,\ldots,d
                        \end{cases}\\
                        \eta_i^{\nu_j}&=\dfrac{i!}{\Gamma(\nu+i+1)}\sum_{k=0}^i d_{i,k}^{\nu_j} \dfrac{\mathbb{E}(\rs{Q}_{(j)}^k)}{\beta_j^k},&j=1,2\\
                        d_{i,k}^{\nu_j}&=\dfrac{(-1)^{i-k}\Gamma(i+\nu_j+1)}{(i-k)!k!\Gamma(\nu_j+k+1)}\\
                        \gamma_{\nu_\ell,\beta_\ell}(z)&=\dfrac{z^{\nu_\ell}e^{-z/\beta_\ell}}{\beta_\ell^{\nu_\ell+1}\Gamma(\nu_\ell+1)}&\ell=1,2\\
                    \end{aligned}
                }
                
                If \(\nu_1\) and \(\nu_2\) are not integers, then we have:
                \begin{align}
                    f_P(q)&=\sum_{i=0}^d\sum_{j=0}^d \dfrac{\xi_{\nu_1,i}\xi_{\nu_2,j}e^{-q/\beta_1}}{\beta_1^{\nu_1+i+1}\beta_2^{\nu_2+j+1}\Gamma(\nu_1+1)\Gamma(\nu_2+1)}\nonumber \\
                    &\times\left(\vphantom{\dfrac{\Gamma(-m_{ij})\Gamma(j+\nu_2+1)}{\Gamma(-i-\nu_1)}q^{m_{ij}}}\beta^{-m_{ij}}\Gamma(m){ij}) {}_1F_1(-i-\nu_1,-m_{ij}+1,q\beta)\right.\nonumber\\
                    &\left.+\dfrac{\Gamma(-m_{ij})\Gamma(j+\nu_2+1)}{\Gamma(-i-\nu_1)}q^{m_{ij}}{}_1F_1(j+\nu_2+1,m_{ij},q\beta)\right) \label{eq:HaProvost2}
                \end{align}
                and
                \begin{align}
                    f_N(q)&=\sum_{i=0}^d\sum_{j=0}^d \dfrac{\xi_{\nu_1,i}\xi_{\nu_2,j}e^{q/\beta_2}}{\beta_1^{\nu_1+i+1}\beta_2^{\nu_2+j+1}\Gamma(\nu_1+1)\Gamma(\nu_2+1)}\nonumber\\
                    &\times\left(\vphantom{\dfrac{\Gamma(-m_{ij})\Gamma(i+\nu_1+1)}{\Gamma(-j-\nu_2)}}\beta^{-m_{ij}}\Gamma(m_{ij}) {}_1F_1(-j-\nu_2,-m_{ij}+1,-q\beta)\right.\nonumber\\
                    &\left.+\dfrac{\Gamma(-m_{ij})\Gamma(i+\nu_1+1)}{\Gamma(-j-\nu_2)}(-q)^{m_{ij}}\times {}_1F_1(i+\nu_1+1,m_{ij}+1,-q\beta)\right) \label{eq:HaProvost3}
                \end{align}
                Ha and Provost \cite{haAccurateApproximationDistribution2013} provide formulae for the CDF as well.
        \subsection{Saddlepoint Approximation}
        \label{ssec:Saddlepoint-Approx}
            Saddlepoint approximation provides an accurate and computationally efficient way to approximate the density and distribution functions of quadratic forms by working directly with their CGF. We are interested in the incomplete form, e.g., \eqref{eqn:QF_rep2},
            \begin{equation*}
                \rs{Q}\sim \sum_{\ell=1}^L \omega_\ell \chi_{\nu_\ell}^2(\delta_\ell^2) + \sigma \mathcal{N}(0,1) + c^{\prime\prime}.
            \end{equation*}
        The saddlepoint method was originally introduced by Daniels in 1954 as an asymptotic approximation technique for the density of the sample mean \cite{1954-Daniels-SaddlepointApproximationsInStatistics}, but it can be applied more generally to random variables that admit a CGF \cite{2001-Kuonen-ComputerintensiveStatisticalMethodsSaddlepointApproximations,1986-Easton-GeneralSaddlepointApproximationsWithApplicationsToL}. For quadratic forms in Gaussian random variables, this approach has been studied in detail in \cite{2001-Kuonen-ComputerintensiveStatisticalMethodsSaddlepointApproximations,2007-Butler-saddlepointAapproximationsWithApplicationst}.
        Let $K_{\rs{Q}}(t)=\ln M_{\rs{Q}}(t)$ denote the CGF of $\rs{Q}$. Provided that $M_{\rs{Q}}(t)$ exists in a strip around the imaginary axis, the density of $\rs{Q}$ can be written through the inverse transform as
        \begin{equation}
        \label{eqn:pdf_saddle}
        f_{\rs{Q}}(q)=\frac{1}{2\pi i}\int_{-i\infty}^{i\infty} e^{K_{\rs{Q}}(t)-tq}\,dt, \qquad t\in\mathbb{C}.
        \end{equation}
        Using Cauchy's theorem, the contour can be deformed to a path of steepest descent passing through the point $t_o$ satisfying
        \begin{equation}
        \label{eqn:saddlepoint}
        K_{\rs{Q}}'(t_o)-q=0.
        \end{equation}
        The point $t_o$ is called the \emph{saddlepoint}. By expanding the exponent locally around $t_o$ and retaining the leading term, one obtains the saddlepoint approximation to the density
        \begin{equation}
        \label{eq:QF_SPA_PDF}
        f_{\rs{Q}}(q)\approx \left[\frac{1}{2\pi K_{\rs{Q}}''(t_o)}\right]^{\frac{1}{2}}
        \exp\!\bigl(K_{\rs{Q}}(t_o)-t_oq\bigr),
        \end{equation}        
    where $t_o$ is the root of \eqref{eqn:saddlepoint}. The terminology ``saddlepoint'' refers to the local behaviour of the real part of the exponent at $t_o$: it is a local minimum in one direction and a local maximum in the orthogonal direction \cite{2011-lipton-OxfordHandbookofCreditDerivatives}. The same approximation can also be derived from a tilted, or indirect Edgeworth expansion \cite{1989-Barndorff-AsymptoticTechniquesForUseInStatistics}.

    The CGF of $\rs{Q}$ is given in \cref{eq:CGF}, and its $m$th derivative is
\begin{equation}
\begin{aligned}
K_{\rs{Q}}^{(m)}(t)
&=2^{m-1}(m-1)!\sum_{\ell=1}^L \omega_\ell^m
\left[
\frac{\nu_\ell}{(1-2\omega_\ell t)^m}
+
\frac{m\delta_\ell^2}{(1-2\omega_\ell t)^{m+1}}
\right] \\
&\quad
+\left(\sigma^2 t+c^{\prime\prime}\right)\dv{1}_{\{m=1\}}
+\sigma^2 \dv{1}_{\{m=2\}}.
\end{aligned}
\end{equation}

As an example, consider the central case
\[
\rs{Q}\sim \sum_{\ell=1}^3 \omega_\ell \chi_{\nu_\ell}^2,
\]
with $\{\omega_\ell\}_{\ell=1}^3=\{0.6,0.3,0.1\}$ and $\{\nu_\ell\}_{\ell=1}^3=\{2,2,1\}$. To approximate the density at $q=1$, we first solve $K'_{\rs{Q}}(t)=1$, which gives $t_o\approx -1.0084$. Substituting into \eqref{eq:QF_SPA_PDF} yields $f_{\rs{Q}}(1)\approx 0.42$; see Fig.~\ref{subfig:pdf_saddle}.

\begin{figure}
    \centering
    \subfloat[Saddlepoint illustration]{\includegraphics[width=0.4\textwidth]{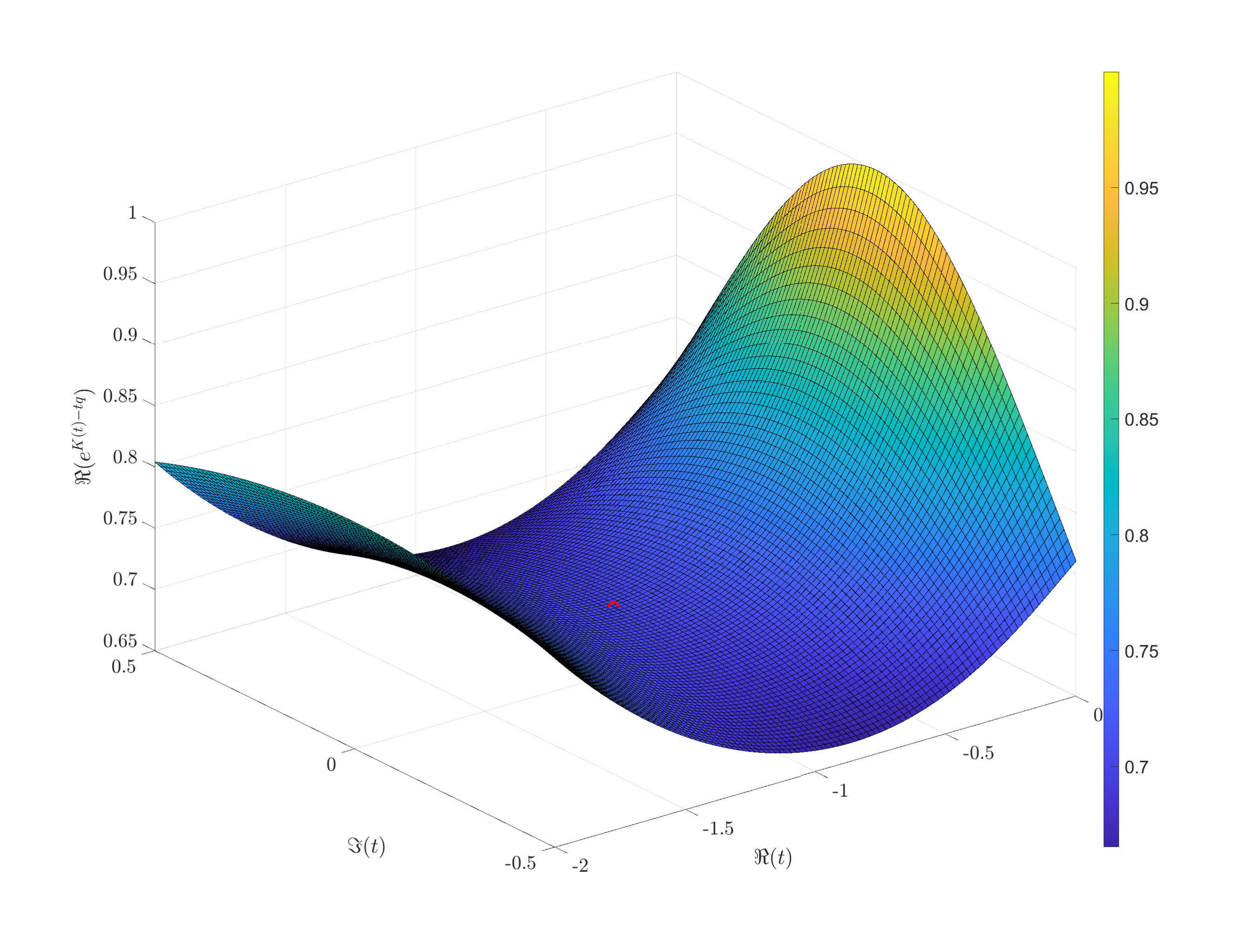}\label{subfig:saddlepoint}}
    \hfil
    \subfloat[Cross section along the real axis]{\includegraphics[width=0.4\textwidth]{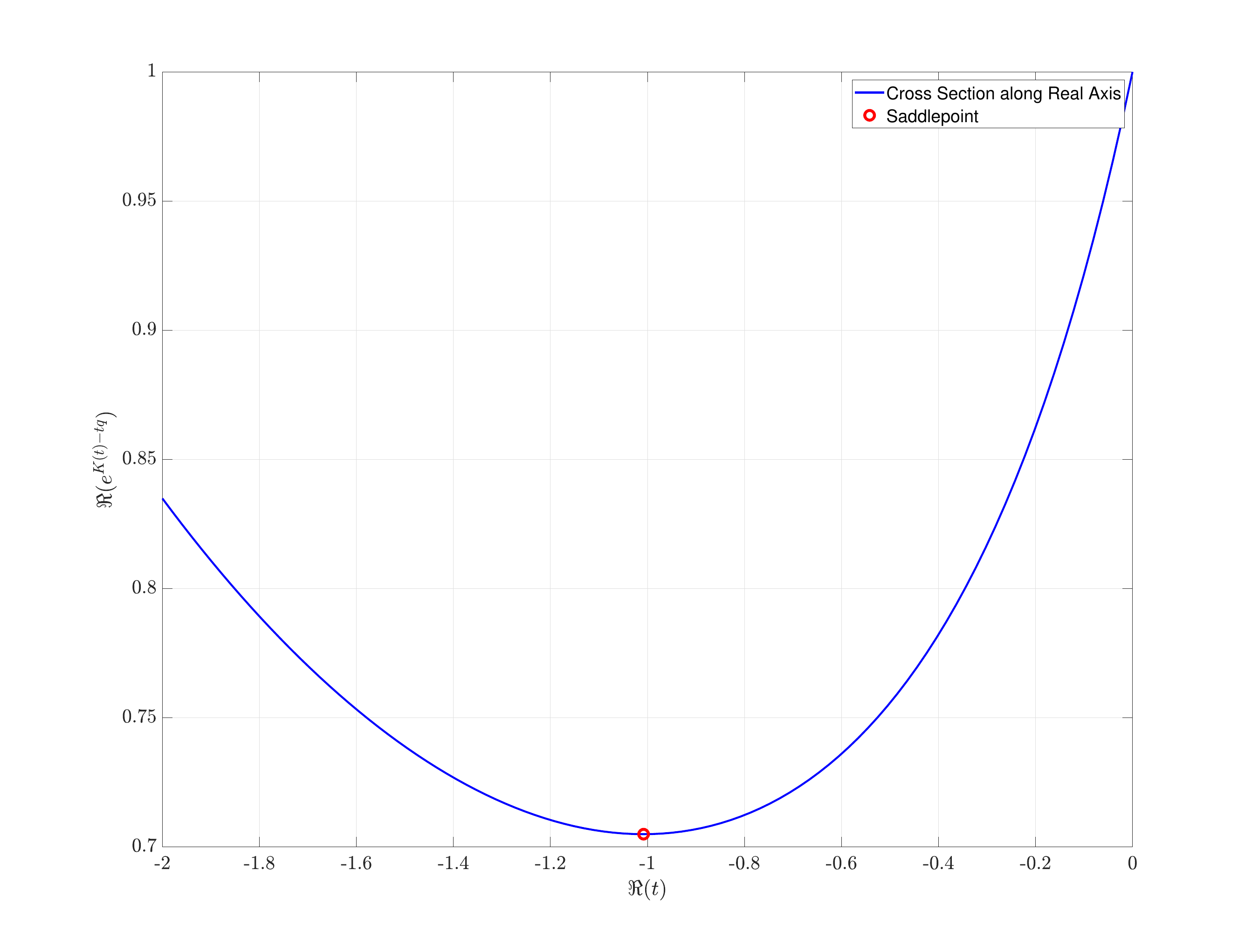}\label{subfig:csr}}
    \hfil
    \subfloat[Cross section along the imaginary axis]{\includegraphics[width=0.4\textwidth]{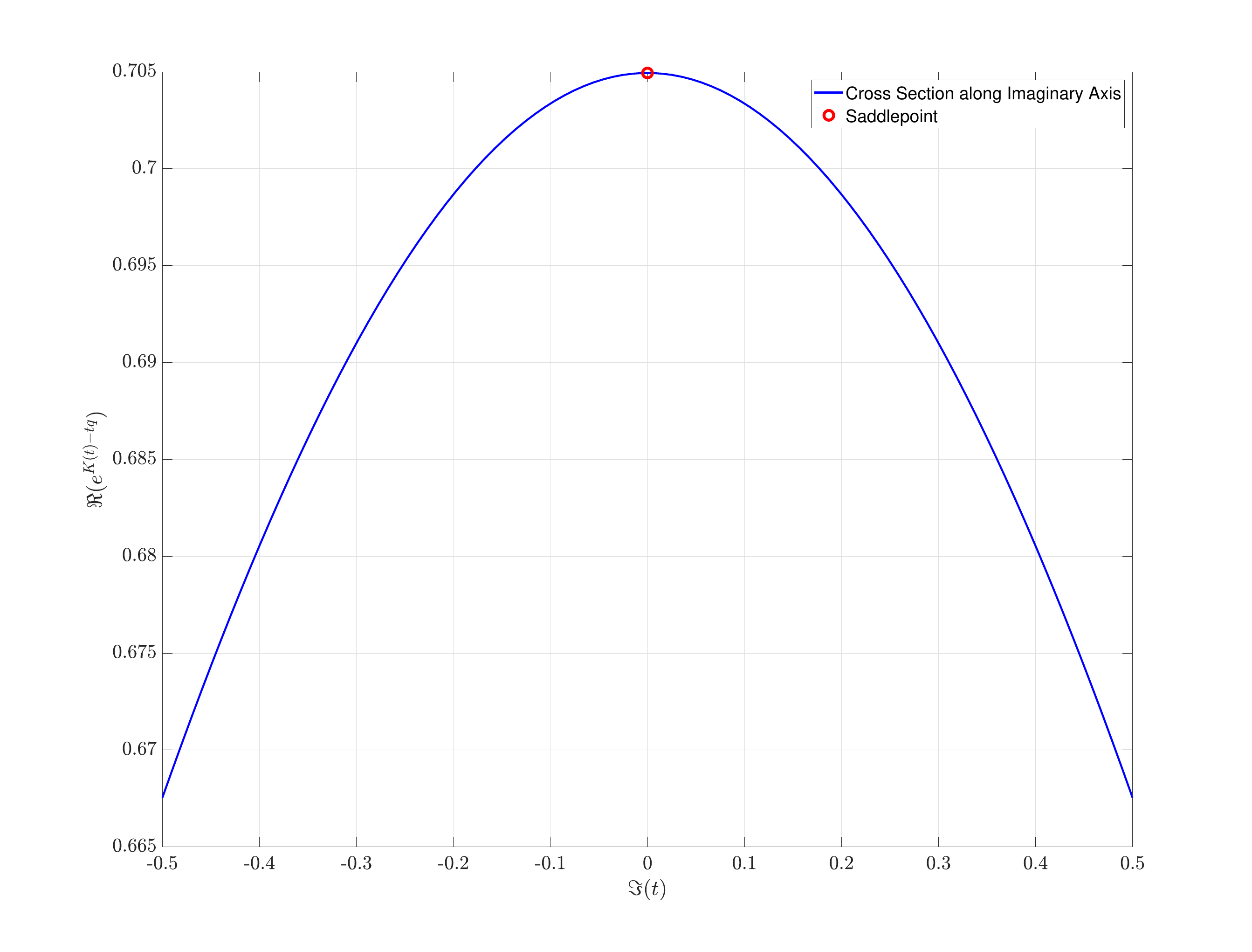}\label{subfig:csi}}
    \hfil
    \subfloat[PDF of $\rs{Q}$]{\includegraphics[width=0.4\textwidth]{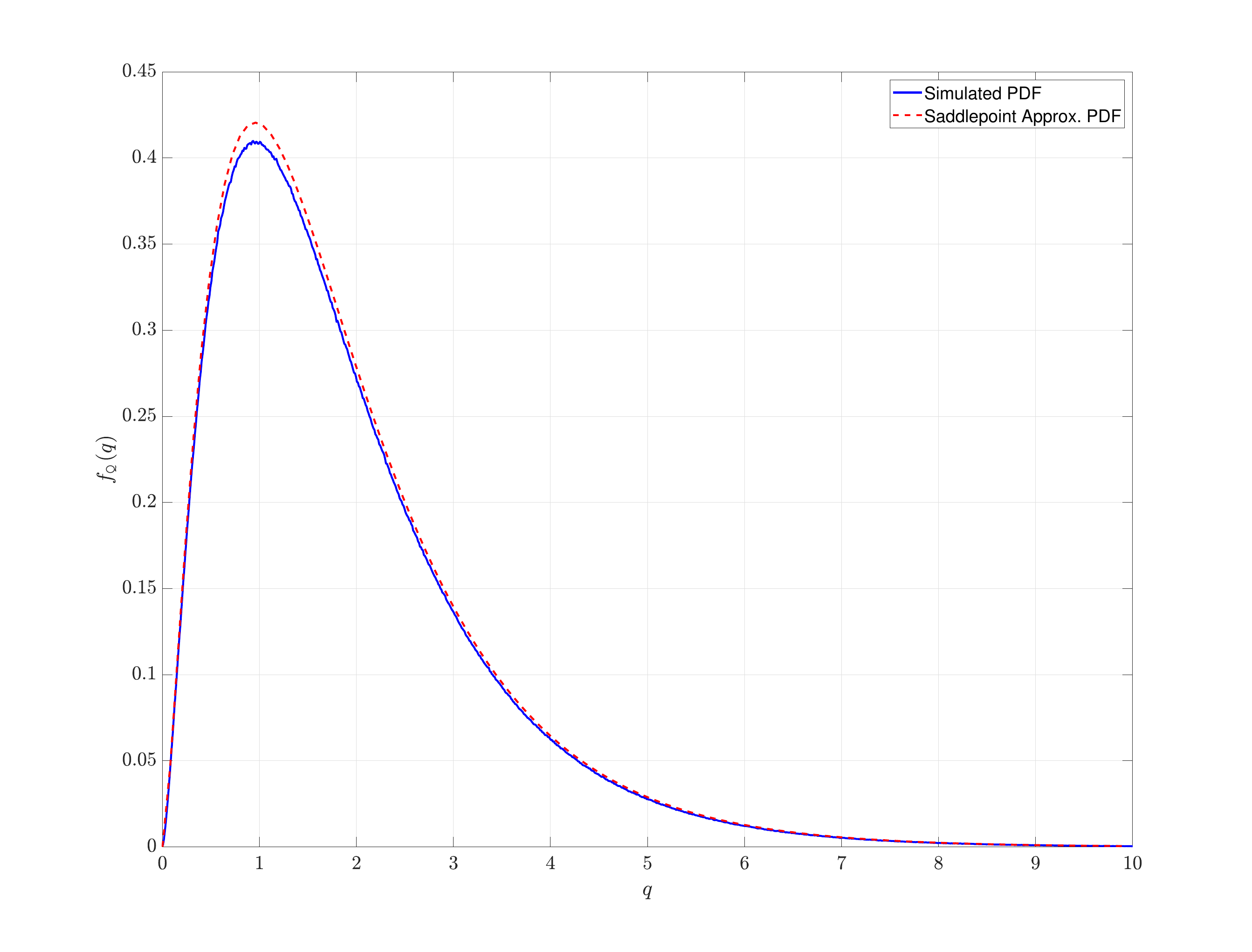}\label{subfig:pdf_saddle}}
    \caption{Saddlepoint approximation method. (a) illustrates the shape of the integrand in 3D. (b) shows that the saddlepoint is a local minimum along the real axis. (c) shows that the saddlepoint is a local maximum along the imaginary axis. (d) compares the PDF obtained using the saddlepoint approximation \eqref{eq:QF_SPA_PDF} with the simulated PDF.}
\end{figure}

For the CDF, the inversion formula becomes
\begin{equation}
\mathbb{P}[\rs{Q}<q]
=
\frac{1}{2\pi i}\int e^{K_{\rs{Q}}(t)-tq}\,\frac{dt}{t}.
\end{equation}
Unlike the density case, the factor $1/t$ introduces a singularity at $t=0$. Consequently, the steepest-descent and tilted-expansion derivations no longer lead to the same tail approximation \cite[p.~542]{2011-lipton-OxfordHandbookofCreditDerivatives}. The Lugannani--Rice formula \cite{1980-Lugannani-SaddlepointApproximationForTheDistributionOfTheSum} handles this singularity by separating the singular part before applying the saddlepoint expansion. For $\rs{Q}$, it gives
\begin{equation}
\label{eq:QF_SPA_CDF_LR}
F_{\rs{Q}}(q)\approx
\Phi(w)+\phi(w)\left(\frac{1}{w}-\frac{1}{v}\right),
\end{equation}
where
\[
w=\operatorname{sgn}(t_o)\sqrt{2\bigl(t_oq-K_{\rs{Q}}(t_o)\bigr)},
\qquad
v=t_o\sqrt{K_{\rs{Q}}''(t_o)}.
\]
At $t_o=0$, we have $w=v=0$, so the second term in \eqref{eq:QF_SPA_CDF_LR} is not defined. It can be shown \cite{1987-Daniels-TailProbabilityApproximations,2007-Butler-saddlepointAapproximationsWithApplicationst} that as $q\to\mu$, where $\mu=\mathbb{E}[\rs{Q}]$ and $t_o\to 0$, the limiting value is
\begin{equation}
F_{\rs{Q}}(\mu)\approx
\frac{1}{2}
+
\frac{K_{\rs{Q}}'''(0)}{6\sqrt{2\pi}\,K_{\rs{Q}}''(0)^{3/2}}.
\end{equation}

An alternative CDF approximation, proposed by Barndorff-Nielsen \cite{barndorff-nielsenApproximateIntervalProbabilities1990a} and used by Kuonen \cite{kuonenSaddlepointApproximationsDistributions1999} for quadratic forms in Gaussian random variables, is
\begin{equation}
\label{eq:QF_SPA_CDF_Bandroff}
F_{\rs{Q}}(q)\approx \Phi\!\left(
w+\frac{1}{w}\ln\frac{v}{w}
\right).
\end{equation}

For approximating the upper tail probability $1-F_{\rs{Q}}(q)$, the term $1-\Phi(\cdot)$ may lead to numerical instability when the tail probability is extremely small, as noted by Daniels \cite[p.~45]{1987-Daniels-TailProbabilityApproximations}. However, Kuonen reported that such instabilities were not encountered in his numerical experiments \cite{2001-Kuonen-ComputerintensiveStatisticalMethodsSaddlepointApproximations}. Chen and Lumley \cite{chenNumericalEvaluationMethods2019} further showed that the saddlepoint approximation reproduces the correct exponential tail behaviour of a linear combination of chi-square random variables in the sense of Berman's result \cite{1992-Berman-TheTailofTheConvolutionofDensitiesAndItsApplication}.

                Al-Naffouri et al. \cite{al-naffouriDistributionIndefiniteQuadratic2016} develop a saddlepoint approximation of complex forms and central real forms. Consider a complex quadratic form rewritten here for convenience \begin{equation*}
                    \rs{Q} = \rv{x}^H\dmt{A}\rv{x}+\real{\dv{b}^H\rv{x}}+c, \qquad \rv{x}\sim \mathcal{CN}_N(\dv{\mu},\dmt{\Sigma}) \tag{CQF}
                \end{equation*}
                Assume it is identically distributed to a linear combination of chi-squares as in \eqref{eq:lcOnlyChiQFCGRV} \begin{equation*}
                     \rs{Q}\stackrel{d}{=}\sum_{n=1}^N \lambda_n|\rs{z}_n+h_n|^2 
                \end{equation*} The authors implicitly invert the MGF of the form. At the head, the CDF is approximated by \begin{equation}
                    F_{\rs{Q}}(q)\approx \dfrac{e^{\tilde{s}(z_0)}}{\sqrt{2\pi|\tilde{s}''(z_0)|}} \label{eq:naffouri-complex-sp-head}
                \end{equation}
                where \[\tilde{s}(z)=qz-\ln(z)-\sum_{i=1}^N\left[\ln(1+\lambda_iz)+|h_i|^2-\dfrac{|h_i|^2}{1+\lambda_iz}\right]\]
                and \(z_0\) is the solution of the equation \(\tilde{s}'(z)=0\). Note that the first and second derivatives of \(\tilde{s}(z)\) are given by:
                \begin{align*}
                    \tilde{s}'(z)&=q-\dfrac{1}{z}-\sum_{i=1}^N\left[\dfrac{\lambda_i}{1+\lambda_iz}+\dfrac{|h_i|^2\lambda_i}{(1+\lambda_iz)^2}\right]\\
                    \tilde{s}''(z)&=\dfrac{1}{z^2}+\sum_{i=1}^N\left[\dfrac{\lambda_i^2}{(1+\lambda_iz)^2}+\dfrac{2|h_i|^2\lambda_i^2}{(1+\lambda_iz)^3}\right]
                \end{align*}
                At the tail, the approximation is given by \begin{equation}
                    F_{\rs{Q}}(q)\approx 1-\dfrac{e^{\tilde{S}(z_0)}}{\sqrt{2\pi|\tilde{S}''(z_0)|}} \label{eq:naffouri-complex-sp-tail}
                \end{equation} 
                where \[\tilde{S}(z)=qz-\ln(-z)-\sum_{i=1}^N\left[\ln(1+\lambda_iz)+|h_i|^2-\dfrac{|h_i|^2}{1+\lambda_iz}\right]\]
                and \(z_0\) is the solution of the equation \(\tilde{S}'(z)=0\). Note that the first and second derivatives of \(\tilde{s}(z)\) are given by
                \begin{align*}
                    \tilde{S}'(z)&=q-\dfrac{1}{z}-\sum_{i=1}^N\left[\dfrac{\lambda_i}{1+\lambda_iz}+\dfrac{|h_i|^2\lambda_i}{(1+\lambda_iz)^2}\right]\\
                    \tilde{S}''(z)&=\dfrac{1}{z^2}+\sum_{i=1}^N\left[\dfrac{\lambda_i^2}{(1+\lambda_iz)^2}+\dfrac{2|h_i|^2\lambda_i^2}{(1+\lambda_iz)^3}\right]
                \end{align*}

                 Consider now a central form
                \[\rs{Q}\stackrel{d}{=}\sum_{j=1}^N\lambda_j\rs{z}_j^2\]
                where \(\rs{z}_j\sim\mathcal{N}(0,1)\). The CDF of the form is approximated by \cite{al-naffouriDistributionIndefiniteQuadratic2016} \begin{equation}
                    F_{\rs{Q}}(q)\approx\dfrac{1}{\sqrt{2\pi}}\exp(\tilde{s}(z_0))\dfrac{1}{\sqrt{\tilde{s}''(z_0)}} \label{eq:naffouri-real-sp}
                \end{equation}
                where \[\tilde{s}(z)=qz-\ln(z)-\dfrac{1}{2}\sum_{j=1}^N\ln(1+2\lambda_jz)\] and \(z_0\) is the unique solution of the equation \(\tilde{s}'(z)=0\).

\section{Ratio of Quadratic Forms} \label{sec:ratioQF}
In this section, we cover the results related to the ratio of quadratic forms \eqref{eq:RQFGRV} and \eqref{eq:RQFCGRV}. First, we present the moments of the quadratic form ratio and their existence. Then, we show two algorithms that compute the moments. In the second part, we present some results concerning the PDF/CDF of the ratio of quadratic forms. 

\subsection{Moments of the Ratio of Quadratic Forms}
As we have seen before, moments of SQFs exist in closed-form expressions. This is not the case for ratios. Firstly, the very existence of moments is not guaranteed. Take, for example, the standard Cauchy distribution $$\rs{R}=\dfrac{\rs{x_1}}{\rs{x_2}}$$ where $\rs{x}_1$ and $\rs{x}_2$ are independent standard normal variables. Now this ratio can be rewritten as a ratio of quadratic forms in two variables $$\rs{R}=\dfrac{\rv{x}^T\dmt{A}\rv{x}}{\rv{x}^T\dmt{B}\rv{x}}, \qquad \rv{x}\sim \mathcal{N}_2(0,\dmt{I}_2)$$ with $$\dmt{A}=\begin{bmatrix}
            0 & \frac12\\
            \frac12 & 0
        \end{bmatrix},\quad \dmt{B}=\begin{bmatrix}
            0 & 0\\
            0 & 1
        \end{bmatrix}$$ It is well-known that this random variable does not admit any finite moment \footnote{In this case, $f_{\rs{R}}(r)=\dfrac{\pi}{1+r^2}$. Hence the functions $r^nf_{\rs{R}}(r)$ are not integrable over $\reals{}$ for any $n>1$.}. Regarding integral transforms, up to our best knowledge, the MGF (or any other integral transform) of a ``general'' ratio has not been studied yet.

         Bao and Kan \cite{baoMomentsRatiosQuadratic2013} provide the necessary and sufficient conditions for the existence of generalized fractional moments \footnote{Generalized fractional moments of a ratio $\rs{R}=\dfrac{\rs{Q}_1}{\rs{Q}_2}$ with $\rs{Q}_1=\rv{x}^T\dmt{A}\rv{x}$ and $\rs{Q}_2=\rv{x}^T\dmt{B}\rv{x}$ are $\expect{\dfrac{Q_1^p}{Q_2^q}}$ where $p$ and $q$ are real numbers.}, generalizing Roberts' work \cite{robertsExistenceMomentsRatios1995}. 
            
            The moments we are reporting here are positive integer moments, i.e., $\expect{\rs{R}^p}, p\in\mathbb{N}$. If the denominator is positive definite, then all the moments exist. If the denominator $\rv{x}^T\dmt{B}\rv{x}$ is positive semi-definite, then we have to eigendecompose the matrix $\dmt{B}$ sorting its eigenvalues $$\dmt{B}=\dmt{PD}\dmt{P}^T=\begin{bmatrix}
               \dmt{P}_1 & \dmt{P}_2 
            \end{bmatrix}\begin{bmatrix}
                \dmt{D_1} & \dmt{0}\\
                \dmt{0} & \dmt{0}
            \end{bmatrix}\begin{bmatrix}
                \dmt{P}_1^T\\
                \dmt{P}_2^T
            \end{bmatrix}=\dmt{P}_1\dmt{D_1}\dmt{P}_1^T$$ where $\dmt{P}\dmt{P}^T=\dmt{I}_N$, $\dmt{P}_1\in\reals{N\times\Rank{B}}$, $\dmt{D}_1\in\reals{\Rank{B}\times\Rank{B}}$, and $\Rank{B}=\rank(B)$. We consider then two cases. If $\dmt{P}_2^T\dmt{A}\dmt{P}_2\neq\dmt{0}$, then the moment $\expect{\rs{R}^p}$ exists if and only if $2p<\Rank{B}$. If $\dmt{P}_2^T\dmt{A}\dmt{P}_2\neq{0}$, then we should consider the matrix $\dmt{P}_1^T\dmt{A}\dmt{P}_2$. If $\dmt{P}_1^T\dmt{A}\dmt{P}_2\neq 0$, then $\expect{\rs{R}^p}$ exists if and only $p<\Rank{B}$. Ohterwise, the moment exists for every $p$. This is summarized in Figure \ref{fig:moment-existence-tree}.

            \begin{figure}[ht!]
                \centering
                \begin{tikzpicture}
    \draw (0,0) node[anchor=south] {$\dmt{B}\succeq \dmt{0}$} -- (0,-0.5);
    \draw (0,-0.5) -- (-2,-0.5);
    \draw[->] (-2,-0.5) -- (-2,-1) node[anchor=north,fill=white] {$\dmt{B}\succ \dmt{0}$} -- (-2,-2) node[anchor=north] {$\expect{\rs{R}^p}$ exists};
    \draw[->] (0,-0.5) -- (4,-0.5) -- (4,-1) node[anchor=north,fill=white] {$\det{B}= 0$} -- (4,-2) -- (2,-2) -- (2,-2.5) node[anchor=north,fill=white] {$\dmt{P}_2^T\dmt{A}\dmt{P}_2 \neq \dmt{0}$} -- (2,-3.5) node[anchor=north] {$\expect{\rs{R}^p}$ exists $\iff 2p<\Rank{B}$};
    \draw[->] (4,-2) -- (6,-2) -- (6,-2.5) node[anchor=north,fill=white] {$\dmt{P}_2^T\dmt{A}\dmt{P}_2= \dmt{0}$} -- (6,-4.5) -- (4,-4.5) -- (4,-5) node[anchor=north,fill=white] {$\dmt{P}_1^T\dmt{A}\dmt{P}_2 \neq \dmt{0}$} -- (4,-6) node[anchor=north] {$\expect{\rs{R}^p}$ exists $\iff p<\Rank{B}$};
    \draw[->] (6,-4.5) -- (8,-4.5) -- (8,-5) node[anchor=north,fill=white] {$\dmt{P}_1^T\dmt{A}\dmt{P}_2 = \dmt{0}$} -- (8,-6) node[anchor=north] {$\expect{\rs{R}^p}$ exists};
\end{tikzpicture}
                \caption{Conditions for the existence of moments of ratios with positive semi-definite denominators}
                \label{fig:moment-existence-tree}
            \end{figure}
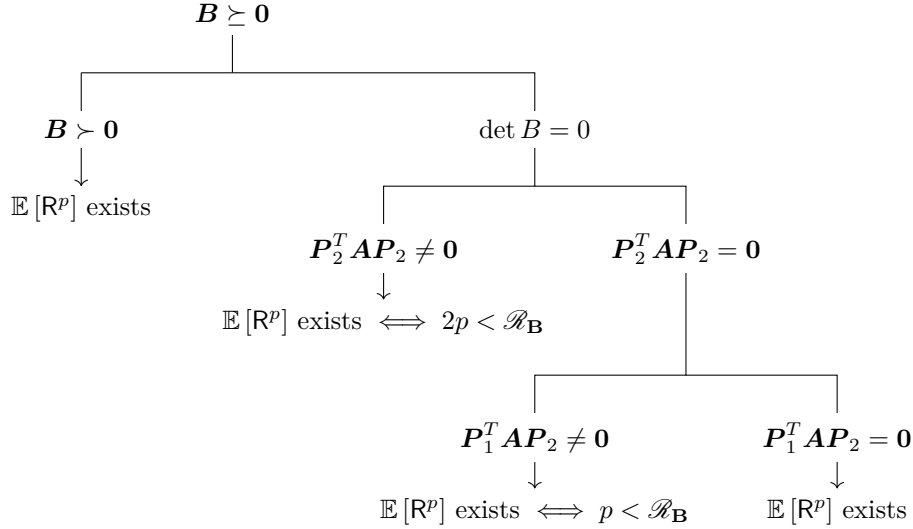

            Suppose that support of the ratio is an interval of lower and upper (possibly infinite) limits $L$ and $U$, respectively. Hence, the moments can be linked to the CDF of an indefinite form via integration by parts \begin{align*}
                \expect{\rs{R}^p}&=\int_L^U r^p \left[\dfrac{d}{dr}F_{\rs{R}}(r)\right]dr\\
                &=U^p-p\int_L^Rr^{p-1}F_{\rs{R}}(r)dr
            \end{align*}
            Bao and Kan \cite{baoMomentsRatiosQuadratic2013} report and produce a number of formulae for the evaluation of the moments. However, in this section, we are only concerned with the relation with indefinite forms.

            Bao and Kan \cite{baoMomentsRatiosQuadratic2013} generalize the work of Hillier et al. \cite{hillierGeneratingFunctionsMoments2014} to provide an infinite series representation for the moment of a ratio of quadratic forms with a positive semi-definite denominator. 

            \noindent Consider a ratio of quadratic forms of Gaussian random variables \eqref{eq:RQFGRV} \begin{equation*}
                \rs{R}=\dfrac{\rv{x}^T\dmt{A}\rv{x}}{\rv{x}^T\dmt{B}\rv{x}}, \qquad \rv{x}\sim \mathcal{N}_N(\dv{\mu},\dmt{I}_N)\footnote{\text{If the covariance matrix $\dmt{\Sigma}$ is non-singular, replace $\dmt{A}$ and $\dmt{B}$ by $\dmt{\Sigma}^{-\frac12}\dmt{A}\dmt{\Sigma}^{-\frac12}$} \text{ and $\dmt{\Sigma}^{-\frac12}\dmt{B}\dmt{\Sigma}^{-\frac12}$ respectively.}} 
            \end{equation*}
            Assume that the denominator is positive semi-definite and that the $p^{\text{th}}$ moment exists (refer to Figure \ref{fig:moment-existence-tree}). Then the $p^{\text{th}}$ moment is given by \begin{equation}
                \expect{\rs{R}^p}=\dfrac{p!\Gamma(N/2)\beta^p}{\Gamma(N/2+p)}\sum_{j=0}^\infty \dfrac{(p)_j}{\Gamma(N/2+j)}h_{p,j}(\dmt{A},\hat{\dmt{B}})
            \end{equation}
            where $\beta$ is an arbitrary parameter in $\left]0,\dfrac{2}{b_{\text{max}}}\right[$, $b_{\text{max}}$ is the maximum eigenvalue of $\dmt{B}$, $\hat{\dmt{B}}=\dmt{I}_N-\beta\dmt{B}$, $(a)_k=a(a+1) \cdots(a+k-1)$ is the Pochammer symbol, and the coefficients $h_{i,j}(\dmt{A}_1,\dmt{A}_2)$ are calculated recursively according to Algorithm \eqref{alg:Bao-hij}. Note that the recursion is over the sum of the indices.

            \begin{algorithm}
                    \caption{Evaluation of $h_{i,j}(\dmt{A}_1,\dmt{A}_2)$ }\label{alg:Bao-hij}
                    \begin{algorithmic}[1]
                        \Require $\dmt{A}_1,\dmt{A}_2\in\reals{N\times N},\dv{\mu}\in\reals{N},p\in\mathbb{N}^*,j_{\text{max}}$
                        \Ensure $h_{i,j}(\dmt{A}_1,\dmt{A}_2)$
                        \State $h_{0,0} \gets 0$
                        \State $\dv{g}_{0,0} \gets \dv{0}_{N\times 1}$
                        \State $\dmt{G}_{0,0} \gets \dmt{0}_{N \times N}$
                        \For{$k=1,2,\ldots,p+j_{\text{max}}$}
                            \For{$i=1,2,\ldots,k$}
                                \State $\dmt{G}_{i,k-i} \gets \dmt{A}_1(h_{i-1,k-i}\dmt{I}_N+\dmt{G}_{i-1,k-i})+\dmt{A}_2(h_{i,k-1-i}\dmt{I}_N+\dmt{G}_{i,k-1-i})$
                                \State $\dv{g}_{i,k-i} \gets (\dmt{G}_{i,k-i}-\dmt{G}_{i,k-1-i})\dv{\mu}-h_{i,k-1-i}\dv{\mu}+\dmt{A}_1\dv{g}_{i-1,k-i}+\dmt{A}_2\dv{g}_{i,k-1-i}$
                                \State $h_{i,k-i} \gets \dfrac{\operatorname{tr}(\dmt{G}_{i,k-i})+\dv{\mu}^T\dv{g}_{i,k-i}}{2k}$
                            \EndFor
                        \EndFor
                    \end{algorithmic}
                \end{algorithm}

                Based on Magnus's work \cite{magnusExactMomentsRatio1986}, Bao and Kan \cite{baoMomentsRatiosQuadratic2013} provide an integral formula for a ratio of quadratic forms with a positive semi-definite denominator. 

                Consider a ratio of quadratic forms of Gaussian random variables \eqref{eq:RQFGRV} \begin{equation*}
                    \rs{R}=\dfrac{\rv{x}^T\dmt{A}\rv{x}}{\rv{x}^T\dmt{B}\rv{x}}, \qquad \rv{x}\sim \mathcal{N}_N(\dv{\mu},\dmt{\Sigma}) 
                \end{equation*}
                Assume that the denominator is positive semi-definite, and that the $p^{\text{th}}$ moment exists (refer to Figure \ref{fig:moment-existence-tree}). Furthermore, assume that $\operatorname{Cov}(\rv{x})=\dmt{I}_N$ \footnote{If the covariance matrix $\dmt{\Sigma}$ is non-singular, replace $\dmt{A}$ and $\dmt{B}$ by $\dmt{\Sigma}^{-\frac12}\dmt{A}\dmt{\Sigma}^{-\frac12}$ and $\dmt{\Sigma}^{-\frac12}\dmt{B}\dmt{\Sigma}^{-\frac12}$ respectively.}. Then the $p^{\text{th}}$ moment is given by \[\mathbb{E}[R^p]=\dfrac{1}{\Gamma(p)}\int_0^\infty t^{p-1}\phi(t)\mathbb{E}[(\rv{w}^T\dmt{C}\rv{w})^p]dt\]
                where $\phi(t)=|\dmt{I}_N+2t\dmt{B}|^{-1/2}\exp\left(\dfrac{1}{2}\dv{\mu}^T[(\dmt{I}_N+2t\dmt{B})^{-1}-\dmt{I}_N]\dv{\mu}\right)$, the  expectation in the right-hand side is over $ \rv{w}\sim \mathcal{N}(\dmt{\dmt{L}^T\dv{\mu}},\dmt{I}_N)$, where $\dmt{C}=\dmt{L}^T\dmt{AL}$, and where $\dmt{L}$ is defined by $\dmt{L}\dmt{L}^T=(\dmt{I}_N+2t\dmt{B})^{-1}$. One can choose $\dmt{L}$ to be for example the square root of $(\dmt{I}_N+2t\dmt{B})^{-1}$. Note that the matrices $\dmt{L}$ and $\dmt{C}$ are dependent on the integration variable $t$. To evaluate the moment, we can use the usual method, which incorporates diagonalizing $\dmt{C}$ to obtain its eigenvalues: $\lambda_1,\lambda_2,\ldots,\lambda_N$, and associated non-centrality parameters $h_1,h_2,\ldots,h_N$, both of which will be functions of $t$. Bao and Kan use Brown's \cite{brownDistributionPositiveDefinite1986} algorithm to accelerate the recursive algorithm of the moment of a single form.


                The moments $\expect{(\rv{w}^T\dmt{C}\rv{w})^p}=2^p p! d_p$, where the sequence $d_p$ is recursively calculated using Algorithm \eqref{alg:Bao-dp}.
                \begin{algorithm}
                    \caption{Evaluation of $d_p=\expect{(\rv{w}^T\dmt{C}\rv{w})^p}/(2^p p!)$ }\label{alg:Bao-dp}
                    \begin{algorithmic}[1]
                        \Require $\dmt{A}\in\reals{N\times N},\dmt{B}>\dmt{0}\in\reals{N\times N},\dv{\mu}\in\reals{N},p\in\mathbb{N}^*,t>0$
                        \Ensure $d_p(t)$
                        \State $\dmt{L} \gets (\dmt{I}_N+2t\dmt{B})^{-\frac12}$
                        \State $\dmt{C} \gets \dmt{L}\dmt{AL}$
                        \State $\tilde{\dv{\mu}}\gets \dmt{L}\dv{\mu}$
                        \State $[\dmt{Q}^T,\dmt{\Lambda}] \gets \operatorname{eig}(\dmt{C})$
                        \For{$n=1,2,\ldots,N$}
                            \State $\lambda_n \gets \Lambda_{nn}$
                        \EndFor
                        \State $[h_1,\ldots,h_N]^T \gets \dmt{Q}^T\tilde{\dv{\mu}}$
                        \State $d_0 \gets 1$
                        \State $u_{n,0} \gets 0$
                        \State $v_{n,0} \gets 0$
                        \For{$k=1,2,\ldots,p$}
                            \State $d_k\gets 0$
                            \For{$n=1,2,\ldots,N$}
                                \State $u_{n,k} \gets \lambda_n(d_{k-1}+u_{n,k-1})$
                                \State $v_{n,k} \gets \lambda_n v_{n,k-1} + h_n u_{n,k}$
                                \State $d_k \gets d_k + u_{n,k} + v_{n,k}$
                            \EndFor
                            \State $d_k\gets d_k/(2k)$
                        \EndFor
                    \end{algorithmic}
                \end{algorithm}        

\subsection{PDF/CDF of QF Ratio}            
In the remaining part of this section, we report some results from the literature that deal with the PDF/CDF problem for the ratio of quadratic forms. First, we show the relation between ratios and indefinite quadratic forms, which can be exploit to obtain the PDF/CDF of the ratio. Then, we proceed to present some results from the literature. 

\subsubsection{Relation Between Ratios and Indefinite Forms}
Consider a ratio of quadratic forms in Gaussian random variables \eqref{eq:RQFGRV} with a non-zero positive semi-definite denominator \[\rs{R}=\dfrac{\rv{x}^T\dmt{A}\rv{x}}{\rv{x}^T\dmt{B}\rv{x}}, \qquad \rv{x}\sim \mathcal{N}_N(\dv{\mu},\dmt{\Sigma}), \qquad \dmt{B} \succcurlyeq \dmt{0}\] Note that the following inequalities are equivalent $$\rs{R}\leq r \iff \rv{x}^T(\dmt{A}-r\dmt{B})\rv{x}\leq 0$$
This is possible because the denominator is positive semi-definite. Note that the set in which the denominator vanishes, i.e., in which the ratio is not well-defined, is a null set. Therefore, the CDF, which is $F_{\rs{R}}(r)=\mathbb{P}[\rs{R}\leq r]$ is the CDF of the indefinite form $\rv{x}^T(\dmt{A}-r\dmt{B})\rv{x}$ evaluated at zero. In addition, in many cases of practical interest, the denominator is positive definite (for instance, \cite{maoModelQuantifyingUncertainty2012}).

At the first glance, the relationship between ratios and indefinite forms seems to be straightforward. If a formula can be used to calculate the CDF of indefinite forms at zero, then it can be utilized to evaluate the CDF of a ratio. However, the situation is a bit more complicated.  Firstly, from a computational perspective, if we are evaluating the CDF of an indefinite form at multiple points, we ought to diagonalize the form first, and then apply the formula at hand multiple times. However, in the case of a ratio, every point produces a new indefinite form, hence the need to rediagonalize every single time. Secondly, this does not seem to be helpful in obtaining a meaningful PDF, i.e., after differentiation with respect to $r$. The PDF of a ratio is only available in a handful of simple cases \cite{nadarajahNoteRatioIndependent2018,nadimiRatioIndependentComplex2018,yanCircularlysymmetricComplexNormal2016,maoModelQuantifyingUncertainty2012}. 
            
\subsubsection{Ratio of Moduli of Complex Random Variables.}
Consider two complex random variables, $\rs{x}_1$ and $\rs{x}_2$. The modulus ratio $\rs{Z}=\frac{|\rs{x}_1|}{|\rs{x}_2|}$  has been extensively studied, and its PDF is available in the literature. Let ($\rv{x}=[\bar{\rs{x}}_1,\bar{\rs{x}}]^H \sim \mathcal{CN}(\dv{\mu},\dmt{\Sigma})$). Then the squared-modulus ratio
\[
    \rs{R}=\frac{|\rs{x}_1|^2}{|\rs{x}_2|^2}
\]
is a particular quadratic-form ratio of \eqref{eq:RQFGRV} with
\[
    \dmt{A}=\begin{bmatrix}1&0\\0&0\end{bmatrix},\qquad
    \dmt{B}=\begin{bmatrix}0&0\\0&1\end{bmatrix}.
\]
Moreover, the PDF of ($\rs{R}$) follows directly from that of ($\rs{Z}$) through a change of variables:
\[
    f_{\rs{R}}(r)=\frac{1}{2\sqrt{r}},f_{\rs{Z}}(\sqrt{r}).
\]
Existing works provide multiple closed-form and semi-closed-form expressions for ($f_{\rs{Z}}$) under various assumptions on ($\dv{\mu}$ ) and ($\dmt{\Sigma}$). In general, these derivations rely on suitable variable transformations, algebraic simplification, and marginalization.

Assume ($\dv{\mu}=[\bar{\mu}_1,\bar{\mu}_2]^H$) and
\[
    \dmt{\Sigma}=\begin{bmatrix}
    \sigma_1^2 & \rho\sigma_1\sigma_2\\
    \bar{\rho}\sigma_1\sigma_2 & \sigma_2^2
                    \end{bmatrix}.
\]
Nadimi \cite{nadimiRatioIndependentComplex2018} and Nadarajah–Kwong \cite{nadarajahNoteRatioIndependent2018} derive the PDF of the modulus ratio for the independent case ($\rho=0$). In particular, the PDF of the associated quadratic-form ratio implied by Nadarajah–Kwong’s result can be written as
\begin{equation}
    \begin{aligned}
    f_{\rs{R}}(r)
    &=\frac{\exp\left(-\frac{\mu_1^2}{\sigma_1^2}-\frac{\mu_2^2}{\sigma_2^2}\right)}
    {2\sigma_1^2\sigma_2^2\left(\dfrac{r}{\sigma_1^2}+\dfrac{1}{\sigma_2^2}\right)^2},
    \Psi_2\left(2,1,1,\frac{r\mu_1^2}{\sigma_1^2 r+\frac{\sigma_1^4}{\sigma_2^2}},\frac{\mu_2^2}{r\frac{\sigma_2^4}{\sigma_1^2}+\sigma_2^2}\right),
    \end{aligned}
    \label{eq:Nadarjah-Kwong_PDF}
    \end{equation}
    where ($\Psi_2$) denotes the Horn confluent hypergeometric series
    \[
        \Psi_2(a,b,c;x,y)=\sum_{n=0}^{\infty}\sum_{k=0}^{\infty}\frac{(a)_{n+k}}{(b)_n(c)_k,n!,k!},x^n y^k.
    \]
    Gu \cite{guQuotientCentralizedNoncentralized2020} obtains the modulus-ratio PDF for ($\mu_1=0$) and ($\rho=0$), which can be recovered as a special case of \eqref{eq:Nadarjah-Kwong_PDF}.
    For correlated complex Gaussian variables with ($\mu_1=\mu_2=0$), Yan \cite{yanCircularlysymmetricComplexNormal2016} derives 
    \begin{equation}
        f_{\rs{R}}(r)=\frac{\sigma_1^2\sigma_2^2(1-|\rho|^2)(\sigma_2^2+r\sigma_1^2)}
        {\left[(\sigma_2^2+r\sigma_1^2)-4r|\rho|^2\sigma_1^2\sigma_2^2\right]^{3/2}}.
        \label{eq:Yan_PDF}
    \end{equation}

 \subsection{Ratios of Central Complex Forms} Al-Naffouri et al. \cite{al-naffouriDistributionIndefiniteQuadratic2016} consider the case of the ratio of two central complex quadratic forms, and gave a formula for the CDF based on their one concerning the central complex case. Let \[\rs{R}=\dfrac{\epsilon_1+\rv{x}^H\dmt{B}_1\rv{x}}{\epsilon_2+\rv{x}^H\dmt{B}_2\rv{x}}\]
where $\dmt{B}_1$ is Hermitian and $\dmt{B}_2$ is positive semi-definite. Noting that the inequality $\rs{R} \leq r$ is equivalent to
\[\rv{x}^H\left(\dmt{B}_1-r\dmt{B}_2\right)\rv{x}\leq -\epsilon_1+r\epsilon_2\]
Then this problem is equivalent to a central problem for each $y$, i.e.,
\begin{equation}F_{\rs{R}}(r)=F_{\rs{Q}_r}(\epsilon_2 r-\epsilon_1)\label{Al-naffouri Ratios}\end{equation}
where $\rs{Q}_r=\rv{x}^T(\dmt{B}_1-r\dmt{B}_2)\rv{x}$.
        
However, it must be noticed that this transformation requires the calculation of the CDF of a different quadratic form for each \(y\), hence diagonalization at each value, i.e., the eigenvalues of the form are functions of \(y\).

\subsubsection{Ratios of Sums of Squared Moduli} \label{par:SimonAlouini}  

Simon and Alouini \cite{simonDifferenceTwoChisquare2001} provide three expressions of the CDF at zero of an indefinite form, and that can be considered the CDF of a ratio of sums of squared moduli.            
Let $\rs{R}=\dfrac{\sum_{k=1}^{L_1}|\rs{x}_k|^2}{\sum_{k=1}^{L_2}|\rs{y}_k|^2}$, $\{\rs{x}_k\}$, $\{\rs{y}_k\}$ be mutually independent Gaussian random variables with means \(\bar{x}_k,\bar{y}_k\) and variances $\sigma_1^2=\mathbb{E}\{|\rs{x}_k-\bar{x}_k|^2\}$, $\sigma_2^2=\mathbb{E}\{|\rs{y}_k-\bar{y}_k|^2\}$ independent of $k$. Using our notation, this accounts for the ratio of complex quadratic forms \(\rs{R}=\dfrac{\rv{x}^H\dmt{I}_{L_1}\rv{x}}{\rv{x}^H\dmt{I}_{L_2}\rv{x}}\), and $\rv{x}\sim\mathcal{N}(\dv{\mu},\dmt{\Sigma})$ with \(\dmt{\Sigma}=\left[\begin{array}{c|c}
\sigma_1^2\dmt{I}_{L_1}       &\dmt{0}  \\ \hline
\dmt{0}               &\sigma_2^2\dmt{I}_{L_2} 
\end{array}\right]\) (\(L_1+L_2=N\)) and \(\dv{\mu}=[\bar{x}_1,\ldots,\bar{x}_{L_1},\bar{y}_1,\ldots,\bar{y}_{L_2}]^T\). The authors provide three expressions for the corresponding indefinite form at 0 which can be used to derive the CDF of the ratio. Two formulae are closed-form expressions and the third one is an integral.
            
One of the expressions is given by
\begin{equation}
\resizebox{\textwidth}{!}{$
\begin{aligned}
F_{\rs{R}}(r)=&Q_1(a,b)-\exp\left[-\frac12(a^2+b^2)\right]I_0(ab)\nonumber\\ &+\dfrac{I_0(ab)\exp\left[-\frac12(a^2+b^2)\right]}{(1+\eta)^{L_1+L_2-1}}\sum_{k=0}^{L_2-1}\binom{L_1+L_2-1}{k}\eta^k\nonumber\\
&+\dfrac{\exp\left[-\frac12(a^2+b^2)\right]}{(1+\eta)^{L_1+L_2-1}}\left[\sum_{n=1}^{L_2-1}I_n(ab)\sum_{k=0}^{L_2-1-n}\binom{L_1+L_2-1}{k}\eta^k\left(\dfrac{b}{a}\right)^n\right.\nonumber\\
&\left.-\sum_{n=1}^{L_1-1}I_n(ab)\sum_{k=0}^{L_1-1-n}\binom{L_1+L_2-1}{k}\eta^{L_1+L_2-1-k}\left(\dfrac{a}{b}\right)^n\right], \quad r>0
\end{aligned}
$}
\label{eq:SimonAlouini_CDF}
\end{equation}
where \begin{align*}
    a&=\sqrt{\dfrac{2r\sum_{k=1}^{L_2}|\bar{y}_k|^2}{\sigma_1^2+r\sigma_2^2}}\\
    b&=\sqrt{\dfrac{2\sum_{k=1}^{L_1}|\bar{x}_k|^2}{\sigma_1^2+r\sigma_2^2}}\\
    \eta&=\dfrac{\sigma_1^2}{r\sigma_2^2}
\end{align*}
The derivation follows the proof provided in \cite[App. B]{proakisDigitalCommunications2008}. The following summarizes the proof in \cite[App.A]{simonDifferenceTwoChisquare2001} referring to the equation numbers.  The proof starts by defining the decision variable (D) in (1) and the probability of interest ($P=\Pr(D<0)$ ) in (3). Using the independence of the terms ($|X_k|^2$) and ($|Y_k|^2$), the characteristic function of (D) is written in (23). After introducing the auxiliary parameters ($\nu_1,\nu_2,\xi_1,\xi_2$) in (24)–(25) and ($\alpha_1,\alpha_2$) in (27), the characteristic function is rewritten in the structured form (28). Applying the conformal transformation (29) converts the inversion integral into the contour representation (31) with integrand (32) and constants defined in (33). Introducing (a) and (b) in (34) and expanding binomial term in (35) leads to the finite sum representation (36)–(37). The resulting contour integrals are evaluated in the three cases (38)–(40) using Bessel and Marcum-(Q) function identities. Substituting these results into (37) gives (41), and combining it with the prefactor in (31) and identity (42) yields the final closed-form expression, equation (4), with parameters (a), (b), and ($\eta$) defined in (5)–(7).

Further, Simon and Alouini \cite{simonDifferenceTwoChisquare2001} provide an integral expression of the CDF in which the integrand is an elementary function. The CDF is given by $$\begin{aligned}
    F_{\rs{R}}(r)&=\dfrac{\eta^{L_2}}{2\pi (1+\eta)^{L_1+L_2-1}}\int_{-\pi}^\pi \left[\dfrac{f(L_2,L_1;\zeta,\eta;\theta}{1+2\zeta\sin\theta+\zeta^2}\right]\\&\times\exp\left(-\dfrac{b^2}{2}[1+2\zeta\sin\theta+\zeta^2]\right)d\theta
    \end{aligned}$$
    where $\zeta:=\dfrac{a}{b}$ is assuemd to be less than $1$, and $$\begin{aligned}
    f(L_1,L_2;\zeta,\eta;\theta)&=\sum_{\ell=1}^{L_1}\binom{L_1+L_2-1}{L_1-\ell}\eta^{-\ell}\zeta^{-\ell+1}\\
    &\times\left[\cos\left((\ell-1)\left(\theta+\dfrac{\pi}{2}\right)\right)-\zeta\cos\left(\ell\left(\theta+\dfrac{\pi}{2}\right)\right)\right]\\
    &+\sum_{\ell=1}^{L_2}\binom{L_1+L_2-1}{L_1-\ell}\eta^{\ell-1}\zeta^\ell\\
    &\times\left[\cos(\ell\left(\theta+\dfrac{\pi}{2}\right)-\zeta\cos\left((\ell-1)\left(\theta+\dfrac{\pi}{2}\right)\right)\right]    \end{aligned}$$

Mao and Todd \cite{maoModelQuantifyingUncertainty2012} consider the central case of $\dmt{A}=\begin{bmatrix}
\dmt{I}_{N/2} &\dmt{0}\\
\dmt{0} &\dmt{0}
\end{bmatrix}$
and $\dmt{B}=\begin{bmatrix}
\dmt{0} &\dmt{0}\\
\dmt{0} &\dmt{I}_{N/2}
\end{bmatrix}$ with the covariance matrix \\
\(\dmt{\Sigma}=\begin{bmatrix}
\sigma_{11}^2\dmt{I}_{N/2} &\sigma_{12}\dmt{I}_{N/2}\\
\sigma_{12}\dmt{I}_{N/2} &\sigma_{22}^2\dmt{I}_{N/2}
\end{bmatrix}\). Equivalently, \[\rs{R}=\dfrac{\sum_{k=1}^{N}|\rs{x}_k|^2}{\sum_{k=1}^{N}|\rs{y}_k|^2}\] where \([\rs{x}_k,\rs{y}_k]^T\sim\mathcal{CN}\left(\dv{0},\begin{bmatrix}
                \sigma_{11} &\sigma_{12}\\
                \sigma_{12} &\sigma_{22}
            \end{bmatrix}\right)\).

\subsubsection{Ratios of Independent Real Quadratic Forms} A frequently studied case is the ratio of quadratic forms of independent variables. It can be perceived as a ratio \(\rs{R}=\dfrac{\rv{x}^T\dmt{A}\rv{x}}{\rv{x}^T\dmt{B}\rv{x}}\) with 
\[
    \dmt{A}=
\begin{bmatrix}
    \dmt{A}_1 &\dmt{0}\\
    \dmt{0} &\dmt{0}    
\end{bmatrix} \text{ and } 
\dmt{B}=
\begin{bmatrix}
\dmt{0} &\dmt{0}\\
\dmt{0} &\dmt{A}_2    
\end{bmatrix}
\]where $\dmt{A}$ and $\dmt{B}$ are conformally partitioned: \(\dmt{A}_1\in \mathbb{R}^{m\times m},\) and \(\dmt{A}_2\in \mathbb{R}^{(N-m)\times (N-m)}\).
            
Suppose \(\rv{x}\) is central. By simultaneously diagonalizing the numerator and denominator such that they are linear combinations of independent (different) chi-square variables, we can write:\[R=\dfrac{\sum_{i=1}^m \lambda_i\rs{Z}_i}{\sum_{i=m+1}^N \lambda_i\rs{Z}_i}\] where $\rs{Z}_i\sim \chi_1^2$

To derive the CDF, start from the ratio ($\rs{R}=\rs{Q}_1/\rs{Q}_2$) (with ($\rs{Q}_2>0$))  and rewrite the CDF as a sign probability:
\[
F_R(c)=\Pr(R<r)=\Pr(\rs{Q}_1-r\rs{Q}_2<0).
\]
Now collect the terms in ($\rs{Q}_1-r\rs{Q}_2$) with positive coefficients into a nonnegative random sum ($\rs{Y}^+$) and the terms with negative coefficients into another nonnegative sum ($\rs{Z}^-$), so that ($\rs{Q}_1-c\rs{Q}_2 = \rs{Y}^+-\rs{Z}^-$). After an optional common scaling by ($\theta$) (to stabilize the series), define ($\rs{Y}=\rs{Y}^+/\theta$) and ($\rs{Z}=\rs{Z}^-/\theta$); then ($F_\rs{R}(r)=\Pr(\rs{Y} < \rs{Z})$), which can be written as the double integral
\[
\Pr(\rs{Y}< \rs{Z})=\int_0^\infty\Big(\int_0^z h_1(y) dy\Big)h_2(z)dz,
\]
where ($h_1,h_2$) are the densities of ($\rs{Y}$) and ($\rs{Z}$).  The key step is that a finite sum of independent gamma variables has a density that can be expanded as a convergent series of gamma-kernel terms (polynomial in ($y$) or ($z$) times an exponential), so write
\[
h_1(y)=\sum_{\nu\ge0}K_v,y^{p_1+\nu-1}e^{-y},\qquad
h_2(z)=\sum_{w\ge0}K_w^{\prime\prime},z^{p_2+w-1}e^{-z},
\]
with coefficients ($K_\nu,K_w^{\prime\prime}$) determined by the gamma-sum parameters. Substituting these series into the double integral, the inner integral ($\int_0^z y^{p_1+v-1}e^{-y}dy$) is an incomplete-gamma term, which is expanded using a standard identity into a finite/infinite sum indexed by ($j$); then the remaining ($z$)-integral becomes a gamma integral and can be evaluated term-by-term. Collecting the resulting factors produces the triple-series in indices ($(v,w,j)$) with gamma-function terms—this collected expression is \cite{mathaiQuadraticFormsRandom1992} . 

\begin{equation}
F_{\rs{R}}(r)=\sum_{\nu=0}^\infty \sum_{w=0}^\infty \sum_{j=0}^\infty K_\nu K'_w\dfrac{\Gamma(\frac{N}{2}+\nu+w+j)2^{-(N/2+\nu+w+j)}}{(m/2+\nu)_{j+1}}
            \label{eq:MathaiRatios_CDF}
\end{equation}
where \begin{align*}
K_\nu&=\sum_{\nu_1+\ldots+\nu_m=\nu}\prod_{j=1}^m \mu_j^{-1/2}\left(\dfrac12\right)_{\nu_j}\dfrac{\gamma_j^{\nu_i-\nu_j}}{\nu_j!\Gamma(m/2+\nu)}\\
K_w&=\sum_{w_{m}+\ldots+w_N=w}\prod_{j=m+1}^N \mu_j^{-1/2}\left(\dfrac12\right)_{w_j}\dfrac{\gamma_j^{w_j}}{w_j!\Gamma((N-m)/2+w)}\\
\gamma_j&=\dfrac{\mu_j-1}{\mu_j}\\
\mu_j&=\begin{cases}\dfrac{\lambda_j}{\theta} &\text{for }j=1,\ldots,m\\
\dfrac{r\lambda_j}{\theta} &\text{for }j=m+1,\ldots,N
\end{cases}
\end{align*}

and \(\theta\) is an arbitrary constant. The Pochammer's symbol is defined as follows:

$$
\left(\frac{m}{2}+\nu\right)_{j+1}= \left(\frac{m}{2}+\nu\right) \left(\frac{m}{2}+\nu+1\right) \left(\frac{m}{2}+\nu+2\right)\cdots \left(\frac{m}{2}+\nu+j\right)
$$

Gardini \cite{gardiniMellinTransformManage2022} considers the ratio of two independent positive definite quadratic forms. After diagonalization, one ends up with \(\rs{R}\sim\dfrac{\sum_{i=1}^{N_1}\lambda_{1,i}\chi_1^2(h_{1,i}^2)}{\sum_{j=1}^{N_2}\lambda_{2,i}\chi_1^2(h_{2,i}^2)}\), \(\lambda_{1,i}>0,\lambda_{2,i}>0\), where all the chi-squares are independent.
The PDF is given by
 \[
    f_\rs{R}(r)=\sum_{k=0}^\infty \sum_{j=0}^\infty c_{1,k}c_{2,j}\left(\dfrac{\beta_1}{\beta_2}\right)^{\alpha_1+k}\dfrac{r^{\alpha_1+k-1}}{B(\alpha_1+k,\alpha_2+j)}\left(1+\dfrac{\beta_2 r}{\beta_1}\right)^{-(\alpha_1+\alpha_2+k+j)}
\]
where: 
\begin{align*}
\alpha_1&=N_1/2 \\\alpha_2&=N_2/2\\
c_{1,0}&=\exp\left(-\dfrac{1}{2}\sum_{j=1}^{N_1}h_{1,j}^2\right)\prod_{j=1}^{N_1}(\beta_1/\lambda_{1,j})^{1/2} \\c_{2,0}&=\exp\left(-\dfrac{1}{2}\sum_{j=1}^{N_2}h_{2,j}^2\right)\prod_{j=1}^{N_2}(\beta_2/\lambda_{2,j})^{1/2}\\
c_{1,k}&=(2k)^{-1}\sum_{r=0}^{k-1}d_{1,k-r}c_{1,r} \\c_{2,k}&=(2k)^{-1}\sum_{r=0}^{k-1}d_{2,k-r}c_{2,r}\\
d_{1,k}&=\sum_{j=1}^{N_1}(1-\beta_1/\lambda_{1,j})^k+k\beta_1\sum_{j=1}^{N_1}(h_{1,j}^2/\lambda_{1,j})(1-\beta_1/\lambda_{1,j})^{k-1} \\d_{2,k}&=\sum_{j=1}^{N_2}(1-\beta_2/\lambda_{2,j})^k+k\beta_2\sum_{j=1}^{N_2}(h_{2,j}^2/\lambda_{2,j})(1-\beta_2/\lambda_{2,j})^{k-1} \\
\end{align*}
                
As Gardini \cite{gardiniMellinTransformManage2022} points out, this equation is not practical, so instead, the inverse Mellin transform is calculated. The Mellin transform is given by \[\hat{f}_{\rs{R}}(z)=\hat{f}_{\rs{Q}_1}(z)\hat{f}_{\rs{Q}_2}(2-z)\]
So the PDF is given by: 
$$\begin{aligned} {f}_{\rs{R}}(r)\approx & \frac{\Delta}{2 \pi}\left(\frac{\beta_2}{\beta_1}\right) \sum_{t=-T}^T\left(\frac{\beta_1}{r \beta_2}\right)^{h+i \Delta t}\\
&\times \left(\sum_{k=0}^K  a_{k, 1}P_k\left(\alpha_1, h-1+i \Delta t\right)\right)\left(\sum_{j=0}^J a_{j, 2} P_j\left(\alpha_2, 1-h-i \Delta t\right)\right)\end{aligned}$$
where $$P_k(\alpha, z-1)= \begin{cases}\frac{\Gamma(z+\alpha-1)}{\Gamma(\alpha)}, & k=0 \\ P_{k-1}(\alpha, z-1)\left(1+\frac{z-1}{\alpha+k+1}\right), & k>0.\end{cases}$$

\subsubsection{Ratio of QF using Saddlepoint Approximation}
For the ratio of quadratic forms in Gaussian random variables, Butler \cite{2007-Butler-saddlepointAapproximationsWithApplicationst} developed saddlepoint approximations for both the density and the CDF. Let the ratio $\rs{R}$ be defined as in \eqref{eq:RQFGRV}. To evaluate its distribution at a given threshold $r$, note that
\[
\mathbb{P}[\rs{R}\le r]
=
\mathbb{P}[\rs{Q}_r\le 0],
\]
where
\[
\rs{Q}_r=\rv{x}^T(\dmt{A}-r\dmt{B})\rv{x}.
\]
Hence, the problem reduces to approximating the distribution of the quadratic form $\rs{Q}_r$ at zero.

Assume that the matrix $\dmt{A}-r\dmt{B}$ is diagonalizable, and write
\[
\dmt{A}-r\dmt{B}=\dmt{P}^T\dmt{\Lambda}\dmt{P},
\]
where
\[
\dmt{\Lambda}=\operatorname{diag}(\omega_1,\omega_2,\ldots,\omega_n).
\]
Then $\rs{Q}_r$ can be expressed as
\[
\rs{Q}_r \sim \sum_{i=1}^n \omega_i \chi_1^2(\delta_i^2),
\]
where the noncentrality parameters are determined by
\[
\dv{\delta}=[\delta_1,\delta_2,\ldots,\delta_n]^T=\dmt{P}\dv{\mu}.
\]
Both the coefficients $\omega_i$ and the noncentrality parameters $\delta_i^2$ depend on the threshold $r$. Using the saddlepoint approach, the density of $\rs{R}$ can be approximated by
\begin{equation}
\label{eq:RQF_SPA_pdf}
f_{\rs{R}}(r)
\approx
\frac{J_r(t_o)}{\sqrt{2\pi K_{\rs{Q}_r}^{\prime\prime}(t_o)}}\,
M_{\rs{Q}_r}(t_o),
\end{equation}
where $t_o$ is the saddlepoint associated with $\rs{Q}_r$, namely the solution of
\[
K_{\rs{Q}_r}'(t)=0.
\]
Here, $M_{\rs{Q}_r}(t)$ is the MGF of $\rs{Q}_r$ given in \eqref{eq:MGF_incmplete_1} with $\sigma=0$ and $c^{\prime\prime}=0$.

The correction factor $J_r(t)$ is given by
\begin{equation}
\label{eq:RQF_SPA_Jr}
J_r(t)
=
\operatorname{tr}\!\left[(\dmt{I}-2t\dmt{\Lambda})^{-1}\dmt{H}\right]
+
\dv{\delta}^T
(\dmt{I}-2t\dmt{\Lambda})^{-1}
\dmt{H}
(\dmt{I}-2t\dmt{\Lambda})^{-1}
\dv{\delta},
\end{equation}
where$\dmt{H}=\dmt{P}\dmt{B}\dmt{P}^T.$

\section{Computational Concerns}
An important question that arises in the context is: \textit{``given this particular quadratic form what method should I use to compute the CDF, PDF and  moments of the quadratic form?''}. This of course is difficult to answer exactly and concisely. 
A rough idea of the expected range of probabilities is required. Then, from our experimentation, the following two quick heuristics are useful: (1) For moderate values of probabilities, use the exact methods or moment matching (2) For small values of probabilities (tail probabilities $< 10^{-8}$), use the saddlepoint approximation. 


 Consider the implementations in \verb|CompQuadForm| of the Davies, Imhoff, Farebrother and Moment Matching methods. For each method proposed there is a ``failure'' region. Note the first three methods are parameterized by accuracy parameters that the user needs to supply. These parameters may lead to inaccurate results. We have collected examples of failures in the repository \href{https://github.com/moaazhaj/QuadFormsCompute/tree/main}{\texttt{QuadFormsCompute}}. Consider the simple case of $Q = 1\chi^2_1(1) + 0.6^4 \chi^2_1(7)$. When computed using imhoff's method, the CCDF at arguments of above 30 produce inaccurate vaules (exceeding 50\% error, and even negative sometimes!). Similar investigation for other methods has been carried out. Noteworthy is Davies's method, which typically produces accurate results (when it produces a result). The implementation in \texttt{CompQuadForm} requires parameters for the number of intervals of numerical integartion and required accuracy, and determining the correct combination of these is done by trial and error. Unintuitively, sometimes increasing accuracy value can lead  to the method failure.  
 
 The moment matching methods are known not to produce accurate results in the tail regions (as expected theoretically). The saddlepoint method can produce inaccurate results in the bulk close to the mean but performs very well in the tail. Demonstartions of accurate tail behaviour can be seen in \cite[Table 3.1]{chenNumericalEvaluationMethods2019}

Consider a scenario where we want to compute the CDF/CCDF for a particular QF at many values. In series based methods such as Farebrother, the series coefficients can be computed once then repeat CDF/CCDF computations can be carried out many times. Numerical integartion methods however require full recomputation when the evaluation point of the CDF/CCDF is changed. In this regard the saddlepoint method's root depends on CDF/CCDF evaluation point and hence will require root recomputation. In general however root computation for  moderate numbers of eigenvalues requires much less computation than numerical integration and series calculations.

A common cost associated with quadratic form computations is the calculation of the eigenvalues (diagonalization) and conversion to effective input. Then the QF can be diagonalized once and the computation repeated many times. Note eigendecomposition of matrices is cubic in the dimension of the matrix. As the dimension increases this can become computationally prohibitive. An example of this is genome wide association studies where the number of variables is the number of gene loci under study. For example, in \cite{chenNumericalEvaluationMethods2019}, Quadratic forms with up to 20000 variables are considered.

The appoarch taken in \cite{chenNumericalEvaluationMethods2019} is based on the huerisitc that when there are many eigenvalues only the largest eigenvaules contribute significantly to the overall computation. These are extracted using randomized methods, \cite{halkoTropp2011FindingStructurewithRandom}. For example in \cite{chenNumericalEvaluationMethods2019}, the authors extract the top 200 eigenvaules out of about 20000. These are then used to evaluate the quadratic form. It is observed that there is no loss in accuracy.  

There are a few packages that can be used to compute quadratic forms inculding: in the R language CompQuadForm, survey, bigQF, in Python: chi2comb and others. For ratio computation the package \verb|qfRatio| provides comptutional methods for CDF, PDF and moments. It employs dedicated methods from \cite{baoMomentsRatiosQuadratic2013,smithExpectationRatioQuadratic1989, hillierGeneratingFunctionsMoments2014} for moments calculations, whilst the CDF is computed by the methods of Davies \cite{daviesAlgorithm155Distribution1980}, Imhoff \cite{imhofComputingDistributionQuadratic1961} and the saddlepoint point method (\cite{2007-Butler-saddlepointAapproximationsWithApplicationst}) after conversion into indefinite quadratic forms. The PDF is computed through the saddlepoint approximation, through Broda's method \cite{broda2009evaluating} or hiller's method.

\begin{figure*}
    \centering
    \includegraphics[width=\linewidth]{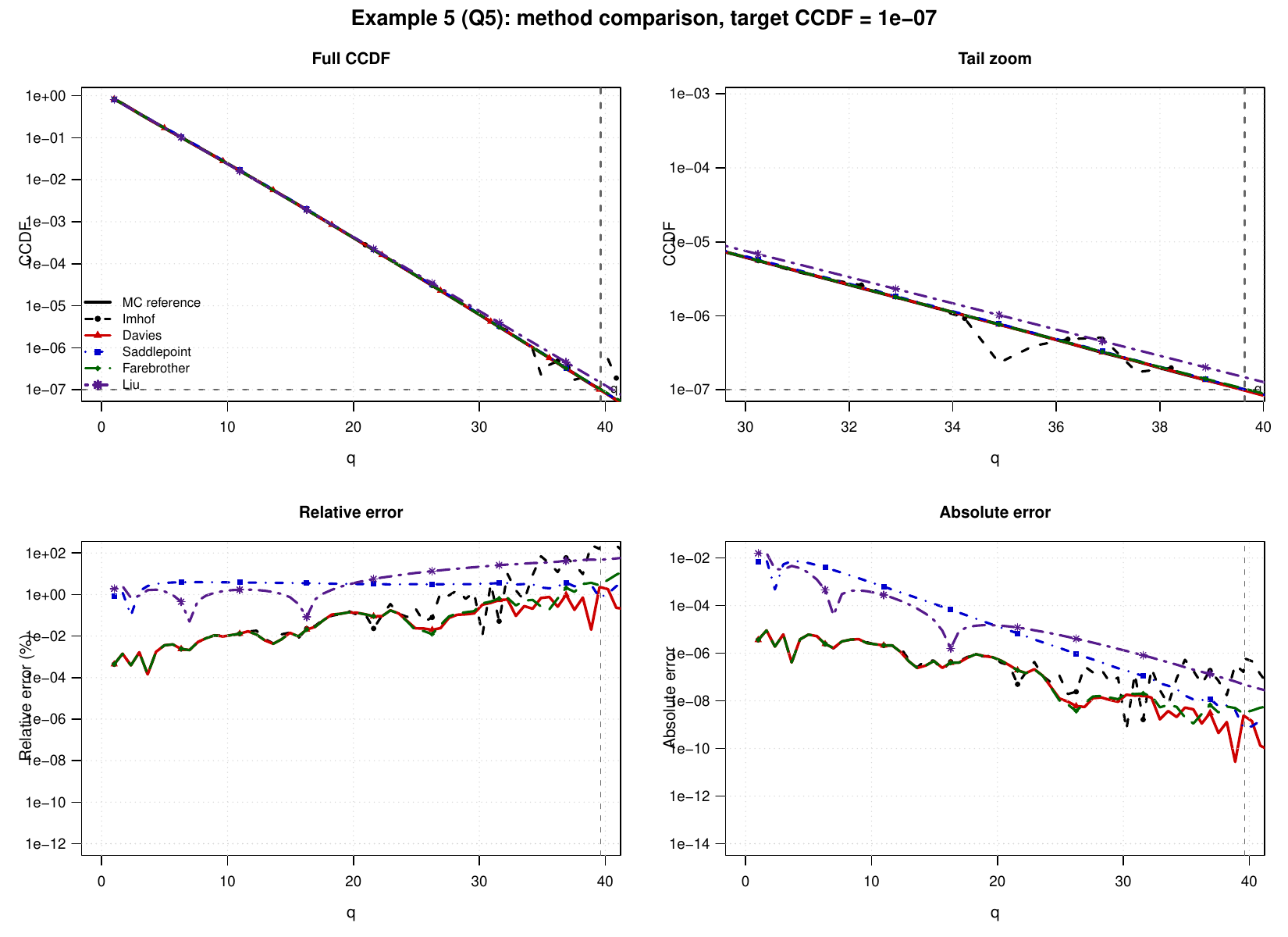}
    \caption{Caption}
    \label{fig:enter-label}
\end{figure*}

\section{Literature Summaries}
    We will summarize on single forms and ratios using a bunch of visual aids. We start by classifying the forms discussed in the literature.
    \begin{figure}
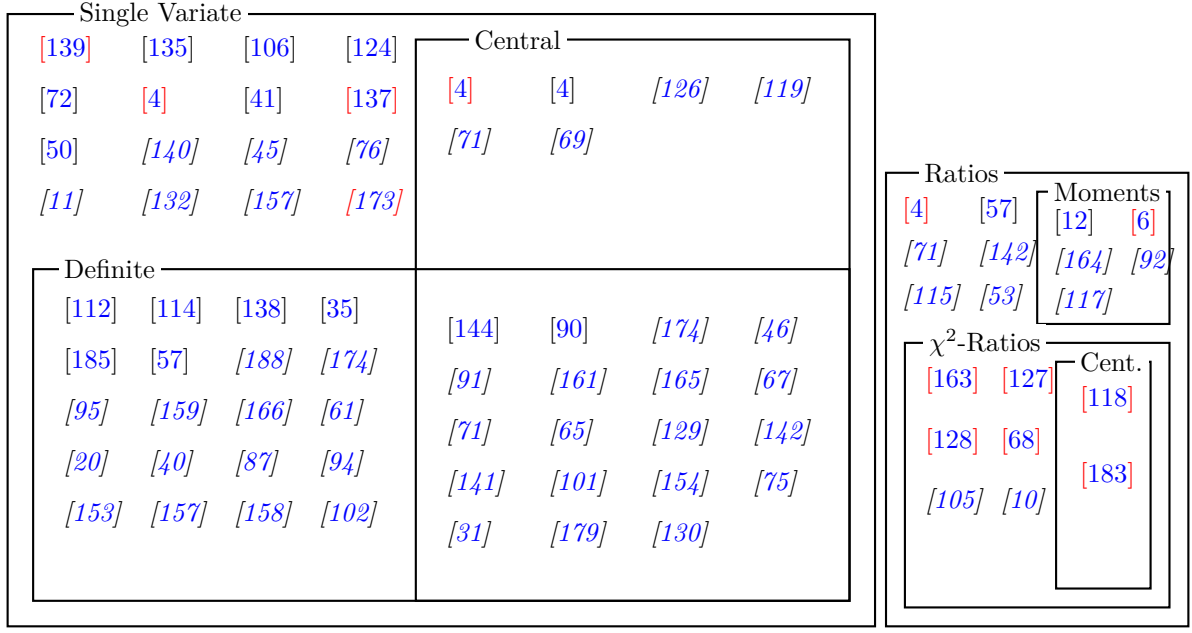

                \begin{center}




\begin{tikzpicture}[scale=1.35,x=1cm, y=1cm]
    \draw[color=black,thick] (0,0) rectangle (8.5,6);
    \Text[color=black,x=1.5,y=6, style={fill=white}]{Single Variate}
    \draw[color=black,thick] (0.25, 0.25) rectangle (8.25,3.5);
    \Text[color=black,x=1.0,y=3.5, style={fill=white}]{Definite}
    \draw[color=black,thick] (4, 0.25) rectangle (8.25,5.75);
    \Text[color=black,x=5,y=5.75, style={fill=white}]{Central}
    \input{Chapter3Modified/Chapter3Figures/Classifs/classifAandB_S}
\end{tikzpicture}
\begin{tikzpicture}
    \draw[color=black,thick] (0,0) rectangle (4,6);
    \Text[color=black,x=1.0,y=6, style={fill=white}]{Ratios}
    \draw[color=black,thick] (2.0,4.0) rectangle (3.75,5.75);
    \Text[color=black,x=2.15,y=5.75, anchor=west, style={fill=white}]{Moments}
    \draw[color=black,thick] (0.25,0.25) rectangle (3.75,3.75);
    \Text[color=black,x=0.5,y=3.75, style={fill=white}, anchor=west]{$\chi^2$-Ratios }
    \draw[color=black,thick] (2.25,0.5) rectangle (3.5,3.5);
    \Text[color=black,x=2.5,y=3.5, style={fill=white}, anchor=west]{Cent.}
    \input{Chapter3Modified/Chapter3Figures/Classifs/classifAandB_R}
\end{tikzpicture}
                \end{center}
                \caption{Classification according to forms}
                \label{fig:classForm}
            \end{figure}
    In Figure \ref{fig:classForm}, red-bracketed citations refer to complex forms, and italic citations refer to papers before 1990.
    Regarding the classification, single forms are classified according to centrality and definiteness. However, we can easily extend the definitions of centrality and definiteness to include certain incomplete forms: whenever the form is equivalently distributed as a linear combination of (central) chi-squares, it is called central, and whenever it is distributed as a linear combination of non-central chi-squares with positive linear coefficients, it is called definite. Ratios, on the other hand, were classified based on chi-squaredness. Since many works on ratios are only devoted to the moments of ratios, they are given a separate category.
    In Table \ref{tab:single-formula-type}, we classify the papers according to their formula type. Evidently, this follows the same organization of Chapter 3.
    \begin{table}[ht!]
        \centering
        \resizebox{\columnwidth}{!}{
        \newcommand{\Smhndcite}[2][]{\citeauthor*{#2} \textcolor{blue}{\cite[#1]{#2}},}
{\centering
\begin{tabular}{|p{8cm}|p{8cm}|}
    \hline
    Exact Formulae & Approximate Formulae \\
    \hline
    \textbf{Finite Expressions} \newline
    \Smhndcite{guQuotientCentralizedNoncentralized2020}  
    \Smhndcite{nadarajahNoteRatioIndependent2018}
    \Smhndcite{al-naffouriDistributionIndefiniteQuadratic2016}
    \Smhndcite{yanCircularlysymmetricComplexNormal2016}
    \Smhndcite{joarderExactDistributionSum2013}
    \Smhndcite{maoModelQuantifyingUncertainty2012}
    \newline \rule{8cm}{1pt}
    \newline
    \textbf{Infinite Series} \newline
    \Smhndcite{nadimiRatioIndependentComplex2018}
    \Smhndcite{cuiExactDistributionProduct2016}
    \Smhndcite{baoMomentsRatiosQuadratic2013}
    \Smhndcite{provostExactDistributionIndefinite1996}
    \Smhndcite{raphaeliDistributionNoncentralIndefinite1996}
    \newline \rule{8cm}{1pt} \newline
    \textbf{Numerical Integartion} \newline
    \Smhndcite{gardiniMellinTransformManage2022}
    \Smhndcite{baoMomentsRatiosQuadratic2013}
    \Smhndcite{royenIntegralRepresentationsConvolutions2007}
    \Smhndcite{simonDifferenceTwoChisquare2001}
    \newline \rule{8cm}{1pt} \newline
    \textbf{Sequences of RVs} \newline
    \Smhndcite{ramirez-espinosaNewApproachStatistical2019}
    \Smhndcite{ramirez-espinosaNewApproximationDistribution2019}
    \Smhndcite{haAccurateApproximationDistribution2013}
    & 
    \textbf{Moment Matching} \newline
    \Smhndcite{zhangEfficientAccurateApproximation2020} 
    \Smhndcite{chenNumericalEvaluationMethods2019}
    \Smhndcite{mohsenipourApproxDistIndef2011}
    \Smhndcite{mohsenipourApproxRatiosDiff2010}
    \Smhndcite{liuNewChisquareApproximation2009}
    \Smhndcite{lindsayMomentBasedApproximationsDistributions2000}
    \Smhndcite{woodApproximationDistributionLinear1989}
    \Smhndcite{buckleyApproximationDistributionQuadratic1988}
    \Smhndcite{hallChiSquaredApproximations1983}
    \Smhndcite{satterthwaiteApproximateDistributionEstimates1946}
    \Smhndcite{pearsonNoteApproximationDistribution1959}
    \newline \rule{8cm}{1pt} \newline
    \textbf{Saddle-point Approximation} \newline
    \Smhndcite{al-naffouriDistributionIndefiniteQuadratic2016}
    \Smhndcite{kuonenSaddlepointApproximationsDistributions1999}
    \\
    \hline
\end{tabular}}}
        \caption{Single Forms and Ratios: Formula Type}
        \label{tab:single-formula-type}
    \end{table}
\newpage

\begin{longtable}{|c|l|c|c|}
    \caption{Summary of all methods } \label{tab:chapter3}\\
    \hline
    SN.& Method/Paper & Equation No.& Code \\ \hline
     1 & Mathai/AlNaffouri PDF (Finite Expressions) & \eqref{eq:PDF_CentralEven}& $\checkmark$\\ \hline
     2 & Mathai CDF (Finite Expressions) & \eqref{eq:CDF_CentEven}& \\ \hline
     3 & AlNaffouri CDF (Finite Expressions) & \eqref{eq:Naffouri-Central-CDF}& $\checkmark$\\ \hline
     4 & Simon PDF & \eqref{eq:Simon-Sum-Central} & \\ \hline
     5 & Simon PDF & \eqref{eq:Simon_PDF_noncent} &  \\  \hline
     6 & Simon CDF & \eqref{eq:Simon_CDF_noncent} &  \\  \hline
     7 & Joarder \& Omar PDF & \eqref{eq:Joarder-PDF}   &  $\checkmark$   \\ \hline
     8 & Joarder \& Omar CDF & \eqref{eq:JoarderOmarCDF} & \\ \hline     
     9  & Cui PDF & \eqref{eq:Cui}    &  $\checkmark$ \\  \hline
     10 & Power SE PDF & \eqref{eq:Power Series_PDF}  &$\checkmark$   \\ \hline
     11 & Power SE CDF & \eqref{eq:Power Series_CDF} & $\checkmark$   \\ \hline
     12 & Laguerre SE PDF & \eqref{eq:Laguerre Series_PDF}  &  $\checkmark$ \\ \hline
     13 & Laguerre SE CDF & \eqref{eq:Laguerre Series_CDF} &  $\checkmark$  \\ \hline
     14 & Chi-square SE PDF & \eqref{eqn:Chi-square Expansion_pdf}  & $\checkmark$  \\ \hline
     15 & Chi-square SE CDF & \eqref{eqn:Chi-square Expansion_cdf} &  $\checkmark$  \\ \hline
     16 & Raphaeli PDF & \eqref{eq:Raphaelli_pdf}  & $\checkmark$  \\ \hline
     17 & Raphaeli CDF & \eqref{eq:Raphaeli_CCDF}  & $\checkmark$  \\ \hline
     18 & Provost Rudiuk PDF1 & \eqref{eq:ProvostRudiuk1_pdf}  & \\ \hline
     19 & Provost Rudiuk PDF2 & \eqref{eq:ProvostRudiuk2_pdf}  & \\ \hline
     20 & Imhof CDF & \eqref{eq:Imhof_CDF} & $\checkmark$ \\ \hline
     21 & Imhof PDF & \eqref{eq:Imhof_PDF} &  $\checkmark$ \\ \hline
     22 & Davies CDF & \eqref{eq:Fc2_davies} & $\checkmark$\\ \hline
     23 & Gardini PDF & \eqref{eq:Gardini-PDF} &  $\checkmark$\\ \hline
     24 & Gardini CDF & \eqref{eq:Gardini-CDF} & $\checkmark$ \\ \hline 
     25 & MM (SW,HBE,WF,LPB4) &  & $\checkmark$ \\  \hline
     26 &  Liu              &    &  $\checkmark$\\ \hline
     27 & Ramirez-Esponoza PD PDF  &  \eqref{eq:Ramirez-PD}  &  $\checkmark$ \\ \hline
     28 & Ha \& Provost     &     &    \\ \hline
     29 & SPA PDF           & \eqref{eq:QF_SPA_PDF} &  $\checkmark$ \\  \hline
     30 & SPA CDF           & \eqref{eq:QF_SPA_CDF_LR} & $\checkmark$  \\  \hline
     31 & Bao and Kan QFRatio Moments           & \eqref{eq:QF_SPA_CDF_LR} & $\checkmark$  \\  \hline
     32 & Nadarajah-Kwong PDF          & \eqref{eq:Nadarjah-Kwong_PDF} & $\checkmark$   \\  \hline
     33 & Yan PDF          & \eqref{eq:Yan_PDF} &   \\  \hline
     34 & Simon \& Alouini CDF          & \eqref{eq:SimonAlouini_CDF} &   \\  \hline
     35 & Mathai Ratio          & \eqref{eq:MathaiRatios_CDF} &   \\  \hline   
     36 & Nadimi PDF          &  & $\checkmark$   \\  \hline
     37 & Gu PDF          &  & $\checkmark$   \\  \hline
     38 & Mao \& Todd  PDF          &  & $\checkmark$   \\  \hline
     
\end{longtable}

\chapter{Classification and Literature Overview: Multi-Forms}
    Recall the definition of multiforms presented earlier and restated here for convenience. The quadratic multiform in real Gaussian random variables of dimension $N$, $\rv{x}\sim \mathcal{N}_N(\dv{\mu},\dmt{\Sigma})  $ is represented mathematically as a vector, $\rv{Q} = [\rs{Q}_1, \rs{Q}_2, \ldots \rs{Q}_M]^T$ 
        \begin{equation}
            \rs{Q}_m = \rv{x}^T\dmt{A}_m\rv{x}+\dv{b}_m^T\rv{x}+c_m,\qquad m=1,2,\ldots,M \tag{MultiQF} 
        \end{equation}
    Similarly, for the quadratic multiforms in complex Gaussian random vectors, $ \rv{x}\sim \mathcal{CN}_N(\dv{\mu},\dmt{\Sigma})$, 
        \begin{equation*}
            \rs{Q}_m = \rv{x}^H\dmt{A}_m\rv{x}+\Re(\dv{b}_m^H\rv{x})+c_m, \qquad m=1,2,\ldots,M \tag{MultiCQF}   
        \end{equation*}

    As before, we are interested in evaluating different quantities of interest that include  MGF, CF, moments, PDF, CDF, etc., of the multiforms. Deriving these quantities is even harder for multiforms than simple forms-aside from the integral transforms. They can be tackled in very limited cases. 
    
      Quadratic multiforms appear in applications in which their dependence cannot be neglected (Refer to Ch2). Independence is equivalent to rather restrictive conditions. They are also very common in performance analysis of communication systems.
      
     This section is organized as follows Subsection 4.1 gives the MGF of (MultiQF) \& (MultiCQF) discusses the evaluation of moments. Subsection 4.2 provides necessary and sufficient conditions for the independence of quadratic forms, illustrating the inaccuracy of an independent approximation. Subsection 4.3 compares the effects of centrality, definiteness and even multiplicities with the case of QFGRVs, presenting the notion of simultaneous diagonalization. Subsection 4.4 presents two methods to deal with simultaneously diagonalizable forms. Subsection 4.5 illustrates, using biforms, some simplifications that arise due to certain conditions on the input. Subsection 4.6 defines multi-chi-squares, which are the bulk of the literature on multiforms. Subsection 4.7 summarizes the recurring proof approaches in the relevant literature. Subsection 4.8 presents a number of special cases of this distribution. Subsection 4.9 gives general methods to evaluate the distribution of MV $\chi^2$. Subsection 4.10 discusses the convolution of MV $\chi^2$. Subsection 4.11 gives techniques for general multiform. Finally, Subsection 4.12 summarizes the pertinent literature using some visual aids.

    \section{MGF and Moments}
        The MGF  of a random vector, $\rv{y}= [\rs{y}_1, \rs{y}_2, \ldots, \rs{y}_M]^T$ is 
        \begin{equation}
            M_{\rv{y}}(\rv{s}) = \mathbb{E}\left[e^{\dv{s}^T\rv{y}}\right] = \mathbb{E}\left[\exp\left(\sum_{m=1}^M s_m\rs{y}_m\right)\right],
        \end{equation}
        where $\dv{s} = [s_1,s_2, \ldots, s_M] \in \mathbb{R}^M$. 
        We can reuse the MGF expression of single quadratic forms to calculate the MGF of quadratic multiforms. For a single form, $\rs{Q} = \rv{x}^T\dmt{A}\rv{x}+\dv{b}^T\rv{x}+c, \quad \rv{x}\sim \mathcal{N}_N(\dv{\mu},\dmt{\Sigma})$, the MGF is given by
        $$
        M_{\rs{Q}}(s) = \mathbb{E}[e^{s\rs{Q}}] = \mathbb{E}\left[e^{\rv{x}^T(s\dmt{A})\rv{x}+(s\dv{b})^T\rv{x}+sc}\right] = h \left(s\dmt{A}, s\dv{b},sc\right),
        $$
        for any deterministic $s\dmt{A} \in \mathbb{R}^{N \times N}$, $s\dv{b}\in\mathbb{R}^N$ and $sc \in \mathbb{R}$, then, we can write the MGF of a multiform,  $\rv{Q} = [\rs{Q}_1, \rs{Q}_2, \ldots \rs{Q}_M]^T$,  $\rs{Q}_m = \rv{x}^T\dmt{A}_m\rv{x}+\dv{b}_m^T\rv{x}+c_m, \quad \rv{x}\sim \mathcal{N}_N(\dv{\mu},\dmt{\Sigma})$, as follows:
        \begin{equation}
        \begin{aligned}
            M_{\rv{Q}}(\dv{s}) &= \mathbb{E}[e^{s\rv{Q}}] = \mathbb{E}\left[e^{\rv{x}^T(\sum_{m=1}^Ms_m\dmt{A}_m)\rv{x}+(\sum_{m=1}^Ms_m\dv{b}_m)^T\rv{x}+\sum_{m=1}^Ms_mc_m}\right]\\ &= h \left(\sum_{m=1}^Ms_m\dmt{A}_m, \sum_{m=1}^Ms_m\dv{b}_m,\sum_{m=1}^Ms_mc_m\right),
        \end{aligned}
        \end{equation}
        but we have seen for single quadratic forms,
        \begin{equation*}
            \begin{aligned}
                h \left(s\dmt{A}, s\dv{b},sc\right) &=  \left|\left(\dmt{I}_N - 2s\dmt{A}\dmt{\Sigma} \right) \right|^{-\frac{1}{2}}\exp \left\{ -\frac{1}{2}(\dv{\mu}^T\dmt{\Sigma}^{-1}\dv{\mu} - 2sc)\right.\\ &\left. + \frac{1}{2} \left[ \dv{\mu}+\dmt{\Sigma}(s\dv{b}) \right]^T \left( \dmt{I}_N - 2s\dmt{A}\dmt{\Sigma} \right)^{-1} \dmt{\Sigma}^{-1} \left[ \dv{\mu}+\dmt{\Sigma}(s\dv{b}) \right] \right\}
            \end{aligned}
        \end{equation*}
        Hence, it is straightforward to extend the result to quadratic multiforms, 
             
        \begin{equation}
            \begin{aligned}
                &M_{\rv{Q}}(\dv{s}) = \left|\dmt{I}_N - 2\sum_{m=1}^Ms_m\dmt{A}_m\dmt{\Sigma} \right|^{-\frac{1}{2}}\exp \Bigg\{ -\frac{1}{2}\left(\dv{\mu}^T\dmt{\Sigma}^{-1}\dv{\mu} - 2\sum_{m}s_mc_m \right) \\
                &+ \frac{1}{2} \left[ \dv{\mu}+\dmt{\Sigma}\left(\sum_{m}s_m\dv{b}_m \right) \right]^T \left[ \dmt{I}_N - 2\left(\sum_{m}s_m\dmt{A}_m\right)\dmt{\Sigma} \right]^{-1}\\
                &\times\dmt{\Sigma}^{-1} \left[ \dv{\mu}+\dmt{\Sigma}\left(\sum_{m}s_m\dv{b}_m \right) \right] \Bigg\}
            \end{aligned} \label{eq:MGF-MultiQF}
        \end{equation}
        Similarly, for (MultiCQF), 
        \begin{equation}
        \resizebox{\textwidth}{!}{$
            \begin{aligned}
            M_{\rv{Q}}(\dv{s}) &= \left|\dmt{I}_N - \left(\sum_{m}s_m\dmt{A}_m\right)\dmt{\Sigma} \right|^{-1}\exp \Bigg\{ -\left(\dv{\mu}^H\dmt{\Sigma}^{-1}\dv{\mu} - \sum_{m}s_mc_m \right) \\
            & +  \left[ \dv{\mu}+\frac{1}{2}\dmt{\Sigma}\left(\sum_{m}s_m\dv{b}_m \right) \right]^H \left[ \dmt{I}_N - \left(\sum_{m}s_m\dmt{A}_m\right)\dmt{\Sigma} \right]^{-1}\\
            &\times \dmt{\Sigma}^{-1} \left[ \dv{\mu}+\frac{1}{2}\dmt{\Sigma}\left(\sum_{m=1}^Ms_m\dv{b}_m \right) \right] \Bigg\}
            \end{aligned}
            $}
        \end{equation}
        Evaluation of the $k$th moment of the $m$th component, i.e., $\mathbb{E}[\rs{Q}_m^k]$, follows the same procedure presented earlier, see Chapter 3, closed-form expression. However, the evaluation of product moments is trickier. In general, 
        \begin{equation}
            \mathbb{E}\left[\prod_{m=1}^M \rs{Q}_m^{k_m}\right] = \dfrac{\partial^{\sum_{m=1}^Mk_m}}{\partial s_1^{k_1} \partial s_2^{k_2} \ldots \partial s_m^{k_m}} M_{\rv{Q}}(\dv{s})\Bigg|_{\dv{s}=\boldsymbol{0}}
        \end{equation}
        In particular, for two complete real forms, $\rs{Q}_1 = \rv{x}^T\dmt{A}_1\rv{x}$, $\rs{Q}_2 = \rv{x}^T\dmt{A}_2\rv{x}$, $\rv{x} \sim \mathcal{N}_N(\dv{\mu}, \dmt{\Sigma})$, the product moments $\mathbb{E}[\rs{Q}_1\rs{Q}_2]$ is        
        \begin{equation}
        \begin{aligned}
             \mathbb{E}[\rs{Q}_1\rs{Q}_2] &= 2 \text{tr} (\dmt{A}_1\dmt{\Sigma}\dmt{A}_2\dmt{\Sigma}) + 4 \dv{\mu}^T \dmt{A}_1 \dmt{\Sigma} \dmt{A}_2 \dv{\mu} \\
             &+ \left[ \dv{\mu}^T \dmt{A}_1\dv{\mu}+\text{tr} \left(\dmt{A}_1\dmt{\Sigma} \right) \right] \left[ \dv{\mu}^T \dmt{A}_2\dv{\mu}+\text{tr} \left(\dmt{A}_2\dmt{\Sigma}\right)\right],
        \end{aligned}
        \end{equation}
        and 
        \begin{equation}
            \text{Cov}(\rs{Q}_1, \rs{Q}_2) = 2 \text{tr}(\dmt{A}_1 \dmt{\Sigma} \dmt{A}_2\dmt{\Sigma})  + 4 \dv{\mu}^T \dmt{A}_1 \dmt{\Sigma} \dmt{A}_2 \dv{\mu}.
        \end{equation}
        For two complete complex forms, 
        \begin{equation}
        \begin{aligned}
            \mathbb{E}[\rs{Q}_1\rs{Q}_2] &=  \text{tr} (\dmt{A}_1\dmt{\Sigma}\dmt{A}_2\dmt{\Sigma}) + \dv{\mu}^H \dmt{A}_1 \dmt{\Sigma} \dmt{A}_2 \dv{\mu} +\dv{\mu}^H \dmt{A}_2 \dmt{\Sigma} \dmt{A}_1 \dv{\mu}\\ 
            & + \left[ \dv{\mu}^H \dmt{A}_1\dv{\mu}+\text{tr} \left(\dmt{A}_1\dmt{\Sigma} \right) \right] \left[ \dv{\mu}^H \dmt{A}_2\dv{\mu}+\text{tr} \left(\dmt{A}_2\dmt{\Sigma}\right)\right], 
        \end{aligned}           
        \end{equation}
        and 
        \begin{equation}
            \text{Cov}(\rs{Q}_1, \rs{Q}_2) =\text{tr} (\dmt{A}_1\dmt{\Sigma}\dmt{A}_2\dmt{\Sigma}) + \dv{\mu}^H \dmt{A}_1 \dmt{\Sigma} \dmt{A}_2 \dv{\mu} +\dv{\mu}^H \dmt{A}_2 \dmt{\Sigma} \dmt{A}_1 \dv{\mu}.
        \end{equation}        
        
        Ghazal \cite{1996-Ghazal-RecurrenceFormulaForExpectationsOfProductsOfQuadratic} derives a recurrence formula that evaluates the expectation of products of arbitrary number of quadratic forms in the form $\rv{x}^T \dmt{A}_i \rv{x}, i=1,2,\ldots M $, where $\rv{x} \sim \mathcal{N}(\dv{0}, \dmt{\Sigma}) $.

    
    \section{Independence}
         An obvious approach toward simplifying multiforms is the independence assumption. Independence reduces product moments, cdf and pdf of multiforms to a set of single-form problems, respectively, as follows
        \begin{align*}
            F_{\rv{Q}}(\dv{q}) &=  \prod_{m=1}^M F_{\rs{Q}_m} (q_m)   \quad \text{and} \quad
            f_{\rv{Q}}(\dv{q}) =  \prod_{m=1}^M f_{\rs{Q}_m} (q_m)\\
            \mathbb{E}\left[\prod_{m=1}^M \rs{Q}_m^{k_m}\right] & = \prod_{m=1}^M \mathbb{E} \left[\rs{Q}_m^{k_m}\right] 
        \end{align*}
        Consider the quadratic forms \(\rv{x}^T\dmt{A}\rv{x}\) and \(\rv{x}^T\dmt{B}\rv{x}\) with \(\rv{x}\sim \mathcal{N}(\dv{\mu},\dmt{\Sigma})\), \(\dmt{\Sigma}>0\), and $\dmt{A}$ and \(\dmt{B}\) are symmetric. Craig \cite{mathaiQuadraticFormsRandom1992} proposes that the forms are independent if and only if $\dmt{A}\dmt{\Sigma}\dmt{B}= \dmt{o}$ by using a theorem of Cochran \cite{cochranDistributionQuadraticForms1934}. He provides an incorrect proof \cite[Sec 5.2]{mathaiQuadraticFormsRandom1992}. His theorem was later proved by a number of authors, see for example \cite{provostOnCraig1996}.

        An incomplete quadratic form $\rs{Q}=\rv{x}^T\dmt{A}\rv{x}+\dv{b}^T\rv{x}+c$ is said to be \emph{degenerate} if $\dmt{A}$ is zero, i.e., if it is algebraically a linear form. Otherwise, it said to be \emph{non-degenerate}. Consider now two non-degenerate incomplete forms $\rs{Q}_m=\rv{x}^T\dmt{A}_m\rv{x}+\dv{b}_m^T\rv{x}+c_m$, $m=1,2$, with $\rv{x}\sim \mathcal{N}(\dv{\mu},\dmt{\Sigma})$, ($\dmt{A}_m\neq\dmt{0}$), and $\dmt{\Sigma} > 0$. The necessary and sufficient conditions for the independence of $\rs{Q}_1$ and  $\rs{Q}_2$  are as follows:
        \begin{enumerate}
            \item $\dmt{\Sigma} \dmt{A}_1\dmt{\Sigma} \dmt{A}_2 = \dmt{0}$
            \item $\dmt{A}_2 \dmt{\Sigma} \dv{b}_1 = \dmt{A}_1 \dmt{\Sigma} \dv{b}_2 = \dv{0}$
            \item $(\dv{b}_1 + 2\dmt{A}_1\dv{\mu})^T \dmt{\Sigma}(\dv{b}_2 + 2\dmt{A}_2\dv{\mu}) = 0$.
        \end{enumerate}
        
        
        For complete Hermitian forms $\rs{Q}_m=\rv{x}^H\dmt{A}_m\rv{x}$, $m=1,2$, and $\rv{x}\sim\mathcal{CN}(\dv{\mu},\dmt{\Sigma})$ with $\dmt{\Sigma}>0$, $\rs{Q}_1$ and $\rs{Q}_2$ are independent if and only if $\dmt{A}_1\dmt{\Sigma}\dmt{A}_2=\dmt{0}$. Indeed, the forms can be reexpressed as $\rs{Q}_m=\rv{y}^T\tilde{\dmt{A}}_m\rv{y}$, where \[\tilde{\dmt{A}}_m=\begin{bmatrix}
            \Re(\dmt{A}_m) & -\Im(\dmt{A}_m)\\
            \Im(\dmt{A}_m) &  \Re(\dmt{A}_m)
        \end{bmatrix}\] and, without loss of generality, $\rv{y}\sim\mathcal{N}(\tilde{\dv{\mu}},\dmt{I}_M)$, with $\tilde{\dv{\mu}}^T=[\Re(\dv{\mu})^T,\Im(\dv{\mu})^T]$. With simple algebraic manipulations, we can see that $\tilde{\dmt{A}_1}\tilde{\dmt{A}_2}=\dmt{0}$ if and only if $\dmt{A}_1\dmt{A}_2=\dmt{0}$.

        \subsection{Maximum Number of Independent Forms}

        How many quadratic forms in $N$ Gaussian random variables can be independent? To answer this question, we need the notion of simultaneous diagonalizability, which will also then find other uses in the coming sections. The $n\times n$ matrices $\{\dmt{A}_\ell\}_{\ell=1}^L$ are said to be \emph{{}simultaneously diagonalizable} if there exists an invertible matrix $\dmt{P}$ such that the matrices $\dmt{P}^{-1}\dmt{A}_\ell\dmt{P}$ are all diagonal, i.e., $\dmt{P}$ simultaneously diagonalize $\{\dmt{A}_\ell\}_{\ell=1}^L$. If the matrices are real symmetric, simultaneous diagonalizability is equivalent to simultaneous orthogonal diagonalizability, i.e., $\dmt{P}$ can be chosen to be orthogonal: $\dmt{P}^T\dmt{P}=\dmt{I}_N$. For Hermitian matrices, we can choose $\dmt{P}$ to be unitary: $\dmt{P}^H\dmt{P}=\dmt{I}_N$. Orthogonality and unitarity are needed to preserve the distribution of the standard normal and complex normal distributions, respectively.

        Consider now a complete quadratic multiform \(\rv{x}^T\dmt{A}_i\rv{x},i=1,\ldots,M\), such that the covariance matrix of \(\rv{x}\) is non-singular. Then we can suppose, without loss of generality, that \(\operatorname{cov}(\rv{x})=\dmt{I}_N\). Applying Craig's theorem on every pair of forms, we know that these components are mutually independent if and only if \(\dmt{A}_i\dmt{A}_j=\dmt{0}\) for all $i\neq j$. Hence the matrices \(\{\dmt{A}_i\}_{i=1}^M\) commute. Referring to theorem 2.5.15 of \cite{hornMatrixAnalysis2012} and noting that the eigenvalues of symmetric matrices are real, hence forbidding similarity with matrices of unreal eigenvalues, we deduce that the matrices \(\dmt{A}_i\) are simultaneously diagonalizable by the same real orthogonal matrix, say \(\dmt{P}\).

        Using this theorem, we can prove that the maximum number of independent quadratic forms \(\rv{x}^T\dmt{A}_i\rv{x}\) is equal to the dimension $N$. Indeed, suppose that we have M quadratic forms \(\rv{x}^T\dmt{A}_i\rv{x}\), where \(\dmt{A}_i\) are considered to be (real) symmetric. Let $\dmt{P}^T\dmt{A}_i\dmt{P}=\dmt{D}_i$ be the diagonal matrices. Hence \(\dmt{D}_i\dmt{D}_j=\dmt{0}\), so the matrices \(\dmt{D}_i\) can only have non-zero eigenvalues in different positions. Therefore, the maximum possible number of independent quadratic forms is $N$. In this case, all the matrices are of rank $1$; hence, the forms are independent scaled chi-square variables, each with one degree of freedom. Similarly, we can verify that at most $N$ incomplete forms can be independent.

        For the Hermitian case, we can use Theorem 2.5.5 of \cite{hornMatrixAnalysis2012} to prove that at most $N$ complete complex forms can be independent.

    \section{Centrality, Definiteness, and Even Degrees}
        In Section 3, we have seen that centrality, definiteness and even degrees can simplify the inversion of the MGF of single forms. Can we generalize to multiforms? In general, we are far less capable of making use of these structures. For simplicity, we consider here complete forms with a non-singular covariance matrix $\dmt{\Sigma}$. Referring to the MGF of central complete multiforms, we have
        $$
             M_{\rv{Q}}(\dv{s}) = \left|\dmt{I}_N - 2\left(\sum_{m=1}^Ms_m\dmt{A}_m\right)\dmt{\Sigma} \right|^{-\frac{1}{2}},
        $$
        and for the central Hermitian forms, 
        $$
            M_{\rv{Q}}(\dv{s}) = \left|\dmt{I}_N - \left(\sum_{m=1}^Ms_m\dmt{A}_m\right)\dmt{\Sigma}\right|^{-1} = \dfrac{1}{\text{polynomial} (s_1, s_2, \ldots, s_M)}.
        $$
        
        One can use some results from multivariate partial fraction decomposition to easily handle the inversion \cite{2022-Heller-multivariateapart}. For a general polynomial, to the best of our knowledge, we do not know if it can be decomposed efficiently into invertible MGFs. One case where we know the path to inversion is clear is \emph{simultaneous diagonalizability}, as defined in the previous section. We will show that by utilizing simultaneous diagonalization one can obtain closed-form expression for the CDF (see: Sec \ref{sec:SimDiag}).
        Simultaneous helps simplify the inversion of the MGF. When the assumption of simultaneous diagonalization is satisfied, the (MultiQF) is written as follows
        \begin{equation}
            \rv{Q} = 
            \begin{bmatrix} 
                \lambda_{11} & \lambda_{12} &  \ldots \lambda_{1N} \\
                \lambda_{21} & \lambda_{22} &  \ldots \lambda_{2N} \\
                \vdots \\
                \lambda_{M1} & \lambda_{M2} &  \ldots \lambda_{MN} \\
            \end{bmatrix} 
            \begin{bmatrix} 
                \rs{z}_1^2  \\
                \rs{z}_2^2\\
                \vdots \\
                \rs{z}_N^2\\
            \end{bmatrix}
            , \label{eq:SimDiagReal}
        \end{equation}
        where $ \rs{z}_m \sim \mathcal{N}(h_m,1)$. For (MultiCQF), simultaneous diagonalizability yields
        \begin{equation}
            \rv{Q} = 
            \begin{bmatrix} 
                \lambda_{11} & \lambda_{12} &  \ldots \lambda_{1N} \\
                \lambda_{21} & \lambda_{22} &  \ldots \lambda_{2N} \\
                \vdots \\
                \lambda_{M1} & \lambda_{M2} &  \ldots \lambda_{MN} \\
            \end{bmatrix} 
            \begin{bmatrix} 
                |\rs{z}_1|^2  \\
                |\rs{z}_2|^2\\
                \vdots \\
                |\rs{z}_N|^2\\
            \end{bmatrix},
            \label{eq:SimDiagComplex}
        \end{equation}
        where $\rs{z}_m \sim \mathcal{CN}(h_m,1)$. In other words, all forms are linear combinations of the same set of chi-squares.
        Simultaneous diagonalizability simplifies the MGF. Indeed, consider the central quadratic forms $ \rv{Q} = [\rs{Q}_1, \rs{Q}_2, \ldots \rs{Q}_M]^T$ with the assumption of simultaneous diagonalization property  (as in \eqref{eq:SimDiagReal} with $h_m=0$), then the MGF is
        \begin{align*}
            M_{\rv{Q}}(\dv{s}) &= \left|\dmt{I}_N - 2\sum_{m}s_m\dmt{\Sigma}^\frac{1}{2}\dmt{A}_m\dmt{\Sigma}^\frac{1}{2}\right|^{-\frac{1}{2}}= \left|\dmt{I}_N - 2\sum_{m}s_m\dmt{P}\dmt{\Lambda}_m\dmt{P}^T\right|^{-\frac{1}{2}}\\
             & = \left|\dmt{I}_N - 2\sum_{m}s_m\dmt{\Lambda}_m \right|^{-\frac{1}{2}}=  \dfrac{1}{\prod_{n=1}^N \sqrt{1-2\sum_{m}s_m\lambda_{mn}} }
        \end{align*}
        If the forms are non-central, the MGF is multiplies by
        \begin{align*}
            &\exp \Bigg\{ -\frac{1}{2} \tilde{\dv{\mu}}^T \left[ \dmt{I}_N - \left( \dmt{I}_N - 2\sum_{m=1}^Ms_m\dmt{\Lambda}_m\right) \right]^{-1}\tilde{\dv{\mu}}^T  \Bigg\},
        \end{align*} 
        where $\tilde{\dv{\mu}} = \dmt{P}^T \dmt{\Sigma}^{-\frac{1}{2}} \dv{\mu}$.
        For a simultaneously diagonalizable central complex multiform (as in \eqref{eq:SimDiagComplex} with $h_m=0$), the MGF is given by
        \begin{align*}
            M_{\rv{Q}}(\dv{s}) &=\left|\dmt{I}_N - 2\sum_{m=1}^Ms_m\dmt{\Sigma}^\frac{1}{2}\dmt{A}_m\dmt{\Sigma}^\frac{1}{2}\right|^{-1} \\ 
             & = \left|\dmt{I}_N - 2\sum_{m=1}^Ms_m\dmt{\Lambda}_m\right|^{-1} \\
             & =  \dfrac{1}{\prod_{n=1}^N (1-2\sum_{m=1}^Ms_m\lambda_{mn})} 
        \end{align*} 
        The latter MGF is a rational fraction, hence it can be decomposed into a sum of invertible MGFs as we shall see later (\ref{sec:SimDiag}). As a conclusion, to obtain a linear combination of invertible MGFs, we need centrality, evenness, and simultaneous diagonalizability. In later sections, we will see cases where we lose one of these properties: \Cref{par:Laverny} \cite{lavernyEstimationMultivariateGeneralized2021} studies simultaneously diagonalizable, positive definite, central real forms; the resulting formula is an infinite series (\cref{eq:Laverny}), and when simultaneous diagonalizability is dropped, say as in \Cref{par:Morales} \cite{morales-jimenezDiagonalDistributionComplex2011}, the closed-form formula is lost as well \cref{eq:Morales}.

    \section{Simultaneously Diagonal Forms} \label{sec:SimDiag}

        In this section, we present two methods from the literature. The first one is a general framework on the calculation of the PDF/CDF of simultaneously diagonalizable central complex forms. It allows indefiniteness. The second is a series expansion of simultaneously diagonalizable positive definite central real forms.

        Al-Naffouri et al. \cite{al-naffouriDistributionIndefiniteQuadratic2016} present a procedure to calculate the CDF of two simultaneously diagonalizable central complex quadratic forms. Suppose the quadratic forms in complex Gaussian random variables $\rs{Q}_1=\rv{x}^H\dmt{A}\rv{x}$ and $\rs{Q}_2=\rv{x}^H\dmt{B}\rv{x}$, where $\rv{x}\sim \mathcal{CN}_N(\dv{0},\dmt{I}_N)$ are simultaneously diagonalizable, i.e., there exits a unitary matrix $\dmt{P}$ such that $\dmt{P}^H\dmt{A}\dmt{P}=\diag(a_1,\ldots,a_N)$ and $\dmt{P}^H\dmt{A}\dmt{P}=\diag(b_1,\ldots,b_N)$. Then the joint CDF is given by 
        \begin{equation}
            \resizebox{\textwidth}{!}{$ 
            F_{\rs{Q}_1,\rs{Q}_2}(q_1,q_2)=-\dfrac{1}{4\pi^2} \displaystyle\int_{\beta_2-i\infty}^{\beta_2+i\infty}\dfrac{e^{q_2 z_2}}{z_2} \displaystyle\int_{\beta_1-i\infty}^{\beta_1+i\infty}\dfrac{e^{q_1 z_1}}{z_1\prod_{n=1}^N(1+a_n z_1+b_n z_2)}dz_1dz_2
            $}
        \end{equation} 
        where $\beta_1$ and $\beta_2$ are arbitrary positive numbers. The internal integral is calculated after performing a partial fraction decomposition with respect to $z_1$. Then the outer one is calculated after performing a partial fraction decomposition with respect to $z_2$. To illustrate the method, a simple example is given.

        Consider the quadratic forms in complex Gaussian random variables $\rs{Q}_1=\rv{x}^H\dmt{A}\rv{x}$ and $\rs{Q}_2=\rv{x}^H\dmt{B}\rv{x}$, where $\rv{x}\sim \mathcal{CN}(\dv{0},\dmt{I}_2)$, $\dmt{A}=\dmt{I}_2$, and $\dmt{B}=\diag(1,-1)$. The matrices are readily diagonal. Then the CDF can be written as 
        \begin{equation*}
        \resizebox{0.95\textwidth}{!}{$ 
            F_{\rv{Q}}(q_1,q_2)=-\dfrac{1}{4\pi^2} \int_{\beta_2-i\infty}^{\beta_2+i\infty}\dfrac{e^{q_2 z_2}}{z_2} \underbrace{\int_{\beta_1-i\infty}^{\beta_1+i\infty}\dfrac{e^{q_1 z_1}}{z_1(1+z_1+z_2)(1+z_1-z_2)}dz_1}_{I(z_2)}dz_2
            $}
        \end{equation*}
        Decomposing the inner integrand $I(z_2)$ with respect to $z_1$, we write 
        $$
        \resizebox{0.95\textwidth}{!}{$ 
        \begin{aligned}
            &\dfrac{1}{z_1(1+z_1+z_2)(1+z_1-z_2)}\\ &=\dfrac{1}{z_1(1+z_2)(1-z_2)}+\dfrac{1}{2z_2(1+z_2)(1+z_1+z_2)}+\dfrac{1}{2z_2(z_2-1)(1+z_1-z_2)}
        \end{aligned}
        $}
        $$
        So the integral is split into a sum of three integrals, each of which is calculated following this complex integration formula \cite{al-naffouriDistributionIndefiniteQuadratic2016} \begin{equation*}
            \frac{1}{2 \pi} \int_{-\infty}^{\infty} \frac{e^{i x p}}{(a+i x)^\nu} d x= \begin{cases}\frac{p^{v-1}}{\Gamma(\nu)} e^{-a p} u(p) & \text { for } a>0 \\ \frac{(-p)^{\nu-1}}{\Gamma(\nu)} e^{-a p} u(-p) & \text { for } a<0\end{cases}
        \end{equation*} and \begin{equation*}
            \int_{-\infty}^{\infty}\dfrac{e^{ix}}{a+ix}dx=2\pi e^{-ap}u(\Re(ap))
        \end{equation*} In particular, we have $$I(z_2)=\pi i u(q_1)\left[\dfrac{2}{(1-z_2)(1+z_2)}+\dfrac{e^{-q_1}e^{-q_1z_2}}{z_2(1+z_2)}+\dfrac{e^{-q_1}e^{q_1z_2}}{z_2(z_2-1)}\right]$$ Note that $\beta_2$ is chosen to be less than one. Now each one of the three fractions is decomposed with respect to $z_2$, and the final result is given by \begin{equation}
            \begin{aligned}
                F_{\rv{Q}}(q_1,q_2)&=u(q_1)\left\{\left(1-\frac12 e^{-q_2}\right)u(q_2)+\frac12 e^{q_2}u(-q_2)\right.\\ &+\frac12e^{-q_1}q_2\left[u(q_2-q_1)-u(q_1+q_2)\right]\\&-\frac12 e^{-q_1}(1+q_1)\left[u(q_2-q_1)+u(q_1+q_2)\right]\\ & \left.+\frac12e^{-q_2}u(q_2-q_1)-\frac12e^{q_2}u(-q_1-q_2) \right\}
            \end{aligned}
        \end{equation} As $q_2$ goes to infinity, the joint CDF converges to $$F_{\rs{Q}_1}(q_1)=u(q_1)[1-e^{-q_1}(1+q_1)]=F_{\frac12\chi_4^2}(q_1)$$ as expected. Moreover, as $q_1$ goes to infinity, we obtain $$F_{\rs{Q}_2}(q_2)=\begin{cases}
            \dfrac{1}{2}e^{y_2} & \text {if }q_2<0\\
            1-\dfrac{1}{2}e^{-y_2} & \text {if }q_2>0
        \end{cases}$$ which agrees with Simon's \cite{simonProbabilityDistributionsInvolving2006} result for the differences of two chi-squares.

        Note that the number of integrals increases exponentially with the number of forms $M$, hence the method becomes less useful as $M$ increases.
        
        \paragraph{Laverny} \label{par:Laverny}
        Laverny et al. provide in \cite{lavernyEstimationMultivariateGeneralized2021} an expansion of the PDF of a multivariate gamma convolution in terms of Laguerre polynomials.

        The gamma distribution of shape parameter $\alpha$ and scale parameter $s$ will be denoted by $\mathcal{G}_{1,1}(\alpha,s)$. The distribution of the sum of $N$ variables, each following $\mathcal{G}_{1,1}(\alpha_k,s_k)$, $k=1,\ldots,N$, is denoted by $\mathcal{G}_{1,N}(\dv{\alpha},\dmt{s})$, where $\dv{\alpha}=(\alpha_1,\ldots,\alpha_N)$ and $\dmt{s}=(s_1,\ldots,s_N)$.
    
        Let $\rv{x}$ be a $M$-variate random vector with support $\mathbb{R}_+^M$. $\rv{x}$ is said to follow a multivariate gamma distribution with shapes $\dv{\alpha}\in \mathbb{R}_+^N$ and scale matrix $\dmt{S}=(s_{ij})\in \mathcal{M}_{N,M}(\mathbb{R_+})$, denoted by $\mathcal{G}_{M,N}(\dv{\alpha},\dmt{S})$, if and only if there exists a vector $\rv{Z}=[Z_1,\cdots,Z_N]^T$, with $Z_i \sim \mathcal{G}_{1,1}(\alpha_i,1)$ mutually independent, such that \[\rv{x}=\dmt{S}^T\rv{Z}\] Then, the density function is given as a linear combination of Laguerre functions.

        Taking $s_{ij}=2$ and $\alpha_i=1/2$ reproduces Royen's original central case. Moreover, for the general case of $M$ central quadratic forms, if the matrices are positive semi-definite and simultaneously orthogonally diagonalizable, then the form follows $\mathcal{G}_{M,N}(\dv{\alpha},\dmt{S})$ with $s_{ij}=2\lambda_{ji}$, $\alpha_i=1/2$ where $\lambda_{ij}$ is the $i^{\text{th}}$ eigenvalue of the $j^{\text{th}}$ matrix. Then, the probability density function is given by \begin{equation}
            f_{\rv{q}}(q_1,\ldots,q_M)=\sum_{\dv{k}\in \mathbb{N}^M}a_{\dv{k}}(f)\varphi_{\dv{k}}(\dv{q}) \label{eq:Laverny}
        \end{equation}
        where 
            \allowdisplaybreaks
            \begin{align*}
                \varphi_{\dv{k}}(\dv{q})&=\prod_{i=1}^M\varphi_{k_i}(q_i)\\
                \varphi_{k_i}(x)&=\sqrt{2}e^{-x}\sum_{\ell=0}^{k_i} \binom{k_i}{l} \dfrac{(-2x)^\ell}{\ell!}\\
                a_{\dv{k}}(f)&=\sqrt{2}^M\sum_{\dv{\ell}\leq \dv{k}} \binom{\dv{k}}{\dv{\ell}}\dfrac{(-2)^{|\dv{\ell}|}}{\dv{\ell}!}\mu_{\dv{\ell},-1}\\
                \mu_{\dv{k},-1}&=\sum_{\dv{\ell}\leq \dv{p}}\mu_{\dv{\ell},-1}\kappa_{\dv{k}-\dv{\ell},-1}\binom{\dv{p}}{\dv{\ell}}\\
            p=(p_1,\ldots,p_M);p_i&=\begin{cases}
                k_i &\text{ if } i\neq d\\
                k_i-1 &\text{ if } i=d
            \end{cases}\\
            d&=\arg\min_i\{k_i\neq 0\}\\
            \kappa_{\dv{0},-1}&=\sum_{i=1}^N \frac{1}{2}\ln\left(\dfrac{1}{1+2\sum_{j=1}^M \lambda_{ji}}\right)\\
            \mu_{\dv{0},-1}&=\exp(\kappa_{\dv{0},-1})\\
            \kappa_{\dv{k},-1}&=\frac12(|\dv{k}|-1)!\sum_{i=1}^N\tilde{\dv{s}}_i^{\dv{k}}\\
            \tilde{\dv{s}}_i&=\dfrac{\dv{\lambda}_i}{1+|\dv{\lambda}_i|}\\
            \dmt{\Lambda}&=[\dv{\lambda}_1,\ldots,\dv{\lambda}_N]\\
            \binom{\dv{a}}{\dv{b}}&=\prod_i\binom{a_i}{b_i},\;\;|\dv{a}|=\sum_i a_i,\;\;\dv{a}^\dv{b}=\prod_ia_i^{b_i}
            \end{align*}

    \section{Special Biforms}   
        We have seen that a quadratic multiform is defined via a set of parameters, one of which is the number of forms, $M$. As complicated problems are studied under simplification assumptions, an obvious one here is taking $M=2$, i.e., studying \emph{quadratic biforms}. In this case, inverting an integral transform would yield a double integral, compared to an $M$-fold integral for the general case.

        \paragraph{Joint Distribution of Sample Mean and Sample Variance}
        
          A linear form, $\dv{b}^T \rv{x}$, can be viewed as a ``degenerate'' incomplete quadratic form, i.e., $\dmt{A} = \boldsymbol{O}_{N \times N}$. Evidently, it is normally distributed. However, one may be interested in the joint distribution of a complete quadratic and a linear form. A good application for this case appears in studying the joint distribution of the sample mean and sample variance. For some applications, relying on the assumption of independence between the sample mean and the sample variance may be too restrictive. The work of Schone \& Schmidt \cite{schoneJointDistributionQuadratic2000} studies the joint PDF and CDF of quadratic and linear forms. Specifically, \cite{schoneJointDistributionQuadratic2000} considers $\rs{Q}_1 = \rv{x}^T \dmt{A} \rv{x}$ and $ \rs{Q}_2 = \dv{b}^T\rv{x}$  where $ \dmt{A} \ge 0,$ $ \dv{b} \in \mathbb{R}^N$, $ \rv{x}\sim \mathcal{N}_N(\boldsymbol{0}, \dmt{\Sigma}), \dmt{\Sigma} > 0 $. The joint MGF is 
        \begin{equation}
            \label{eqn:MGF_ShoneAndScmid}
            M_{\rv{Q}}(\dv{s}) = \prod_{j=1}^{\rankFinal} (1-2s_1 \lambda_j)^{-\frac{1}{2}} \exp \left[  {\frac{1}{2}} s_2^2 \left( \sum_{j=1}^{\rankFinal} \dfrac{h_j^2}{1-2s_1\lambda_j} + \sum_{j=\rankFinal+1}^N h_j^2  \right) \right], 
        \end{equation}
        where $ \rankFinal = \text{rank} (\dmt{\Sigma}^{\frac{1}{2}} \dmt{A}\dmt{\Sigma}^{\frac{1}{2}})=\text{rank} (\dmt{A})$, $\dv{h} = [h_1, h_2, \ldots h_N]^T = \dmt{P}^T \dmt{\Sigma}^{\frac{1}{2}} \dv{b}$, and $\{\lambda_j\}_{j=1}^{\rankFinal}$ are the nonzero eigenvalues of $\dmt{\Sigma}^{\frac{1}{2}} \dmt{A}\dmt{\Sigma}^{\frac{1}{2}}$ (arranged in non-increasing order). To derive the joint density, the MGF in \eqref{eqn:MGF_ShoneAndScmid} is re-expressed in integral form using variable transformation. The joint density is derived from the new MGF as a series of generalized Laguerre functions with coefficients that can be computed recursively. The joint density is 
        \begin{equation}
            f_{\rv{Q}} (\dv{q}) = \sum_{i=0}^\infty \sum_{j=0}^\infty \frac{f_{i,j}}{\sqrt{2\pi}} g_i^{(\rankFinal)}(q_1)q_2^{2j},
        \end{equation}

        where $g_i^{(j)}(q)$ is a type of generalized Laguerre function weighted by a gamma distribution, and the coefficient $f_{i,j}$ is given by 
            
        $$ 
        f_{i, j}=\sum_{k=0}^i \sum_{l=0}^k a_l c_{k-l} e_{i-k, j}, \quad i, j \ge 0
        $$
        
        where
        
        $$
        \begin{array}{llrl}
        a_0=1, & a_i=\frac{1}{2 i} \sum_{k=0}^{i-1} a_k \sum_{l=1}^\rankFinal (1-\frac{\lambda_l}{\beta})^{i-k}, & i \ge 1, \\
        b_0=\sum_{j=1}^N h_j^{ 2}, & b_i=\sum_{k=1}^\rankFinal h_k^{ 2} (1-\frac{\lambda_k}{\beta})^{i-1}\left((1-\frac{\lambda_k}{\beta})-1\right), & i \ge 1, \\
        c_0=\frac{1}{\sqrt{b_0}}, & c_i=-\frac{1}{2 i b_0} \sum_{k=0}^{i-1}(i+k) c_k b_{i-k}, & i \ge 1, \\
        d_0=-\frac{1}{2 b_0}, & d_i=-\frac{1}{b_0} \sum_{k=0}^{i-1} d_k b_{i-k}, & i \ge 1
        \end{array}
        $$
        
        and
        
        $$
        \begin{aligned}
        & e_{i, 0}=0, \quad i \ge 1, \quad e_{0, j}=\frac{d_0^j}{j!}, \quad j \ge 0 \\
        & e_{i, j}=\frac{1}{i} \sum_{k=0}^{i-1}(i-k) e_{k, j-1} d_{i-k}, \quad i \ge 1, \quad j \ge 1.
        \end{aligned}
        $$
       The joint distribution is obtained from the joint density as follows: 
        \begin{equation}
        \begin{aligned}
            F_{\rs{Q}_1, \rs{|Q_2|}}(q_1,q_2) &= \sum_{i=0}^\infty \sum_{j=0}^\infty \frac{f_{i+1,j}}{\sqrt{2\pi}} \frac{4\beta}{2j+1}g_i^{(\rankFinal+2)}(q_1)q_2^{2j+1}\\ &+ F_{\chi^2_{\rankFinal}}\left(\frac{q_1}{\beta}\right) \left[ \Phi \left( \frac{q_2}{\beta}  \right) - \Phi \left( -\frac{q_2}{\beta}  \right) \right],
        \end{aligned}
        \end{equation}
        where $\beta > \frac{\lambda_1}{2} $ and $ F_{\chi^2_{\rankFinal}}(u)$ is the CDF of a central chi-square random variable with $\rankFinal$ degrees of freedom evaluated at $u$.

        \paragraph{Sums of Squares of Complex Normal Random Variables}
        
        Tavares \& Tavares \cite{tavaresStatisticsSumSquared2007} study the distribution of a sum of squares of modulo of complex normal variables. The sum itself cannot be written in the form (QF) or (CQF); however, the real and imaginary parts of the sum can be written as  (MultiQF). To illustrate, let $\{\rs{x}_j\}_{j=1}^n$ be a set of mutually IID complex Gaussian random variables, where $\rs{x}_j \sim \mathcal{CN}(\mu_j, 2\sigma^2)$. The real and imaginary parts are independent. Denote $\rs{x}_{j}^I := \Re({\rs{x}_j})$, $\rs{x}_{j}^Q := \Im({\rs{x}_j})$, $\mu_{j}^I := \mathbb{E}[\rs{x}_{j}^I] $ and $\mu_{j}^Q := \mathbb{E}[\rs{x}_{j}^Q] $. For each $j \in \mathbb{N}$,  $\rs{x}_j^2 = (\rs{x}_{j}^I)^2 - (\rs{x}_{j}^Q)^2 + i 2\rs{x}_{j}^I \rs{x}_{j}^Q$. We are interested in finding the distribution of $\sum_{j=1}^n \rs{x}_j^2$. Define $\rs{Q}_1 := \sum_{j=1}^n (\rs{x}_{j}^I)^2 - (\rs{x}_{j}^Q)^2$ and $\rs{Q}_2 := \sum_{j=1}^n 2\rs{x}_{j}^I \rs{x}_{j}^Q$. Then, $\rs{Q}_1$ and $\rs{Q}_2$ can be written as
                \begin{align*}
        \rs{Q}_1 &= \rv{x}^T \dmt{A}_1\rv{x}, \quad \dmt{A}_{1}=\begin{bmatrix}
                                                                        \dmt{I}_n & \dmt{0}\\
                                                                        \dmt{0}   & -\dmt{I}_n
                                                                    \end{bmatrix}, \\
        \rs{Q}_2 &= \rv{x}^T \dmt{A}_2 \rv{x}, \quad \dmt{A}_{2}=\begin{bmatrix}
                                                                        \dmt{0}   & \dmt{I}_n\\
                                                                        \dmt{I}_n & \dmt{0}
                                                                    \end{bmatrix}, \\            
        \end{align*}   
         where $\rv{x}=[\rs{x}_{1}^I, \rs{x}_{2}^I, \ldots \rs{x}_{n}^I, \rs{x}_{1}^Q, \rs{x}_{2}^Q, \ldots \rs{x}_{n}^Q ]^T$  $ \sim \mathcal{N}_N([\dv{\mu}_I^T, \dv{\mu}_Q^T]^T, \sigma^2 \dmt{I}_N)$, $N =2n$, $\dv{\mu}_I=[\mu_{1}^I, \mu_{2}^I, \ldots \mu_{n}^I]^T$, and $\dv{\mu}_Q=[\mu_{1}^Q, \mu_{2}^Q, \ldots \mu_{n}^Q ]^T$. The MGF of $\rv{Q} = [\rs{Q}_1, \rs{Q}_2]^T$ 
        \begin{equation}
            \label{eqn:MGF_Tavares}
            \begin{aligned}
                 M_{\rv{Q}}(\dv{s}) &= \dfrac{1}{1-4\sigma^4(s_1^2 + s_2^2)^{N/4}}\\ &\times \exp \left[ {\dfrac{2\sigma^2(s_1^2 + s_2^2) \sum_{j=1}^n|\mu_j|^2 + s_1\sum_{j=1}^n \Re(\mu_j^2) +  s_2\sum_{j=1}^n \Im(\mu_j^2) }{1-4\sigma^4(s_1^2 + s_2^2)}} \right]. 
            \end{aligned}
        \end{equation}
        The joint PDF is obtained from a direct Fourier transform of the characteristic function, which 
        \begin{equation}
            \label{eqn:PDF_Tavares}
            \begin{aligned}
                f_{\rv{Q}} (\dv{q}) &= \frac{1}{8\pi\sigma^4} \int_0^\infty f(u)\\
                &\times J_0 \left( \frac{u}{2\sigma^2} \sqrt{\left( q_1 - \frac{\sum_{j=1}^n \Re{(\mu_j^2)}}{1+u^2}\right)^2 + \left( q_2 - \frac{\sum_{j=1}^n \Im{(\mu_j^2)}}{1+u^2}\right)^2} \right) du, 
            \end{aligned}
        \end{equation}
        where $J_\nu(x)$ is the Bessel function of the first kind and order $\nu$ and 
        $$f(u) \triangleq \dfrac{u}{(1+u^2)^{\frac{n}{2}}} \exp{\left( \frac{-1}{2\sigma^2}  \frac{u^2\sum_{j=1}^n|\mu_j|^2}{1+u^2} \right)}$$. 

        \paragraph{Bivariate Rician Distribution}
        
        The work presented by Beaulieu \& Hemachandra in \cite{beaulieuNovelRepresentationsBivariate2011} studies the joint distribution of the PDF and CDF of two correlated Rician random variables. To demonstrate the proposed method, consider the quadratic biform, $\rv{Q} = [\rs{Q}_1, \rs{Q}_2]^T \triangleq [|\rs{x}_1|^2, |\rs{x}_2|^2]^T$, where  $|\rs{x}_k|$ is Rician distributed which can be generated from a complex random variable, $\rs{x}_k = \Re{(\rs{x}_k)} + i \Im{(\rs{x}_k)}$, with $
        \Re{(\rs{x}_k)} \sim \mathcal{N}(\sigma_k \lambda_k m_1, \frac{\sigma^2_k}{2})$ and $
        \Im{(\rs{x}_k)} \sim \mathcal{N}(\sigma_k \lambda_k m_2,\frac{\sigma^2_k}{2})$; hence, $
        \rs{x}_k \sim \mathcal{CN}(\sigma_k \lambda_k(m_1+im_2), \sigma_k^2)$, where $k \in \{1,2\}$, $\lambda_k \in (-1,1) \backslash \{0\}$, $m_1, m_2 \in \mathbb{R}$. For a given $k$, the real and imaginary parts are independent. The quadratic biform components, $\rs{Q}_1$ and $\rs{Q}_2$, can be represented as follows: 
        \begin{align*}
        \rs{Q}_1 &= \rv{x}^H \dmt{E}_{11} \rv{x}, \quad \dmt{E}_{11}=\begin{bmatrix}
                                                                        1 & 0\\
                                                                        0 & 0
                                                                    \end{bmatrix} \\
        \rs{Q}_2 &= \rv{x}^H \dmt{E}_{22} \rv{x}, \quad \dmt{E}_{22}=\begin{bmatrix}
                                                                        0 & 0\\
                                                                        0 & 1
                                                                    \end{bmatrix}, \\            
        \end{align*}
        where $\rv{x} = [\rs{x}_1,\rs{x}_2]^T  \sim \mathcal{CN}_2(\dv{\mu},\dmt{\Sigma})$, $\dv{\mu} = [\sigma_1\lambda_1(m_1 + im_2), \sigma_2\lambda_2(m_1 + im_2)]^T$, $\dmt{\Sigma}=\begin{bmatrix}
                    \sigma_1^2                          & \sigma_1\sigma_2\lambda_1\lambda_2\\
                     \sigma_1\sigma_2\lambda_1\lambda_2 & \sigma_2^2
                    \end{bmatrix}.$
        The assumed model of $\rs{x}_k$ in \cite{beaulieuNovelRepresentationsBivariate2011} is $\rs{x}_k = \sigma_k \sqrt{1-\lambda_k^2}\rs{v}_k + \sigma_k\lambda_k\rs{v}_0  + i \left( \sigma_k \sqrt{1-\lambda_k^2}\rs{w}_k + \sigma_k\lambda_k\rs{w}_0  \right) $ where $\rs{v}_k$, $\rs{w}_k \sim \mathcal{N}(0, \frac{1}{2})$ are independent; and $\rs{v}_0 \sim \mathcal{N}(m_1, \frac{1}{2})$ and $\rs{w}_0 \sim \mathcal{N}(m_2, \frac{1}{2})$ are also independent.            
        The joint PDF of $|\rs{x}_k|$ is derived from conditioning on $\rs{v}_0$ and $\rs{w}_0$. Conditioning on $\rs{v}_0$ and $\rs{w}_0$, $|\rs{x}_k|$ is still Rician distributed because the real and imaginary parts have equal variance and non-zero means.
        Also, the random variables $\{\rs{x}_k\}_{k=1}^2$ are independent after the conditioning operation. Defining $\rv{y} \triangleq  [|\rs{x}_1|, |\rs{x}_2|]^T $, the joint conditional PDF can be written as follows: 
        \begin{equation}
            f_{\rv{y}|{\rs{v}_0, \rs{w}_0}}(\dv{y}|v_0,w_0) = f_{\rs{y}_1|{\rs{v}_0, \rs{w}_0}}(y_1|v_0,w_0) f_{\rs{y}_2|{\rs{v}_0, \rs{w}_0}}(y_2|v_0,w_0).
        \end{equation}
        The unconditional joint PDF is
        \begin{equation}
            f_{\rv{y}}(\dv{y}) = \int_{w_0}\int_{v_0} f_{\rv{y}|{\rs{v}_0, \rs{w}_0}}(\dv{y}|v_0,w_0) f_{{\rs{v}_0, \rs{w}_0}}(v_0,w_0) dv_0 dw_0.
        \end{equation}
        Because the $\rs{v}_0$ and $\rs{w}_0$ are independent Gaussian random variables, their joint PDF is at hand; hence, one can ditrectly substitute to get the joint unconditional PDF of two correlated Rician random variables. After variable transformation, the joint PDF is expressed as
        \begin{equation}
            \label{eqn:PDF_BiformBeaulieu}
            \begin{aligned}
                f_{\rv{y}}(\dv{y}) &= \int_0^\infty \prod_{k=1}^2 \exp{(-ts)} \exp{[-(m_1^2 + m_2^2)]} I_0\left( 2\sqrt{t} \sqrt{m_1^2 + m_2^2} \right)\\ 
                &\times\frac{y_k}{\Omega_k^2} \exp{\left( -\frac{y_k^2}{2\Omega_k^2}\right)} I_0\left( \frac{y_k\sqrt{\sigma_k^2\lambda_k^2}\sqrt{t}}{2\Omega_k^2}\right) dt
            \end{aligned}
        \end{equation}
        at the value of Laplace argument, 
        \begin{equation}
            s = 1 + \sum_{k=1}^2 \frac{\sigma_k^2\lambda_k^2}{2\Omega_k^2},
        \end{equation}
        where $\Omega_k \triangleq \frac{\sigma_k^2(1 - \lambda_k^2)}{2}$, $I_0(.)$ is the modified Bessel function of the first kind and zero order.
        Note that there is one single integration is needed in Equation \eqref{eqn:PDF_BiformBeaulieu} to calculate the joint PDF. 

        A similar approach is used to obtain the result for the CDF,
        \begin{equation}
            \label{eqn:CDF_BiformBeaulieu}
            \begin{aligned}
                F_{\rv{y}}(\dv{y}) &= \int_0^\infty \exp{-t} \exp{\left(-(m_1^2 + m_2^2) \right)} I_0\left( 2\sqrt{t} \sqrt{m_1^2 + m_2^2} \right)\\
                &\times\prod_{k=1}^2 \left[ 1 - Q\left( \frac{\sqrt{t}\sqrt{\sigma_k^2\lambda_k^2}}{\Omega_k}, \frac{y_k}{\Omega_k}\right)  \right]dt.
            \end{aligned}
        \end{equation} These formulas are the bivariate special case of what we shall see later in \eqref{Beaulieu-Hemachandra:Rician}.

        \paragraph{Others} There are other examples of bivariate forms including: Simon \cite[Ch 3]{simonProbabilityDistributionsInvolving2006}, Hagedorn's bivariate \eqref{eq:Hagedorn-bivariate}.

    \section{Definition of Multi-Chi-Square Distribution}


        Multi-chi-squares\footnote{We are using multi-chi-squares as a shorthand for what is more commonly called multivariate chi squares} are generalizations of the univariate chi-squares such that each component (marginal) of the multivariate vector is a chi-square. There are many ways to define such variables. Kotz in \cite[Ch. 48]{kotzContinuousMultivariateDistributions2000} surveys more than 14 multivariate generalizations of the univariate chi-squares! The definition we will use in this section is perhaps one of the most general. Many of \cite[Ch. 48]{kotzContinuousMultivariateDistributions2000} definitions are special cases of the definition we will use. The earliest such definition is perhaps that of Krishnamoorthy and Parasarathy \cite{krishnamoorthyMultivariateGammaTypeDistribution1951}, under the name of multivariate gamma. Although multi-variate gamma is intimately related to the multi-chi-square they are not identical.

        \subsection{Multi-Chi-Square}

            The random vector can be generated starting from a bunch of Gaussian random variables. This method will be shown in three straightforwardly equivalent ways. We will start with the formulation that appears in Royen's 1991 paper \cite{royenExpansionsMultivariateChisquare1991}. The second one appears in Royen's 1995 paper \cite{royenCENTRALNONCENTRALMULTIVARIATE1995}. The third one is the equivalent formulation in terms of quadratic forms. 
    
            Consider a set of
            independent random vectors $\rv{y}_k \sim \mathcal{N}_M (\dv{\mu}_k, \dmt{R})$, $k=1,2,\ldots, K$. Stack them as columns in a matrix $\rmt{Y} \in \mathbb{R}^{M \times K}$, 
            $$
                \rmt{Y} = [\rv{y}_1, \rv{y}_2, \ldots, \rv{y}_K ]
            $$
            The sought distribution is that of the squared Euclidean norms of the rows of the
            random matrix $\rmt{Y}$, as illustrated below
            \[ 
            \resizebox{\textwidth}{!}{$
            \rmt{Y} = 
            \begin{bmatrix} 
                (\rs{y}_1)_1 & (\rs{y}_2)_1 &  \ldots & (\rs{y}_K)_1 \\
                (\rs{y}_1)_2 & (\rs{y}_2)_2 &  \ldots & (\rs{y}_K)_2 \\
                \vdots \\
                (\rs{y}_1)_M & (\rs{y}_2)_M &  \ldots & (\rs{y}_K)_M \\
            \end{bmatrix}
            \mathbf{\longrightarrow}
            \begin{bmatrix} 
                (\rs{y}_1)^2_1 + (\rs{y}_2)^2_1 +  \ldots + (\rs{y}_K)^2_1 \\
                (\rs{y}_1)^2_2 + (\rs{y}_2)^2_2 +  \ldots + (\rs{y}_K)^2_2 \\
                \vdots \\
                (\rs{y}_1)^2_M + (\rs{y}_2)^2_M +  \ldots + (\rs{y}_K)^2_M \\
            \end{bmatrix}
             \triangleq \begin{bmatrix}
                \rs{Q}_1\\
                \rs{Q}_2\\
                \vdots\\
                \rs{Q}_M
            \end{bmatrix} = \rv{Q}
            $}
             \]
            Clearly, each component of $\rv{Q}\in\mathbb{R}^M$ is a scaled non-central chi-square random variables of $K$ degrees of freedom each. We say $\rv{Q} $ follows a multi-chi-square distribution parameterized by $K$, $\dmt{R}$ and $\dmt{\Delta} = \mathbb{E}[\rmt{Y}]\mathbb{E}[\rmt{Y}]^T = \sum_{k=1}^K \dv{\mu}_k \dv{\mu}_k^T$. The MGF, Equation \eqref{eq:Royen-MGF}, depends only on $M,K,\dmt{R},\text{ and } \dmt{\Delta}$. Hence the distribution can be parametrized by the aforementioned parameters. Symbolically, we write 
            $$
            \rv{Q} \sim \chi^2_M(K,\dmt{R}, \dmt{\Delta})
            $$
            For $M=1$, the (scaled) univariate non-central chi-square distribution is recovered \footnote{Beware of the notational difference: the subscript of $\chi$ does not refer to the degrees of freedom here.}.
    
            The second formulation starts from a non-central Wishart matrix. Let $\dmt{R}$ be a positive definite matrix in $\mathbb{R}^{M\times M}$. Let $\rv{y}_k\sim \mathcal{N}_M(\dv{\mu}_k,\dmt{R})$ for $k=1,2,\ldots,K$, and $\rmt{y}=[\rv{y}_1,\rv{y}_2,\ldots,\rv{y}_K]\in\mathbb{R}^{M\times K}$. Hence, we say that $\rmt{W}=\rv{Y}\rv{Y}^T\in\mathbb{R}^{M\times M}$ follows a non-central Wishart distribution of $K$ degrees of freedom, covariance matrix $\dmt{R}$, and non-centrality matrix $\dmt{\Delta}=\mathbb{E}[\rmt{Y}]\mathbb{E}[\rmt{Y}]^T$. Symbolically, we write $$
                \rmt{W} \sim \mathcal{W}^2_M(\dmt{R}, \dmt{\Delta})
            $$ Its main diagonal then follows the multi-chi-square distribution $$\diag(\dmt{W}) \sim \chi^2_M(K,\dmt{R}, \dmt{\Delta})$$
    
            As this monograph is concerned with quadratic multiforms, we would like to use the latter notion to define the multi-chi-square distribution. Indeed, let $\rv{x}\sim \mathcal{N}_{N}(\dv{\mu},\dmt{\Sigma})$. We assume that the parameters $N$ and $\dmt{\Sigma}$ satisfy the following (1) $N=KM$, and (2) $\dmt{\Sigma}=\dmt{I}_K\otimes\dmt{R}=\diag(\dmt{R},\ldots,\dmt{R})$ for some positive semi-definite matrix $\dmt{R}\in\mathbb{R}^{M\times M}$. The quadratic forms are then defined by \begin{equation}
                \rs{Q}_m=\rv{x}^T(\dmt{I}_K\otimes \dmt{E}_{mm})\rv{x},\quad m=1,\ldots,M \label{eq:MVChi2}
            \end{equation} where $\dmt{E}_{mm}$ denotes the $M\times M$ matrix with zero entries except for the $m^{\text{th}}$ diagonal position which is one, i.e., $\dmt{E}_{mm}=(\delta_{im}\delta_{j m})_{1\leq i,j \leq M}$, where $\delta_{ij}$ is the Kronecker delta. Again, the three formulations are equivalent.
    
            \paragraph{MGF}
            
                As we have expressed the distribution as a quadratic form in Gaussian random variables, we can use the general MGF expression \eqref{eq:MGF-MultiQF} to evaluate its MGF. We start with deriving the MGF of a central multi-chi-square using that of a general quadratic multiform. The MGF can be derived from the general MGF as \begin{align*}
                    M_{\chi^2_M(K,\dmt{R}, \dmt{0})}&=\left|\dmt{I}_{KM}-2\left(\sum_{m=1}^Ms_m \dmt{I}_K\otimes \dmt{E}_{mm}\right)(\dmt{I}_K\otimes \dmt{R})\right|^{-1/2}\\
                    =&\left|\dmt{I}_K\otimes\dmt{I}_M-2\sum_{m=1}^Ms_m \dmt{I}_K\otimes (\dmt{E}_{mm}\dmt{R})\right|^{-1/2}\\
                    =&\left|\dmt{I}_M-2\sum_{m=1}^Ms_m\dmt{E}_{mm}\dmt{R}\right|^{-p}\\
                \end{align*}
                Let $\dmt{S}=\diag(s_1,s_2,\ldots,s_M)$. Hence $\sum_{m=1}^Ms_m\dmt{E}_{mm}=\dmt{S}$. Therefore the central part of the MGF can be written as 
                $$M_{\chi^2_M(K,\dmt{R}, \dmt{0})}=|\dmt{I}_M-2\dmt{SR}|^{-p}$$
                Note that for the general non-central case, we can write the MGF as that of the central multiform, multiplied by the exponential term, as follows: \[M_{\chi^2_M(K,\dmt{R}, \dmt{\Delta})}=M_{\chi^2_M(K,\dmt{R}, \dmt{0})}\times\tilde{M}\] where \begin{equation*}
                    \begin{aligned}
                        \tilde{M}&=\exp \Bigg\{ -\frac{1}{2}\left(\dv{\mu}^T\tilde{\dmt{\Sigma}}^{-1}\dv{\mu} \right)+ \frac{1}{2} \dv{\mu}^T \left[ \dmt{I}_N - 2\left(\sum_{m}s_m\dmt{A}_m\right)\tilde{\dmt{\Sigma}} \right]^{-1}\tilde{\dmt{\Sigma}}^{-1} \dv{\mu} \Bigg\}
                    \end{aligned}
                \end{equation*} with $\dmt{A}_m=\dmt{I}_K\otimes\dmt{E}_{mm}$, $\tilde{\dmt{\Sigma}}=\dmt{I}_K\otimes\dmt{R}$ and $\dv{\mu}=[\dv{\mu}_1^T,\ldots,\dmt{\mu}_K^T]^T$. Recall that the vectors $\dv{\mu}_k$ are the means of the underlying Gaussians. Then simplify the matrix in the quadratic term of the argument of the exponential function \begin{align*}
                    &\left[\dmt{I}_{KM}-\left(\dmt{I}_{KM}-2\left(\sum_{m=1}^Ms_m \dmt{I}_K\otimes \dmt{E}_{mm}\right)(\dmt{I}_K\otimes \dmt{R})\right)^{-1}\right](\dmt{I}_K\otimes \dmt{R})^{-1}\\
                    =&\left\{\dmt{I}_K\otimes\left[\dmt{I}_M-(\dmt{I}_M-2\dmt{SR})^{-1}\right]\right\}\dmt{I}_K\otimes\dmt{R^{-1}}
                \end{align*} Using the fact that $(\dmt{I}_M-2\dmt{SR})^{-1}(\dmt{I}_M-2\dmt{SR})=\dmt{I}_M$, we can write $$\dmt{I}_M-(\dmt{I}_M-2\dmt{SR})^{-1}=-2\dmt{SR}(\dmt{I}_M-2\dmt{SR})^{-1}$$ Hence the inner matrix is equal to \begin{align*}
                    &\dmt{I}_K\otimes \left(-2\dmt{SR}(\dmt{I}_M-2\dmt{SR})^{-1}\dmt{R}^{-1}\right)
                    = \dmt{I}_K\otimes \left(-2\dmt{S}(\dmt{I}_M-2\dmt{RS})^{-1}\right)
                \end{align*} Then the non-central part of the MGF can be written as \begin{align*}
                    \tilde{M}&=\exp\left\{-\dfrac12\dv{\mu}^T\left[\dmt{I}_K\otimes\left(-2\dmt{S}(\dmt{I}_M-2\dmt{RS})^{-1}\right)\right]\dv{\mu}\right\}\\
                    =&\exp\left\{\sum_{k=1}^K\dv{\mu}_k^T\dmt{S}(\dmt{I}_M-2\dmt{RS})^{-1}\dv{\mu}_k\right\}\\
                    =&\exp\left\{\sum_{k=1}^K\tr\left[\dmt{S}(\dmt{I}_M-2\dmt{RS})^{-1}\dv{\mu}_k\dv{\mu}_k^T\right]\right\}\\
                    =&\exp\left\{\tr\left[\dmt{S}(\dmt{I}_M-2\dmt{RS})^{-1}\sum_{k=1}^K\dv{\mu}_k\dv{\mu}_k^T\right]\right\}\\
                    =&\etr\left\{\dmt{S}(\dmt{I}_M-2\dmt{RS})^{-1}\dmt{\Delta}\right\}\\
                \end{align*}
                Therefore the MGF of the multi-chi-square distribution $\rv{Q}\sim\chi_M^2(K,\dmt{R},\dmt{\Delta})$ is given by \begin{equation}
                    M_{\rv{Q}}(\dv{s})=|\dmt{I}_M-2\dmt{SR}|^{-p}\etr\left\{\dmt{S}(\dmt{I}_M-2\dmt{RS})^{-1}\dmt{\Delta}\right\} \label{eq:Royen-MGF}
                \end{equation} with $\dmt{S}=\diag(s_1,s_2,\ldots,s_M)$. As the MGF is defined in a neighborhood of zero, it uniquely determines the distribution. As a result, the multi-chi-square is parameterized by $\dmt{\Delta}$ instead of the set of means $\{\dv{\mu}_k\}_{k=1}^K$. Note that the matrix $\dmt{\Delta}$ is the multivariate analogue of the non-centrality parameter of a univariate non-central chi-square. In particular, if $\{\rs{u}_n\}_{n=1}^N$ are independent standard Gaussians, the random variable $Q=\sum_{n=1}^N(\rs{u}_n+\mu_n)^2$ is a non-central chi-square with a non-centrality parameter $\delta=\sum_{n=1}^N \mu_n^2$.

            \paragraph{Correlation vs Covariance}
    
                Before proceeding further with the distribution, we will assume, without loss of generality, that the covariance matrix $\dmt{R}\in\mathbb{R}^{M\times M}$ \footnote{$\dmt{R}$ is necessarily a correlation matrix, since, otherwise, the components would not be unscaled chi-squares.} is a correlation matrix, i.e., its diagonal entries are units. In this case, the components follow (unscaled) non-central chi-square distribution. In fact, suppose that $\dmt{R}=(r_{ij})_{1\leq i,j \leq M}$ with $$r_{ij}=\begin{cases}
                    \sigma_{ii}^2 &\text{ if } i=j\\
                    \rho_{ij}\sigma_i\sigma_j &\text{ if } i\neq j
                \end{cases}$$ Hence we can write $\rs{Q}_m=\sigma_m^2\tilde{\rv{x}}^T(\dmt{I}_K\otimes\dmt{E}_{mm})\tilde{\rv{x}}$, where $\tilde{\rv{x}}\sim\mathcal{N}(\tilde{\dv{\mu}},\tilde{\dmt{R}})$, with $\tilde{\dmt{R}}=(\tilde{r}_{ij})_{1\leq i,j \leq M}$ a correlation matrix $$\tilde{r}_{ij}=\begin{cases}
                    1 &\text{ if } i=j\\
                    \rho_{ij} &\text{ if } i\neq j
                \end{cases}$$ Note that $\tilde{\dv{\mu}}$ can be obtained as follows $$\tilde{\mu}_{\alpha M+m}=\mu_{\alpha M+m}/\sigma_m \text{ for }\alpha=0,1,\ldots,K-1 \text{ and } m=1,\ldots,M$$ Therefore, we can obtain the general distribution (with an arbitrary underlying covariance matrix) by scaling the components of that with an underlying correlation matrix.  

                Note that it is common in the literature to study the distribution with an underlying correlation matrix instead of a general covariance matrix (see \cite{royenMultivariateGammaDistributions1991,royenExpansionsMultivariateChisquare1991,royenCENTRALNONCENTRALMULTIVARIATE1995}).

            \paragraph{A Counterexample}
    
                Apparently, a necessary condition for a multi-chi-square in this sense is that all the marginals have the same number of degrees of freedom. Is this condition sufficient? We construct a counterexample to prove it isn't. Indeed, consider the seemingly bivariate chi-square given by \(\rv{y}=[\rs{x}_1^2+\rs{x}_3^2,\rs{x}_2^2+\rs{x}_4^2]^T\), where \(\rv{x}=[\rs{x}_1,\rs{x}_2,\rs{x}_3,\rs{x}_4]^T\sim\mathcal{N}(\dv{0},\dmt{\Sigma})\), where \begin{align*}
                    \dmt{\Sigma}=\begin{bmatrix}
                    1 &0 &0 &\alpha\\
                    0 &1 &0 &0\\
                    0 &0 &1 &0\\
                    \alpha &0 &0 &1
                \end{bmatrix}
                \end{align*}
                with $0<\alpha<1$. The moment generating function is then \[M_{\rv{y}}(s_1,s_2)=(1-2s_1)^{-\frac12}(1-2s_2)^{-\frac12}[(1-2s_1)(1-2s_2)-4\alpha^2s_1s_2]^{-\frac12}\] which cannot be put into the form of the Laplace transform of a multi-chi-square. Indeed, a central bivariate chi-square random vector would have an MGF of the form $$ M_{\rv{Q}}(\dv{s})=|\dmt{I}_M-2\dmt{SR}|^{-1}$$
                with correlation matrix $\dmt{R}$ given by $$\dmt{R}=\begin{bmatrix}
                    1 &\rho\\
                    \rho &1
                \end{bmatrix}.$$
                Hence, the MGF can be written as $$M_{\rv{Q}}(s_1,s_2)=[(1-2s_1)(1-2s_2)-4\rho s_1s_2]^{-1}$$
                Searching for some $\rho\in(0,1)$ so that $\rv{y}$ and $\rv{Q}$ are identically distributed gives the equation $$\begin{aligned}
                    \left[(1-2s_1)(1-2s_2)-4\rho s_1s_2\right]^2=(1-2s_1)(1-2s_2)[(1-2s_1)(1-2s_2)-4\alpha^2s_1s_2]
                \end{aligned}$$
                that holds for all $s_1,s_2$ in some neighborhood of zero. It is equivalent to $$(4\rho^4-8\rho^2+16\alpha^2)s_1^2s_2^2+(16\rho^2-8\alpha^2)s_1^2s_2+(16\rho^2-8\alpha^2)s_1s_2^2+(4\alpha^2-8\rho^2)s_1s_2=0$$
                which holds for every $s_1,s_2$ if and only if $\rho=\alpha=0$. Therefore, $\rv{y}$ does not follow a multi-chi-square distribution.

        \subsection{Complex Multi-Chi-Square}
    
            The complex counterpart of multi-chi-square can be similarly defined. Consider the independent vectors $\rv{y}_k \sim \mathcal{CN}_M (\dv{\mu}_k, \dmt{R})$, $k=1,2,\ldots, K$. Stack them as columns in a matrix $\rmt{Y} \in \mathbb{C}^{M \times K}$, 
            $$
                \rmt{Y} = [\rv{y}_1, \rv{y}_2, \ldots, \rv{y}_K ]
            $$
            The sought distribution is that of the squared Euclidean norms of the rows of $\rmt{Y}$, as illustrated below
            $$
            \resizebox{\textwidth}{!}{$
            \rmt{Y} = 
            \begin{bmatrix} 
                (\rs{y}_1)_1 & (\rs{y}_2)_1 &  \ldots & (\rs{y}_K)_1 \\
                (\rs{y}_1)_2 & (\rs{y}_2)_2 &  \ldots & (\rs{y}_K)_2 \\
                \vdots \\
                (\rs{y}_1)_M & (\rs{y}_2)_M &  \ldots & (\rs{y}_K)_M \\
            \end{bmatrix}
            \rightarrow
            \begin{bmatrix} 
                |\rs{y}_1|^2_1 + |\rs{y}_2|^2_1 +  \ldots + |\rs{y}_K|^2_1 \\
                |\rs{y}_1|^2_2 + |\rs{y}_2|^2_2 +  \ldots + |\rs{y}_K|^2_2 \\
                \vdots \\
                |\rs{y}_1|^2_M + |\rs{y}_2|^2_M +  \ldots + |\rs{y}_K|^2_M \\
            \end{bmatrix}
            \triangleq \rv{Q}
            $}
            $$
            Clearly, the components of $\rv{Q}\in\mathbb{R}^M$ are scaled non-central chi-square random variables each of $K$ degrees of freedom. We say $\rv{Q} $ follows a complex multi-chi-square distribution parameterized by $K$, $\dmt{R}$ and $\dmt{\Delta}=\mathbb{E}[\rmt{Y}]\mathbb{E}[\rmt{Y}]^H = \sum_{k=1}^K \dv{\mu}_k \dv{\mu}_k^H$.  Symbolically, we write 
            $$
                \rv{Q} \sim \mathcal{C}\chi^2_M(K,\dmt{R}, \dmt{\Delta})
            $$ 
    
            Again, the second formulation starts from a non-central complex Wishart matrix. Let $\dmt{R}$ be a positive definite matrix in $\mathbb{C}^{M\times M}$. Let $\rv{y}_k\sim \mathcal{CN}_M(\dv{\mu}_k,\dmt{R})$ for $k=1,2,\ldots,K$, and $\rmt{Y}=[\rv{y}_1,\rv{y}_2,\ldots,\rv{y}_K]\in\mathbb{C}^{M\times K}$. Hence, we say that $\rmt{W}=\rmt{Y}\rmt{Y}^H\in\mathbb{C}^{M\times M}$ follows a non-central Wishart distribution of $K$ degrees of freedom, covariance matrix $\dmt{R}$, and non-centrality matrix $\dmt{\Delta}=\mathbb{E}[\rmt{Y}]\mathbb{E}[\rmt{Y}]^H$. Symbolically, we write $$
                \rmt{W} \sim \mathcal{CW}^2_M(\dmt{R}, \dmt{\Delta})
            $$ Its main diagonal then follows the complex multi-chi-square distribution $$\diag(\dmt{W}) \sim \mathcal{C}\chi^2_M(K,\dmt{R}, \dmt{\Delta})$$
    
            We would also like to define the multi-chi-square distribution using quadratic forms. Indeed, let $\rv{x}\sim \mathcal{CN}_{N}(\dv{\mu},\dmt{\Sigma})$. We assume that the parameters $N$ and $\dmt{\Sigma}$ satisfy the following (1) $N=KM$, and (2) $\dmt{\Sigma}=\dmt{I}_K\otimes\dmt{R}=\diag(\dmt{R},\ldots,\dmt{R})$ for some positive semi-definite matrix $\dmt{R}\in\mathbb{C}^{M\times M}$. The quadratic forms are then defined by 
            \begin{equation}
                \rs{Q}_m=\rv{x}^H(\dmt{I}_K\otimes \dmt{E}_{mm})\rv{x},\quad m=1,\ldots,M \label{eq:ComplexMVChi2}
            \end{equation}
    
            Similarly, the MGF is given by \begin{equation}
                M_{\rv{Q}}(\dv{s})=|\dmt{I}_M-\dmt{SR}|^{-K}\etr\left\{\dmt{S}(\dmt{I}_M-\dmt{RS})^{-1}\dmt{\Delta}\right\} \label{eq:Royen-Complex-MGF}
            \end{equation} with $\dmt{S}=\diag(s_1,s_2,\ldots,s_M)$. 

            \paragraph{Relation between Real and Complex Multi-Chi-Squares}
    
            The definition of the complex multi-chi-square distribution is quite straightforward. However, the relation with its real counterpart isn't. One could immediately verify that a complex chi-square with a real covariance matrix $\dmt{R}$ can be written as a (real) multi-chi-square. Starting with the definition as in Equation \ref{eq:ComplexMVChi2}, let $\rv{\tilde{x}}\sim\mathcal{N}([\Re(\dv{\mu})^T,\Im(\dv{\mu})^T]^T,\dmt{I}_{2K}\otimes\dmt{R})$. Hence we can rewrite the forms as $$\rv{Q}_m=\rv{\tilde{x}}^T(\dmt{I}_{2K}\otimes\dmt{E}_{mm})\rv{\tilde{x}}$$ The case of a non-real covariance is not as easy. Firstly, we will shift our focus to the central case, studying the MGFs of the real and complex distributions. Recall that, for the real case, the MGF is given by $$M_{\rs{Q}_R}(\dv{s})=|\dmt{I}_M-2\dmt{S}\dmt{R}_R|^{-K_R/2}$$ and for the complex case, it is given by $$M_{\rs{Q}_C}(\dv{s})=|\dmt{I}_M-\dmt{S}\dmt{R}_C|^{-K_C}$$ The questions boil up to the search for $K_R$ and $\dmt{R}_R$ so that $$M_{\rs{Q}_R}(\dv{s})=M_{\rs{Q}_C}(\dv{s})$$ for all $\dv{s}$ in a neighborhood of zero. Since the inverses of the left- and right-hand sides are polynomials, the equality holds for all $\dv{s}\in\mathbb{R}^M$ \cite{rongMultivariatePolynomials}, and it's attained if, and only if, the polynomials are identical: $$|\dmt{I}_M-2\dmt{S}\dmt{R}_R|^{K_R/2}=|\dmt{I}_M-\dmt{S}\dmt{R}_C|^{K_C}$$ Since the degrees of the left- and right-hand sides are $\frac{MK_R}{2}$ and $MK_C$ respectively, we deduce that $K_R=2K_C$. There remains to search for the covariance matrix $\dmt{R}_R$. Write $$\dmt{R}_C=\begin{bmatrix}
                \sigma_1^2 & \rho_{12}^*\sigma_1\sigma_2 &\cdots & \rho_{1M}^*\sigma_1\sigma_M\\
                \rho_{12}\sigma_1\sigma_2 & \sigma_2^2 &\cdots & \rho_{2M}^*\sigma_1\sigma_M\\
                \vdots &\vdots &\ddots &\vdots\\
                \rho_{1M}\sigma_1\sigma_M &\rho_{2M}\sigma_1\sigma_M &\cdots & \sigma_M^2
            \end{bmatrix}$$ Using an approach similar to \cite{StackQuestionMVChi2}, we can prove that the correlations of the real and complex matrices, in the case of equivalence, verify $\rho^R_{ij}=\pm \frac12|\rho^C_{ij}|$. In particular, substituting for $\dmt{S}$ the matrix $\frac{1}{x}\dmt{D}$, where $\dmt{D}$ is any real diagonal matrix, it can be shown that $2\dmt{\Sigma}_R \dmt{D}$ and $\dmt{\Sigma}_C \dmt{D}$ have the same eigenvalues. Substituting particular diagonal matrices can be used to establish the equalities. For $M=2$, we can find an equivalent real distribution. Indeed, take $$\dmt{R}_R=\dfrac{1}{2}\begin{bmatrix}
                \sigma_1^2 & |\rho_{12}|\sigma_1\sigma_2\\
                |\rho_{12}|\sigma_1\sigma_2 &\sigma_2^2\\
            \end{bmatrix}$$ For $M=3$, such a matrix exists if and only if \begin{equation}
                \rho_{12}\rho_{13}^*\rho_{23}\in \mathbb{R}
            \end{equation}
            For $M=4$,  the condition is more complicated. An obvious condition is the equivalence of every combination of $3$ chi-squares. The full set of conditions entails an additional condition: \begin{equation}
            \rho_{12}\rho_{13}^*\rho_{23}\in\mathbb{R}, \; 
            \rho_{12}\rho_{14}^*\rho_{24}\in\mathbb{R}, \;
            \rho_{23}\rho_{24}^*\rho_{34}\in\mathbb{R}, \;
            \rho_{13}\rho_{14}^*\rho_{34}\in\mathbb{R},\label{eq:M4Cond1}
            \end{equation}
            and \begin{align}
            &\rho_{13}^*\rho_{24}^*\rho_{14}\rho_{23}+\rho_{13}^*\rho_{34}^*\rho_{12}\rho_{24}+\rho_{14}^*\rho_{12}\rho_{23}\rho_{34}\in\mathbb{R}.\label{eq:M4Cond5}
            \end{align} Note that conditions \mref{eq:M4Cond1} ensure that each combination of three chi-squares form a complex multi-chi-square random vector that is equivalent to a real multi-chi-square. It is worth noting that in case that no more than one correlation $\rho_{jk}$ vanishes, condition \eqref{eq:M4Cond5} becomes superfluous.

    \section{Overview of Derivation Approaches of the CDF and PDF of Multi-Chi-Squares}

        \paragraph{Phase Marginalization}There are many examples of this approach in the literature, such as \cite{tekinayMomentsQuadrivariateRayleigh2020,hagedornTrivariateChisquaredDistribution2006,chenInfiniteSeriesRepresentations2005,dharmawansaDiagonalDistributionComplex2009,dharmawansaNewSeriesRepresentation2009,peppasTrivariateNakagamimDistribution2009,wiegandSeriesApproximationsRayleigh2019,wiegandSeriesRepresentationMultidimensional2018}. By definition a chi-squared variable is a sum of squared Gaussian magnitudes. The marginalization appoarch starts by:
        \begin{enumerate}
            \item  Replacing the Gaussians or the sum squared of Gaussians by a magnitude and direction (phase) quantity. \cite[Eq. 40]{tekinayMomentsQuadrivariateRayleigh2020}, \cite[Eq. 10]{hagedornTrivariateChisquaredDistribution2006}, \cite[Eq. 1]{chenInfiniteSeriesRepresentations2005}, \cite[Eq. 5]{dharmawansaDiagonalDistributionComplex2009} and \cite[Pg. 2]{wiegandSeriesRepresentationMultidimensional2018}.
            \item The Gaussian density is expressed in terms of these transformed variables and we then attempt to marginalize out the direction/phase quantities.
        \end{enumerate}
        The approach to marginalization differs now depending on the function and assumptions of the problem. For example, for chi-squares from single complex Gaussians, and with the assumption that the covariance matrix is real (no cross-correlation between real and imaginary variables) leads the density function that has exponentials of cosine functions (an even function). In this case an exponential of a cosine can be replaced by a fourier series which contains only cosines: $e^{A \cos \theta} = I_0(A) + \sum_{k=1}^\infty I_k(A) \cos k\theta$. The marginalization over phases ($\theta$) can now be carried out by integrating term by term. 
        
        The result of this process, for example, in the trivariate case is a density $f(q_1,q_2,q_3)$ which is usually an infinite sum of Bessel function products i.e. $I_k(c_1\sqrt{q_1q_2})I_m(c_2\sqrt{q_2q_3})I_n(c_3\sqrt{q_1q_3})$. In the quadrivariate case, there are six bessel products (every pair of $q_iq_j$). To obtain the CDF the appoarch taken in \cite{tekinayMomentsQuadrivariateRayleigh2020,peppasTrivariateNakagamimDistribution2009,chenInfiniteSeriesRepresentations2005} is to expand the bessel functions in power series form \cite[Eq. 8.445]{gradshteinTableIntegralsSeries2007}. We can then integrate with respect to $q_1, q_2, q_3$ to obtain the CDF. Note each bessel function expanded adds one more infinite series. For example the quadrivariate CDF series in \cite[Eq. 5]{tekinayMomentsQuadrivariateRayleigh2020} has nine infinite sums. Discussions of the truncation error can be found in \cite[Sec II.B]{chenInfiniteSeriesRepresentations2005}, \cite[Eq. 12]{peppasTrivariateNakagamimDistribution2009}. Finally discussions for the moments can be found in \cite[Sec II.C]{chenInfiniteSeriesRepresentations2005}\cite{tekinayMomentsQuadrivariateRayleigh2020}.
        
        The integrals and series become more and more difficult as the number of variables increases. This approach has been applied up to four variables ($M=4$) with various conditions on the covariance/inverse covariance matrices to make the process tractable. To the best of the authors knowledge this approach dates back to Miller 1958. In fact, Miller \cite{millerGeneralizedRayleighProcesses1958} proposes a more general procedure where marginalization is carried out over a surface of (all directions) rather than just a phase. This approach is used exactly in \cite{hagedornTrivariateChisquaredDistribution2006,chenInfiniteSeriesRepresentations2005,peppasTrivariateNakagamimDistribution2009}. It is worth mentioning that the result proposed in \cite{wiegandSeriesApproximationsRayleigh2019}, although in principle applicable to $M > 4$, the authors have only implemented it up to $M = 4$. The numbers of coefficients numbers expand combinatorially as the dimensions increase. Finally, we note that all the above methods are for central multi-chi-squares (generalized  multivariate rayleigh). In \cite{dharmawansaDiagonalDistributionComplex2009,dharmawansaNewSeriesRepresentation2009} extensions to the non-central case but have been reported with the restriction of $(\Sigma^{-1})_{13} = 0$
        
        \paragraph{Deconditioning} \label{para:DeconditioningMethod} If the underlying covariance structure admits an $m$-factorial representation ($\Sigma = D + VV^T, D: \text{ diagonal}, V \in \reals{N\times m}$)with $m$ a small number, the multi-chi-squares can be seen to depend on two types of Gaussian variables: $m$ ``shared'' variables associated with the factor $V$, and each chi-square depends separately on a Gaussian associated with the diagnonal $D$. Conditioning on the ``shared'' variables makes each chi-square independent. We then decondition for each of the shared variables; with each deconditioning step costing us one integration. Examples of works using this approach are: \cite{beaulieuNovelRepresentationsBivariate2011,beaulieuNovelSimpleRepresentations2011,hemachandraNovelRepresentationsEquicorrelated2011,bithasNovelResultsMultivariate2019,khammassiNewAnalyticalApproximation2022}.
        
        The main difficulty to applying this method when the structure of $\Sigma$ is not predetermined to be $m$-factorial is finding the smallest $m$ and the factors $D$ and $V$. A pragmatic approach is taken in \cite{khammassiNewAnalyticalApproximation2022}, where an approximate m-factorial representation is obtained by using the eigendecomposition of the covariance model and then selecting  a few dominant eigenvalues and their corresponding eigenvectors. Thus if there are $m$ eigenvalues above $\epsilon$ we obtain an $m$-factorial representation.
        
        \paragraph{Inversion} The connection between the MGF/CF/Laplace Transform and the PDF/CDF is straightforward, in principle. However other than the simplest cases, the inversion is intractable in direct form. One must identify a suitable class of functions in which we are able to expand the MGF. Preferably, the class includes a function and its powers/derivatives. The most common examples are the Laguerre functions, gamma density and its derivatives, chi-square densities. The MGF is expanded in terms of this function and the coefficients are identified. We then utilize the linearity of the transforms to invert term by term. Note that exchanging the inversion (integral) with an infinite sum requires showing the convergence of the infinite sum.
        This technique has been used in: \cite{krishnamoorthyMultivariateGammaTypeDistribution1951,royenExpansionsMultivariateChisquare1991,royenMultivariateGammaDistributions1991,royenCENTRALNONCENTRALMULTIVARIATE1995,royenIntegralRepresentationsApproximations2007,morales-jimenezDiagonalDistributionComplex2011}. Note this technique is not unique to multivariate chi-squares but rather has been used for single quadratic forms \cite[Sec 4.2abc]{mathaiQuadraticFormsRandom1992} and other multi-quadratic forms \cite{lavernyEstimationMultivariateGeneralized2021,schoneJointDistributionQuadratic2000,tavaresStatisticsSumSquared2007}. 
        
        Unlike the methods based on marginalization, inversion methods have produced analytical formulae for any number of chi-squares, whether they are central or non-central.

    \section{Multi-Chi-Square PDF/CDF: Special Cases}
        As usual, we will handle a few more special cases before approaching Royen's general solution.

            \subsection{Rayleigh and Rician Distributions}

                \paragraph{Bivariate Rayleigh Distribution}

                The modulus of a complex random variable following \(\mathcal{CN}(0,\sigma^2)\) is Rayleigh distributed and has the PDF \[f_R(r)=\frac{2r}{\sigma^2}e^{-\frac{x^2}{\sigma^2}}.\] The square of the modulus evidently follows a scaled chi-square distribution: \(\frac{\sigma^2}{2}\chi_2^2\).
    
                Let $\rs{R}_1=|\rs{x}_1|$ and $\rs{R}_2=|\rs{x}_2|$, where \[\begin{bmatrix}
                    \rs{x}_1\\
                    \rs{x}_2
                \end{bmatrix}\sim\mathcal{CN}\left(\begin{bmatrix}
                    0\\
                    0
                \end{bmatrix},\begin{bmatrix}
                    \sigma_1^2 & \rho \sigma_1 \sigma_2\\
                    \rho \sigma_1 \sigma_2 & \sigma_2^2
                \end{bmatrix}\right),\] the vector of moduli \([\rs{R}_1,\rs{R}_2]^T\) follows a bivariate Rayleigh distribution. The squared moduli \(\rs{R}_1^2,\rs{R}_2^2\) follow then a (scaled) bivariate chi-square distribution. The PDF of the moduli is given by \cite{simonProbabilityDistributionsInvolving2006} is given by \[\begin{aligned}
                f_{\rs{R}_1, \rs{R}_2}\left(r_1, r_2\right) & =\frac{4r_1 r_2}{\left(1-\rho^2\right)\left(\sigma_1 \sigma_2\right)^{2}} \exp \left[-\frac{1}{\left(1-\rho^2\right)}\left(\frac{r_1^2}{\sigma_1^2}+\frac{r_2^2}{\sigma_2^2}\right)\right] \\
                & \times I_{0}\left(\frac{2r_1 r_2|\rho|}{\sigma_1 \sigma_2\left(1-\rho^2\right)}\right), r_1 \geq 0, r_2 \geq 0
                \end{aligned}\]
                Then the PDF of the  squared moduli, which follow a bivariate chi-square distribution \[\begin{bmatrix}
                    \rs{R}_1^2\\
                    \rs{R}_2^2
                \end{bmatrix}\sim\mathcal{C}\chi_1^2\left(2,\begin{bmatrix}
                    \sigma_1^2 & \rho \sigma_1 \sigma_2\\
                    \rho \sigma_1 \sigma_2 & \sigma_2^2
                \end{bmatrix},\begin{bmatrix}
                    0 & 0\\
                    0 & 0
                \end{bmatrix}\right),\] is given by
                \[\begin{aligned}
                f_{\rs{R}_1^2, \rs{R}_2^2}\left(q_1, q_2\right) & =\frac{1}{\left(1-\rho^2\right)\left(\sigma_1 \sigma_2\right)^{2}} \exp \left[-\frac{1}{\left(1-\rho^2\right)}\left(\frac{q_1}{\sigma_1^2}+\frac{q_2}{\sigma_2^2}\right)\right] \\
                & \times I_{0}\left(\frac{2\sqrt{q_1 q_2}|\rho|}{\sigma_1 \sigma_2\left(1-\rho^2\right)}\right), q_1 \geq 0, q_2 \geq 0
                \end{aligned}\]
    

                \paragraph{Bivariate Rician Distribution}
            
                Many forms appear in the literature under the title of Rician distribution. A Rician variable is the square root of a scaled non-central chi-square variable of two degrees of freedom. Indeed, let \(\rs{x}\sim\mathcal{CN}(\mu,\sigma^2)\), where \(\mu\in\mathbb{C}\). The modulus \(\rs{T}=|\rs{x}|\) follows the Rician distribution, and its PDF is given by \[f_{\rs{T}}(t)=\dfrac{2t}{\sigma^2}\exp\left(-\dfrac{t^2+|\mu|^2}{\sigma^2}\right)I_0\left(2\dfrac{|\mu|}{\sigma^2}t\right)\] 
    
                If  \(\begin{bmatrix}
                    \rs{x}_1\\
                    \rs{x}_2
                \end{bmatrix}\sim\mathcal{CN}\left(\begin{bmatrix}
                    \mu_1\\
                    \mu_2
                \end{bmatrix},\left[\begin{smallmatrix}
                    \sigma_1^2 & \rho \sigma_1 \sigma_2\\
                    \rho \sigma_1 \sigma_2 & \sigma_2^2
                \end{smallmatrix}\right]\right)\), the vector of moduli \(\begin{bmatrix}
                    \rs{R}_1\\
                    \rs{R}_2
                \end{bmatrix}\) follows a bivariate Rician distribution. The squared moduli \(\rs{R}_1^2,\rs{R}_2^2\) follow then a bivariate chi-square distribution. The PDF of the moduli, in the case of $|\mu_1|=|\mu_2|=:a$, is given by \cite{millerGeneralizedRayleighProcesses1958,simonProbabilityDistributionsInvolving2006} \[\begin{aligned}
                &f_{R_1, R_2}\left(r_1, r_2\right)  =\frac{4r_1 r_2}{\sigma_1^2 \sigma_2^2\left(1-\rho^2\right)}\\
                &\times\exp \left[-\frac{1}{2\left(1-\rho^2\right)}\left(\frac{2r_1^2}{\sigma_1^2}+\frac{2r_2^2}{\sigma_2^2}+\left(\frac{\sigma_1^2+\sigma_2^2-2 \rho \sigma_1 \sigma_2}{\sigma_1^2 \sigma_2^2}\right) a^2\right)\right] \\
                & \times \sum_{i=0}^{\infty} \varepsilon_i I_i\left(\frac{2r_1 r_2 \rho}{\left(1-\rho^2\right) \sigma_1 \sigma_2}\right) I_i\left(\frac{2a r_1\left(1-\rho \sigma_1 / \sigma_2\right)}{\left(1-\rho^2\right) \sigma_1^2}\right) \\
                & \times I_i\left(\frac{2a r_2\left(1-\rho \sigma_2 / \sigma_1\right)}{\left(1-\rho^2\right) \sigma_2^2}\right),\; r_1 \geq 0,\;r_2 \geq 0
                \end{aligned}\] where $\epsilon_0=1$, and $\epsilon_i=2$ for $i>0$. 
                Hence the PDF of the  squared moduli which follow a bivariate chi-square distribution  \[\begin{bmatrix}
                    \rs{R}_1^2\\
                    \rs{R}_2^2
                \end{bmatrix}\sim\mathcal{C}\chi_1^2\left(2,\begin{bmatrix}
                    \sigma_1^2 & \rho \sigma_1 \sigma_2\\
                    \rho \sigma_1 \sigma_2 & \sigma_2^2
                \end{bmatrix},\begin{bmatrix}
                    a^2 & *\\
                    * & a^2
                \end{bmatrix}\right) \footnote{Note that the off-diagonal terms are irrelevant.},\] is given by
                \[\begin{aligned}
                f_{\rs{Q}_1, \rs{Q}_2}\left(q_1, q_2\right) & =\frac{1}{\sigma_1^2 \sigma_2^2\left(1-\rho^2\right)}\\
                &\times\exp \left[-\frac{1}{2\left(1-\rho^2\right)}\left(\frac{2q_1}{\sigma_1^2}+\frac{2q_2}{\sigma_2^2}+\left(\frac{\sigma_1^2+\sigma_2^2-2 \rho \sigma_1 \sigma_2}{\sigma_1^2 \sigma_2^2}\right) a^2\right)\right] \\
                & \times \sum_{i=0}^{\infty} \varepsilon_i I_i\left(\frac{2\sqrt{q_1 q_2} \rho}{\left(1-\rho^2\right) \sigma_1 \sigma_2}\right) I_i\left(\frac{2a \sqrt{q_1}\left(1-\rho \sigma_1 / \sigma_2\right)}{\left(1-\rho^2\right) \sigma_1^2}\right) \\
                & \times I_i\left(\frac{2a \sqrt{r_2}\left(1-\rho \sigma_2 / \sigma_1\right)}{\left(1-\rho^2\right) \sigma_2^2}\right), q_1 \geq 0, q_2 \geq 0
                \end{aligned}\]

                \paragraph{Quadrivariate Rayleigh Distribution: Real Toeplitz Correlation Matrix}

                Tekinay and Beard (2020) study the quadrivariate Rayleigh distribution \cite{tekinayMomentsQuadrivariateRayleigh2020} with a Toeplitz underlying correlation matrix. They basically marginalize over the phases. They provide formulae for the PDF, CDF, CF, and moments. Let \(\rv{Q}\sim\mathcal{C}\chi_1^2(4,\dmt{\Sigma},\dmt{0})\), and assume, withoput loss of generality, that all diagonal elements of $\dmt{\Sigma}$ are equal: \((\dmt{\Sigma})_{k,k}=2\zeta\) for all $k$. The inverse covariance matrix will have the form: $$\dmt{\Sigma}^{-1}=\dfrac{1}{2\zeta}\begin{bmatrix}
                    \phi_1 &\phi_3 &\phi_4 &\phi_5\\
                    \phi_3 &\phi_2 &\phi_6 &\phi_4\\
                    \phi_4 &\phi_6 &\phi_2 &\phi_3\\
                    \phi_5 &\phi_4 &\phi_3 &\phi_1
                \end{bmatrix} $$
                
                Then, the joint PDF of $\rs{Q}_i$ is given by
                \UglyAlign{
                    &f_{\rs{Q}_1,\rs{Q}_2,\rs{Q}_3,\rs{Q}_4}(q_1,q_2,q_3,q_4)=|\dmt{\Sigma}|^{-1}\exp\left(-\dfrac{\phi_1}{2\zeta}(q_1+q_4)+\dfrac{\phi_2}{2\zeta}(q_2+q_3)\right)\nonumber\\
                    &\times \sum_{\ell=-\infty}^\infty \sum_{j=-\infty}^\infty \sum_{m=-\infty}^\infty I_\ell\left(-\dfrac{\phi_3}{\zeta}\sqrt{q_1q_2}\right)I_j\left(-\dfrac{\phi_4}{\zeta}\sqrt{q_1q_3}\right)I_{\ell+j}\left(-\dfrac{\phi_5}{\zeta}\sqrt{q_1q_4}\right)\nonumber\\
                    &\times I_{\ell+m}\left(-\dfrac{\phi_6}{\zeta}\sqrt{q_2q_3}\right)I_{\ell+j+m}\left(-\dfrac{\phi_3}{\zeta}\sqrt{q_3q_4}\right)I_m\left(-\dfrac{\phi_4}{\zeta}\sqrt{q_2q_4}\right) \label{tekinay}
                }

                and the CDF is given by
                \UglyAlign{
                    & F_{\rv{Q}}\left(q_1, q_2, q_3, q_4\right) \nonumber\\
                    & =\eta^3 \sum_{l, m, j=-\infty}^{\infty} \sum_{b, q, b^{\prime}, h, f, h^{\prime}=0}^{\infty} \frac{\left(-2 \phi_3 \eta\right)^{|l|+|l+j+m|+2 b^{\prime}+2 b}}{2^{\frac{\left(v_1+v_2+v_3+v_4\right)}{2}}}\nonumber\\
                    &\times\frac{\left(-2 \phi_4 \eta\right)^{|j|+|m|+2 h^{\prime}+2 h}\left(-2 \phi_5 \eta\right)^{|l+j|+2 f}}{\left(\phi_1 \eta\right)^{\frac{v_1+v_4+4}{2}}\left(\phi_2 \eta\right)^{\frac{v_2+v_3+4}{2}}} \nonumber\\\
                    & \times \frac{\left(-2 \phi_6 \eta\right)^{|l+m|+2 q}}{b ! f ! h !(b+|l|) !(h+|j|) !} \frac{\gamma\left(\frac{\nu_1+2}{2}, \frac{\phi_1}{2 \zeta} q_1\right)}{(f+|j+l|) ! q !} \frac{\gamma\left(\frac{\nu_2+2}{2}, \frac{\phi_2}{2 \zeta} q_2\right)}{b^{\prime} !(q+|l+m|) !}\nonumber\\
                    &\times\frac{\gamma\left(\frac{\nu_3+2}{2}, \frac{\phi_2}{2 \zeta} q_3\right)}{\left(b^{\prime}+|l+m+j|\right) !} \frac{\gamma\left(\frac{v_4+2}{2}, \frac{\phi_1}{2 \zeta} q_4\right)}{\left(h^{\prime}+|m|\right) ! h^{\prime} !} 
                }
                where
                \UglyAlignS{
                    \eta&=\dfrac{|\dmt{\Sigma}|}{(2\zeta)^4}\\
                    \nu_1&=|l|+|j|+|j+l|+2b+2h+2f\\
                    \nu_2&=|l|+|m|+|l+m|+2b+2q+2h'\\
                    \nu_3&=|j|+|l+m|+|j+l+m|+2h+2q+2h'\\
                    \nu_4&=|m|+|j+l|+|j+l+m|+2f+2b'+2h'
                }
                
                The product moments are given by
                \UglyEquation{
                    \begin{aligned}
                        &\mathbb{E}\left[\rs{q}_1^{\beta_1} \rs{q}_2^{\beta_2} \rs{q}_3^{\beta_3} \rs{q}_4^{\beta_4}\right] 
                        &\\
                        &=\eta^3(\eta \zeta)^{\beta} \sum_{l, m, j=-\infty}^{\infty} \sum_{b, q, b^{\prime}, h, f, h^{\prime}=0}^{\infty} \frac{\left(-2 \phi_5 \eta\right)^{|l+j|+2 f}}{2^{\frac{v_1+v_2+v_3+v_4-\beta}{2}}}\\
                        &\times \frac{\left(-2 \phi_4 \eta\right)^{|j|+|m|+2 h^{\prime}+2 h}\left(-2 \phi_3 \eta\right)^{|l|+|l+j+m|+2 b^{\prime}+2 b}}{\left(\phi_1 \eta\right)^{\frac{v_1+v_4+4+2\beta_1+2\beta_4}{2}}\left(\phi_2 \eta\right)^{\frac{v_2+v_3+4+2\beta_2+2\beta_3}{2}}} \\
                        & \times \frac{\left(-2 \phi_6 \eta\right)^{|l+m|+2 q}}{b ! h ! h^{\prime}(b+|l|) !(h+|j|) ! f !} \frac{\Gamma\left(\frac{v_1+2\beta_1+2}{2}\right)}{q ! b^{\prime} !(f+|j+l|) !} \\
                        &\times \frac{\Gamma\left(\frac{v_2+2\beta_2+2}{2}\right)}{(q+|l+m|) !} \frac{\Gamma\left(\frac{v_3+2\beta_3+2}{2}\right)}{\left(b^{\prime}+|l+m+j|\right) !} \frac{\Gamma\left(\frac{v_4+2\beta_4+2}{2}\right)}{\left(h^{\prime}+|m|\right) !},\\ &\beta_1, \beta_2, \beta_3, \beta_4 \geq-1/2
                    \end{aligned}
                }

                \paragraph{Multivariate Rayleigh Distribution}

                Wiegand and Nadarjah \cite{wiegandSeriesApproximationsRayleigh2019} study the multivariate Rayleigh distribution. Let $\rv{Q}\sim\mathcal{C}\chi_1^2(M,\dmt{\Sigma},\dmt{0})$. The joint PDF is given by the following
                \[f_{\rs{Q}_1,\ldots,\rs{Q}_M}(q_1,\ldots,q_M)=\pi^p\gamma_{\dv{q},\dmt{K}}\sum_{j_1=0}^\infty \sum_{j_2=0}^{j_1} \cdots \sum_{j_p=0}^{j_{p-1}}\prod_{t=1}^p b_{t,j_t^*}\sum_{\rho\in \{-1,1\}^p}\prod_{w=1}^M \delta_{\{\Sigma_{x_W}=0\}}\label{wiegand}\]
                where each $t$ corresponds to a couple $(k,\ell)$ with $\ell=1,\ldots,M$, $k=\ell+1,\ldots,M$ as illustrated in the next example, the number of these couples is $p=M(M-1)/2$, \(\delta_A=1\) if the proposition \(A\) is true, otherwise it is zero, and 
                \UglyAlignS{
                    \gamma_{\dv{q},K}&=\dfrac{1}{2^{2M}\pi^M|\dmt{K}|^{1/2}}\exp\left(-\dfrac{1}{2|\dmt{K}|}\sum_{i=1}^M q_ic_{ii}\right)\\
                    \dmt{K}&=\dfrac{1}{2}\dmt{\Sigma}^{-1}\\
                    \dmt{C}&=|\dmt{\Sigma}|\dmt{K}^T\\
                    j_t^*&=j_{p-t+1}-j_{p-t+2}\\
                    b_{t,j_t}&=\begin{cases}
                        I_0(|a_t|), &\text{ if } j_t=0\\
                        2(-1)^{j_t\mathbb{I}_{(a_t<0)}}I_{j_t}(|a_t|), &\text{ if } j_t>0
                    \end{cases}\\
                    \Sigma_{x_i}&=A_{i,1}+\sum_{\omega=1}^{M-i}A_{w,i}-\sum_{\omega=1}^{i-1}A_{i-w,1+w}\\
                    A_{s,t}&=\alpha_{\sum_{i=0}^{p-1}(M-i+1)+s}\\
                    \alpha_t&=j_{t^*}\rho_t
                } 

                \paragraph{Trivariate Rayleigh Distribution: Series Expression}

                Chen and Tellambura \cite{chenInfiniteSeriesRepresentations2005} study the distribution of three arbitrarily correlated Rayleigh variables. Evidently, the squares of the variable are then scaled chi-squares with two degrees of freedom each. In our terms, these scaled chi-squares correspond to a trivariate quadratic form \(\rv{Q} \sim \mathcal{C}\chi_1(3,\dmt{\Sigma}, \dmt{0})\). It is worth mentioning that \(\dmt{\Sigma}\) is not necessarily real, i.e., there can be cross-correlation between the real part of a component of \(\rv{x}\), and the imaginary part of another component. This is the actual meaning of arbitrary correlation. In order to derive the PDF, Chen and Tellambura start with an expression of Miller \[\begin{aligned}
                f_{\rv{R}}\left(r_1, r_2, r_3\right)= & 8 \operatorname{det}(\dmt{\Phi}) r_1 r_2 r_3 e^{-\left(r_1^2 \phi_{11}+r_2^2 \phi_{22}+r_3^2 \phi_{33}\right)}\\
                &\times\sum_{k=0}^{\infty} \varepsilon_k(-1)^k \cos (k \chi) I_k\left(2 r_1 r_2\left|\phi_{12}\right|\right) \\
                & \times I_k\left(2 r_2 r_3\left|\phi_{23}\right|\right) I_k\left(2 r_3 r_1\left|\phi_{31}\right|\right)
                \end{aligned}\] where \(\rv{R}=[\sqrt{\rs{Q}_1},\sqrt{\rs{Q}_2},\sqrt{\rs{Q}_3}]^T\), \(\dmt{\Phi}=(\phi_{ij})=\dmt{\Sigma}^{-1}\), \(\varepsilon_k\) is the Neumann factor, which is equal to \(1\) if \(k=0\), and equal to \(2\) otherwise, and \(\chi=\chi_{12}+\chi_{23}+\chi_{31}\), with \(\chi_{ij}\) being a phase of \(\phi_{ij}\). After expressing the modified Bessel functions as a power series, they express the PDF as an infinite multiple series given by \[\begin{aligned}
                & f_{\rv{R}}\left(r_1, r_2, r_3\right) =8 \operatorname{det}(\dmt{\Phi}) e^{-\left(r_1^2 \phi_{11}+r_2^2 \phi_{22}+r_3^2 \phi_{33}\right)} \sum_{k=0}^{\infty} \varepsilon_k(-1)^k \cos (k \chi) \\
                & \times \sum_{l, m, n=0}^{\infty} \frac{\left|\phi_{12}\right|^{2 l+k}}{l !(l+k) !} \frac{\left|\phi_{23}\right|^{2 m+k}}{m !(m+k) !} \frac{\left|\phi_{31}\right|^{2 n+k}}{n !(n+k) !} r_1^{2(l+n+k)+1} r_2^{2(l+m+k)+1} r_3^{2(m+n+k)+1}
                \end{aligned}\]
                This makes the variables separable so that they can be integrated to obtain the CDF \[\begin{aligned}
                &F_{\rv{R}}\left(r_1, r_2, r_3\right)=  \frac{\operatorname{det}(\dmt{\Phi})}{\phi_{11} \phi_{22} \phi_{33}} \sum_{k=0}^{\infty} \varepsilon_k(-1)^k \cos (k \chi) \\
                & \times \sum_{l, m, n=0}^{\infty} C \nu_{12}^{l+\frac{k}{2}} \nu_{23}^{m+\frac{k}{2}} \nu_{31}^{n+\frac{k}{2}} \gamma\left(\delta_1, r_1^2 \phi_{11}\right) \gamma\left(\delta_2, r_2^2 \phi_{22}\right) \gamma\left(\delta_3, r_3^2 \phi_{33}\right)
                \end{aligned}\] where
                \UglyAlignS{
                    C&=\dfrac{1}{\ell!(\ell+k)!m!(m+k)!n!(n+k)!}\nonumber\\
                    \nu_{jk}&=\dfrac{|\phi_{jk}|^2}{\phi_{jj}\phi_{kk}}\nonumber\\
                    \delta_1&=\ell+n+k+1\\
                    \delta_2&=m+\ell+k+1\\
                    \delta_3&=n+m+k+1\\
                }
                Regarding the scaled chi-squares, i.e., the trivariate quadratic form \((\rs{Q}_1,\rs{Q}_2,\rs{Q}_3)=(R_1^2,R_2^2,R_3^2)\), the PDF can be simply calculated by \[f_{\rv{Q}}(q_1,\ldots,q_N)=\dfrac{1}{\prod_{i=1}^N\sqrt{q_i}}f_{\rv{R}}\left(\sqrt{q_1},\ldots,\sqrt{q_N}\right)\]
                and the CDF is obtained using \[F_{\rv{Q}}(q_1,\ldots,q_N)=f_{\rv{R}}\left(\sqrt{q_1},\ldots,\sqrt{q_N}\right)\]

                Hence, the PDF of the quadratic form, that is, the squares of Chen and Tellambura's is given by 
                \begin{align}
                    f_{\rv{Q}}(q_1,q_2,q_2)&=8|\dmt{\Phi}|\exp(-(\phi_{11}q_1+\phi_{22}q_2+\phi_{33}))\nonumber \sum_{k=0}^\infty \epsilon_k(-1)^k\cos(k\chi)\nonumber\\
                    &\times \sum_{\ell,m,n=0}^\infty \dfrac{|\phi_{12}|^{2\ell+k}}{\ell!(\ell+k)!}\dfrac{|\phi_{23}|^{2m+k}}{m!(m+k)!}\dfrac{|\phi_{31}|^{2n+k}}{n!(n+k)!} q_1^{\ell+n+k}q_2^{\ell+m+k}q_3^{m+n+k}
                \end{align}
                and the CDF can be written as
                \UglyAlign{
                    F_{\rv{Q}}(q_1,q_2,q_3)&=\dfrac{|\dmt{\Phi}|}{\phi_{11}\phi_{22}\phi_{33}}\sum_{k=0}^\infty \epsilon_k(-1)^k\cos(k\chi)\nonumber\\
                    &\times \sum_{\ell,m,n=0}^\infty C\nu_{12}^{\ell+\frac{k}{2}}\nu_{23}^{m+\frac{k}{2}}\nu_{31}^{n+\frac{k}{2}}\gamma(\delta_1,q_1\phi_{11})\nonumber\\
                    &\times \gamma(\delta_2,q_2\phi_{22})\gamma(\delta_3,q_3\phi_{33})
                }

                \paragraph{Trivariate and Quadrivariate Rayleigh: Integral Expressions}

                Beaulieu and Zhang (2017) \cite{beaulieuNewSimplestExact2017} study a trivariate Rayleigh distribution with a real underlying covariance matrix. Evidently, the formulation is identical to that in Wiegand and Nadarjah \cite{wiegandSeriesApproximationsRayleigh2019} we showed before, with the additional assumption of $M=N=3$. The authors generalize their formulae to the quadrivariate case ($M=N=4$).  They start with the joint distribution of the real and imaginary parts of the components of a complex random vector distributed as \(\mathcal{CN}(\dv{0},\dmt{\Sigma})\), where \(\dmt{\Sigma}\) is real. Let \(\operatorname{Var}(\rs{x}_i)=\sigma_i^2\), and \(\operatorname{Cov}(\rs{x}_i,\rs{x}_j)=\rho_{ij}\sigma_i\sigma_j\), where \(\rho\in\mathbb{R}\). The joint distribution of the real and imaginary parts of an RV can be simply converted to the joint distribution of its modulus and an argument. Employing this change for the three complex variables, they obtain the joint distribution of the three sought Rayleigh variables with the phases. Hence, the sought PDF can be written as a triple integral. After some manipulations, the triple integral is transformed into a single integral; this transformation is the novel aspect of this method. The PDF is then written as \[f_{\rv{R}}(r_1,r_2,r_3)=g_1(r_1)g_2(r_2)g_3(r_3)\texttt{INT}(r_1,r_2,r_3)\]
                where \(\texttt{INT}(r_1,r_2,r_3)\) is a single integral. Employing the fact that \((\rs{Q}_1,\rs{Q}_2,\rs{Q}_3)=(R_1^2,R_2^2,R_3^2)\), we can obtain the formula of the PDF of the quadratic form. It is given by
                    \[f_{\rs{Q}_1,\rs{Q}_2,\rs{Q}_3}(q_1,q_2,q_3)=g_1(q_1)g_2(q_2)g_3(q_3)\texttt{INT}(q_1,q_2,q_3)\label{Beaulieu-Zhang}\]
                    where
                \begin{align*}
                    g_i(q_i)&=\dfrac{1}{2\sigma_i^2}\exp\left\{-\dfrac{1}{|\dmt{\Omega}|}\left[\dfrac{q_i(1-\rho_{jk}^2)}{\sigma_i^2}\right]\right\}\\
                    \texttt{INT}(q_1,q_2,q_3)&=\int_0^\pi \dfrac{\exp\{L_2\cos(t)\}}{\pi|\dmt{\Omega}|}I_0(\sqrt{L_1^2+L_3^2+2L_1L_3\cos(t)})dt\\
                    L_i&=\dfrac{2\sqrt{q_j}\sqrt{q_k}(\rho_{jk}-\rho_{ij}\rho_{ik})}{\sigma_j\sigma_k|\dmt{\Omega}|}\\
                    \dmt{\Omega}&=\begin{bmatrix}
                        1 &\rho_{12} &\rho_{13}\\
                        \rho_{12} &1 &\rho_{23}\\
                        \rho_{13} &\rho_{23} &1
                    \end{bmatrix} \text{ (correlation matrix)}
                \end{align*} where $(i,j,k)$ is a permutation of $(1,2,3)$. Note that \(\rv{Q}\sim\mathcal{C}\chi_1^2(3,\dmt{\Sigma},\dmt{0})\).
                
                For the quadrivariate counterpart, a double integral representation is possible. The PDF is given by
                {\small\UglyAlign{
                    &f_{\rs{Q}_1,\rs{Q}_2,\rs{Q}_3,\rs{Q}_4}(q_1,q_2,q_3,q_4)=\dfrac{1}{64\pi^2\sigma_1^2\sigma_2^2\sigma_3^2\sigma_4^2|\dmt{\Omega}|}\exp\left(-\dfrac{1}{2|\dmt{\Omega}|}\sum_{i=1}^4\dfrac{r_i|\dmt{\Phi_i}|}{\sigma_i^2}\right)\nonumber\\
                    &\times \int_0^{2\pi}\int_0^{2\pi} \exp\left[L_{13}\cos(t_1)+L_{14}\cos(t_2)+L_{34}\cos(t_2-t_1)\right]\nonumber\\
                    &\times I_0\left\{\sqrt{L_{12}^2+L_{23}^2+L_{24}^2+2L_{12}L_{23}\cos(t_1)+2L_{12}L_{24}\cos(t_2)+2L_{23}L_{24}\cos(t_1-t_2)}\right\}dt_1dt_2
                }}
                where
                \UglyAlignS{
                    \dmt{\Phi_i}&=\begin{bmatrix}
                        1 &\rho_{jk} &\rho_{j\ell}\\
                        \rho_{ik} &1 &\rho_{k\ell}\\
                        \rho_{j\ell} &\rho_{k\ell} &1
                    \end{bmatrix}\\
                    L_{ij}&=\dfrac{\sqrt{r_ir_j}}{4\sigma_i\sigma_j|\dmt{\Omega}|}\begin{vmatrix}
                        \rho_{ij} &\rho_{jk} &\rho_{j\ell}\\
                        \rho_{ik} &1 &\rho_{k\ell}\\
                        \rho_{i\ell} &\rho_{k\ell} &1 
                    \end{vmatrix}\\
                    \dmt{\Omega}&=\begin{bmatrix}
                        1 &\rho_{12} &\rho_{13} &\rho_{14}\\
                        \rho_{12} &1 &\rho_{23} &\rho_{24}\\
                        \rho_{13} &\rho_{23} &1 &\rho_{34}\\
                        \rho_{14} &\rho_{24} &\rho_{34} &1
                    \end{bmatrix}
                } where $(i,j,k,\ell)$ is a permutation of $(1,2,3,4)$. Note that with the exception of this case, all formulae for multivariate Rayleigh and Rician distributions incorporated infinite series. Again, \(\rv{Q}\sim\mathcal{C}\chi_1^2(4,\dmt{\Sigma},\dmt{0})\).

            \subsection{Tri-Chi-Square}

                \paragraph{General Real Case}
            
                For $M=3$, Royen \cite{royenExpansionsMultivariateChisquare1991} gives three expressions for the CDF, of which we provide one. The CDF of the random vector $\rv{Q}\sim\chi_3^2(2p,\dmt{R})$ (i.e., $p$ is half the number of degrees of freedom) is given by \begin{equation*}
                    F_{\rv{Q}}(q_1,q_2,q_2)=|\dmt{Q}|^p\left(\prod_{j=1}^3 \dfrac{\gamma(p,d_j q_j)}{\Gamma(p)}+\sum_{n=2}^\infty P_n(q_1,q_2,q_3)\right)
                \end{equation*}
                where \[ 
                \resizebox{\textwidth}{!}{$
                    P_n(q_1,q_2,q_3)=\begin{cases}
                        \begin{aligned}
                            &\dfrac{1}{\Gamma(p)}\displaystyle\sum_{n_1+n_2+n_3=n/2}\left(\sum_{m=0}^{\min n_j}\dfrac{2^{2m}}{(2m)!}\dfrac{\Gamma(p+\frac{n}{2}-m)}{\prod_{j=1}^3(n_j-m)!}\right)\\
                            &\times \prod_{j=1}^3c_j^{2n_j}\dfrac{\gamma(p+n/2-m_j,d_j q_j)}{\Gamma(p+N/2-m_j)},\text{ if $n$ is even,}
                        \end{aligned} \\
                        \begin{aligned}
                            &\dfrac{1}{\Gamma(p)}\displaystyle\sum_{n_1+n_2+n_3=(n-1)/2}\left(\sum_{m=0}^{\min n_j}\dfrac{2^{2m+1}}{(2m+1)!}\dfrac{\Gamma(p+\frac{n-1}{2}-m)}{\prod_{j=1}^3(n_j-m)!}\right)\\
                            &\times \prod_{j=1}^3c_j^{2n_j+1}\dfrac{\gamma(p+(n-1)/2-m_j,d_j q_j)}{\Gamma(p+N/2-m_j)}, \text{ if $n$ is odd.}
                        \end{aligned}     
                    \end{cases} 
                    $} 
                \] and \begin{align*}
                    d_j&=\dfrac{r^{jj}}{2}\\
                    (r^{ij})&=\dmt{R}^{-1}\\
                    c_j&=-q_{i\ell} \text{ where $(j,i,\ell)$ is a permutation of $(1,2,3)$}\\
                    q_{ij}&=\dfrac{r^{ij}}{\sqrt{r^{ii}r^{jj}}}
                \end{align*}

                \paragraph{Complex Tri-Chi-Square with Tridiagonal Precision Matrix}

                Dharmawansa and McKay \cite{dharmawansaDiagonalDistributionComplex2009} study the trivariate complex chi-square distribution distribution with a tridiagonal precision matrix and a special case of non-centrality. Let \[\rv{Q}\sim \mathcal{C}\chi_K^2(3,\dmt{\Sigma},\delta^2\dmt{J}_3)\] where $\dmt{\Sigma}^{-1}=\begin{bmatrix}
                    \varphi_{11} &\varphi_{12} &0\\
                    \bar{\varphi}_{12} &\varphi_{22} &\varphi_{23}\\
                    0 &\bar{\varphi}_{23} &\varphi_{33}
                \end{bmatrix}$ and $\dmt{J}_3$ is the $3\times 3$ matrix of ones.
                
                The PDF, for $K>1$, is given by
                \UglyAlign{
                    &f_{\rs{Q}_1,\rs{Q}_2,\rs{Q}_3}(q_1,q_2,q_3)=\dfrac{\exp\left\{-(\varphi_{11}q_1+\varphi_{22}q_2+\varphi_{33}q_3+d\delta^2)\right\}}{|\dmt{\Sigma}|^n(|a||b||c||\varphi_{12}|\varphi_{23}|\delta^3\sqrt{q_2})^{K-1}}\nonumber\\
                    &\times \sum_{k,\ell=0}^\infty \sum_{m=0}^\ell \sum_{p=0}^{\min(m,\lfloor\frac{k}{2}\rfloor)}\sum_{q=0}^{\min(k-2p,\ell-m)}(-1)^{k+\ell} K_{k,l,m,p,q}C_{\nu_{k,\ell}^{p,q}}^{\nu_m}(\cos\alpha)\nonumber\\
                    &\times C_{\ell-m}^{\nu_m}(\cos\beta)\sin^m\alpha\sin^m\beta I_{\nu_k}(2|a|\delta\sqrt{q_1})I_{\nu_k}(2|\varphi_{12}|\sqrt{q_1q_2})\nonumber\\
                    &\times I_{\nu_\ell}(2|c|\delta\sqrt{q_3})I_{\nu_\ell}(2|\varphi_{23}|\sqrt{q_2q_3})I_{\nu'^{p,q}_{k,l}}(2|b|\delta\sqrt{q_2}) \label{Dharmawansa-McKay}
                }
                where\\ \UglyTabular{{ll}
                        \(a=\varphi_{11}+\bar{\varphi}_{12}\) & \(A=\Re(a)\Re(\varphi_{12})-\Im(a)\Im(\varphi_{12})\) \\
                        \(b=\varphi_{22}+\varphi_{12}+\bar{\varphi}_{23}\) & \(B=\Im(a)\Re(\varphi_{12})+\Re(a)\Im(\varphi_{12})\)\\
                        \(c=\varphi_{33}+\varphi_{23}\)  & \(C=\Re(c)\Re(\varphi_{23})+\Im(c)\Im(\varphi_{23})\)\\
                        \(d=\varphi_{11}+\varphi_{22}+\varphi_{33}+2\Re(\varphi_{12}+\varphi_{23})\) & \(D=\Im(c)\Re(\varphi_{23})-\Re(c)\Im(\varphi_{23})\) \\
                        \(\alpha=\operatorname{Arg}(Z_S)-\operatorname{Arg}(Z_Q)\) & \(\beta=\operatorname{Arg}(Z_R)-\operatorname{Arg}(Z_Q)\)       
                }\\
                and
                \UglyAlignS{
                \nu_\lambda&=K+\lambda-1\\
                \nu_{k,\ell}^{p,q}&=k+\ell+m-2p-2q\\
                \nu_{k,\ell}^{'p,q}&=k+\ell+K-2p-2q-1\\
                K_{\ell,m,p,q}&=\dfrac{2^{2m}(\ell-m)!\Gamma(2K+m-3)\Gamma^2(\nu_m)\nu_k\nu_\ell(2\nu_m-1)}{m!\Gamma(\nu_m+\nu_\ell)}B_{p,q}A_{q,k-2p,\ell-m}\\
                B_{\xi,\eta}&=\dfrac{(-m)_\xi(K-1)_{K-\xi}(K-2\xi+\nu_m)}{(K+m)_{\eta-\xi}\nu_m\xi!}\\
                A_{\xi,\eta,\kappa}&=\dfrac{(\eta+\kappa-2\xi+\nu_m)(\nu_m)_\xi(\nu_m)_{\eta-2\xi}(\nu_m)_{\kappa-\xi}(2\nu_m)_{\eta+\kappa-\xi}(\eta+\kappa-2\xi)!}{(\eta+\kappa-\xi+\nu_m)\xi!(\eta-\xi)!(\kappa-\xi)!(\nu_m)_{\eta+\kappa-\xi}(2\nu_m)_{\eta+\kappa-2\xi}}\\
                (a)_n&=\begin{cases}
                \dfrac{\Gamma(a+n)}{\Gamma(a)} &\text{ for } a\notin \mathbb{Z}^-\cup \{0\}\\
                \dfrac{(-1)^n(-a)!}{(-a-n)!} &\text{ for } a\notin \mathbb{Z}^-\cup \{0\} \text{ and } 0\leq n\leq -a\\
                0 &\text{ for } a\in \mathbb{Z}^- \text{ and } n>-a.
                \end{cases}
                }

                \paragraph{Real Tri-Chi-Square with Tridiagonal Precision Matrix}

                Dharmawansa et. al. \cite{dharmawansaNewSeriesRepresentation2009} study the same distribution, but for the real Wishart matrix. Dharmawansa and McKay \cite{dharmawansaDiagonalDistributionComplex2009} reprove the result exploiting the complex case. Let \[\rv{Q}\sim\chi_K^2(3,\dmt{\Sigma},\delta^2\dmt{J}_3)\] where $\dmt{\Sigma}^{-1}=\begin{bmatrix}
                    \varphi_{11} &\varphi_{12} &0\\
                    \bar{\varphi}_{12} &\varphi_{22} &\varphi_{23}\\
                    0 &\bar{\varphi}_{23} &\varphi_{33}
                \end{bmatrix}$ and $\dmt{J}_3$ is the $3\times 3$ matrix of ones.
                The joint PDF is given, for \(K>2\), by
                \UglyAlign{
                    &f_{\rs{Q}_1,\rs{Q}_2,\rs{Q}_3}(q_1,q_2,q_3)=\dfrac{2^{K-5}\Gamma^2(p-1)\exp\{-\frac{1}{2}(\varphi_{11}q_1+\varphi_{22}q_2+\varphi_{33}r_3+d\|\dv{\mu}\|^2)\}}{|\dmt{\psi}|^p(abc\varphi_{12}\varphi_{23}\|\mu\|^3\sqrt{q_2})^{p-1}}\nonumber\\
                    &\times \sum_{k=0}^\infty \sum_{\ell=0}^\infty \sum_{q=0}^{\min(k,\ell)}(-1)^{k+\ell}\binom{K+k+\ell-2q-3}{K-3}A_{q,k,\ell}\nonumber\\
                    &\times (p+k-1)(p+\ell-1)I_{p+\ell-1}(a\|\dmt{\mu}\|\sqrt{q_1})I_{p+k-1}(\varphi_{12}\sqrt{q_1q_2})\nonumber\\
                    &\times I_{p+\ell-1}(c\|\mu\|\sqrt{q_3})I_{p+\ell-1}(\varphi_{23}\sqrt{q_2q_3})I_{p+k+\ell-2q-1}(b\|\mu\|\sqrt{q_2})
                }
                where
                \UglyEquation{
                    \begin{aligned}
                        a&=\varphi_{11}\\
                        b&=\varphi_{22}+\varphi_{12}+\varphi_{23}\\
                        c&=\varphi_{33}+\varphi_{23}\\            d&=\varphi_{11}+\varphi_{22}+\varphi_{33}+2\varphi_{12}+2\varphi_{23}\\
                        A_{q,k,\ell}&=\dfrac{(k+\ell+\lambda-2q)(\lambda)_q(\lambda)_{k-q}(\lambda)_{\ell-q}(2\lambda)_{k+\ell-q}(k+\ell-2q)!}{(k+\ell+\lambda-q)q!(k-q)!(\ell-q)!(\lambda)_{k+\ell-q}(2\lambda)_{k+\ell-2q}}\\
                        \lambda&=p-1
                    \end{aligned}\nonumber
                }
                The joint CDF is given by
                {\footnotesize
                \UglyAlign{
                    &F_{\rs{Q}_1,\rs{Q}_2,\rs{Q}_3}(q_1,q_2,q_3)=|\dmt{\Psi}|^{p}\Gamma^2(p-1)\exp(-d\|\dv{\mu}\|^2/2)\nonumber\\
                    &\times \sum_{k,\ell=0}^\infty \sum_{m=0}^{\min(k,\ell)}(-1)^{k+\ell}A_{m,k,\ell}(p+k-1)(p+\ell-1)\binom{K+k+\ell-2m-3}{K-3}\nonumber\\
                    &\times \sum_{i_1,i_2,i_3,i_4,i_5=0}^\infty \dfrac{\gamma(\delta_1,\frac{\varphi_{11}}{2}q_1)\gamma(\delta_2,\frac{\phi_22}{2}q_2)\gamma(\delta_3,\frac{\phi_33}{2}q_3)}{2^{\frac12(\lambda_1+\lambda_2+\lambda_3)}1_1!i_2!i_3!i_4!i_5!\Gamma(i_1+p+k)\Gamma(i_2+p+k)\Gamma(i_3+p+\ell)}\nonumber\\
                    &\times\dfrac{\|\dv{\mu}\|^{\lambda_1+\lambda_2+\lambda_3}a^{\lambda_1}b^{\lambda_2}c^{\lambda_3}\varphi_{12}^{\lambda_4}\varphi_{23}^{\lambda_5}}{\Gamma(i_4+p+\ell)\Gamma(i_5+p+k+\ell-2m)\varphi_{11}^{\delta_1}\varphi_{22}^{\delta_2}\varphi_{33}^{\delta_3}} \label{Dharmawansa et al} 
                }}
                where
                \UglyEquation{
                    \begin{aligned}
                        \lambda_1&=2i_1+k & \lambda_5&=2i_4+\ell  \\
                        \lambda_2&=2i_5+k+\ell-2m & \delta_1&=i_1+i_2+k+p \\
                        \lambda_3&=2i_3+\ell & \delta_2&=i_2+i_4+i_5+k+p-m\\
                        \lambda_4&=2i_2+k & \delta_3&=i_3+i_4+\ell+p \\
                    \end{aligned}\nonumber
                }

                \paragraph{Complex Tri-Chi-Square}
                
                Hagedorn et. al. \cite{hagedornTrivariateChisquaredDistribution2006} study the distribution of \(\rv{Q}\sim\mathcal{C}\chi_K^2(3,\dmt{\Sigma},\dmt{0})\). Write the covariance matrix $\dmt{\Sigma}$ as 
                \[2\begin{bmatrix}
                    \sigma_1^2 &\sigma_{12A}+i\sigma_{12C} &\sigma_{13A}+i\sigma_{13C}\\
                    \sigma_{12A}-i\sigma_{12C} &\sigma_2^2 &\sigma_{23A}+i\sigma_{23C}\\
                    \sigma_{13A}-i\sigma_{13C} &\sigma_{23A}-i\sigma_{23C} &\sigma_3^2
                \end{bmatrix}\]

                The joint probability distribution, for $K>1$, is given by
                \[\begin{aligned}
                    &f_{\rs{Q}_1,\rs{Q}_2,\rs{Q}_3}(q_1,q_2,q_3)=\dfrac{\exp(-\frac{1}{2}(a_1q_1+b_1q_2+c_1q_3))}{8(K-1)D^{p}(d_1d_2d_3)^{K-1}}\\
                    &\times\sum_{k=K-1}^\infty k(-1)^{k-K+1}C^{K-1}_{k-K+1}(\cos(\gamma))I_k\left(d_1\sqrt{q_1q_2}\right)I_k\left(d_2\sqrt{q_2q_3}\right)I_k\left(d_3\sqrt{q_2q_3}\right)
                \end{aligned}\]
                where
                \UglyAlignS{
                    d_1^2&=a_2^2+a_3^2,\;\;
                    d_2^2=b_2^2+b_3^2,\;\;
                    d_3^2=c_2^2+c_3^2,\\
                    d_4&=(a_2c_2+a_3c_3)b_2+(a_2c_3-a_3c_2)b_3,\;\;
                    \cos(\gamma)=\dfrac{d_4}{d_1d_2d_3}\\    D&=\left|\sigma_1^2(\sigma_{23A}^2+\sigma_{23C}^2)+\sigma_2^2(\sigma_{13A}^2+\sigma_{13C}^2)+\sigma_3^2(\sigma_{12A}^2+\sigma_{23C}^2)-\sigma_1^2\sigma_2^2\sigma_3^2\right.\\ &\left.-2(\sigma_{12A}\sigma_{13A}\sigma_{23A}+\sigma_{12A}\sigma_{13C}\sigma_{23C}+\sigma_{12C}\sigma_{13C}\sigma_{23A}-\sigma_{12C}\sigma_{13A}\sigma_{23C})\right|
                }
                and the coefficients \(a_i,b_i,c_i\) account to \UglyEquation{\dmt{\Sigma}^{-1}=\dfrac{1}{2}\begin{bmatrix}
                    a_1 &a_2-ia_3 &c_2-ic_3\\
                    a_2+ia_3 &b_1 &b_2-ib_3\\
                    c_2+ic_3 &b_2+ib_3 &c_1
                \end{bmatrix}\nonumber}

                The joint probability distribution, for $K=1$, is given by
                \UglyAlign{
                f_{\rs{Q}_1,\rs{Q}_2,\rs{Q}_3}(q_1,q_2,q_3)&=\dfrac{\exp(-\frac{1}{2}(a_1q_1+b_1q_2+c_1q_3))}{8D} \nonumber\\ &\times \left[\vphantom{\sum_{k=1}^\infty}I_0\left(d_1\sqrt{q_1q_2}\right)I_0\left(d_2\sqrt{q_2q_3}\right)I_0\left(d_3\sqrt{q_1q_3}\right)\right.\nonumber\\
                &\left.+ 2\sum_{k=1}^\infty (\cos(k\gamma))I_k\left(d_1\sqrt{q_1q_2}\right)I_k\left(d_2\sqrt{q_2q_3}\right)I_k\left(d_3\sqrt{q_1q_3}\right) \right]
                }

                Moreover, the joint PDF of the bi-variate distribution is given by
                \begin{align}
                    f_{\rs{Q}_2,\rs{Q}_2}(q_1,q_2)&=\dfrac{\exp\left(-\dfrac{\sigma_2^2q_1+\sigma_1^2q_2}{2(\sigma_1^2\sigma_2^2-\sigma_{12A}^2-\sigma_{12C}^2)}\right)(q_1q_2)^{(K-1)/2}}{(\sigma_1^2\sigma_2^2-\sigma_{12A}^2-\sigma_{12C}^2)\Gamma(K)2^{K+1}(\sigma_{12A}^2+\sigma_{12C}^2)^{(K-1)/2}}\nonumber\\
                    &\times I_{K-1}\left(\dfrac{\sqrt{\sigma_{12A}^2+\sigma_{12C}^2}\sqrt{q_1q_2}}{\sigma_1^2\sigma_2^2-\sigma_{12A}^2-\sigma_{12C}^2}\right) \label{eq:Hagedorn-bivariate}
                \end{align}

                Compared to Dharamawansa and McKay's work \cite{dharmawansaDiagonalDistributionComplex2009}, we can see that Hagedorn formula allows only central variables with no restriction whatsoever on the covariance structure, compared to a trivariate precision matrix restriction in the former case. Formula-wise, we can see that both are infinite series in modified bessel functions of the first kind, however, Dharmawansa's series is a double sum and Hagedorn's is a single sum.

                
                Peppas and Sagias \cite{peppasTrivariateNakagamimDistribution2009} study Hagedorn's distribution with a cross-covariance of zero. So the covariance matrix $\dmt{\Sigma}$ is real and has the following form: \[2\begin{bmatrix}
                    \sigma_1^2 &c_{12} &c_{13}\\
                    c_{12} &\sigma_2^2 &c_{23}\\
                    c_{13} &c_{23} &\sigma_3^2
                \end{bmatrix}\]  The PDF is directly derived from Hagedorn. The CDF, $F_{\rs{Q}_1,\rs{Q}_2,\rs{Q}_3}(q_1,q_2,q_3)$, for \(K>1\), is given by
                \UglyAlign{
                &\dfrac{T^{2K}}{2(K-1)q_1q_2q_3}\sum_{k=K-1}^\infty \sum_{l_1,l_2,l_3=0}^\infty k(-1)^{k-K+1}\binom{K+k-2}{2K-3}\nonumber\\
                &\times \prod_{\ell=1}^3\dfrac{\omega_\ell^{2l_\ell+k+1-K}n_\ell !}{\psi_\ell^{n_\ell+1}l_\ell ! (l_\ell+k)!}\left[1-\exp\left(-\dfrac{n\psi_\ell}{T\Omega_\ell}q_\ell\right)\sum_{p=0}^{n_\ell}\dfrac{1}{p!}\left(\dfrac{n\psi_\ell}{T\Omega_\ell}q_\ell\right)^p\right] \label{PeppasSagias}
                }
                where    
                \UglyTabular{{lll}
                    \multicolumn{3}{l}{$T=1-(\rho_{12}+\rho_{23}+\rho_{13})+2\sqrt{\rho_{12}\rho_{22}\rho_{13}}$}\\
                    $n_1=l_1+l_3+k$, &$\psi_1=1-\rho_{23}$, &$\omega_1=-\sqrt{\rho_{12}}+\sqrt{\rho_{23}\rho_{13}}$,\\
                    $n_2=l_1+l_2+k$, &$\psi_2=1-\rho_{13}$, &$\omega_2=-\sqrt{\rho_{23}}+\sqrt{\rho_{12}\rho_{13}}$,\\
                    $n_3=l_2+l_3+k$, &$\psi_3=1-\rho_{12}$, &$\omega_3=-\sqrt{\rho_{13}}+\sqrt{\rho_{12}\rho_{23}}$\\
                    $\rho_{\ell\ell'}=\dfrac{4n^2c_{\ell\ell'}^2}{\Omega_\ell^2\Omega_{\ell'}^2}$, &$\Omega_\ell=2n\sigma_\ell^2$, &\(\psi_i=1-\rho_{jk}\).
                }


            \subsection{Tree Structure}

                A matrix $\dmt{R}$ can be represented as a graph with nodes $i$ and $j$ connected if $\dmt{R}_{ij}\neq 0$. A graph is a tree if it has no cycles. For example tridiagonal matrices represent a chain of nodes. The correlation matrix or its inverse having such a structure offers some simplification. The general case has been studied in Royen \cite{royenMultivariateGammadistributionsConnected1994}. Special tridiagonal structures where studied in \cite{dharmawansaDiagonalDistributionComplex2009,dharmawansaNewSeriesRepresentation2009} and a penta-diagonal inverse in \cite{chenInfiniteSeriesRepresentations2005}. Note m-factorial matrices are in general not tree like structures.

                The CDF of a central real-multi-chi-squares with covariance matrices that have a tree-type take the form: \cite[Eq. 3.8]{royenMultivariateGammadistributionsConnected1994}: 
                
                \begin{align}
                F_{\rv{Q}}&(\dv{q}) =\sum_{n=0}^\infty \sum_{(\dv{n})} c(\dv{n}) \prod_{i\in A}G_{p+n_i}^{(n_i)}(q_i) \prod_{i\in B}\int_0^\infty g_{p+n_i}(t_i)\nonumber\\
                & \prod_{k\in L_i}H_{p}(q_k,-r_{ki}^2t_i)
                (\sqrt{q_i/y_i})^{p} I_{p}(2\sqrt{q_it_i}) dt_i \label{eqn:RoyenTreeEquation}
                \end{align}
                The $M$ variables $\dv{q}$ are split into three disjoint sets: $A, B, L$. The variables in $L$ are nodes in the graph connected to only one other node, i.e., they are leaves. Set $B$ are nodes connected to a node in $L$, i.e., predecessors of leaves. Set $A$ are nodes not connected to any in nodes in $L$.  Taking the $i^{th}$ node in $B$, $L_i$ are the nodes in $L$ connected to $i$. Hence, $L = \cup_{i\in B} L_i$. Let the indices $(i,j)$ of the non-zero elements in $\dmt{R}$ be $\mathcal{R}$. Then the partition $(\dv{n})$ is $\sum_{(i,j)\in\mathcal{R}} n_{ij} = n$ is over the non-zero elements of the correlation matrix. The constant $c(\dv{n})$ depends only on $n_{ij}$ and $\dmt{R}$. Finally, $n_i = \sum_{j} n_{ij}$
                
                The second line in \Cref{eqn:RoyenTreeEquation}, represents the conditional distribution of $L_i$ variables conditioned on the $i \in B$ variable. The integral in the first line ``deconditions'' this distribution, and finally the product over $A$ accounts for the ``unconnected'' variables. Utilizing the tree correlation structure of the covariance matrix represents another variant of the deconditioning approach. Note \Cref{eqn:RoyenTreeEquation} has simpler equivalents \cite[Eq. 3.7, 3.9]{royenMultivariateGammadistributionsConnected1994}, but these hide the conditioning/deconditioning action.
                
            \subsection{$m$-Factorial Matrices}

                \paragraph{Definitions}
                An $M\times M$ correlation matrix $\dmt{R}$ is said to be $m$-factorial, if $m$ is the smallest integer allowing the representation \[\dmt{R}=\dmt{D}+\dmt{A}\dmt{A}^T\]
                where $\dmt{D}>0$ \footnote{In \cite{royenIntegralRepresentationsApproximations2007}, the term \(m\)-factorial is relaxed to allow for a complex matrix \(\dmt{D}\); however, we adopt the notation of other Royen's papers, as in \cite{royenCENTRALNONCENTRALMULTIVARIATE1995}, for example.} is a diagonal matrix, and $\dmt{A}$ is an $M\times m$ matrix of rank $m$. Note that any $M\times M$ covariance matrix is at most $(M-1)$-factorial \footnote{Let $\lambda_{\text{min}}$ be the smallest eigenvalue of $\dmt{\Sigma}$. Eigendecomposing with decreasing eigenvalues, we can write $\dmt{\Sigma}=\dmt{U}^T\dmt{D U}=\dmt{U}^T(\dmt{D}-\lambda_{\text{min}}\dmt{I}_N+\lambda_{\text{min}}\dmt{I}_N)\dmt{U}=\lambda_{\text{min}}\dmt{I}_N+\dmt{A}\dmt{A}^T$ where $\dmt{A}=\dmt{U}^T\diag(\sqrt{\lambda_1},\ldots,\sqrt{\lambda_{N-1}})$.}. It can be shown, for a real matrix $\dmt{R}$, that the columns of $\dmt{A}$ in the above representation are either real or pure imaginary. In the case of a real matrix \(\dmt{A}\), the components of \(\dv{x}\) can be written as \(x_j=d_j^{1/2}U_j+\sum_{k=1}^ma_{jk}Z_k\), with i.i.d \(\mathcal{N}(0,1)\) \(U_j,Z_k\). So we have \(m\) common ``factors'', the fact which justifies the wording.
                
                The one-factorial case is well-studied in the literature (see, for example, \cite{royenMultivariateGammaDistributions1991}, \cite{beaulieuNovelSimpleRepresentations2011}, \cite{bithasNovelResultsMultivariate2019}), as it shows up in some applications and as it simplifies the calculations considerably. It is easy to verify that an $M\times M$ matrix $\dmt{R}$ is one-factorial if and only if one of the two conditions below holds:
                \begin{align}
                    \forall i\; &\exists a_i \;\; -1<a_i<1 \text{ and } \forall i\neq j \;\; r_{ij}=a_i a_j \label{3.10}\\
                    \forall i \; &\exists a_i\in\mathbb{R} \; \forall i\neq j\;\; r_{ij}=-a_i a_j \text{ and }R\geq 0. \label{3.11}
                \end{align} Note that \eqref{3.10} corresponds to a real column vector $\dv{a}$, and \eqref{3.11} corresponds to a pure imaginary one (with scaling out common imaginary unit).
    
                An $M\times M$ correlation matrix $\dmt{R}$ is said to be 'real $m$-factorial', if $m$ is the smallest integer allowing the representation \[\dmt{R}=\dmt{D}+\dmt{A}\dmt{A}^T\]
                where $\dmt{D}$ is any real diagonal matrix (can have negative entries), and $\dmt{A}$ is an $M\times m$ matrix of rank $m$. 
        
                \paragraph{One-Factorial Matrices}
        
                Assume $\dmt{R}$ is one-factorial \[\dmt{R}=\dmt{D}+\dv{a}\dv{a}^T.\] Royen \cite{royenCENTRALNONCENTRALMULTIVARIATE1995} proves that for $\rv{Q}\sim\chi_K^2(M,\dmt{R},\dmt{\Delta})$, the CDF $F_{\rv{Q}}(q_1,\ldots,q_M)$ is given by
                \UglyEquation{\sum_{n=0}^\infty \int_0^\infty \left(\sum_{(2n)}d(n_1,\ldots,n_M)\prod_{j=1}^MG_{r,n_j}(\frac{1}{2}w_j^2q_j,d_{jj}+c_j^2t)c_j^{n_j}\right)g_{p+n}(t)dt \label{one-factorial non-central}}
                where
                \[d(n_1,\ldots,n_M)=\alpha^n\sum_{\substack{n_{.j}=n_j-n_{jj} \\ j=1,\ldots,M }}\prod_{1\leq i < j \leq M}\dfrac{{d'}_{ij}^{n_{ij}}}{n_{ij}!}\]
                and \begin{align*}
                    d'_{ij}&=(2-\delta_{ij})d_{ij}\\
                    \dmt{D}&=(d_{ij})=\dfrac{1}{2}\dmt{W}\dmt{\Delta}\dmt{W}\\
                    \dv{c}&=\dmt{D}^{1/2}\dv{a}
                \end{align*}
                Note that \(\displaystyle\sum_{(2n)}\) means summation over all $M$-partitions of $2n$, \(n_1+n_2+\ldots+n_M=2n\).
                

                Suppose that $\dmt{R}$ is one-factorial and $\rank(\dmt{\Delta})=1$. Hence we can write $\dmt{\Delta}=(\delta_i\delta_j)$. Then the distribution function is given by
                \[\int_0^\infty \int_0^\pi \prod_{j=1}^M G_{\frac{K}{2}}\left(\frac{1}{2}w_j^2q_j,\frac{1}{2}w_j^2q_j^2+b_j^2t+2b_j\delta_j\sqrt{\frac{1}{2}w_j^2t}\cos(\phi)\right)f_{\frac{K}{2}}(\phi)d\phi g_{\frac{K}{2}}(t)dt \label{one-factorial non-central rank 1}\]
                where \[f_p(\phi)=\dfrac{(\sin^2\phi)^{p-1}}{B(\frac{1}{2},p-\frac{1}{2})} \label{3.13}\]
        
                The CDF $F(q_1,\ldots,q_M)$, in the case of a central distribution \cite{royenMultivariateGammaDistributions1991}, is given by
                \[F(q_1,\ldots,q_M)=\int_0^\infty g_p(t) \prod_{j=1}^M G_p\left(\dfrac{q_j/2}{1+a_j^2},\dfrac{a_j^2t}{1-a_j^2}\right)dt \label{onefactorialold}\]
                for the case \eqref{3.10} 
                    and 
                \[F(q_1,\ldots,q_M)=\int_0^\infty g_p(t) \prod_{j=1}^M G_p\left(\dfrac{q_j/2}{1+a_j^2},\dfrac{-a_j^2t}{1-a_j^2}\right)dt\]
                for the case \eqref{3.11}  \cite{royenMultivariateGammaDistributions1991} with $K=2p$.

                Nakagami-$m$ distribution can be perceived as the square root of the sum of square moduli of $m$ independent central complex normal vectors, say following $\mathcal{CN}(0,\sigma^2)$. Thus, Rayleigh distribution corresponds to $m=1$. Hence, the Nakagami-$m$ distribution is equivalent to \(\sqrt{\dfrac{\sigma^2}{2}\chi^2_{2m}}\). Starting from the PDF of a chi-square variable of $2m$ degrees of freedom, and after square-rooting and scaling, we obtain the following PDF of a Nakagami-$m$ variable $T$: \[f_T(t)=\dfrac{2}{\Gamma(m)\sigma^{2m}}t^{2m-1}\exp\left(-\dfrac{t^2}{\sigma^2}\right)\] The Nakagami-$m$ distribution is equivalently parametrized by $\Omega$ and $m$ so that $\Omega=m\sigma^2$. 

                Clearly, the square of a Nakagami-$m$ variable corresponds to a single quadratic form $\rs{Q}=\rv{x}^H\dmt{I}_m\rv{x}$, where \(\rv{x}\sim\mathcal{CN}(\dv{0},\sigma^2\dmt{I}_m)\). Note that, in our paper, $m$ is usually replaced by $N$ to unify 
                the notation. 
                
                Beaulieu and Hemachandra \cite{beaulieuNovelSimpleRepresentations2011} study the correlated Rayleigh distribution with a one-factorial underlying covariance matrix, the square roots of a multi-chi-square. Let \(\rv{Q}\sim\mathcal{C}\chi_1^2(M,\dmt{\Sigma},\dmt{0})\) where  \(\dmt{\Sigma}=(\rho_{ij}\sigma_i\sigma_j)_{1\leq i,j \leq n}\), and \(\rho_{ij}= 1 \) if $i=j$ and $\lambda_i\lambda_j$ else. Equivalently, we can write \[\Sigma=\dmt{D}+\dv{a}\dv{a}^T\] where $\dmt{D}=\diag(\sigma_1^2-\lambda_1^2,\ldots,\sigma_N^2-\lambda_N^2)$ and $\dv{a}^T=[\lambda_1,\ldots,\lambda_N]$. The joint PDF is given by
                \[
                \resizebox{\textwidth}{!}{$
                f_{\rv{Q}}(q_1,\ldots,q_N)=\dfrac{1}{2^N}\displaystyle\int_0^\infty e^{-t}\prod_{k=1}^N\dfrac{1}{\Omega_k^2}\exp\left(-\dfrac{q_k+\sigma_k^2\lambda_k^2t}{2\Omega_k^2}\right)I_0\left(\dfrac{\sqrt{q_k}\sqrt{t\sigma_k^2\lambda_k^2}}{\Omega_k^2}\right)dt\label{Beaulieu-Hemachandra:Rayleigh}
                $}
                \]
                where \(\Omega_k^2=\sigma_k^2(1-\lambda_k^2)/2\) and the CDF is given by
                \[
                F_{\rv{Q}}(q_1,\ldots,q_N)=\displaystyle\int_0^\infty e^{-t}\prod_{k=1}^N\left[1-Q_1\left(\dfrac{\sqrt{t}\sqrt{\sigma_k^2\lambda_k^2}}{\Omega_k},\dfrac{\sqrt{q_k}}{\Omega_k}\right)\right]dt
                \]

                Beaulieu and Hemachandra \cite{beaulieuNovelSimpleRepresentations2011} study correlated Rician distribution with a one-factorial underlying covariance structure.
                Let \(\rv{x}\sim \mathcal{CN}(\dv{\mu},\dmt{\Sigma})\) with \(\dmt{\Sigma}=(\rho_{ij}\sigma_i\sigma_j)_{1\leq i,j \leq n}\), \(\rho_{ij}=\begin{cases}
                    1 &\text{if } i=j\\
                    \lambda_i\lambda_j &\text{else}
                \end{cases}\) and \(\dv{\mu}=(m_1+im_2)[\sigma_1\lambda_1,\ldots,\sigma_n\lambda_n]^T\).
                Consider the $N$ quadratic forms \(\rs{Q}_j=\rv{x}^T\dmt{A}_j\rv{x}\) (hence $M=N$) where \(\dmt{A}_j\) is the diagonal matrix whose only non-zero entry is a unit at the \(j^{\text{th}}\) position. Equivalently, \(\rv{Q}\sim\mathcal{C}\chi_1^2(M,\dmt{\Sigma},\dmt{\Delta})\), with $\dmt{\Sigma}$ defined as in the previous case, and \(\dmt{\Delta}=(m_1^2+m_2^2)\dv{v}\dv{v}^T\), with \(\dv{v}^T=[\lambda_1\sigma_1,\ldots,\lambda_N\sigma_N]\). The joint PDF is given by
                \begin{align}
                    f_{\rv{Q}}(q_1,\ldots,q_N)&=\dfrac{1}{2^N}\int_0^\infty e^{-t-m_1^2-m_2^2}I_0\left(2\sqrt{t}\sqrt{m_1^2+m_2^2}\right)\nonumber\\
                    &\times\prod_{k=1}^N\dfrac{1}{\Omega_k^2}\exp\left(-\dfrac{q_k+\sigma_k^2\lambda_k^2t}{2\Omega_k^2}\right)I_0\left(\dfrac{\sqrt{q_k}\sqrt{t\sigma_k^2\lambda_k^2}}{\Omega_k^2}\right)dt\label{Beaulieu-Hemachandra:Rician}
                \end{align}
                where \(\Omega_k^2=\sigma_k^2(1-\lambda_k^2)/2\) and the CDF is given by
                \[\begin{aligned}
                    F_{\rv{Q}}(q_1,\ldots,q_N)&=\int_0^\infty e^{-t-m_1^2-m_2^2}I_0\left(2\sqrt{t}\sqrt{m_1^2+m_2^2}\right)\\
                    &\times\prod_{k=1}^N\left[1-Q_1\left(\dfrac{\sqrt{t}\sqrt{\sigma_k^2\lambda_k^2}}{\Omega_k},\dfrac{\sqrt{q_k}}{\Omega_k}\right)\right]dt
                \end{aligned}\] As expected, the joint distribution depends only on the sum of squares of $m_1$ and $m_2$, not on their distinct values.
                
                The same authors \cite{beaulieuNovelRepresentationsBivariate2011} derive the same formula for the particular bivariate case.
                
                Bithas \cite{bithasNovelResultsMultivariate2019} generalizes this to account for \[\dv{\mu}=[\sigma_1\lambda_1(m_{1,1}+im_{2,1}),\ldots,\sigma_N\lambda_N(m_{1,N}+im_{2,N})]^T\] which accounts for any mean. The PDF and CDF are given by \begin{align}
                f_{\rv{Q}}\left(q_1, q_2, \ldots, q_N\right)&=\frac{1}{\pi} \int_{-\infty}^{\infty} \int_{-\infty}^{\infty} \exp \left(-x_0^2-q_0^2\right)\\
                &\times\prod_{k=1}^N \frac{r_k}{b_k} \exp \left(-\frac{q_k+s_k^2}{2 b_k}\right) I_0\left(\frac{\sqrt{q_k} s_k}{b_k}\right) d x_0 d q_0 \\
                F_{\rv{Q}}\left(q_1, q_2, \ldots, q_N\right)&=\frac{1}{\pi} \int_{-\infty}^{\infty} \int_{-\infty}^{\infty} \exp \left(-x_0^2-q_0^2\right)\\
                &\times\prod_{k=1}^N\left[1-Q_1\left(\frac{s_k}{\sqrt{b_k}}, \frac{\sqrt{q_k}}{\sqrt{b_k}}\right)\right] d x_0 d q_0 .
                \end{align}
                where \(s_k=\sigma_k \lambda_k\left[\left(m_{1, k}+x_0\right)^2+\left(m_{2, k}+q_0\right)^2\right]^{1 / 2}, b_k=\sigma_k^2 \frac{1-\lambda_k^2}{2}\).
                
                
                Beaulieu and Hemachandra \cite{beaulieuNovelSimpleRepresentations2011} study the correlated Nakagami-m distribution with a one-factorial underlying covariance matrix. The vector of squares follow a central complex multi-chi-square distribution: $\rv{Q}\sim\mathcal{C}\chi_K^2(M,\dmt{\Sigma},\dmt{0})$, where \(\dmt{\Sigma}=(\rho_{ij}\sigma_i\sigma_j)_{1\leq i,j \leq n}\), and \(\rho_{ij}= 1 \) if $i=j$ and $\lambda_i\lambda_j$ else. The joint PDF is given by:
                \begin{align}
                    &f_{\rv{Q}}(q_1,\ldots,q_M)=\dfrac{1}{2^M}\int_0^\infty \dfrac{t^{K-1}}{\Gamma(K)}e^{-t}\nonumber\\
                    &\times\prod_{\ell=1}^M\dfrac{1}{(\sigma_\ell^2\lambda_\ell^2t)^{\frac{K-1}{2}}}\dfrac{q_\ell^{\frac{M-1}{2}}}{\Omega_\ell^2}\exp\left(-\dfrac{q_\ell+\sigma_\ell^2\lambda_\ell^2t}{2\Omega_\ell^2}\right)I_{K-1}\left(\dfrac{\sqrt{q_\ell}\sqrt{t\sigma_\ell^2\lambda_\ell^2}}{\Omega_\ell^2}\right)dt\label{Beaulieu-Hemachandra:Nakagami-m}
                \end{align}
                where \(\Omega_\ell^2=\sigma_\ell^2(1-\lambda_\ell^2)/2\). The CDF is given by:
                \[F_{\rv{Q}}(q_1,\ldots,q_M)=\int_0^\infty \dfrac{t^{K-1}}{\Gamma(K)}e^{-t}\prod_{\ell=1}^M\left[1-Q_K\left(\dfrac{\sqrt{t}\sqrt{\sigma_\ell^2\lambda_\ell^2}}{\Omega_\ell},\dfrac{\sqrt{q_\ell}}{\Omega_\ell}\right)\right]dt\]
                    
                Hemachandra and Beaulieu \cite{hemachandraNovelRepresentationsEquicorrelated2011} study a particular type (constant correlation) of 1-factorial multi-chi-square. Using our notation, it corresponds to \(\dv{\mu}=\dv{e}\otimes\dv{\mu}_0\), \(\dv{e}\in\mathbb{R}^M\) is a vector of ones, \(\dv{\mu}_0\in\mathbb{R}^K\), \(\dmt{\Sigma}=\dmt{\tilde{\Sigma}}\otimes \dmt{I}_K\), \((\dmt{\tilde{\Sigma}})_{ii}=\sigma_i^2\), \((\dmt{\tilde{\Sigma}})_{ij}=\lambda^2\sigma_i^2\) for \(i\neq j\), and \(\dmt{A}_j=\dmt{E}_{jj}\otimes \dmt{I}_K\) where \(j=1,\ldots,M\). Equivalently, $\rv{Q}\sim \chi_K^2(M,\dmt{\tilde{\Sigma}},\|\dv{\mu}_0\|^2\dv{v}\dv{v}^T)$, where $\dv{v}=[\sigma_1,\ldots,\sigma_M]$.
            
                The PDF is given by \begin{align}
                    f_{\rv{Q}}(q_1,q_2,\ldots,q_M)&=\int_0^\infty \left(\dfrac{t}{s^2}\right)^{\frac{K-1}{4}}\exp(-(s^2+t^2))I_{\frac{K}{2}-1}(2s\sqrt{t})\nonumber\\
                    &\times\prod_{\ell=1}^M\dfrac{1}{2\Omega_\ell^2}\left(\dfrac{q_\ell}{\sigma_\ell^2\lambda^2t}\right)^{\dfrac{K-2}{4}}\nonumber\\
                    &\times \exp\left(-\dfrac{\sigma_\ell^2\lambda^2t+q_\ell}{2\Omega_\ell^2}\right)I_{\frac{K}{2}-1}\left(\dfrac{\sqrt{q_\ell}\sqrt{\sigma_\ell^2\lambda^2t}}{\Omega_\ell^2}\right)dt
                \end{align}
                
                where \(
                    s^2=\sum_{\ell=1}^K\left(\dfrac{{\dv{\mu}_0}_\ell}{\lambda}\right)^2,\;
                    \Omega_\ell^2=\dfrac{\sigma_\ell^2(1-\lambda_\ell^2)}{2},\;
                    \Delta_\ell^2=\dfrac{2\sigma_\ell^2\lambda^2t}{1-\lambda^2}
                \)
                The CDF is given by \UglyAlign{
                    F_{\rv{Q}}(q_1,q_2,\ldots,q_M)&=\int_0^\infty \left(\dfrac{t}{s^2}\right)^{\frac{K-1}{4}}\exp(-(s^2+t^2))I_{\frac{K}{2}-1}(2s\sqrt{t})\nonumber\\
                    &\times\prod_{\ell=1}^M \sum_{\nu=0}^\infty \exp\left(\dfrac{\Delta^2}{2}\right)\dfrac{(\Delta^2/2)^\nu}{\nu!}\nonumber\\
                    &\times \dfrac{\gamma(\nu+\frac{k}{2},(1-\lambda^2)q_\ell/2)}{\Gamma(\nu+p)}dt\\
                    &=\int_0^\infty \left(\dfrac{t}{s^2}\right)^{\frac{K-1}{4}}\exp(-(s^2+t^2))I_{\frac{K}{2}-1}(2s\sqrt{t})\nonumber\\
                    &\times\prod_{\ell=1}^M\left[1-Q_{\frac{K}{2}}\left(\dfrac{\sqrt{\sigma_\ell^2\lambda^2t}}{\Omega^\ell},\dfrac{\sqrt{r_\ell}}{\Omega_\ell}\right)\right]dt \label{HemachandraEvenK} 
                }
                where Equation \ref{HemachandraEvenK} holds for even \(K\).

                Beaulieu and Hemachandra \cite{beaulieuNovelSimpleRepresentations2011} study also the ``generalized Rician distribution'', with a one-factorial underlying covariance matrix. Let \(\rv{x}\sim \mathcal{CN}(\dv{\mu},\dmt{\Sigma})\) with \(\dmt{\Sigma}=\dmt{I}_K\otimes\dmt{\tilde{\Sigma}}\), \(\dmt{\tilde{\Sigma}}=(\rho_{ij}\sigma_i\sigma_j)_{1\leq i,j \leq n}\), \(\rho_{ij}= 1 \) if $i=j$ and $\lambda_i\lambda_j$ else, \(\dv{\mu}=[\dv{\tilde{\mu}_1}^T,\ldots,\dv{\tilde{\mu}}_K^T]^T\), and \(\dv{\tilde{\mu}}_\ell^T=(m_{1\ell}+m_{2\ell})[\sigma_1\lambda_1,\ldots,\sigma_n\lambda_n]^T\).
                Consider the $M$ quadratic forms \(\rs{Q}_j=\rv{x}^T\dmt{A}_j\rv{x}\) (hence $M=KN$) where \(\dmt{A}_j=\dmt{I}_K\otimes\dmt{\tilde{A}}_j\) and \(\dmt{\tilde{A}}_j\) is the diagonal matrix whose only non-zero entry is a unit at the \(j^{\text{th}}\) position. Equivalently, \(\rv{Q}\sim\mathcal{C}\chi_K^2(M,\dmt{\Sigma},\dmt{\Delta})\), with $\dmt{\Sigma}$ defined as in the previous case, and \(\dmt{\Delta}=(m_1^2+m_2^2)\dv{v}\dv{v}^T\), with \(\dv{v}^T=[\lambda_1\sigma_1,\ldots,\lambda_N\sigma_N]\). The joint PDF is given by
                \begin{align}
                   & f_{\rv{Q}}(q_1,\ldots,q_M)=\dfrac{1}{2^M}\int_0^\infty \dfrac{t^{\frac{K-1}{2}}}{S^{K-1}}e^{-t-S^2}I_{K-1}(2S\sqrt{t})\nonumber\\
                    &\times\prod_{\ell=1}^M\dfrac{1}{(\sigma_\ell^2\lambda_\ell^2t)^{\frac{K-1}{2}}}\dfrac{q_\ell^{\frac{M-1}{2}}}{\Omega_\ell^2}\exp\left(-\dfrac{q_\ell+\sigma_\ell^2\lambda_\ell^2t}{2\Omega_\ell^2}\right)I_{K-1}\left(\dfrac{\sqrt{q_\ell}\sqrt{t\sigma_\ell^2\lambda_\ell^2}}{\Omega_\ell^2}\right)dt\label{Beaulieu-Hemachandra:Generalized Rician}
                \end{align}
                
                where \(\Omega_\ell^2=\sigma_\ell^2(1-\lambda_\ell^2)/2\), \(S=\sum_{\ell=1}^K(m_{1\ell}^2+m_{2\ell}^2)\sum_{j=1}^M\sigma_j^2\lambda_j^2\) and the CDF is given by
                \[\begin{aligned}
                    F_{\rv{Q}}(q_1,\ldots,q_M)&=\int_0^\infty \dfrac{t^{\frac{K-1}{2}}}{S^{K-1}}e^{-t-S^2}I_{K-1}(2S\sqrt{t})\\
                    &\times\prod_{\ell=1}^M\left[1-Q_K\left(\dfrac{\sqrt{t}\sqrt{\sigma_\ell^2\lambda_\ell^2}}{\Omega_\ell},\dfrac{\sqrt{q_\ell}}{\Omega_\ell}\right)\right]dt
                \end{aligned}\]

                \paragraph{2-Factorial Matrices}
                Royen \cite{royenCENTRALNONCENTRALMULTIVARIATE1995} provides three formulae for the CDF of the central two-factorial multi-chi-square distribution. Let $\rv{Q}\sim\chi_M^2(2p,\dmt{R})$, with $\dmt{R}=\dmt{D}+\dmt{A}\dmt{S}^T$, where $\dmt{D}>0$ is diagonal and $\dmt{A}\in\mathbb{C}^{M\times 2}$. Define $\dmt{W}=\diag(w_1,\ldots,w_M)=\dmt{D}^{-1/2}$. Assume, without loss of generality, that the columns of $\dmt{B}=\dmt{W}\dmt{A}$ are orthogonal \footnote{Apply SVD to $\dmt{B}$.}. One of them is given by
                \begin{equation}
                    \begin{aligned}
                        &F_{\rv{Q}}(q_1,\ldots,q_M)=
                        \sum_{n=0}^\infty \sum_{(n)} c(n_1,\ldots,n_M)\prod_{j=1}^MF_{r_jn_j}(\frac{1}{2}w_j^2q_j)
                    \end{aligned}
                \end{equation}
                where \[c(n_1,\ldots,n_M)=\dfrac{1}{\Gamma(p)}\sum_{n_{.j}=n_j}\Gamma(p+\dfrac{1}{2}(n+\sum_{j=1}^M n_{jj}))\prod_{1\leq i \leq j \leq M} \dfrac{\gamma_{ij}^{n_{ij}}}{n_{ij}!} \] with \[\gamma_{ij}=\begin{cases}
                    c_{ij}^2-c_{ii}c_{jj} &\text{if } i\neq j\\
                    -c_jj & \text{else}
                \end{cases}\]

                \paragraph{Real 1-Factorial Matrices}
                
                Suppose that \(\dmt{R}\) is real one-factorial, but not one-factorial:\[\dmt{R}=\dmt{D}+\dv{a}\dv{a}^T,\] where $\dmt{D}$ is real diagonal but not positive. Then the CDF \cite{dickhausSurveyMultivariateChisquare2015} of \(\rv{Q}\sim\chi_K^2(M,\dmt{R},\dmt{0})\) is given by
                \[F(q_1,\ldots,q_M)=\lambda^{-p}\int_0^\infty \left[\prod_{j=1}^MG_{M}^*\left(\dfrac{q_j}{1-a_j^2},\dfrac{a_j^2}{1-a_j^2}\dfrac{t}{\lambda}\right)\right]g_{p}(t)dt\]
                with \(\lambda=1+\sum_{j=1}^M\dfrac{a_j^2}{1-a_j^2}\).

                \paragraph{General $m$-Factorial}
                Royen \cite{royenCENTRALNONCENTRALMULTIVARIATE1995} gives formulae for general $m$-factorial multi-chi-square distributions. However, these formulae entail Monte-Carlo simulations, which are beyond the scope of our monograph. 

                Khammassi et al. \cite{khammassiNewAnalyticalApproximation2022} study the multivariate Rayleigh distribution with an underlying real covariance matrix. Let \(\rv{x}=[\rs{x}_1,\ldots,\rs{x}_N]^T\sim \mathcal{CN}(\dv{0},\dmt{\Sigma})\) with \(\dmt{\Sigma}\in \mathbb{R}^{N\times N}\), and, without loss of generality, \((\Sigma)_{k,k}=\sigma^2\) for all $k$. Let \(\dmt{\Sigma}=\dmt{U}\dmt{\Lambda}\dmt{U}^T\) be the eigendecomposition of \(\dmt{\Sigma}\), with \(\dmt{\Lambda}=\diag(\lambda_1,\ldots,\lambda_N)\) and \(\lambda_1\geq \lambda_2\geq \lambda_N\). Write \(\dmt{U}=[\dv{u}_1,\ldots,\dv{u}_N]\) and \(\dv{u}_\ell=[u_{1,\ell},\ldots,u_{N,\ell}]^T\). Then the random variable \(\rs{x}_k\) is distributed as \(\sum_{\ell=1}^N\sqrt{\lambda_\ell}u_{k,\ell}(\rs{a}_\ell+i\rs{b}_\ell)\), \(k=1,\ldots,N\), where the variables \(\rs{a}_\ell+i\rs{b}_\ell\) are independent and follow \(\mathcal{CN}(0,1)\). Consider the quadratic forms \(\rs{Q}_i=\rv{x}^H\dmt{A}_j\rv{x}\) for $j=1,\ldots,N$ $(M=N)$ where \(\dmt{A}_j\) is the diagonal matrix whose only nonzero entry is a unit on the $j^{\text{th}}$ position, i.e., \(\rs{Q}_j=|\rs{x}_j|^2\). 

                The authors' approach is to approximate the vector \(\rv{x}\) by a vector \(\rv{\hat{x}}=\rv{\hat{x}}(\epsilon)\) which converges in distribution to \(\rv{x}\) as \(\epsilon\) tends top zero. Let \begin{equation}
                    \hat{\rs{x}}_k=\sqrt{\sigma^2-\sum_{\ell=1}^{\epsilon\text{-rank}}\lambda_\ell u_{k,\ell}^2}(\rs{x}_k+i\rs{y}_k)+\sum_{\ell=1}^{\epsilon\text{-rank}}\sqrt{\lambda_\ell}u_{k,\ell}(\rs{a}_\ell+i\rs{b}_\ell) \label{eq:Khammassi-Approx}
                \end{equation}
                with \(k=1,\ldots,N\), \((\rs{a}_\ell+i\rs{b}_\ell)\), \((\rs{x}_\ell+i\rs{y}_\ell)\) are independent and follow \(\mathcal{CN}(0,1)\), and \(\epsilon\text{-rank}=\min\{n:\lambda_n\leq \epsilon\}\). Let $\hat{\rs{Q}}_j=\rv{\hat{x}}^H\dmt{A}_j\rv{\hat{x}}$. Then its CDF is given by \eqref{eq:Khammassi-Formula}
                \begin{equation} 
                    \begin{aligned}\scriptscriptstyle &F_{\hat{\rs{Q}}_1,\ldots,\hat{\rs{Q}}_N}(\hat{q}_1,\ldots,\hat{q}_N)=\int\ldots\int \prod_{\ell=1}^{\epsilon\text{-rank}}\dfrac{1}{\pi}\exp(-(a_\ell^2+b_\ell^2))\nonumber\\
                        &\times \prod_{k=1}^N\left(1-Q_1\left(\dfrac{\sqrt{2\left(\displaystyle\sum_{\ell=1}^{\epsilon\text{-rank}}\sqrt{\lambda_\ell u_{l,\ell}a_\ell}\right)^2+2\left(\displaystyle\sum_{\ell=1}^{\epsilon\text{-rank}}\sqrt{\lambda_\ell u_{l,\ell}b_\ell}\right)^2}}{\sigma_{\epsilon,k}},\dfrac{\sqrt{2\hat{q}_k}}{\sigma_{\epsilon,k}}\right)\right)\nonumber\\
                        &da_1\ldots da_{{\epsilon\text{-rank}}}db_1\ldots db_{{\epsilon\text{-rank}}}
                    \end{aligned} 
                \label{eq:Khammassi-Formula}
                \end{equation}
                where \begin{align*}
                    \sigma_{\epsilon,k}&=\sqrt{\sigma^2-\sum_{\ell=1}^{\epsilon-\rank}s_\ell u_{k,\ell}^2}.
                \end{align*}
                The approximate random vector $\rv{\hat{Q}}$ has an $\epsilon-\text{rank}$-factorial covariance matrix. In fact, it is given by \[\dfrac{1}{\sigma^2}\operatorname{Cov}(\rv{\hat{Q}})=\dmt{D}+\dmt{A}\dmt{A}^T\] where $\dmt{A}=[\sqrt{\lambda_1}\dmt{u}_1,\ldots,\sqrt{\lambda_{\epsilon-\text{rank}}}\dmt{u}_{\epsilon-\text{rank}}]$. So the authors effectively approximate the original covariance matrix by an approximate $\epsilon-\text{rank}$-factorial one.


    \section{Multi-Chi-Squares PDF/CDF: General Methods}
        \paragraph{Useful Functions}
        Here is a list of typically used functions by Royen. This is not meant to cover the relations between these functions nor to provide the best computational methods. For further information, refer to the cited Royen's paper.
        
        \begin{itemize}
            \item \(g_\alpha(x)=\dfrac{x^{\alpha-1}e^{-x}}{\Gamma(\alpha)}\)
            \item \(\begin{aligned}[t]
                G_\alpha(x)&=\int_0^xg_\alpha(t)dt=\dfrac{\gamma(\alpha,x)}{\Gamma(\alpha)}\\
                &=\begin{cases}
            1-e^{-x}\sum_{j=0}^{\alpha-1}\dfrac{x^j}{j!}&2\alpha \text{ is even}\\
            \erf(x^{1/2})-e^{-x}\sum_{j=1}^{\alpha-1/2}\dfrac{x^{j-1/2}}{\Gamma(j+\frac{1}{2})}&2\alpha \text{ is odd}
        \end{cases}
            \end{aligned}\)
            \item \(G_{\alpha+n}^{(n)}(x)=\dfrac{d^n}{dx^n}G_{\alpha+n}(x)=\binom{\alpha+n-1}{n-1}^{-1}L^{(\alpha)}_{n-1}(x)g_{\alpha+1}(x)\)
            \item \(\begin{aligned}[t]
                G_\alpha(x,y)&=e^{-y}\displaystyle\sum_{n\geq 0}G_{\alpha+n}(x)\dfrac{y^n}{n!}
            \end{aligned}\)
            \item \(G_\alpha^*(x,y)=e^yG_\alpha(x,y)\)
            \item \(\begin{aligned}[t]
                \mathcal{G}_\alpha(x,y)&=\displaystyle\sum_{n\geq0}G_{\alpha+n}(x)y^n\\
                &=\begin{cases}
                \dfrac{1}{y-1}(G_\alpha(x)-y^{1-\alpha}e^{x(y-1)}G_\alpha(xy)), &y\neq 1,\alpha >0\\
                xg_\alpha(x)+(1-\alpha+x)G_\alpha(x) &y=1.
            \end{cases}
            \end{aligned}\) \footnote{Note that the efficient computation of this function for a nonreal second argument is an open problem \ref{op:gamma-complex}}
            
            \item \(H_{\alpha,n}(x)=\displaystyle\sum_{m\geq 0} (-1)^{n-m}\binom{n}{m}2^mG_{\alpha+m}(x)\)
            \item \(H_{\alpha,n}(x,y)=e^{-y}\displaystyle\sum_{m\geq 0}H_{\alpha+m,n}(x)\dfrac{y^m}{m!}\)
        \end{itemize}


        \subsection{Royen's General Method}
        Royen \cite{royenIntegralRepresentationsConvolutions2007} gives 3 $(M-1)$-variate integral representations of the CDF. Note that in this paper, he generalizes the distribution to account for scaled chi-squares. We'll report the first one. The CDF \(F_{\rv{Q}}(q_1,\ldots,q_M)\) is given by the following \[\dfrac{1}{2\pi^{M-1}}\int_{\mathcal{C}_{M-1}}\mathcal{G}_{p}\left(\dfrac{w_M^2q_M}{1-z_0},\dfrac{y}{1-z_0}\right)K_{p}\prod_{j=1}^{M-1}\dfrac{1}{1-z_j}G_{p}\left(w_j^2q_j,\dfrac{z_j}{z_j-1}\right)d\phi_j\] where \(\mathcal{C}_{M-1}=(-\pi,\pi]^{M-1}\), \(\dmt{W}=\diag(w_1,\ldots,w_M)\) is an arbitrary scale matrix, \(z_j=r\mathrm{e}^{\phi_j}\) with \(r\) being an arbitrary positive number, and \UglyAlignS{\dmt{B}&=\begin{bmatrix}
            \dmt{B}_{MM} &\dv{b}_M\\
            \dv{b}_M^T &b_{MM}
        \end{bmatrix}=2\dmt{W\Sigma W}-\dmt{I}_M\\
        \dmt{D}&=\begin{bmatrix}
            \dmt{D}_{MM} &\dv{d}_M\\
            \dv{d}^M &d_{MM}
        \end{bmatrix}=4\dmt{W\Delta\Sigma W}\\
        z_0=z_0(z_1,\ldots,z_{M-1})&=\dv{b}_M^T(\dmt{Z}_{MM}+\dmt{B}_{MM}^{-1}\dv{b}_M)-b_{MM}\\
        \dmt{Z}&=\diag(z_1,\ldots,z_M)=\begin{bmatrix}
            \dmt{Z}_{MM} &\dv{z}_M\\
            \dv{z}_M^T &z_{MM}
        \end{bmatrix}\\
        y=y(z_1,\ldots,z_{M-1})&=\\
        [\dv{b}_M^T(\dmt{Z}_{MM}+\dmt{B}_{MM})^{-1}&,-1]\dmt{D}[\dv{b}_M^T(\dmt{Z}_{MM}+\dmt{B}_{MM})^{-1},-1]^T\\
        K_{p}=K_{p}(z_1,\ldots,z_{M-1})&=\etr\left(-(\dmt{Z}_{MM}+\dmt{B}_{MM})^{-1}\right)|\dmt{I}+\dmt{B}_{MM}\dmt{Z}_{MM}^{-1}|^{-p}
        }
        Note that \(\dmt{W}\) and \(r\) must verify \(\|\dmt{B}\|<r<1\).
        
        Royen (2016) gives \cite{royenNonCentralMultivariateChiSquare2016} $(M)$- and $(M-1)$-variate integral representation of the CDF over a bounded region. A covariance matrix \(\dmt{\Sigma}\) is used instead of a correlation matrix \(\dmt{R}\). The first one is given by what follows
        \begin{align}
            F_{\rv{Q}}(q_1\ldots q_M)&=|v\dmt{\Sigma}|^{-2K}\etr(-\dmt{\Delta}\dmt{\Sigma}^{-1})\nonumber\\
            &\times(2\pi)^{-M}\int_{\mathcal{C}_M}\dfrac{\etr(\dmt{Z}^{-1}(\dmt{I}_M+\dmt{B}\dmt{Z}^{-1}\dmt{D})}{|\dmt{I}_M+\dmt{B}\dmt{Z}^{-1}|^{M}}\prod_{j=1}^{2K}\mathcal{G}_{2K}(vq_j,z_j)d\varphi_j \label{Royen Non-Central Chi-Square}
        \end{align}
        where \(\dmt{B}=2(v\dmt{\Sigma})^{-1}-\dmt{I}_M\), \(\dmt{D}=2v^{-1}\dmt{\Sigma}^{-1}\dmt{\Delta}\dmt{\Sigma}^{-1}\), $v$ is any scale factor satisfying \(|\dmt{B}\|\leq 1\), \(\dmt{Z}=\diag(z_1,\ldots,z_M)\), \(z_j=re^{i\varphi_j},r>\|\dmt{B}\|\), and \(\mathcal{C}_M=(-\pi,\pi]^M\). The formula in \eqref{Royen Non-Central Chi-Square} is derived by manipulating \eqref{Royen Non-Central Chi-Square} into series and manipulating the Laplace transform of multi-chi-square into a similar series and showing that they form a Laplace pair. Of course, coming up with a formula like \eqref{Royen Non-Central Chi-Square} is not an intuitive task. {At this point, Royen had been working on multi-chi-squares for 25 years.}
        
            Royen further attempts to reduce the number of integrals from M to M-1. Actually, this is only useful for small M, so the reduction by one dimension is helpful. The \((M-1)\)-variate integral representation is given by
        \begin{align}
            &F_{\rv{Q}}(q_1,\ldots,q_M)=\dfrac{\etr(-\dmt{\Delta}\dmt{\Sigma}^{-1})}{|\dmt{W}\dmt{\Sigma}\dmt{W}|^{2K}}\dfrac{1}{(2\pi)^{M-1}}\nonumber\\
            &\times\int_{\mathcal{C}_{M-1}}\dfrac{L_{2K}(z_1,\ldots,z_{M-1})}{(1-z)^{2K}}G_{2K}^*\left((1-z_0)w_Mq_M,(1-z_0)^{-1}y\right)\nonumber\\
            &\times\prod_{j=1}^{M-1}\mathcal{G}_{2K}(w_j^2q_j,z_j)d\varphi_j
            \end{align}
        where
        \begin{equation*}
            \begin{aligned}
                L_\alpha(z_1,\ldots,z_{M-1})&=\etr\left(\dmt{Z}_{MM}^{-1}(\dmt{I}_{M-1}+\dmt{B}_{MM}\dmt{Z}_{MM}^{-1})^{-1}\dmt{D}_{MM}\right)\\&\times|\dmt{I}_{M-1}+\dmt{B}_{MM}\dmt{Y}_{MM}^{-1}|^\alpha\\
                \dmt{B}&=2(\dmt{W}\dmt{\Sigma}\dmt{W})^{-1}-\dmt{I}_M=\begin{pmatrix}
            \dmt{B}_{MM} &\dv{b}_M\\
            \dv{b}_M^T &b_{MM}
            \end{pmatrix} \\
            &\text{with } \dmt{B}_{MM}\in\mathbb{R}^{(M-1)\times(M-1)}\\        \dmt{D}&=2\dmt{W}^{-1}\dmt{\Sigma}^{-1}\dmt{\Delta}\dmt{\Sigma}^{-1}\dmt{W}^{-1}=\begin{pmatrix}
            \dmt{D}_{MM} &\dv{d}_M\\
            \dv{d}_M^T &d_{MM}
            \end{pmatrix}\\
            &\text{with } \dmt{D}_{MM}\in\mathbb{R}^{(M-1)\times(M-1)}\\
            \dmt{W} &\text{ is a scale diagonal matrix.}\\
            \dmt{Z}_{MM}&=\diag(z_1,\ldots,z_{M-1}),z_j=re^{i\varphi_j},r>\|\dmt{B}\|\\
            y=y(z_1,\ldots,z_{M-1})&=\\
            [\dv{b}_M^T(\dmt{Z}_{MM}+\dmt{B}_{MM})^{-1}&,-1]\dmt{D}[\dv{b}_M^T(\dmt{Z}_{MM}+\dmt{B}_{MM})^{-1},-1]^T\\
            z_0&=z_0(z_1,\ldots,z_{M-1})=\dv{b}_M^T(\dmt{Z}_{MM}+\dmt{B}_{MM})^{-1}\dv{b}_M-b_{MM}\\
            \end{aligned}
        \end{equation*}
        
        Note that Royen's 2007 result \cite{royenIntegralRepresentationsConvolutions2007} is slightly different. The matrix $\dmt{D}$ is different and not symmetric. Moreover, Royen gives in 1995 \cite{royenCENTRALNONCENTRALMULTIVARIATE1995} a general expression for the non-central multi-chi-square. This is based on the m-factorial factorization of a matrix, but it depends on a Monte Carlo simulation. As this is outside the scope of our survey, it is not covered. However, corollaries that work for specific cases will be reported.

        \subsection{Complex Multi-Chi-Square} 
        
        \paragraph{Central Complex Multi-Chi-Square} \label{par:Morales}
        Morales and Jimenez \cite{morales-jimenezDiagonalDistributionComplex2011} derive the following CDF in a similar manner to Royen's \cite{royenExpansionsMultivariateChisquare1991}. For $\rv{Q}\sim\mathcal{C}\chi_K^2(M,\dmt{\Sigma},\dmt{0})$, we have \begin{equation}
            F_{\rs{Q}_1,\ldots,\rs{Q}_M}(q_1,\ldots,q_M)=\sum_{n=0}^\infty \sum_{(n)} c(n_1,\ldots,n_M)\prod_{j=1}^M\Delta_{n_j}^K (w_j^2 q_j) \label{eq:Morales}
        \end{equation}
        where $w_j$ are scale factors chosen so that \(\|\dmt{I}-\dmt{WRW}\|_2<1\), with \(\dmt{W}=\diag(w_1,\ldots,w_M)\), the coefficients $c(n_1,\ldots,n_M)$ are given by \[\sum_{(n)}c(n_1,\ldots,n_M)\prod_{j=1}^M u_j^{n_j}=\sum_{\ell_1+2\ell_2+\ldots+k\ell_M\leq n}\dfrac{\Gamma(p+\ell_1+\ldots+\ell_M)}{\Gamma(p)}\prod_{j=1}^M\dfrac{(-D_j)^{\ell_j}}{\ell_j!}\]
        where \(u_j\) are indeterminates, \(D_j=(-1)^j\sum_{|S|=j}|\dmt{A}_S|\prod_{m\in S}u_m\),  \(S\) is a non-empty subset of $\{1,2,\ldots,M\}$, \(|S|\) is its cardinal, and $\dmt{A}_S$ is the principal submatrix of \(\dmt{A}=\dmt{I}-\dmt{WRW}\) with rows and columns from the set $S$, and
        \begin{equation}
            \begin{aligned}
                \Delta_n^K(x)&=\begin{cases}
            G_K(x), &\text{ if }n=0\\
            g_{K+n}^{(n-1)}(x)=\dfrac{(n-1)!}{(K+n-1)!}e^{-x}x^K L_{n-1}^K(x), &\text{ if }n>0
            \end{cases}
            \end{aligned}
        \end{equation}

        \paragraph{Non-Central Complex Multi-Chi-Squares}

        Royen \cite{royenNonCentralMultivariateChiSquare2016} extends his CDF formulas to the complex case. The $M$-fold integral for \[Q\sim\mathcal{C}\chi_K(M,\dmt{\Sigma},\dmt{\Delta})\] becomes \begin{align}
            F_{\rv{Q}}(q_1\ldots q_M)&=|v\dmt{\Sigma}|^{-K}\etr(-\dmt{\Delta}\dmt{\Sigma}^{-1})\nonumber\\
            &\times(2\pi)^{-M}\int_{\mathcal{C}_M}\dfrac{\etr(\dmt{Z}^{-1}(\dmt{I}_M+\dmt{B}\dmt{Z}^{-1}\dmt{D})}{|\dmt{I}_M+\dmt{B}\dmt{Z}^{-1}|^{K}}\prod_{j=1}^{M}\mathcal{G}_{2K}(vq_j,z_j)d\varphi_j 
        \end{align}
        where \(\dmt{B}=2(v\dmt{\Sigma})^{-1}-\dmt{I}_M\), \(\dmt{D}=2v^{-1}\dmt{\Sigma}^{-1}\dmt{\Delta}\dmt{\Sigma}^{-1}\), $v$ is any scale factor satisfying \(\|\dmt{B}\|\leq 1\), \(\dmt{Z}=\diag(z_1,\ldots,z_M)\), \(z_j=re^{i\varphi_j},r>\|\dmt{B}\|\), and \(\mathcal{C}_M=(-\pi,\pi]^M\).

        The $(M-1)$-fold integral can be extended to \begin{align}
            &F_{\rv{Q}}(q_1,\ldots,q_M)=\dfrac{\etr(-\dmt{\Delta}\dmt{\Sigma}^{-1})}{|\dmt{W}\dmt{\Sigma}\dmt{W}|^{K}}\dfrac{1}{(2\pi)^{M-1}}\nonumber\\
            &\times\int_{\mathcal{C}_{M-1}}\dfrac{L_{2K}(z_1,\ldots,z_{M-1})}{(1-z)^{K}}G_{K}^*\left((1-z_0)w_Mq_M,(1-z_0)^{-1}y\right)\nonumber\\
            &\times\prod_{j=1}^{M-1}\mathcal{G}_{K}(w_j^2q_j,z_j)d\varphi_j
            \end{align}
        where
        \begin{equation*}
            \begin{aligned}
                L_\alpha(z_1,\ldots,z_{M-1})&=\etr\left(\dmt{Z}_{MM}^{-1}(\dmt{I}_{M-1}+\dmt{B}_{MM}\dmt{Z}_{MM}^{-1})^{-1}\dmt{D}_{MM}\right)\\&\times|\dmt{I}_{M-1}+\dmt{B}_{MM}\dmt{Y}_{MM}^{-1}|^K\\
                \dmt{B}&=2(\dmt{W}\dmt{\Sigma}\dmt{W})^{-1}-\dmt{I}_M=\begin{pmatrix}
            \dmt{B}_{MM} &\dv{b}_M\\
            \dv{b}_M^H &b_{MM}
            \end{pmatrix} \\
            &\text{with } \dmt{B}_{MM}\in\mathbb{R}^{(M-1)\times(M-1)}\\        \dmt{D}&=2\dmt{W}^{-1}\dmt{\Sigma}^{-1}\dmt{\Delta}\dmt{\Sigma}^{-1}\dmt{W}^{-1}=\begin{pmatrix}
            \dmt{D}_{MM} &\dv{d}_M\\
            \dv{d}_M^H &d_{MM}
            \end{pmatrix}\\
            &\text{with } \dmt{D}_{MM}\in\mathbb{R}^{(M-1)\times(M-1)}\\
            \dmt{W} &\text{ is a scale diagonal matrix.}\\
            \dmt{Z}_{MM}&=\diag(z_1,\ldots,z_{M-1}),z_j=re^{i\varphi_j},r>\|\dmt{B}\|\\
            y=y(z_1,\ldots,z_{M-1})&=\\
            [\dv{b}_M^H(\dmt{Z}_{MM}+\dmt{B}_{MM})^{-1}&,-1]\dmt{D}[\dv{b}_M^T(\dmt{Z}_{MM}+\dmt{B}_{MM})^{-1},-1]^T\\
            z_0&=z_0(z_1,\ldots,z_{M-1})=\dv{b}_M^H(\dmt{Z}_{MM}+\dmt{B}_{MM})^{-1}\dv{b}_M-b_{MM}\\
            \end{aligned}
        \end{equation*}
                

    \section{Royen's Convolution}

    Royen gives three types of integral representation for the CDF of the convolution of the original distributions whose Laplace transform is given by
    \[\hat{f}(t_1,\ldots,t_M)=|\dmt{I}_M+2\dmt{\Sigma} \dmt{T}|^{-p}\etr(-4\dmt{\Sigma} \dmt{T}(\dmt{I}+2\dmt{\Sigma} \dmt{T})^{-1}\dmt{\Delta})\]
    where $\dmt{T}=\diag(t_1,\ldots,t_M)$,$t_1,\ldots,t_M\geq 0$. 
    
    The Laplace transform of the convolution is given by
    \[\hat{f}(t_1,\ldots,t_M)=\prod_{k=1}^n|\dmt{I}_M+2\dmt{\Sigma}_k \dmt{T}|^{-K_k/2}\etr(-4\dmt{\Sigma}_k \dmt{T}(\dmt{I}+2\dmt{\Sigma}_k \dmt{T})^{-1}\dmt{\Delta}_k)\]

    In terms of quadratic forms, the studied distribution is that of \(\rs{Q}_j=\rv{x}^T\dmt{A}_j\rv{x}\), where \(\rv{x}\sim\mathcal{N}_{MK}(\dv{\mu},\dmt{\Sigma})\), with \(\dv{\mu}=\begin{bmatrix}
        \dv{\tilde{\mu}}^{(1)}\\
        \vdots\\
        \dv{\tilde{\mu}}^{(n)}
    \end{bmatrix}\), \(\dv{\tilde{\mu}^{(k)}}=\begin{bmatrix}
        \dv{\mu}_1^{(k)}\\
        \vdots\\
        \dv{\mu}_{K_k}^{(k)}
    \end{bmatrix}\), \(K=\sum_{k=1}^n K_k\), \(\dmt{\Sigma}=\diag(\dv{\tilde{\Sigma}}^{(1)},\ldots,\dmt{\tilde{\Sigma}}^{(n)})\), \(\dmt{\tilde{\Sigma}}^{(k)}=\dmt{I}_{K_k}\otimes\dmt{\Sigma}^{(k)}\), \(\dmt{A}_j=\diag(\dmt{A}_j^{(1)},\ldots,\dmt{A}_j^{(n)})\), \(\dmt{A}_j^{(k)}=\dmt{I}_{K_k}\otimes\dmt{E}_{jj}\). Evidently, the convolution is identically distributed to a multi-chi-square if and only if all covariance matrices $\{\dmt{\Sigma}_k\}_{k=1}^n$ are identical.
        
    
    
    Royen \cite{royenIntegralRepresentationsConvolutions2007} provides three expressions for the the CDF, one of which is given by
    \begin{equation}
        \begin{aligned}
            F_{\rv{Q}}(q_1,\ldots,q_M)&=(2\pi)^{-M}\int_{C_M}\prod_{k=1}^n \etr(- (\dmt{Z}+\dmt{B}_k)^{-1}\dmt{D}_k)|\dmt{I}+\dmt{B}_k\dmt{Z}^{-1}|^{-\alpha_k} \\
            &\times\prod_{j=1}^M F_{p}(\nu q_j,z_j)d\varphi_j
        \end{aligned}
    \end{equation} where \begin{gather*}
        \mathcal{C}_M=[-\pi,\pi]^M\\
        F_\alpha(x,y)=\dfrac{1}{1-y}G_\alpha(x,\frac{y}{y-1})\\
        \dmt{B}_k=2\nu \dmt{\Sigma}_k-\dmt{I},\; \nu \text{ is arbitrarily chosen so that } \nu<1/\max \|\dmt{\Sigma}_k\|\\
        \dmt{D}_k=2\dmt{\Delta}_k(\dmt{I}+\dmt{B}_k)\\
        \dmt{Z}=\diag(z_1,\ldots,z_M),\;z_j=re^{i\varphi_j},\; r \text{ is chosen in } (\max \|\dmt{B}_k\|,1)\\
    \end{gather*}

    \paragraph{Royen's Formulae} were derived in different cases of multi-chi-squares.  A common element in the formulae is some form of the gamma function. With the exception of the one-factorial cases, the formulae are computationally inefficient, even for reasonably small $M$ values. There are two main reasons: the complexity of the formulae, and the complexity of computing Royen's gamma functions. Royen uses conflicting notations to describe the variety of gamma function he has used. 

        From \cite[Eq. 8]{morales-jimenezDiagonalDistributionComplex2011} (Gamma CDF and PDF):
    \[
    G_{k}(x) = \gamma(k, x) / \Gamma(k), \quad g_k(x) = \frac{d}{dx}G_{k}(x)
    \]
    and
    \[
    G^{(n)}_{k + n}(x) = g^{(n-1)}_{k + n}(x) = \frac{d^n}{dx^n}G_{k + n}(x)
    \]
    From \cite[Pg. 2, Para. 3]{royenExpansionsMultivariateChisquare1991} a similar formula is defined as $G_{p+n}^{(n)}(x)$
    From \cite[Eq. 3.1 as H extra factor $e^{-x}$]{royenMultivariateGammaDistributions1991} and \cite[Eq. 2.4, G not H]{royenMultivariateGammadistributionsConnected1994}:
    \[
    H_{k}(x, y) = e^{-y}\sum_{n = 0}^\infty G_{k + n}(x) y^n/n!
    \]
    Finally, in \cite[Eq. 2.14]{royenNonCentralMultivariateChiSquare2016} the caligraphic G $\mathcal{G}$ functions is the similar to $H$ but has differences: 
    \[
    \mathcal{G}_p(x, y) = \sum_{n = 0}^\infty G_{k + n}(x) y^n
    \]
    
    

    \section{General Multiforms}
        Shephard \cite{shephardCharacteristicFunctionDistribution1991} gives a method for the inversion of a multivariate characteristic function generalizing Gil-Palaez's \cite{gil-pelaezNoteInversionTheorem1951} method that is used by Imhof \cite{imhofComputingDistributionQuadratic1961}. 

        As a quadratic multiform satisfies the conditions of his theorem, i.e., it has a mean and its density function is Lebesgue integrable, the method can be used to recursively compute the CDF as follows 
       \[
       \resizebox{0.95\textwidth}{!}{$ 
       \frac{(-1)^M}{(\pi)^M} \lim _{n \rightarrow \infty} \int_0^n \cdots \int_0^n \prod_{j=1}^M\left[1-\frac{t_j}{n}\right]_{t_1} \Delta_{t_2} \cdots \Delta_{t_M}\left[\frac{\varphi(t) e^{-\rv{q}^T \dv{t}}}{i t_1 i t_2 \cdots i t_M}\right] d \dv{t}=u^{\dagger}(\rv{q})
       $}
       \]
        where $\phi(\dv{t})=\phi(t_1,\ldots,t_M)=M_{\rv{q}}(i\dv{t})$ is the characteristic function, \(\Delta \eta_t(t)=\eta(t)+\eta(-t)\) for any function \(\eta\) of \(t\), and \[\begin{aligned}
        u^{\dagger}(\rv{q})= & 2^M F\left(q_1, \ldots, q_M\right) \\
        & -2^{M-1}\left[F\left(q_2, q_3, \ldots, q_M\right)+\cdots+F\left(q_1, \ldots, q_{M-2}, q_{M-1}\right)\right] \\
        & +2^{M-2}\left[F\left(q_3, q_4, \ldots, q_M\right)+\cdots+F\left(q_1, \ldots, q_{M-3}, q_{M-2}\right)\right] \\
        & +\cdots+(-1)^M.
        \end{aligned}\]
        Hence the CDF of $M$ quadratic forms is computed recursively using all the CDFs of $(m-1)$ combinations of $m$ forms for every $m<N$, computing an $m$-multiple integral at each step.  
        Lasserre \cite{lasserreComputingGaussianExponential2017} provides a computational method that calculates the Gaussian measure of semi-algebraic sets; it can be consequently used to compute the CDF of any quadratic multiform. However, it comes with computational drawbacks.
    
        A semi-algebraic set is a finite intersection of sets defined by polynomial inequalities; a set \(B\) is semi-algebraic if it can be written in the form \[B=\{\dv{x}\in\mathbb{R}^N:g_j(\dv{x})\leq 0,j=1,2,\ldots, M\}\]
        for (multivariate) polynomials \(g_j(\dv{x})\). To estimate its probability content when \(\rv{x}\sim \mathcal{N}_N(\dv{0},\dmt{I}_N)\), Lasserre provides two pairs of sequences that sandwich and converge to the CDF. The first pair \(\underline{o}_d,\bar{o}_d\) which satisfies, as mentioned, \(\underline{o}_d\leq F_{\rv{q}}(q_1,\ldots,q_M) \leq \bar{o}_d\) and \(\displaystyle\lim_{d\to\infty}\underline{o}_d=\displaystyle\lim_{d\to\infty}\bar{o}_d=F_{\rv{q}}(q_1,\ldots,q_M)\), is given by 
        \[
        \resizebox{\textwidth}{!}{$ 
        \bar{o}_d=\sup_{\dv{u},\dv{v}\in\mathbb{R}[\dv{x}]^*_{2d}}\left\{u_{\dv{0}} \text{ s.t. } u_\alpha+v_\alpha=g_\alpha;\forall \alpha\in\mathbb{N}_{2d}^N;\right.
        \left.\dmt{M}_d(\dv{u}),\dmt{M}_d(\dv{v})\geq 0;\dmt{M}_{d-d_j}(g_j\dv{u})\geq 0,j=1,\ldots,M\right\}
        $}
        \] 
        where
        \[
        \resizebox{0.95\textwidth}{!}{$ 
        \begin{aligned}
            \mathbb{R}[\dv{x}]_{2d} &= \text{ space of polynomials in } N \text{ variables with total degree at most } 2d\\
            \mathbb{N}_{2d}^N&=\{(m_1,\ldots,m_N)\in\mathbb{N}^N:m_1+\ldots+m_N\leq 2d\}\\
            \dmt{M}_d(\dv{z})&=(z_{\dv{\alpha}+\dv{\beta}})_{\dv{\alpha},\dv{\beta}\in\mathbb{N}_{d}^N}\\
            \dmt{M}_{d-d_j}(g_j\dv{z})&=\left(\sum_{\dv{\gamma}} (g_j)_{\dv{\gamma}} z_{\dv{\alpha}+\dv{\beta}+\dv{\gamma}}\right)_{\alpha,\beta\in\mathbb{N}_{d}^N}\\
        \end{aligned}
        $}
        \]
        where \((g_j)_\gamma\) are the coefficients of \(g_j\), i.e., \(g_j(\dv{x})=\sum_{\dv{\gamma}}(g_j)_\gamma \dv{x}^{\dv{\gamma}}\). Now \(\underline{o}_d\) can be calculated from the upper bound on the complement set, which is a disjoint union of semi-algebraic sets, \(\mathbb{R}^N- B=\bigcup_{i=1}^MB_i\), where \[B_i=\left[\bigcap_{j=1}^{i-1}\{\dv{x}:g_j(\dv{x})\leq 0\}\right]\cap\{g_i(\dv{x})>0\}\] with the convention of the empty intersection being the whole \(\mathbb{R}^N\). Then \(\underline(o)_d=1-\bar{o}_d(B_i)\)
    
        For a multiform \(\rs{Q}_j=\rv{x}^T\dmt{A}_j\rv{x}+\dv{b}_j^T\rv{x}+c_j\), \(j=1,\ldots, M\), \(\rv{x}\sim \mathcal{N}_N(\dv{\mu},\dmt{\Sigma})\), computing the CDF \(F_{\rv{q}}(q_1,\ldots,q_M)\) is equivalent to computing the (standard) Gaussian measure of the set \(B=\{\rv{x}\in\mathbb{R}^N:g_j(\rv{x})\leq 0,j=1,2,\ldots, M\}\) for \(-\rv{x}^T\dmt{\Lambda}_j\rv{x}-\dv{\tilde{b}}_j^T\rv{x}+\tilde{c}_j\), where \(\dmt{\Sigma}^{1/2}\dmt{A}_j\dmt{\Sigma}^{1/2}=\dmt{P}_j^T\dmt{\Lambda}_j\dmt{P}_j\) is an orthogonal diagonalization, \(\dv{\tilde{b}}_j=\dmt{P}_j\dmt{\Sigma}^{1/2}\dv{b}_j+2\dmt{P}_j\dmt{\Sigma}^{1/2}\dmt{A}_j\dv{\mu}\) and \(\tilde{c}_j=q_j-\dv{\mu}^T\dmt{A}_j\dv{\mu}-\dv{b}_j^T\dv{\mu}-c_j\).
    
        A pair of faster sequences \(\underline{\omega}_d,\bar{\omega}_d\) is obtained after strengthening \(\bar{o}_d\) by adding the constraints \(L_{\dv{u}}(p_{i,\dv{\alpha}})=0,\dv{\alpha}\in\mathbb{N}^N_{r(d)},i=1,\ldots,N\), where \(L_{\dv{u}}(g)=\sum_{\dv{\gamma}}g_{\dv{\gamma}}\dv{u}^{\gamma}\) and \(r(d)=2d-1\).

        \section{Literature Summary and Statistics}

            We will summarize on single forms and ratios using a bunch of visual aids. We start by classifying the forms discussed in the literature.
            \begin{figure}[ht!]
                \begin{center}
\begin{tikzpicture}[scale=1.35,x=1cm, y=1cm]
    \draw[color=black,thick] (0,0) rectangle (7.5,6);
    \Text[color=black,x=1.5,y=6, style={fill=white}]{Multi-variate}

    \draw[color=black,thick] (5,5.75) rectangle (7.0, 3.15);    
    \Text[color=black,x=5.15, y=5.75, style={fill=white}, anchor=west]{Bivariate}

    \draw[color=black,thick] (3.0, 5.75) rectangle (4.75, 4.75);
    \Text[color=black,x=3.15, y=5.75, style={fill=white}, anchor=west]{SimDiag}
    
    \draw[color=black,thick] (0.25,0.25) rectangle (7.375,4.5);
    \Text[color=black,x=0.5, y=4.5, style={fill=white}, anchor=west]{Multivariate $\chi^2$}

    \draw[color=black,thick] (0.5, 3.0) rectangle (7.25, 4.25);
    \Text[color=black,x=1.0,y=2.9, style={fill=white}, anchor=west]{1-Factorial}

    \draw[color=black,thick] (4.0, 1.75) rectangle (7.25, 0.5);
    \Text[color=black,x=4.5,y=1.75, style={fill=white}, anchor=west]{Trivariate}

    \draw[color=black,thick] (4.0, 2.75) rectangle (7.25, 2.0);
    \Text[color=black,x=4.5,y=2.75, style={fill=white}, anchor=west]{Quadrivariate}

    \input{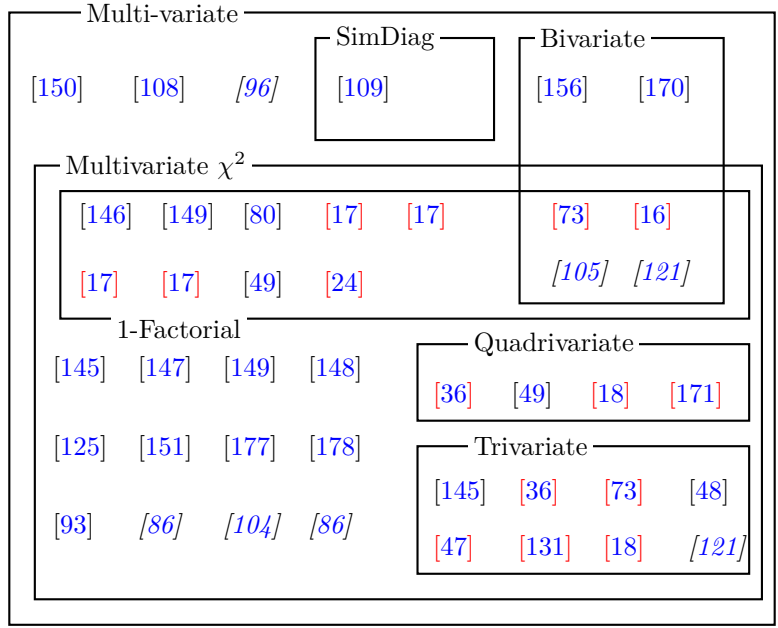}
\end{tikzpicture}
                \end{center}
                \caption{Classification according to forms}
                \label{fig:classFormMulti}
            \end{figure}
            In Figure \ref{fig:classFormMulti}, red-bracketed citations refer to complex forms, and italic citations refer to papers before 1990.

            In the following graph \ref{fig:subMulti1}, we illustrate the relations between the papers. We draw an arrow from Paper 1 to Paper 2 if the quadratic form studied in Paper 2 is a special case of the studied in Paper 1. This implies that the formula derived in Paper 1 applies to the form studied in Paper 2.
            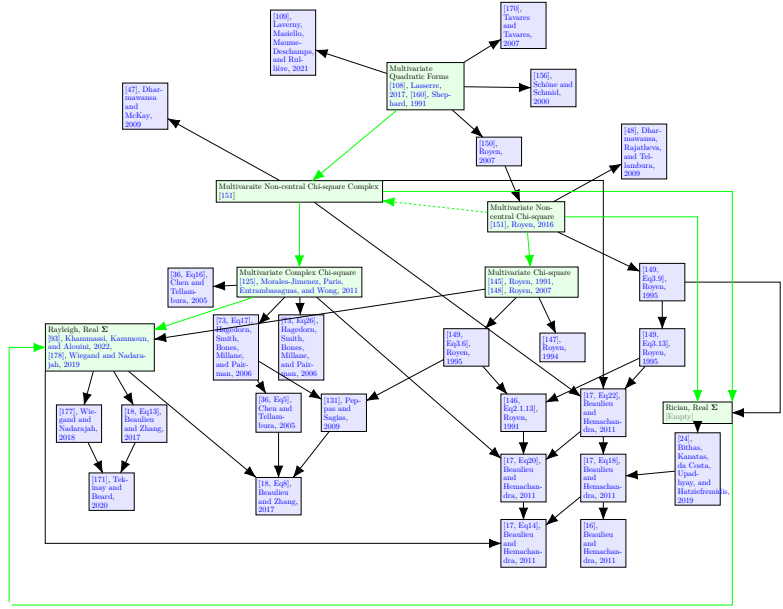
\begin{figure}[ht!]
                \begin{center}
                    \scalebox{.3}{

  \begin{tikzpicture}[scale=4,every node/.append style={draw,rectangle,fill=blue!10,text width=50},]
 
      \draw
        (0.63, 0.75) node (Tekinay){\citeF{tekinayMomentsQuadrivariateRayleigh2020}}
        (0.5, 2.34) node[style={text width = 130, fill=green!10}] (Wiegand2019){Rayleigh, Real $\mathbf{\Sigma}$ \\ \citeF{khammassiNewAnalyticalApproximation2022},\\ \citeF{wiegandSeriesApproximationsRayleigh2019}}
        (0.27, 1.5) node (Wiegand2018){\citeF{wiegandSeriesRepresentationMultidimensional2018}}
        (2.47, 0.7) node (Beaulieu2017Tri){\citeF[Eq8]{beaulieuNewSimplestExact2017}}
        (0.99, 1.5) node (Beaulieu2017Quadri){\citeF[Eq13]{beaulieuNewSimplestExact2017}}
        (5.17, 0.18) node (Beaulieu2011Rayleigh){\citeF[Eq14]{beaulieuNovelSimpleRepresentations2011}}
        (7.09, 1.62) node[style={text width = 80, fill=green!10}] (Bithas){Rician, Real $\mathbf{\Sigma}$\\ \textcolor{gray}{[Empty]}}
        (7.09, 1) node[style={text width = 50}] (BithasAct){ \citeF{bithasNovelResultsMultivariate2019}}
        (8, 1.62) node[style={text width = 0, fill=white, draw=white}] (Bithas2){}
        (-0.5, -0.5) node[style={text width = 0, fill=white, draw=white}] (Rayleigh2){}
        (6.05, 0.18) node (BeaulieuBivariate){\citeF{beaulieuNovelRepresentationsBivariate2011}}
        (6.05, 0.9) node (Beaulieu2011Rician){\citeF[Eq18]{beaulieuNovelSimpleRepresentations2011}}
        (4.09, 5.22) node[style={text width = 90, fill=green!10}] (General){Multivariate Quadratic Forms \citeF{lasserreComputingGaussianExponential2017}, \citeF{shephardCharacteristicFunctionDistribution1991}}
        (2.63, 5.7) node (Laverny){\citeF{lavernyEstimationMultivariateGeneralized2021}}
        (5.2, 3.78) node[style={text width = 90,fill=green!10}] (Royen2016){Multivariate Non-central Chi-square \citeF{royenNonCentralMultivariateChiSquare2016}}
        (5.17, 0.9) node (Beaulieu2011Nakagami){\citeF[Eq20]{beaulieuNovelSimpleRepresentations2011}}
        (6.05, 1.62) node (Beaulieu2011GenRician){\citeF[Eq22]{beaulieuNovelSimpleRepresentations2011}}
        (2.7, 4.06) node[style={text width = 202, fill=green!10}] (GeneralChiSqComplex){Multivaraite Non-central Chi-square Complex\\ \cite{royenNonCentralMultivariateChiSquare2016}}
        (2.7, 3.06) node[style={text width = 150, fill=green!10}] (Morales){Multivariate Complex Chi-square \\ \citeF{morales-jimenezDiagonalDistributionComplex2011}}
        (1, 5) node (DhamrmaComplex){\citeF{dharmawansaDiagonalDistributionComplex2009}}
        (6.5, 4.5) node (DhamrmaReal){\citeF{dharmawansaNewSeriesRepresentation2009}}
        (3.19, 1.62) node (Peppas){\citeF{peppasTrivariateNakagamimDistribution2009}}
        (4.9, 4.5) node (Royen2007Conv){\citeF{royenIntegralRepresentationsConvolutions2007}}
        (5.17, 5.89) node (Tavares){\citeF{tavaresStatisticsSumSquared2007}}
        (2.0, 2.34) node (HagedornTri){\citeF[Eq17]{hagedornTrivariateChisquaredDistribution2006}}
        (2.72, 2.34) node (HagedornBi){\citeF[Eq26]{hagedornTrivariateChisquaredDistribution2006}}
        (2.47, 1.62) node (ChenTri){\citeF[Eq5]{chenInfiniteSeriesRepresentations2005}}
        (1.5, 3) node (ChenQuadri){\citeF[Eq16]{chenInfiniteSeriesRepresentations2005}}
        (5.5, 5.2) node (Schone){\citeF{schoneJointDistributionQuadratic2000}}
        (6.7, 3.06) node (Royen1995OneFact){\citeF[Eq3.9]{royenCENTRALNONCENTRALMULTIVARIATE1995}}
        (6.7, 2.34) node (Royen1995OneFactCond){\citeF[Eq3.13]{royenCENTRALNONCENTRALMULTIVARIATE1995}}
        (4.54, 2.34) node (Royen1995TwoFact){\citeF[Eq3.6]{royenCENTRALNONCENTRALMULTIVARIATE1995}}
        (5.6, 2.34) node (Royen1994){\citeF{royenMultivariateGammadistributionsConnected1994}}
        (5.26, 3.06) node[style={text width = 110, fill=green!10}] (Royen1991Gen){Multivariate Chi-square \citeF{royenExpansionsMultivariateChisquare1991},\\ \citeF{royenIntegralRepresentationsApproximations2007}}
        (5.17, 1.62) node (Royen1991OneFact){\citeF[Eq2.1.13]{royenMultivariateGammaDistributions1991}};
      \begin{scope}[-{Latex[length=5mm]}]
      
        \draw (Wiegand2019) to (Wiegand2018);
        \draw (Wiegand2019) to (Beaulieu2017Tri);
        \draw (Wiegand2019) to (Beaulieu2017Quadri);
        \draw (Wiegand2019.south west) to (Wiegand2019.south west|-Beaulieu2011Rayleigh.west) to (Beaulieu2011Rayleigh.west);
        \draw (Wiegand2018) to (Tekinay);
        \draw (Beaulieu2017Quadri) to (Tekinay);
        \draw (Bithas) to (BithasAct);
        \draw (BithasAct) to (Beaulieu2011Rician);
        \draw (Beaulieu2011Rician) to (BeaulieuBivariate);
        \draw (Beaulieu2011Rician) to (Beaulieu2011Rayleigh);
        \draw (General) to (Laverny);
        \draw[green] (General) to (GeneralChiSqComplex);
        \draw[green] (GeneralChiSqComplex) to (Morales);
        \draw (GeneralChiSqComplex) to (DhamrmaComplex);
        \draw (General) to (Royen2007Conv);
        \draw (General) to (Tavares);
        \draw (General) to (Schone);
        \draw (Royen2016) to (DhamrmaReal);
        \draw (Royen2016) to (Royen1995OneFact);
        \draw[green] (Royen2016) to (Royen1991Gen);
        \draw[green,dashed] (Royen2016) to (GeneralChiSqComplex);
        \draw (Beaulieu2011Nakagami) to (Beaulieu2011Rayleigh);
        \draw (Beaulieu2011GenRician) to (Beaulieu2011Rician);
        \draw (Beaulieu2011GenRician) to (Beaulieu2011Nakagami);
        \draw[green] (Morales) to (Wiegand2019);
        \draw (Morales) to (Beaulieu2011Nakagami);
        \draw (Morales) to (HagedornTri);
        \draw (Morales) to (HagedornBi);
        \draw (Morales) to (ChenQuadri);
        \draw (Peppas) to (Beaulieu2017Tri);
        \draw (Royen2007Conv) to (Royen2016);
        \draw (HagedornTri) to (Peppas);
        \draw (HagedornTri) to (ChenTri);
        \draw (ChenTri) to (Beaulieu2017Tri);
        \draw (Royen1995OneFact) to (Royen1995OneFact-|Bithas2)  to (Bithas2|-Bithas) to (Bithas);
        \draw (Royen1995OneFact) to (Royen1995OneFactCond);
        \draw (Royen1995OneFactCond) to (Beaulieu2011GenRician);
        \draw (Royen1995OneFactCond) to (Royen1991OneFact);
        \draw (Royen1995TwoFact) to (Peppas);
        \draw (Royen1995TwoFact) to (Royen1991OneFact);
        \draw (Royen1991Gen) to (Wiegand2019);
        \draw (Royen1991Gen) to (Royen1995TwoFact);
        \draw (Royen1991Gen) to (Royen1994);
        \draw (Royen1991OneFact) to (Beaulieu2011Nakagami);
        \draw[green] (GeneralChiSqComplex) to (GeneralChiSqComplex-|Bithas.north east) to (Bithas.north east) ;
        \draw (GeneralChiSqComplex) to (Beaulieu2011GenRician);
        \draw (GeneralChiSqComplex.north east) to (GeneralChiSqComplex.north east-|Beaulieu2011GenRician.north) to (Beaulieu2011GenRician.north);
        
        \draw[green] (Bithas.south east) to (Bithas.south east|-Rayleigh2) to (Rayleigh2) to (Rayleigh2|-Wiegand2019) to (Wiegand2019);
        \draw[green] (Royen2016) to (Royen2016-|Bithas.north) to (Bithas.north);
      \end{scope}
    \end{tikzpicture}}
                \end{center}
                \caption{Subsumption Graph}
                \label{fig:subMulti1}
            \end{figure}
            
            Next, we classify formulae according to their types in Table \ref{tab:FourmulaTypeMulti}.
            \begin{table}[ht!]
                \centering
                \newcommand{\Smhndcite}[2][]{\citeauthor*{#2} \textcolor{blue}{\cite[#1]{#2}},}
\centering
\begin{tabular}{|p{8cm}|}
    \hline
    Exact Formulae  \\
    \hline
    \textbf{Finite Expressions} \newline
    \Smhndcite{al-naffouriDistributionIndefiniteQuadratic2016}
    
    \rule{8cm}{1pt} \newline
    \textbf{Infinite Series} \newline
    \Smhndcite{lavernyEstimationMultivariateGeneralized2021}
    \Smhndcite{tekinayMomentsQuadrivariateRayleigh2020}
    \Smhndcite{bithasNovelResultsMultivariate2019}
    \Smhndcite{wiegandSeriesApproximationsRayleigh2019}
    \Smhndcite{wiegandSeriesRepresentationMultidimensional2018}    
    \Smhndcite{morales-jimenezDiagonalDistributionComplex2011}
    \Smhndcite{dharmawansaDiagonalDistributionComplex2009}
    \Smhndcite{dharmawansaNewSeriesRepresentation2009}
    \Smhndcite{peppasTrivariateNakagamimDistribution2009}
    \Smhndcite{hagedornTrivariateChisquaredDistribution2006}
    \Smhndcite{chenInfiniteSeriesRepresentations2005}
    \Smhndcite{schoneJointDistributionQuadratic2000}    
    \Smhndcite{royenCENTRALNONCENTRALMULTIVARIATE1995}
    \Smhndcite{royenMultivariateGammadistributionsConnected1994}
    \Smhndcite{royenExpansionsMultivariateChisquare1991}
    \Smhndcite{royenMultivariateGammaDistributions1991}

    \rule{8cm}{1pt} \newline
    \textbf{Numerical Integartion} \newline
    \Smhndcite{khammassiNewAnalyticalApproximation2022}
    \Smhndcite{beaulieuNewSimplestExact2017}
    \Smhndcite{royenNonCentralMultivariateChiSquare2016}    
    \Smhndcite{beaulieuNovelRepresentationsBivariate2011}
    \Smhndcite{beaulieuNovelRepresentationsBivariate2011}
    \Smhndcite{royenIntegralRepresentationsApproximations2007}    
    \Smhndcite{tavaresStatisticsSumSquared2007}
    \Smhndcite{soongUsingComplexIntegration1997}
    \Smhndcite{royenCENTRALNONCENTRALMULTIVARIATE1995}
    \Smhndcite{royenMultivariateGammadistributionsConnected1994}
    \Smhndcite{royenMultivariateGammaDistributions1991}
    
    \rule{8cm}{1pt} \newline
    \textbf{Sequences of RVs} \newline
    \Smhndcite{khammassiNewAnalyticalApproximation2022} 
    \\
    \hline
\end{tabular}
                \caption{MultiQF Formula Type}
                \label{tab:FourmulaTypeMulti}
            \end{table}
            As we can see, all the formulae are exact. As for Khammassi's \cite{khammassiNewAnalyticalApproximation2022} formula, although the formula \eqref{eq:Khammassi-Formula} is approximate since it requires $\epsilon-\text{rank}\neq N$, the approximate random vector \eqref{eq:Khammassi-Approx} approaches the original one as $\epsilon$ tends to zero. Actually, when $\epsilon$ becomes smaller than the minimum eigenvalue of $\dmt{\Sigma}$, the approximate vector becomes identically identically distributed to the original random vector.

            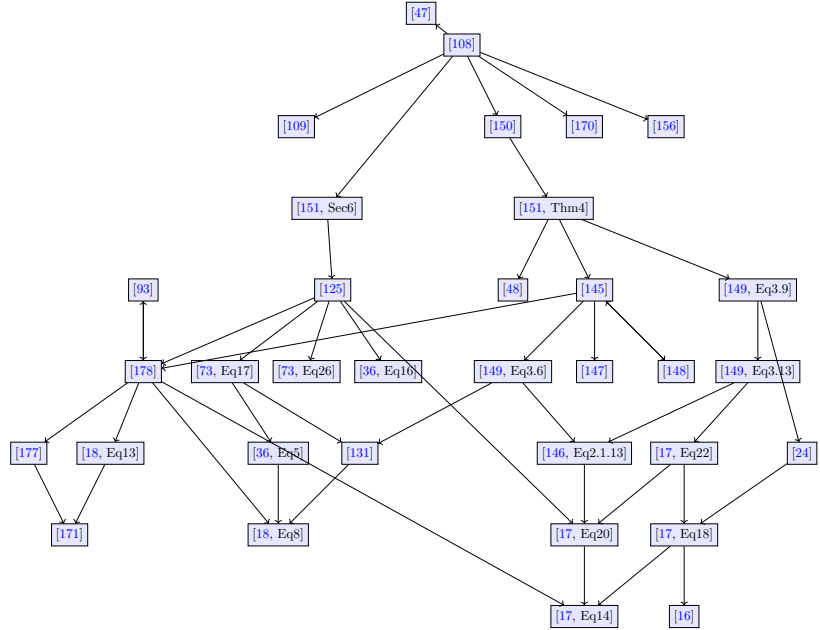
\begin{figure}[ht!]
                \begin{center}
                    \scalebox{.5}{\begin{tikzpicture}[scale=3,every node/.append style={draw,rectangle,fill=blue!10},]
      \draw
        (1.28, 3.06) node (0){\cite[]{khammassiNewAnalyticalApproximation2022}}
        (0.63, 0.9) node (2){\cite[]{tekinayMomentsQuadrivariateRayleigh2020}}
        (1.28, 2.34) node (4){\cite[]{wiegandSeriesApproximationsRayleigh2019}}
        (0.27, 1.62) node (5){\cite[]{wiegandSeriesRepresentationMultidimensional2018}}
        (2.47, 0.9) node (6){\cite[Eq8]{beaulieuNewSimplestExact2017}}
        (0.99, 1.62) node (7){\cite[Eq13]{beaulieuNewSimplestExact2017}}
        (5.17, 0.18) node (11){\cite[Eq14]{beaulieuNovelSimpleRepresentations2011}}
        (7.09, 1.62) node (3){\cite[]{bithasNovelResultsMultivariate2019}}
        (6.05, 0.18) node (10){\cite[]{beaulieuNovelRepresentationsBivariate2011}}
        (6.05, 0.9) node (12){\cite[Eq18]{beaulieuNovelSimpleRepresentations2011}}
        (4.09, 5.22) node (8){\cite[]{lasserreComputingGaussianExponential2017}}
        (2.63, 4.5) node (1){\cite[]{lavernyEstimationMultivariateGeneralized2021}}
        (4.9, 3.78) node (9){\cite[Thm4]{royenNonCentralMultivariateChiSquare2016}}
        (2.9, 3.78) node (RoyenComplex){\cite[Sec6]{royenNonCentralMultivariateChiSquare2016}}
        (5.17, 0.9) node (13){\cite[Eq20]{beaulieuNovelSimpleRepresentations2011}}
        (6.05, 1.62) node (14){\cite[Eq22]{beaulieuNovelSimpleRepresentations2011}}
        (2.95, 3.06) node (15){\cite[]{morales-jimenezDiagonalDistributionComplex2011}}
        (3.73, 5.5) node (16){\cite[]{dharmawansaDiagonalDistributionComplex2009}}
        (4.54, 3.06) node (17){\cite[]{dharmawansaNewSeriesRepresentation2009}}
        (3.19, 1.62) node (18){\cite[]{peppasTrivariateNakagamimDistribution2009}}
        (5.98, 2.34) node (19){\cite[]{royenIntegralRepresentationsApproximations2007}}
        (4.45, 4.5) node (20){\cite[]{royenIntegralRepresentationsConvolutions2007}}
        (5.17, 4.5) node (21){\cite[]{tavaresStatisticsSumSquared2007}}
        (2.0, 2.34) node (22){\cite[Eq17]{hagedornTrivariateChisquaredDistribution2006}}
        (2.72, 2.34) node (23){\cite[Eq26]{hagedornTrivariateChisquaredDistribution2006}}
        (2.47, 1.62) node (24){\cite[Eq5]{chenInfiniteSeriesRepresentations2005}}
        (3.44, 2.34) node (25){\cite[Eq16]{chenInfiniteSeriesRepresentations2005}}
        (5.89, 4.5) node (26){\cite[]{schoneJointDistributionQuadratic2000}}
        (6.7, 3.06) node (27){\cite[Eq3.9]{royenCENTRALNONCENTRALMULTIVARIATE1995}}
        (6.7, 2.34) node (28){\cite[Eq3.13]{royenCENTRALNONCENTRALMULTIVARIATE1995}}
        (4.54, 2.34) node (29){\cite[Eq3.6]{royenCENTRALNONCENTRALMULTIVARIATE1995}}
        (5.26, 2.34) node (30){\cite[]{royenMultivariateGammadistributionsConnected1994}}
        (5.26, 3.06) node (31){\cite[]{royenExpansionsMultivariateChisquare1991}}
        (5.17, 1.62) node (32){\cite[Eq2.1.13]{royenMultivariateGammaDistributions1991}};
      \begin{scope}[->]
        \draw (0) to (4);
        \draw (4) to (0);
        \draw (4) to (5);
        \draw (4) to (6);
        \draw (4) to (7);
        \draw (4) to (11);
        \draw (5) to (2);
        \draw (7) to (2);
        \draw (3) to (12);
        \draw (12) to (10);
        \draw (12) to (11);
        \draw (8) to (1);
        \draw (8) to (RoyenComplex);
        \draw (RoyenComplex) to (15);
        \draw (8) to (16);
        \draw (8) to (20);
        \draw (8) to (21);
        \draw (8) to (26);
        \draw (9) to (17);
        \draw (9) to (27);
        \draw (9) to (31);
        \draw (13) to (11);
        \draw (14) to (12);
        \draw (14) to (13);
        \draw (15) to (4);
        \draw (15) to (13);
        \draw (15) to (22);
        \draw (15) to (23);
        \draw (15) to (25);
        \draw (18) to (6);
        \draw (19) to (31);
        \draw (20) to (9);
        \draw (22) to (18);
        \draw (22) to (24);
        \draw (24) to (6);
        \draw (27) to (3);
        \draw (27) to (28);
        \draw (28) to (14);
        \draw (28) to (32);
        \draw (29) to (18);
        \draw (29) to (32);
        \draw (31) to (4);
        \draw (31) to (19);
        \draw (31) to (29);
        \draw (31) to (30);
        \draw (32) to (13);
      \end{scope}
    \end{tikzpicture}}
                \end{center}
                \caption{Subsumption Graph}
                \label{fig:subMulti2}
            \end{figure}

\chapter{Numerical Concerns and Open Problems}
    
    In this section, we discuss a select collection of numerical concerns that one may face while applying the methods discussed in Sections 1, 3, and 4. In addition to that, we end our monograph with a collection of pertinent open problems.
    
    \section{Numerical Concerns}

        \subsection{Curse of Dimensionality}

        In Section 1, we demonstrated that a non-trivial quadratic form can be represented as a linear combination of chi-squares (possibly with an added Gaussian). Although the result hold theoretically for any dimension $N$, the linear coefficients are calculated via (an equivalent of) matrix diagonalization. The computational complexity is $\mathcal{O}(N^3)$. In genetics \cite{chenNumericalEvaluationMethods2019}, the number of Gaussian variables $N$ is in the range of thousands. 

            \subsubsection{Randomized Projection Methods}
             Randomized projection methods are computational techniques used to reduce the dimensionality of high-dimensional data. Low-rank matrix approximations, e.g., the truncated singular value decomposition and the rank-revealing QR decomposition, play a crucial role in data analysis and computation. Randomized algorithms for constructing low-rank matrix approximations use random sampling to identify a subspace that captures most of the action of a matrix. The input matrix is compressed into this subspace, explicitly or implicitly, and further manipulated deterministically to achieve the desired low-rank factorization. 

             The collection of standard matrix decompositions includes the pivoted QR factorization, the eigenvalue decomposition, and the singular value decomposition (SVD), each of which reveals the (numerical) range of a matrix. Truncated forms of these decompositions are frequently employed to represent a low-rank approximation of a matrix: 
             \begin{equation}
                 \label{eqn:Low-RankApprox}
                 \underset{m \times n}{\dmt{A}} \quad \approx \quad \underset{m \times k}{\dmt{B}} \qquad \underset{k \times n}{ \dmt{C}}
             \end{equation}
             The inner dimension is called the numerical rank of the matrix. If the numerical rank is much smaller than $m$ or $n$, a factorization of the form \eqref{eqn:Low-RankApprox} allows inexpensive storage of the matrix and faster multiplication with vectors or matrices.

             Matrix approximation is divided into two computational stages:
                \begin{itemize}
                    \item \textbf{Stage A:} Construct a low-dimensional subspace that captures the matrix's significant features, represented by an orthonormal basis matrix \(\dmt{G}\), and
                     \begin{equation}
                         \dmt{A} \approx \dmt{G}\dmt{G}^H\dmt{A}
                     \end{equation}
                    \item \textbf{Stage B:} Use the basis \( \dmt{G} \) to perform standard matrix factorizations (e.g., SVD, QR) of $\dmt{A}$. Note that $\dmt{B} = \dmt{G}$ and $\dmt{C} = \dmt{G}^H\dmt{A}$.  
                \end{itemize}
                Randomized algorithms efficiently address Stage A by leveraging random sampling to approximate the matrix's range. This significantly reduces computational complexity while maintaining accuracy. Stage B employs established deterministic methods to refine the approximations further.

                Approximating the range of a matrix via randomness can be formulated in the following two problems:  
                \begin{itemize}
                    \item \textbf{Fixed-Precision Approximation:} Construct a basis matrix \( \dmt{G} \) that captures most of the matrix's action within a specified error tolerance. We seek a matrix $\dmt{G}$ with $k$ orthonormal columns such that
                    \begin{equation}
                        \| \dmt{A} - \dmt{G}\dmt{G}^H\dmt{A} \| \le \epsilon.
                    \end{equation}
                    \item \textbf{Fixed-Rank Approximation:} Generate a basis with a pre-defined number of columns to minimize approximation error. Given a matrix $\dmt{A}$, a target rank $k$, and an oversampling parameter $p$ (adding a few extra columns), we seek to construct a matrix $\dmt{G}$ with $k + p$ orthonormal columns such that 
                    \begin{equation}
                       \| \dmt{A} - \dmt{G}\dmt{G}^H\dmt{A} \| \approx \min_{\text{rank}(\dmt{Z}) \le k } \| \dmt{A} - \dmt{Z} \|.
                    \end{equation}
                \end{itemize}  
                The process involves generating random test matrices, projecting them onto the matrix's range, and orthonormalizing the results to form the desired basis. Oversampling ensures robustness and accuracy, even in challenging cases. These methods are computationally efficient and highly stable compared to traditional approaches. 

        \subsection{Machine Precision Limitations}

            Theoretical equality does not guarantee computability. Since real numbers are rounded to a certain precision in a machine, it is possible that certain operations will not be represented correctly by the machine. A straightforward example is the successive division by 2 of a real number. At some point, the quotient will be smaller than the machine precision. 

            Consider the chi-square density expansion of the PDF of the QF as in Section 3 \begin{align*}
                f_{\rs{Q}}(q)&=u(q)\sum_{k=0}^\infty c_k\dfrac{q^{N/2+k-1}e^{-q/2\beta}}{(2\beta)^{N/2+k}\Gamma(N/2+k)}=\sum_{k=0}^\infty c_k f_{\beta \chi_{N+2k}^2}(q)
            \end{align*}
            As the $k^{\text{th}}$ term depends on the fraction $\dfrac{q^{\frac{N}{2}+k-1}}{\Gamma(\frac{N}{2}+k)}$, higher order terms cannot be stored in standard machine bits. Hence, in MATLAB, the numerator and the denominator are flagged infinite, an the term becomes a NaN. 

            As a numerical example, take the quadratic form \begin{equation}
                \rs{Q}=\rv{x}^T\dmt{A}\rv{x},\quad x\sim\mathcal{N}_2(\dv{\mu},\dmt{\Sigma}) \label{eq:QF-chi-series-example}
            \end{equation} where $$\dmt{A}=\begin{bmatrix}
                2 & 1\\
                1 & 2\\
            \end{bmatrix}, \; \dmt{\Sigma}=\begin{bmatrix}
                1 & \frac{1}{2}\\
                \frac{1}{2} & 1
            \end{bmatrix}, \; \dmt{\mu}=\begin{bmatrix}
                5 \\
                5
            \end{bmatrix}$$

            \begin{figure}[ht!]
                \centering
                \includegraphics[width=0.5\linewidth]{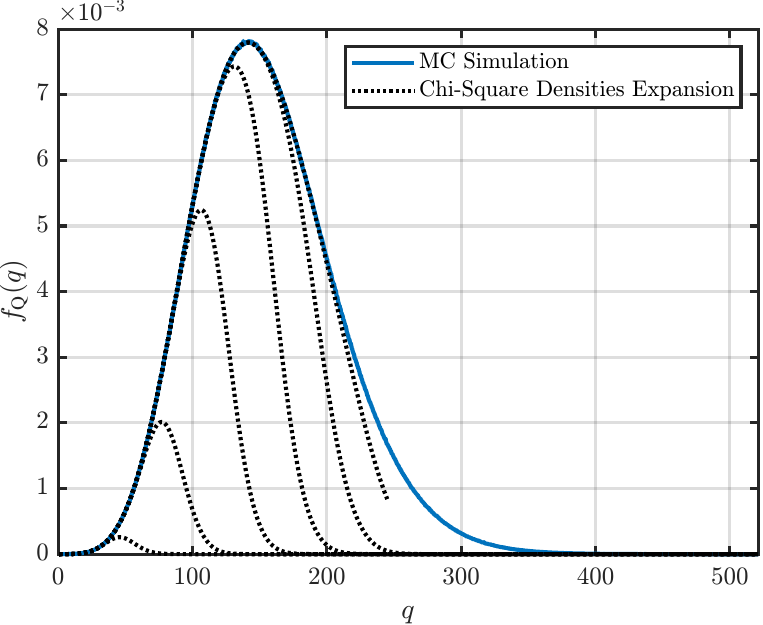}
                \caption{Partial sums of chi-square density expansion of the PDF of \eqref{eq:QF-chi-series-example} using MATLAB standard precision}
                \label{fig:chi-partial-sums}
            \end{figure}
            In Figure \ref{fig:chi-partial-sums}, partial sums of the chi-square density expansion of the probability density function are plotted with the Monte-Carlo simulation. For the chi-square density expansion, MATLAB is used with standard machine precision. The Monte-Carlo simulation was generated using $10^8$ instances of $\rv{x}$. Note that, in Figure \ref{fig:chi-partial-sums}, the truncation index ranges from $30$ to $170$. We can clearly see that, for low truncation indices and high PDF arguments, the truncated chi-square density is significantly inaccurate. Incorporating more terms adds NaNs to the sum, so we can see that the support is shrinking.

            A possible solution to that is to work with non-standard machine precision. Another solution is to resort to logarithmic calculations. In particular, we can rewrite the problematic terms as
            \begin{equation*}
            \dfrac{q^{\frac{N}{2}+k-1}}{\Gamma\left(\frac{N}{2}+k\right)}=\exp\left[\left(\dfrac{N}{2}+k-1\right)\ln(q)-\ln\left(\Gamma\left(\frac{N}{2}+k-1\right)\right)\right]
            \end{equation*}
            The natural logarithm of the gamma function is a numerically built-in function in MATLAB. Using the latter function, we overcome the problem, and reproduce the same graph \ref{fig:chi-partial-sums} showing convergence, as in Figure \ref{fig:chi-partial-sums-log}. Note that the truncation index ranges again from $30$ to $170$.
            \begin{figure}[ht!]
                \centering
                \includegraphics[width=0.5\linewidth]{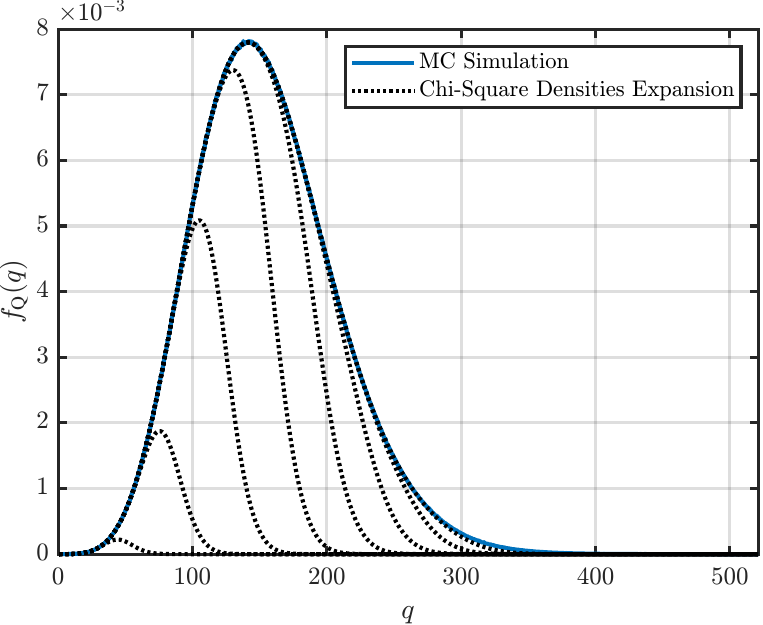}
                \caption{Partial sums of chi-square density expansion of the PDF of \eqref{eq:QF-chi-series-example} using MATLAB standard precision with logarithmic correction}
                \label{fig:chi-partial-sums-log}
            \end{figure}

        \subsection{Evaluation of Special Functions}

            Not all special functions are built-in in common programming languages. Moreover, some are available only as symbolic functions. While symbolic functions are usually better from an accuracy perspective, they slow down the implementation, possibly to the point of infeasibility. 

            Consider, for example, the series expansion of the PDF of an indefinite form (refer to the previous equation), developed by Provost and Rudiuk \cite{provostExactDistributionIndefinite1996}. The Whitaker-W function is available only as a symbolic scheme in MATLAB (refer to the table).

    \section{Open Problems}




        \begin{openProblem} Let $\rv{Q}$ be multivariate quadratic form in $\rv{x} \sim \mathcal{N}(\mu, \Sigma)$, $x \in \reals{N}$ with $\rs{Q}_m = \rv{x}^T \dmt{A}_m \rv{x} + \dv{b}_m^T \rv{x} + c_m$. Define the set $\Omega = \cap_{m=1}^M \{Q_m(x) < q\}$. We seek the computational complexity of $\mathbb{P}{\Omega}$ in terms of the number of quadratic forms $M$ and the numbers of guassians $N$, for the general multivariate form and for some classes of quadratic forms: central forms, positive definite forms.

        \begin{remark}
        Even in the single case, the many complicated infinite series, suggest that the problem is inherently complex. Which class of complexity does this problem fall into is the question we are asking. 
        \end{remark}
            
        \end{openProblem}
        
        \begin{openProblem} [Conditions for the equivalence between Real and complex Multi Chi-squares] Consider the real and complex multi-chi-square MGFs (ref) and (ref). What conditions can be imposed on $R_R,K_r, R_C, K_C$ such that both forms are the same.
        \end{openProblem}
        \begin{remark}
            In the central case, the question is a linear algebaric question about when the polynomials from the determinants are equal $|I-2SR_R|^{K_R/2} = |I-SR_C|^{K_C}$. For the low dimensional cases $M = 2,3,4$ one can explicitly find the conditions and list them out as we have done in (ref). We seek ``simple'' expression(s) in the general case.
        \end{remark}

        \begin{openProblem} [Multivariate saddlepoint for Quadratic forms]
            \label{openProb:multivariateSaddlepoint}
            Let $\rv{Q} = \begin{bmatrix}\rs{Q}_1 & \rs{Q}_2 & \ldots \rs{Q}_M\end{bmatrix} $ be multivariate quadratic form in $\rv{x} \sim \mathcal{N}(\mu, \Sigma)$ with $\rs{Q}_m = \rv{x}^T \dmt{A}_m \rv{x} + \dv{b}_m^T \rv{x} + c_m$. Let $M_{\rv{Q}(\dv{s})}$ be the moment generating function (MGF) of $\rv{Q}$ and $K(s) = \log M_Q(s)$ is cumulant generating function (CGF). Let $s^*$ be a solution to $\nabla_sK_Q(s) = \dv{q}$. We seek a function $F_{\text{spa}}(s^*, q)$ a function of the the saddlepoint $s^*$ and the CDF evaluation point $q$ such that $F_Q(\dv{q})$ and $F_{\text{spa}}(s^*, \dv{q})$ have the same tail decay rate in all directions as $\dv{q} \to \infty$.
        \end{openProblem}
        
        \begin{remark}
            Note in \Cref{openProb:multivariateSaddlepoint} we avoided requiring an absolute error gurantee such as $|F_Q(q) - F_{\text{spa}}| \leq \epsilon$, since even in the univariate case the saddlepoint methods \emph{do not} satisfy such a requirement. In the literature as we have mentioned there are several attempts to produce multivariate saddlepoint methods. In particular, we refer the reader to \cite{2007-Butler-saddlepointAapproximationsWithApplicationst} for a textbook introduction to saddlepoint methods. We refer also to \cite{kolassaMultivariateSaddlepointTailProbabilityApproximations2003,jixinLiKolassa2010MultivariateSaddlepoint}, for derivations of multivariate saddlepoint methods. Finally a more detailed description of multivariate saddlepoint methods can be found \cite{kolassaSeriesapproximationmethodsinstatistics2006}
        \end{remark}
        \begin{remark}
            The saddlepoint PDF is available as a straightforward extension of the univariate case using the Hessian of the CGF: $$f_{Q}(q) \approx [2\pi \det{\nabla^2_sK(s^*)}]^{-\frac{M}{2}}\exp{[K(s^*) - s^{*T}q]}$$
            However, integrating this approximate PDF does not well approximate the CDF. 
        \end{remark}
        
        
        \begin{openProblem} \label{op:gamma-complex}
            Let $x \in \reals{}, y \in \complexs{}$ Consider the gamma function, $\mathcal{G}_\alpha(x,y) = \sum_{n=0}^\infty G_{\alpha + n}(x) y^n $ with $G_{a}(x) = \gamma(\alpha, x)/\Gamma(\alpha)$ being the normalized  lower incomplete gamma function. We seek an efficient evaluation method of $G_\alpha(\cdot, \cdot)$ when the second argument is not pure real.
        \end{openProblem}
        \begin{remark}
            This function, or closely related ones, appear in many of Thomas Royen's multivariate chi-squares papers. In \cite[Eq 2.14]{royenNonCentralMultivariateChiSquare2016} it is expressed as: $G_\alpha(x,y) = [G_\alpha(x) - y^{1-\alpha}e^{x(y-1)}G_\alpha(xy) ] / (1 - y)$. Thus, the problem reduces to evaluating the incomplete gamma function with complex argument. Unlike its real counterpart, to the best of our knowledge, there are no efficient widely known algorithms to evaluate it. For example MATLAB's implementation \href{https://www.mathworks.com/help/symbolic/sym.igamma.html}{igamma} is a symbolic. No software is listed in \cite[\href{https://dlmf.nist.gov/5.4.E1}{DLMF (8.28)}]{NIST:DLMF}, for the complex argument incomplete gamma function. Efficient implementation of this function, would be a starting block for ``efficient'' methods for computing multivariate chi-squares.
        \end{remark}
        
        
        \begin{openProblem}[{Lassere's method with Hermite polynomials}] Let $\gamma$ be the standard gaussian measure on $\reals{n}$. Let $\Omega \subset \reals{n}$ be a semialgebaraic set. Let $\mu, \nu$ be measures with $\mu$ restricted to $\Omega$. We seek to approximate the solution of the program: $\sup_{\mu,\nu} \mu(\Omega)\,\, \text{s.t.}\,\, \gamma = \mu + \nu$ by forming programs that moment match up to the $k^{th}$ order Hermite moments defined as: $h_{\alpha}(\mu) = \int_\Omega \text{He}(x) d\mu$, with $\alpha \in \mathbb{N}^{n}, \sum_i \alpha_i \leq k$.
        \end{openProblem}
        \begin{remark}
            In \cite{lasserreComputingGaussianExponential2017}, it is conjectured that using Hermite polynomials would have significantly better convergence and numerical properties compared to the standard monomial moments. This needs further study and validation. Note the univariate saddlepoint method in a sense accounts for the first Hermite moment and obtains extremely accurate results. To see this derivation of the saddlepoint approximation in terms of hermite polynomials see \cite{1954-Daniels-SaddlepointApproximationsInStatistics},\cite[Ch. 2]{2001-Kuonen-ComputerintensiveStatisticalMethodsSaddlepointApproximations}.
        \end{remark}

\section*{Acknowledgements}
The authors are grateful to Ulrike Fischer, who designed the original style files, and Neal Parikh, who laid the groundwork for those style files.

\appendix
\chapter{Notation}
In this monograph, we adopt the following notation: for random quantities, we use non-italic sans-serif font while for deterministic quantities, we use italic normal font. For vectors and matrices, we use bold fonts for random and deterministic quantities and non-bold fonts for scalar quantities. The following list summarizes the common symbols we use.

\begin{table}[ht!]
\centering
\begin{tabularx}{0.95\textwidth}{@{}cl@{}}
   $\mathbb{R}$ & set of real numbers\\
   $\mathbb{C}$ & set of complex numbers\\
   $\Re(.)$ & real part of a complex number\\
   $\Im(.)$ & imaginary part of a complex number\\
   $\rv{X}_j$ & column $j$ of the random matrix $\rmt{x}$. \\
   $\rs{x}_{ij}$ & element $i,j$ of the random matrix $\rmt{X}$. \\
   $\dmt{I}_N$ & identity matrix of dimension $N$\\
   $|a|$ & absolute value of a complex number $a$\\
   $\|.\|$ & $\ell^2$-norm of a vector or its induced matrix norm\\
   $|\dmt{A}|$ & determinant of a matrix $\dmt{A}$\\
   $.^T$ & transpose of a matrix\\
   $.^H$ & conjugate transpose of a matrix\\
   $\diag(a_1,\ldots,a_n)$ & diagonal matrix with diagonal entries $a_1,\ldots,a_n$\\
   $\diag(\dmt{A})$ & diagonal of the matrix $\dmt{A}$\\
   $\mathcal{N}(\mu,\sigma^2)$ & Gaussian random variable with mean $\mu$ and variance $\sigma^2$\\
   $\mathcal{N}_N(\dv{\mu},\dmt{\Sigma})$ & Gaussian $N$-random vector with mean \\
    & $\dv{\mu}$ and covariance matrix $\dmt{\Sigma}$\\
   $\chi^2_\nu$ & chi-square random variable with $\nu$ degrees of freedom\\
   $\chi^2_\nu(\delta^2)$ & chi-square random variable with $\nu$ degrees of freedom\\
    & and non-centrality parameter $\delta^2$\\
   $\mathbb{P}(.)$ & probability of an event\\
   $\mathbb{E}(.)$ & expectation of a random object\\
   $\binom{m}{n}$ & binomial coefficient
\end{tabularx}
\end{table}

\chapter{Select Proofs}
            \section{Series Expansions for Positive Definite Forms}\label{append:series-expansions}
                This appendix rephrases Kotz's \cite{kotzSeriesRepresentationsDistributions1967Central,kotzSeriesRepresentationsDistributions1967NonCentral} unified proofs of series expansion of positive definite quadratic forms in Gaussian random variables.

                We start from the moment generating function, which coincides, up to the sign of the argument, with the (one-sided) Laplace transform of the probability density function. Note that this coincidence is particularly due to the definiteness of the form.

                We seek a series expansion of the PDF of the form: \begin{equation}
                    f_{\rs{Q}}(q) = \sum_{k\geq 0} c_k h_k(q)
                \end{equation}
                where $\{h_k(q)\}_{k\geq 0}$ is a sequence of well-known functions. Assume that these sequences satisfy $$\sum_{k\geq 0} |c_k| |h_k(q)|\leq Ae^{bq}$$ for some real numbers $A$ and $b$. Define $$g(q) = \sum_{k\geq 0}c_Kh_k(q)$$ Then, we can easily show, using Lebesgue dominated convergence theorem, that the Laplace transforms $\hat{h}$ and $\hat{g}$ exist, and that we can interchange the Laplace transform and the infinite summation: $$\hat{g}(s)=\sum_{k\geq 0}c_k\hat{h}_k$$Assume further that Laplace transforms can be written as $$\hat{h}_k(s)=\xi(s)\eta^k(s)$$
                where $\xi(s)$ is a non-vanishing function, analytic for $\Re(s)>b$, and $\eta(s)$ is also analytic for $\Re(s)>b$, and invertible with $\zeta(\eta(s))=s$. 
                Then \begin{align*}
                    (\hat{g}\circ \zeta)/(\xi\circ \zeta)(\theta)&=\dfrac{1}{\xi(\zeta(\theta))}\sum_{k=0}^\infty c_k\hat{h}_k(\zeta(\theta))\\
                    &=\dfrac{1}{\xi(\zeta(\theta))}\sum_{k=0}^\infty c_k\xi(\zeta(\theta))\eta^k(\zeta(\theta))\\
                    &=\sum_{k=0}^\infty c_k[\eta(\zeta(\theta))]^k\\
                    &=\sum_{k=0}^\infty c_k \theta^k &\text{for } \Re(\zeta(\theta))>b
                \end{align*}
                Let $M(\theta)=(\hat{f}_{\rs{Q}}\circ \zeta)/(\xi\circ \zeta)(\theta)$. Choose the sequences $c_k$ and $h_k$ so that $M(\theta)=\sum_{k=0}^\infty c_k\theta^k$. Prove that the corresponding series $g(q)$ is dominated by an exponential function. Hence $$(\hat{f}_{\rs{Q}}\circ \zeta)/(\xi\circ \zeta)(\theta)=(\hat{g}\circ \zeta)/(\xi\circ \zeta)(\theta)$$ or $$(\hat{f}_{\rs{Q}}\circ \zeta)(\theta)=(\hat{g}\circ \zeta)(\theta)$$ Apply the equality at $\theta=\eta(s)$. Thus $$\hat{f}_{\rs{Q}}(s)=\hat{g}(s) \text{ for } \Re(s)\geq b$$ By the uniqueness of Laplace transform, we can deduce that $$f_{\rs{Q}}(q)=g(q)=\sum_{k=0}^\infty c_kh_k(q)$$




\bibliographystyle{plainnat}
\bibliography{bibfiles/BibFinal2_arxiv_bibtex}

\end{document}